\documentclass[cleveref, english,nolineno]{socg-lipics-v2021}

\newif\iflong
\newif\ifshort
\longtrue
\iflong
\else
\shorttrue
\fi

\ifshort
\EventEditors{Xavier Goaoc and Michael Kerber}
\EventNoEds{2}
\EventLongTitle{38th International Symposium on Computational Geometry (SoCG 2022)}
\EventShortTitle{SoCG 2022}
\EventAcronym{SoCG}
\EventYear{2022}
\EventDate{June 7--10, 2022}
\EventLocation{Berlin, Germany}
\EventLogo{}
\SeriesVolume{224}
\ArticleNo{XX}  
\fi

\usepackage{amsmath}
\usepackage{amsfonts}
\usepackage{amssymb}
\usepackage{amsthm}
\usepackage{graphicx}
\usepackage{todonotes}
\usepackage{enumitem,boxedminipage}
\usepackage{bm}

\title{Parameterized Algorithms for Upward Planarity}

\iflong
\relatedversion{An extended abstract of this manuscript was accepted for publication in the proceedings of the 38th International Symposium on Computational Geometry (SoCG 2022).}
\fi

\author{Steven Chaplick}{Maastricht University, Maastricht, The Netherlands}{s.chaplick@maastrichtuniversity.nl}{
https://orcid.org/
0000-0003-3501-4608}{}

\author{Emilio Di Giacomo}{Università degli Studi di Perugia, Perugia, Italy}{emilio.digiacomo@unipg.it}{https://orcid.org/0000-0002-9794-1928}{MIUR, grant 20174LF3T8, Dip. Ing. - UNIPG, grants RICBA19FM and RICBA20EDG}

\author{Fabrizio Frati}{Roma Tre University, Rome, Italy}{frati@dia.uniroma3.it}{https://orcid.org/0000-0001-5987-8713}{MIUR, grant 20174LF3T8}

\author{Robert Ganian}{Technische Universität Wien, Wien, Austria}{rganian@ac.tuwien.ac.at}{https://orcid.org/
0000-0002-7762-8045}{Austrian Science Fund (FWF) Project Y1329}

\author{Chrysanthi N. Raftopoulou}{National Technical University of Athens, Athens, Greece}{crisraft@mail.ntua.gr}{https://orcid.org/0000-0001-6457-516X}{NTUA research program $\mathrm{\Pi}$EBE 2020}

\author{Kirill Simonov}{Technische Universität Wien, Wien, Austria}{ksimonov@ac.tuwien.ac.at}{}{Austrian Science Fund (FWF) Project P31336}

\authorrunning{S. Chaplick, E. Di Giacomo, F. Frati, R. Ganian, C. N. Raftopoulou, K. Simonov}

\acknowledgements{The authors thank Fabrizio Montecchiani and Giuseppe Liotta for fruitful discussions on the topic of upward planarity. This research was initiated at Dagstuhl Seminar 21293: Parameterized Complexity in Graph Drawing~\cite{Dagstuhlseminar21}.}

\Copyright{S. Chaplick, E. Di Giacomo, F. Frati, R. Ganian, C. N. Raftopoulou, K. Simonov}




\newcommand{\td}{\textit{td}}
\newcommand{\tw}{\textit{tw}}
\newcommand{\clos}{\ensuremath{\text{clos}}}

\newcommand{\cG}{\ensuremath{\mathcal{G}}}

\newcommand{\signG}{\ensuremath{\widetilde{G}}}
\newcommand{\signH}{\ensuremath{\widetilde{H}}}
\newcommand{\bigoh}{\ensuremath{\mathcal{O}}}
\newcommand{\fut}{\emph{fut}}
\newcommand{\past}{\emph{past}}

\newcommand{\W}{\ensuremath{\mathrm{W}}}
\newcommand{\FPT}{\ensuremath{\mathrm{FPT}}}
\newcommand{\XP}{\ensuremath{\mathrm{XP}}}

\newcommand{\NP}{\ensuremath{\mathrm{NP}}}

\newcommand{\UP}{\textup{\textsc{Upward Planarity}}}

\newboolean{longShapeDesc}
\newcommand{\shapeDesc}[8]{
\ifthenelse{\boolean{longShapeDesc}}
{\ensuremath{\langle #1,\allowbreak #2,\allowbreak #3,\allowbreak #4,\allowbreak #5,\allowbreak #6,\allowbreak #7,\allowbreak #8 \rangle}}
{\ensuremath{\langle #1,\allowbreak #2,\allowbreak #3,\allowbreak #5,\allowbreak #6\rangle}}
}

\newboolean{longShapeSequence}
\newcommand{\shapeSequence}[8]{
\ifthenelse{\boolean{longShapeSequence}}
{\shapeDesc{#1}{#2}{#3}{#4}{#5}{#6}{#7}{#8}}
{\ensuremath{\langle #1,\allowbreak #2,\allowbreak #3,\allowbreak #5\rangle}}
}

\begin{document}

%

	\maketitle

	\begin{abstract}
We obtain new parameterized algorithms for the classical problem of determining whether a directed acyclic graph admits an upward planar drawing. Our results include a new fixed-parameter algorithm parameterized by the number of sources, an XP-algorithm parameterized by treewidth, and a fixed-parameter algorithm parameterized by treedepth. All three algorithms are obtained using a novel framework for the problem that combines SPQR tree-decompositions with parameterized techniques. Our approach unifies and pushes beyond previous tractability results for the problem on series-parallel digraphs, single-source digraphs and outerplanar digraphs.
	\end{abstract}
	

\section{Introduction}\label{sec:intro}

A digraph is called \emph{upward planar} if it admits an upward planar drawing, that is, a planar drawing where all edges are oriented upward. The problem of upward planarity testing (\UP) and constructing an associated upward planar drawing arises, among others, in the context of visualization of hierarchical network structures; application domains include project management, visual languages and software engineering~\cite{DBLP:journals/siamcomp/BertolazziBMT98}. Upward planarity is the most prominent notion of planarity that is inherently directed, and also has classical connections to the theory of ordered sets: the orders arising from the transitive closure of upward planar single-source digraphs have bounded dimension~\cite{DBLP:journals/jct/TrotterM77}.

Since the introduction of the notion, \UP\ has become the focus of extensive theoretical research. 
The problem has been shown to be $\NP$-complete more than 25 years ago~\cite{GargT94,gt-ccu-01}, but the first polynomial-time algorithms for restricted variants of \UP\ have been published even earlier~\cite{HuttonL91,HuttonL96}. 
Among others, the problem is known to be polynomial-time tractable when $G$ is provided with a planar embedding~\cite{bdl-udtd-94} (which also implies polynomial-time tractability for triconnected DAGs, since these admit a single planar embedding), or when restricted to the class of outerplanar DAGs~\cite{Papakostas94}, DAGs whose underlying graph is series-parallel~\cite{DBLP:journals/siamdm/DidimoGL09}, and most prominently single-source DAGs~\cite{DBLP:journals/siamcomp/BertolazziBMT98,BrucknerHR19,HuttonL96}.

In spite of the number of results on \UP\ that analyze the classical complexity of the problem on specific subclasses of instances, the problem was up to now mostly unexplored from the more fine-grained perspective of parameterized complexity analysis\iflong\footnote{A basic introduction to parameterized complexity is provided in the preliminaries.}\fi~\cite{CyganFKLMPPS15,DowneyF13}. In particular, while it was known that  \UP\ is fixed-parameter tractable when parameterized by the cyclomatic number of the input DAG (or, equivalently, the feedback edge number of the underlying undirected graph)~\cite{Chan04}, by the number of triconnected components and cut vertices~\cite{HealyL06}, or the number of triconnected components plus the maximum diameter of a split component~\cite{DBLP:journals/siamdm/DidimoGL09}, the complexity of the problem under classical structural parameterizations has remained completely open.

\smallskip
\noindent
\textbf{Contribution.}\quad
We develop a novel algorithmic framework for solving \UP\ which combines parameterized dynamic programming with the SPQR-tree decompositions of planar graphs~\cite{dt-opl-96,gm-lti-00,HopcroftT73}. In essence, our framework uses a characterization of the ``shapes'' of faces in an upward planar drawing that is inspired by earlier work on the notion of spirality~\cite{bdl-udtd-94,DBLP:journals/siamdm/DidimoGL09} and
reduces \UP\ to the task of handling the ``rigid'' nodes in these decompositions. Informally, the task that needs to be handled there can be stated as follows: what are all the possible ways to combine the possible shapes of the children of a rigid node to obtain an upward planar drawing for the node itself?
The framework is formalized in the form of a general ``Interface Lemma'' (Lemma~\ref{lem:R_node_general}) which can be complemented with numerous parameterizations as well as other algorithmic approaches. 

In the remainder of this article, we use this framework to push the boundaries of tractability for \UP. Our first result in this direction is a fixed-parameter algorithm for \UP\ parameterized by the number of sources in the input graph\iflong\footnote{It is worth noting that this immediately also implies tractability parameterized by the number of sinks: inverting all arcs results in an equivalent instance.}\fi. This result generalizes the polynomial-time tractability of the single-source case~\cite{DBLP:journals/siamcomp/BertolazziBMT98,BrucknerHR19} and answers an open question from a recent Dagstuhl seminar~\cite{Dagstuhlseminar21}. On a high level, we use the Interface Lemma to reduce the problem to a case where almost all children of a rigid node have a simple shape, and we show how this can be handled via a flow network approach.

Having established the tractability of instances with few sources, we turn towards understanding which structural properties of the underlying undirected graph can be used to solve \UP\ efficiently. In this context, 
\iflong there is a well-known hierarchy of fundamental structural parameters of graphs (see, e.g., \cite[Figure 1]{Belmonte0LMO20}) and the study of their algorithmic applications is a prominent research direction that has been pursued not only for a number of problems in computational geometry~\cite{CabelloK09,HuszarS19}, but also in many other fields such as in artificial intelligence~\cite{BliemMMW20,GanianO18} and constraint satisfaction~\cite{SamerS10,GanianRS17}. Yet, \fi
apart from the fixed-parameter tractability of \UP\ parameterized by the feedback edge number~\cite{Chan04}\iflong (which is one of the by far most restrictive parameterizations in the aforementioned hierarchy)\fi, nothing was known about whether the more widespread ``decompositional'' parameters 
can be used to solve the problem. The parameters that will be of interest here are \emph{treewidth}~\cite{RobertsonS84}, the most prominent structural graph parameter, and \emph{treedepth}~\cite{sparsity}, the arguably best known parameter that lies below treewidth in the parameter hierarchy\ifshort\ (see, e.g., \cite[Figure 1]{Belmonte0LMO20})\fi.

To obtain new boundaries of tractability for \UP\ with respect to these two parameters, we first show that the problem posed by the Interface Lemma can be restated as a purely combinatorial problem on a suitable combinatorization of the embedding of the graph represented by the rigid node, and---crucially---that a bound on the input graph's treewidth also implies a bound for the treewidth of this combinatorization. Once that is done, we design a non-trivial dynamic program that exploits this treewidth bound to handle the rigid nodes, which together with the Interface Lemma allows us to solve \UP. This yields an $\XP$-algorithm for \UP\ parameterized by the treewidth of the underlying undirected graph---a result which unifies and generalizes the polynomial-time tractability of \UP\ on outerplanar as well as series-parallel graphs~\cite{DBLP:journals/siamdm/DidimoGL09,Papakostas94}. Furthermore, a more detailed analysis of the dynamic program reveals that the same algorithm runs in fixed-parameter time when parameterized by treedepth. 

\ifshort
Due to space limitations some proofs are omitted and can be found in~\cite{arxiv}.
\fi

\iflong
\subparagraph{Paper Organization.} After introducing the necessary preliminaries in Section~\ref{sec:prelim}, we introduce a characterization of components of the SPQR trees used by our framework (Section~\ref{sec:comp_shapes}), provide the bulk of our framework for digraphs whose underlying undirected graph is 2-connected (Section~\ref{sec:general_algo}), and then complete the description of the framework by handling general digraphs (Section~\ref{se:singly-connected}). Sections~\ref{sec:fpt_sources} and~\ref{sec:tw} are then dedicated to applying the framework to \UP\ parameterized by the number of sources and treewidth along with treedepth, respectively. 
\fi

%
%
%
%

	\section{Preliminaries}\label{sec:prelim}
	\iflong
		We refer to the classical books for basic graph and graph drawing terminology~\cite{Diestel12,BattistaETT99}.
	\fi	
	\ifshort
	We refer to the usual sources for graph drawing and parameterized complexity terminology~\cite{CyganFKLMPPS15,BattistaETT99,Diestel12,DowneyF13}.
	\fi
	We use $N_G(v)$ to denote the set of vertices adjacent to a vertex $v$ in a graph $G$.

\iflong
	\subsection{Upward planar drawings and embeddings} \label{subse:upward-definitions}
	\fi
	
	\ifshort
\subparagraph{Upward planar drawings and embeddings.}
\fi

\iflong
	A \emph{drawing} of a graph maps each vertex to a point in the plane and each edge to a Jordan arc between the end-points of the edge. A drawing is \emph{planar} if no two edges intersect, except at common end-points. A planar drawing partitions the plane into regions, called \emph{faces}. The bounded faces are called \emph{internal}, while the unbounded face is the \emph{outer face}. 
	Two planar drawings of a graph are \emph{equivalent} if: (1) they have the same \emph{rotation system}, that is, for each vertex $v$, the clockwise order of the edges incident to $v$ is the same in both drawings; and (2) their outer faces are delimited by the same walk, that is, the order of the edges encountered when clockwise traversing the boundary of the outer face is the same in both drawings. A \emph{planar embedding} of a graph is an equivalence class of planar drawings of that graph. 
	\fi
	\ifshort
	 A \emph{planar embedding} is an equivalence class of planar drawings of a graph, where two drawings are equivalent if the clockwise order of the edges incident to each vertex is the same and the outer faces are delimited by the same walk. 
	\fi	 
	 
	\iflong
	Thus, a planar embedding of a graph consists of a rotation system and a choice for the walk delimiting the outer face. We often talk about a \emph{face of a planar embedding}, meaning a face of any planar drawing that respects the planar embedding. The \emph{flip} of a planar embedding is the planar embedding obtained by reversing the clockwise order of the edges incident to each vertex and by reversing the order of the edges encountered when clockwise traversing the boundary of the outer face. \fi
	
	
\iflong	
	Throughout the paper, we use the term \emph{digraph} as short for ``directed graph''. A digraph is \emph{acyclic} if it contains no directed cycle. An acyclic digraph is usually called \emph{DAG}, for short. A vertex in a digraph is a \emph{source} if it is only incident to outgoing edges and it is a \emph{sink} if it is only incident to incoming edges. \fi
	A vertex in a digraph is a \emph{switch} if it is a source or a sink, and it is a \emph{non-switch} otherwise. The \emph{underlying graph} of a digraph is the undirected graph obtained from the digraph by ignoring the edge directions. 
	\iflong
	A \emph{plane digraph} is a digraph together with a prescribed planar embedding for its underlying graph.

	\fi	
	A drawing of a digraph is \emph{upward} if every edge is represented by a Jordan arc monotonically increasing from the source to the sink of the edge, and it is \emph{upward planar} if it is both upward and planar. A digraph is \emph{upward planar} if it admits an upward planar drawing; we use \UP\ to denote the problem of determining whether a digraph is upward planar; w.l.o.g., we assume that the input digraph is connected. 

\ifshort
In an upward planar drawing $\Gamma$ of a digraph $G$, an \emph{angle} represents an incidence between a vertex $v$ and a face $f$. The angle is either \emph{flat} (if precisely one of the two edges incident to $v$ and $f$ is incoming at $v$), \emph{large} (if $v$ is a switch vertex and the angle has more than $180^{\circ}$ in $\Gamma$), or \emph{small} (otherwise)~\cite{bdl-udtd-94}; the latter two cases are jointly called \emph{switch angles}. \fi
\iflong
	Consider an upward planar drawing $\Gamma$ of a digraph $G$. An \emph{angle} $\alpha$ of a face $f$ of $\Gamma$ is a triple $(e_1,v,e_2)$, where $e_1$ and $e_2$ are two edges of $G$ that are incident to the vertex $v$, that are incident to the face $f$, and that are consecutive in the order of the edges encountered when clockwise traversing the boundary of $f$. We say that $\alpha$ is \emph{flat} if one between $e_1$ and $e_2$ is incoming $v$ and the other one is outgoing $v$, otherwise $\alpha$ is a \emph{switch angle}.	\fi	
	Then $\Gamma$ defines an \emph{angle assignment}, which assigns the value $-1$, $0$, and $1$ to each small, flat, and large angle, respectively, in every face of $\Gamma$. The angle assignment, together with the planar embedding of the underlying graph of $G$ in $\Gamma$, constitutes an \emph{upward planar embedding} of $G$.

	\iflong
	A switch angle at a vertex $v$ is hence delimited by two outgoing or by two incoming edges for $v$. Each switch angle is further classified as \emph{large} or \emph{small} as follows. Consider a switch angle $\alpha=(e_1,v,e_2)$ at a vertex $v$ in a face $f$ delimited by two outgoing edges (resp.\ by two incoming edges) and consider a disk $D$ centered at $v$, sufficiently small so that its boundary has a single intersection with every edge incident to $v$. The edges $e_1$ and $e_2$ divide $D$ into two regions, one of which contains part of $f$ and contains no portion of any edge incident to $v$ in its interior; call $D'$ this region. Then we say that $\alpha$ is \emph{large} if $D'$ contains a suitably short vertical segment that has $v$ as its highest (resp.\ lowest) end-point, it is \emph{small} otherwise.

	An upward planar drawing hence defines an \emph{angle assignment}, which is an assignment of the value $-1$, $0$, and $1$ to each small, flat, and large angle, respectively, in every face of $\Gamma$. This angle assignment, together with the planar embedding of the underlying graph of $G$ in $\Gamma$, constitutes an \emph{upward planar embedding} of $G$. An \emph{upward plane digraph} is a digraph together with a prescribed upward planar embedding.
		\fi
	

	The angle assignments that enhance a planar embedding into an upward planar embedding have been characterized by Didimo et al.~\cite{DBLP:journals/siamdm/DidimoGL09}, building on the work by Bertolazzi et al.~\cite{bdl-udtd-94}. Note that, once the planar embedding $\mathcal E$ of a digraph $G$ is specified, then so are the angles of the faces of $\mathcal E$; in particular, whether an angle is flat or switch only depends on $\mathcal E$. Consider an angle assignment for $\mathcal E$. If $v$ is a vertex of $G$, we denote by $n_i(v)$ the number of angles at $v$ that are labeled $i$, with $i \in \{-1,0,1\}$. If $f$ is a face of $G$, we denote by $n_i(f)$ the number of angles of $f$ that are labeled $i$, with $i \in \{-1,0,1\}$. The cited characterization is as follows.
		
	
\begin{theorem}[\cite{bdl-udtd-94,DBLP:journals/siamdm/DidimoGL09}]\label{th:upward-conditions}
		Let $G$ be a digraph, $\mathcal E$ be a planar embedding of the underlying graph of $G$, and $\lambda$ be an assignment of each angle of each face in $\mathcal E$ to a value in $\{-1,0,1\}$. Then $\mathcal E$ and $\lambda$ define an upward planar embedding of $G$ if and only if the following properties hold:
		\begin{description}
			\item[UP0] If $\alpha$ is a switch angle, then $\lambda(\alpha)\in\{-1,1\}$, and if $\alpha$ is a flat angle, then $\lambda(\alpha)=0$.
			\item[UP1] If $v$ is a switch vertex of $G$, then $n_1(v)=1$, $n_{-1}(v)=\deg(v)-1$, $n_0(v)=0$.
			\item[UP2] If $v$ is a non-switch vertex of $G$, then $n_1(v)=0$, $n_{-1}(v)=\deg(v)-2$, $n_0(v)=2$.
			\item[UP3] If $f$ is a face of $G$, then $n_1(f)=n_{-1}(f)-2$ if $f$ is an internal face and $n_{1}(f)=n_{-1}(f)+2$ if $f$ is the outer face.
		\end{description}
	\end{theorem}

\iflong
 
	Theorem~\ref{th:upward-conditions} has the following algorithmic consequence.
	
\begin{theorem}[\cite{bdl-udtd-94,bkm-msms-17}]\label{th:upward-fixed-embedding}
		Let $G$ be an $n$-vertex digraph and $\mathcal E$ be a planar embedding of the underlying graph of $G$. It is possible to test in $\bigoh(n \log^3 n)$ time whether there exists an angle assignment $\lambda$ such that $\mathcal E$ and $\lambda$ define an upward planar embedding of $G$. Consequently, it can be tested in $\bigoh(n \log^3 n)$ time whether $G$ admits an upward planar drawing $\Gamma$ which respects $\mathcal E$; in the positive case, $\Gamma$ can be constructed within the same time bound.
	\end{theorem}	

	Theorem~\ref{th:upward-fixed-embedding} was originally proved in~\cite{bdl-udtd-94} with an $\bigoh(n^2)$ running time. The main idea for its proof is the following. Construct a \emph{planar bipartite network} $\mathcal N(S,T,A)$, where:
	\begin{itemize}
		\item $S$ is a set of sources; there is a source $s_w$ for each vertex $w$ of $G$; each source can supply a single unit of flow;
		\item $T$ is a set of sinks; there is a sink $t_f$ for each face $f$ of $\mathcal E$; each sink $t_f$ demands a number of units of flow equal to $n_f/2 -1$ or $n_f/2 +1$, depending on whether $f$ is an internal face or the outer face of $\mathcal E$, where $n_f$ is the number of switch angles incident to $f$; and
		\item $A$ is a set of arcs from the sources to the sinks; there is an arc from a source $s_w$ to a sink $t_f$ if the vertex $w$ corresponding to $s_w$ is incident to the face $f$ corresponding to $t_f$; each arc $a$ has a capacity $c_a$ of a single unit of flow. 
	\end{itemize}
	It was proved in~\cite{bdl-udtd-94} that there exists an angle assignment $\lambda$ such that $\mathcal E$ and $\lambda$ define an upward planar embedding of $G$ if and only if the bipartite network $\mathcal N$ admits a flow whose value is the sum of the demands of the sinks in $T$ (or the sum of the supplies of the sources in $S$). A \emph{flow} is an assignment of a value $\phi_a\leq c_a$ to each arc $a$ such that the sum of the values assigned to the arcs outgoing each node $s_w$ does not exceed the supply of $s_w$ and  the sum of the values assigned to the arcs incoming each node $t_f$ does not exceed the demand of $t_f$. The \emph{value} of a flow is the sum of the values assigned to its arcs. 
	
	Borradaile et al.~\cite{bkm-msms-17} presented an $\bigoh(n\log^3 n)$ algorithm to find a maximum flow in a planar network with multiple sources and sinks. This result, plugged into the described framework by Bertolazzi et al.~\cite{bdl-udtd-94}, gives us a proof of Theorem~\ref{th:upward-fixed-embedding}.
	
	\fi
	\iflong
	\subsection{Parameterized Complexity}
	In parameterized complexity~\cite{CyganFKLMPPS15,DowneyF13},
the complexity of a problem is studied not only with respect to the input size, but also with respect to some problem parameter(s). The core idea behind parameterized complexity is that the combinatorial explosion resulting from the \NP-hardness of a problem can sometimes be confined to certain structural parameters that are small in practical settings. We now proceed to the formal definitions.

A {\it parameterized problem} $Q$ is a subset of $\Omega^* \times
\mathbb{N}$, where $\Omega$ is a fixed alphabet. Each instance of $Q$ is a pair $(I, \kappa)$, where $\kappa \in \mathbb{N}$ is called the {\it
parameter}. A parameterized problem $Q$ is
{\it fixed-parameter tractable} (\FPT)~\cite{DowneyF13,CyganFKLMPPS15}, if there is an
algorithm, called a {\em fixed-parameter algorithm}, that decides whether an input $(I, \kappa)$
is a member of $Q$ in time $f(\kappa) \cdot |I|^{\bigoh(1)}$, where $f$ is a computable function.  The class \FPT{} denotes the class of all fixed-parameter tractable parameterized problems.
A weaker notion of tractability is that of \XP: a parameterized problem $Q$ is in the class \XP\ if there is an algorithm that decides whether an input $(I, \kappa)$
is a member of $Q$ in time $(|I|+1)^{f(\kappa)}$, where $f$ is a computable function.
	\fi

\iflong
	\subsection{Treewidth and Treedepth} \fi
\ifshort	
\subparagraph{Treewidth and Treedepth.} \fi
	Here we consider the treewidth and treedepth of the underlying graphs\footnote{Directed alternatives to treewidth exist, but are typically not well-suited for algorithmic applications~\cite{GanianHK0ORS16}.}.
	A \emph{tree-decomposition}~$\mathcal{T}$ of a graph $G=(V,E)$ is a pair 
$(T,\chi)$, where $T$ is a tree (whose vertices we call \emph{nodes}) rooted at a node $r$ and $\chi$ is a function that assigns each node $t$ a set $\chi(t) \subseteq V$ such that the following holds: 
\ifshort
for every $uv \in E$ there is a node	$t$ such that $u,v\in \chi(t)$, and for every vertex $v \in V$, the set of nodes $t$ satisfying $v\in \chi(t)$ forms a nonempty subtree of~$T$. 
%
%
%
\fi
\iflong
\begin{itemize}[noitemsep]
	\item For every $uv \in E$ there is a node	$t$ such that $u,v\in \chi(t)$.
	\item For every vertex $v \in V$, the set of nodes $t$ satisfying $v\in \chi(t)$ forms a nonempty subtree of~$T$.
\end{itemize}

A tree-decomposition is \emph{nice} if the following two conditions are also satisfied:
\begin{itemize}[noitemsep]
	\item $|\chi(\ell)|=1$ for every leaf $\ell$ of $T$ and $|\chi(r)|=0$.
	\item There are only three kinds of non-leaf nodes in $T$:
	\begin{itemize}[noitemsep,label=]
        \item \textbf{Introduce node:} a node $t$ with exactly
          one child $t'$ such that $\chi(t)=\chi(t')\cup
          \{v\}$ for some vertex $v\not\in \chi(t')$.
        \item \textbf{Forget node:} a node $t$ with exactly
          one child $t'$ such that $\chi(t)=\chi(t')\setminus
          \{v\}$ for some vertex $v\in \chi(t')$.
        \item \textbf{Join node:} a node $t$ with two children $t_1$,
          $t_2$ such that $\chi(t)=\chi(t_1)=\chi(t_2)$.
	\end{itemize}
\end{itemize}
\fi
The \emph{width} of a tree-decomposition $(T,\chi)$ is the size of a largest set $\chi(t)$ minus~$1$, and the \emph{treewidth} of the graph $G$,
denoted $\tw(G)$, is the minimum width of a tree-decomposition of~$G$. 
\iflong It is known that a tree-decomposition can be transformed into a nice tree-decomposition of the same width in linear time. 
Efficient fixed-parameter algorithms are known for computing a nice tree-decomposition of near-optimal width~\cite{BodlaenderDDFLP16,Kloks94}. 

\begin{proposition}[\cite{BodlaenderDDFLP16}]\label{fact:findtw}%
	There exists an algorithm which, given an $n$-vertex graph $G$ and an integer~$k$, in time $2^{\bigoh(k)}\cdot n$ either outputs a tree-decomposition of $G$ of width at most $5k+4$ and $\bigoh(n)$ nodes, or determines that $\tw(G)>k$.
\end{proposition}  

We let $T_t$ denote the subtree of $T$ rooted at a node $t$, and we use $\chi(T_t)$ to denote the set $\bigcup_{t'\in V(T_t)}\chi(t')$. In the context of dynamic programming, the set $\past(t)=\chi(T_t)\setminus \chi(t)$ is called the \emph{past} while the set $\past(t)=V\setminus\chi(T_t)$ is called the \emph{future}.

\fi	
The second structural parameter that we will be considering here is the \emph{treedepth} of a graph $G$, denoted $\td(G)$~\cite{sparsity}. A useful way of thinking about graphs of bounded treedepth is that they are (sparse) graphs with no long paths.

\iflong	
	A \emph{rooted forest} is a disjoint union of rooted trees. 
	For a vertex~$x$ in a tree~$T$ 
	of a rooted forest, the \emph{height} (or {\em depth})
	of~$x$ in the forest is the number of vertices in the path from 
	the root of~$T$ to~$x$. The \emph{height of a rooted forest} is the maximum height of a vertex of the forest. 
	\sloppy
	\begin{definition}[Treedepth]\label{def:td}
		Let the \emph{closure} of a rooted forest~$\cal F$ be the graph
		$\clos({\cal F})=(V_c,E_c)$ with the vertex set 
		$V_c=\bigcup_{T \in \cal F} V(T)$ and the edge set
		$E_c=\{xy \colon \text{$x$ is an ancestor of $y$ in some $T\in\cal F$}\}$.
		A \emph{treedepth decomposition}
		of a graph $G$ is a rooted forest $\cal F$ such that $G \subseteq \clos(\cal F)$.
		The \emph{treedepth} $\td(G)$ of a graph~$G$ is the minimum height of
		any treedepth decomposition of $G$. 
	\end{definition}
	
	\noindent We will later use $T_x$ to denote the vertex set of the subtree of $T$ rooted at a vertex $x$ of $T$. 
	The following properties of treedepth will be crucial for our considerations.

\begin{proposition}[\cite{sparsity}]
		\label{pro:tdfacts}
		For every path of length $d$ in a graph $G$, it holds that $\td(G)\leq d\leq 2^{\td(G)}$. Moreover, $\tw(G) \leq \td(G)$.
	\end{proposition} 
\fi

\iflong
	\subsection{Expansion}
	\fi
	\ifshort
	\subparagraph{Expansion.}
\fi
\ifshort
	In our algorithms, we will employ a linear-time preprocessing step called expansion to simplify the input digraphs so that every vertex has at most one incoming edge (in which case it is a \emph{top} vertex) or at most one outgoing edge (in which case it is a \emph{bottom} vertex)~\cite{DBLP:journals/siamcomp/BertolazziBMT98}. The expansion is obtained by replacing each non-switch vertex $v$ with two new vertices $v_1$ and $v_2$, which inherit the incoming and outgoing edges of $v$, respectively, and the edge $(v_1,v_2)$ (called the \emph{special edge} of $v_1$ and $v_2$). It is known that expansion preserves upward planarity, and it is possible to observe that it preserves biconnectivity, does not create new sources, and only increases treewidth and treedepth by at most a factor of $2$.
	\fi
\iflong
	In our algorithms, we will employ a linear-time preprocessing step to simplify the input digraphs so that every vertex has at most one incoming edge or at most one outgoing edge; this operation, which was introduced by Bertolazzi et al.~\cite{DBLP:journals/siamcomp/BertolazziBMT98}, is called \emph{expansion} and is defined as follows. We begin by marking each non-switch vertex in the digraph $G=(V,E)$ as \emph{unprocessed}. We then loop over each non-switch unprocessed vertex $v$, for which we create two new \emph{processed} vertices $v_1$ and $v_2$ and add the following edges to the digraph: $\{v_1v_2\}\cup \{(av_1~|~av\in E\}\cup \{v_2b~||~ vb\in E\}$, and then we delete $v$. In an expanded graph, that is, a graph that has gone through the expansion operation, every vertex has at most one outgoing or at most one incoming edge; in the former case we say it is a \emph{bottom} vertex, while in the latter case it is a \emph{top} vertex. The only outgoing (incoming) edge of a bottom (top) vertex of an expanded digraph is the \emph{special edge} of that vertex. Since a non-switch vertex of degree $2$ has exactly one incoming and one outgoing edge, it could be considered both a top and a bottom vertex. To avoid this ambiguity, we treat degree-$2$ vertices as top vertices.   
	
	A digraph expansion preserves upward planarity~\cite{DBLP:journals/siamcomp/BertolazziBMT98}; furthermore, it preserves biconnectivity and the number of sources, and at most doubles the number of vertices. Moreover, below we show that applying expansion on graphs of bounded treedepth results in graphs which also have bounded treedepth. Hence, in the remainder, we will assume, without loss of generality, that each digraph for which we aim to test upward planarity has been expanded and thus each of its vertices has at most one incoming edge or at most one outgoing edge.
 \begin{observation} 
		Let $G$ be a digraph and $G'$ be its expansion. The treedepth and treewidth of $G'$ are at most twice the treedepth and treewidth of $G$, respectively. Moreover, $G$ and $G'$ have the same number of sources.
	\end{observation}
\begin{proof}
		Consider a treedepth decomposition $\cal F$ of $G$, replace every vertex of this decomposition with an edge to obtain a rooted forest $\cal F'$.
		We claim that $\cal F'$ is a treedepth decomposition of $G'$: For a vertex $v$ of $G$ corresponding to a vertex $u$ in $\cal F$, associate the respective vertices $v_1$, $v_2$ of $G'$ with the endpoints of the edge in $\cal F'$ that replaced $u$. It can be easily seen that no ancestor-descendant relation is violated, and that the height of $\cal F'$ is at most twice the height of $\cal F$.
		
		For treewidth, it is possible to simply transform the original tree-decomposition of $G$ into one of $G'$ by replacing each vertex by the two vertices that replaced it, resulting in a new tree-decomposition of $G'$ with at most twice the width of the original. Regarding the number of sources, it can be observed that each vertex that is not a source is replaced by two vertices which are not sources, and each source is replaced by two vertices of which precisely one is a source.
	\end{proof}\fi
	

\iflong
	\subsection{SPQR-tree decomposition}\label{subsec:spqr}
	\fi
	\ifshort
	\subparagraph{SPQR-tree decomposition.}
	\fi
	Let $G$ be a biconnected undirected graph. 
	A pair of vertices is a \emph{separation pair} if its removal disconnects $G$. 
	A \emph{split pair} is either a separation pair or a pair of adjacent vertices. 
	A \emph{split component} of $G$ with respect to a split pair $\{u,v\}$ is either an edge $(u,v)$ or a maximal subgraph $G_{uv} \subset G$ such that $\{u,v\}$ is not a split pair of $G_{uv}$. 
	A split pair $\{s',t'\}$ of $G$ is \emph{maximal} with respect to a split pair $\{s,t\}$ of $G$, if for every other split pair $\{s^*,t^*\}$ of $G$, there is a split component that includes the vertices $s',t',s$ and $t$. 
	

	An \emph{SPQR-tree} $T$ of $G$ with respect to an edge $e^*$ is a rooted tree that describes a recursive decomposition of $G$ induced by its split pairs~\cite{dt-opl-96}. 
	\iflong In what follows, we call \emph{nodes} the vertices of $T$, to distinguish them from the vertices of $G$. \fi
\iflong The nodes of $T$ are of four types S, P, Q, and R. \fi	
	Each node $\mu$ of $T$ is associated with a split pair $\{u,v\}$ of $G$, where $u$ and $v$ are the \emph{poles} of $\mu$, with a subgraph $G_\mu$ of $G$, called the \emph{pertinent graph} of $\mu$, which consists of one or more split components of $G$ with respect to $\{u,v\}$, and with a multigraph $\textrm{sk}(\mu)$, called the \emph{skeleton of $\mu$}, which represents the arrangement of such split components in $G_\mu$. The edges of $\textrm{sk}(\mu)$ are called \emph{virtual edges}. Each node $\mu$ of $T$ whose pertinent graph is not a single edge has some children, each corresponding to a split components of $G$ in $G_\mu$. Each of these children is the root of a subtree of $T$.  	
	\ifshort
The nodes of $T$ are of four types S, P, Q, and R. Q-nodes correspond to edges of $G$, while S-, P- and R-nodes correspond to so-called series, parallel and rigid compositions of the pertinent graphs of the children of the given node~\cite{dt-opl-96}.
	\fi

	\iflong
	Formally, $T$ is defined as follows.	
The root $\rho$ of $T$ is a Q-node corresponding to the edge $e^*=(u,v)$. The child $\sigma$ of $\rho$ has $u$ and $v$ as poles, and its pertinent graph is defined as $G_\sigma = G \setminus e^*$. Now consider a node $\mu$ of $T$ with poles $s$ and $t$ and a pertinent graph $G_{\mu}$; assume that $G_{\mu}$ is \emph{$st$-biconnectible}, i.e., it is either biconnected or it becomes so if the edge $(s,t)$ is added to it. We distinguish some cases.

	\textbf{Base case}: $G_\mu$ consists of a single edge $e$ between $s$ and $t$. Then, $\mu$ is a Q-node whose skeleton is the edge $(s,t)$. The node $\mu$ is a leaf of $T$. 

	\textbf{Series case}: $G_\mu$ is not biconnected; since $G_\mu$ is $st$-biconnectible, each cut-vertex of $G_{\mu}$ (i.e., each vertex whose whose removal disconnects $G_\mu$) lies on any path from $s$ to $t$. Then, $\mu$ is an S-node. Let $v_1,\dots,v_{k-1}$, where $k \geq 2$, be the cut vertices of $G_\mu$, in the order in which they are encountered in any path from $s$ to $t$. The skeleton of $\mu$ is a path consisting of the virtual edges $e_1,\dots,e_k$; for $i=1,\dots,k$, we have $e_i= (v_{i-1},v_i)$, where $v_0=s$ and $v_k=t$. Further, $\mu$ has $k$ children $\nu_1,\dots,\nu_k$; for $i=1,\dots,k$, the poles of $\nu_i$ are $v_{i-1}$ and $v_{i}$, and the pertinent graph $G_{\nu_i}$ of $\nu_i$ is the union of all the split components with respect to $\{v_{i-1},v_{i}\}$ that do not contain both $s$ and $t$. The decomposition recurs on the nodes $\nu_1,\dots,\nu_k$.

	\textbf{Parallel case}: $G_\mu$ is not biconnected and $\{s,t\}$ is a split pair of $G_\mu$. Then $\mu$ is a P-node. If the split pair $\{s,t\}$ defines $k$ maximal split components of $G_\mu$, then the skeleton of $\mu$ is a set of $k$ parallel edges $e_1,\dots,e_k$ between $s$ and $t$.  Further, $\mu$ has $k$ children $\nu_1,\dots,\nu_k$; for $i=1,\dots,k$, the poles of $\nu_i$ are $s$ and $t$, and the pertinent graph $G_{\nu_i}$ of $\nu_i$ is a maximal split component with respect to $\{s,t\}$. The decomposition recurs on the nodes $\nu_1,\dots,\nu_k$.

	\textbf{Rigid case}: None of the other cases is applicable. Let $\{s_1,t_1\},\dots,\{s_k,t_k\}$ be the maximal split pairs of $G$ with respect to $\{s,t\}$ ($k \geq 1$) such that $\{s_i,t_i\}$ belongs to $G_\mu$, for $i=1,\dots,k$. 
	Then $\mu$ is an R-node whose skeleton is the graph whose vertex set is $\{s_1,t_1\}\cup \dots\cup \{s_k,t_k\}$ and whose edge set is $\{(s_1,t_1),\dots,(s_k,t_k)\}$. Then $\mu$ has $k$ children $\nu_1,\dots,\nu_k$; for $i=1,\dots,k$, the poles of $\nu_i$ are $s_i$ and $t_i$, and the pertinent graph $G_{\nu_i}$ of $\nu_i$ is the union of all the split components with respect to $\{s_{i},t_{i}\}$ that do not contain both $s$ and $t$. The decomposition recurs on the nodes $\nu_1,\dots,\nu_k$.
	
	\medskip
	\fi

	\begin{figure}[htb]
		\centering
		\begin{subfigure}{.2\textwidth}
			\centering
			\includegraphics[width=\columnwidth, page=1]{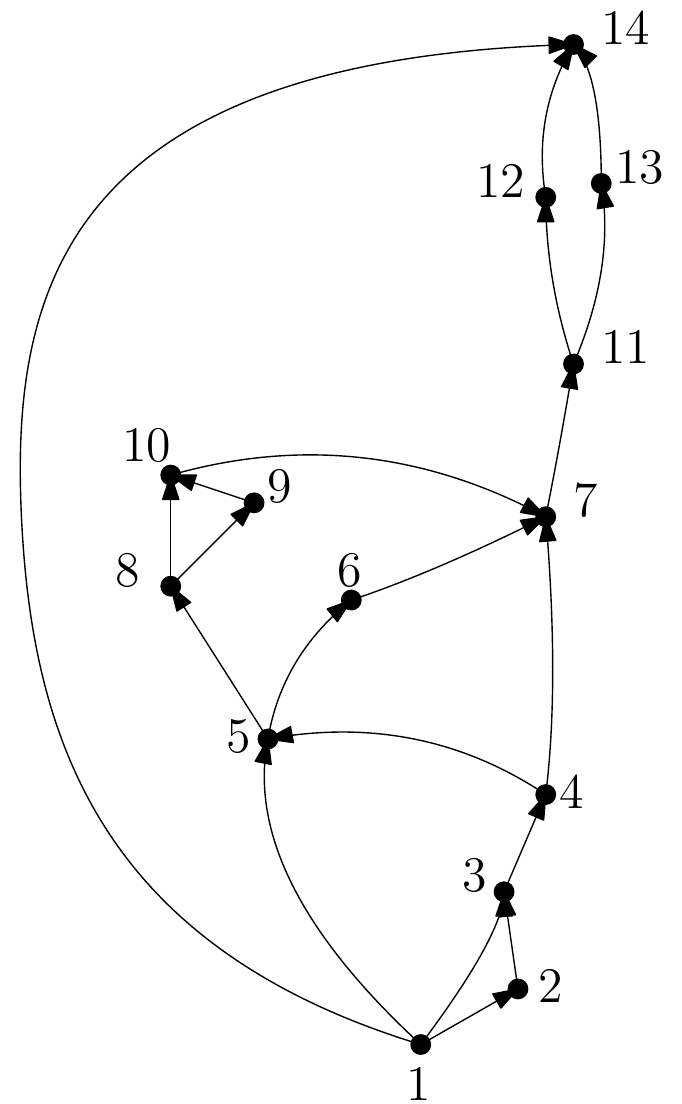}
			\subcaption{}
			\label{fi:SPQR-example-a}
		\end{subfigure}\\
		
		\begin{subfigure}{.49\textwidth}
			\centering
			\includegraphics[width=\columnwidth, page=2]{figures/SPQR-example}
			\subcaption{}
			\label{fi:SPQR-example-b}
		\end{subfigure}
		\hfil
		\begin{subfigure}{.49\textwidth}
			\centering
			\includegraphics[width=\columnwidth, page=3]{figures/SPQR-example}
			\subcaption{}
			\label{fi:SPQR-example-c}
		\end{subfigure}
		\caption{\label{fi:SPQR-example}(a) A planar DAG $G$. (b) An SPQR-tree of $G$. For each node that is not a Q-node, the skeleton is depicted together with a dashed edge to represent the rest of the graph; for each Q-node, the corresponding edge is shown.  (c) An SPQR-tree of $G$ whose S-nodes have exactly two children.}
	\end{figure}

	Note that each virtual edge $e_i$ in the skeleton of a node $\mu$ of $T$ \emph{corresponds} to the pertinent graph $G_{\nu_i}$ of a child $\nu_i$ of $\mu$. We say that $G_{\nu_i}$ is a \emph{component} of  $G_\mu$. Figs.~\ref{fi:SPQR-example-a} and~\ref{fi:SPQR-example-b} show a planar graph and its SPQR-tree. To simplify our algorithms, we assume that every S-node of $T$ has two children. If this is not the case, we can modify $T$ to achieve this property (see Fig.~\ref{fi:SPQR-example-c}). An SPQR-tree $T$ of an $n$-vertex planar graph has $\bigoh(n)$ Q-, S-, P-, and R-nodes. Also, the total number of vertices of the skeletons for the nodes in $T$ is $\bigoh(n)$~\cite{dt-opl-96}.
	
	
	When talking about an SPQR-tree $T$ of a biconnected directed graph $G$, we mean an SPQR-tree of its underlying graph. Let $\mu$ be a node of $T$ with poles $u$ and $v$. A \emph{$uv$-external upward planar embedding} of $G_{\mu}$ is an upward planar embedding of $G_{\mu}$ such that $u$ and $v$ are incident to the outer face. In our algorithms, when testing the upward planarity of a digraph $G$, the fact that its SPQR-tree $T$ is rooted at an edge $e^*$ of $G$ corresponds to the requirement that $e^*$ is incident to the outer face of the upward planar embedding $\mathcal E$ of $G$ we are looking for. For each node $\mu$ of $T$, the restriction of $\mathcal E$ to the vertices and edges of the pertinent graph $G_{\mu}$ of $\mu$ is a $uv$-external upward planar embedding of $G_{\mu}$.




\section{The Shapes of Components}\label{sec:comp_shapes}
	\setboolean{longShapeDesc}{true}
	
	Let $G$ be a biconnected DAG, let $T$ be an SPQR-tree of $G$ rooted at an edge $e^*$, let $\mu$ be a node of $T$ with poles $u$ and $v$, and let $\mathcal E_{\mu}$ be a $uv$-external upward planar embedding of $G_{\mu}$. Let $\lambda$ be the angle assignment defined by $\mathcal{E}_{\mu}$. The poles $u$ and $v$ identify two paths on the boundary of the outer face $f_0$ of $\mathcal{E}_{\mu}$: the \emph{left outer path} $P_l=\langle v_0=u, v_1, \dots, v_k=v \rangle$ is the path that leaves $f_0$ on the left when walking from $u$ to $v$; the \emph{right outer path} $P_r=\langle w_0=u,w_1,\dots,w_h=v\rangle$ of $\mathcal E_{\mu}$ is the path that leaves $f_0$ on the right when walking from $u$ to $v$; see Fig.~\ref{fig:shape_desc_example}. 
	\iflong Notice that $P_l$ and $P_r$ may share some vertices other than $u$ and $v$ if $G_{\mu}$ is not biconnected.\fi  
	For $i=0,1,\dots,k$, let $\alpha_i$ denote the angle at $v_i$ inside $f_0$ and, for $i=0,1,\dots,h$, let $\beta_i$ denote the angle at $w_i$ inside $f_0$. 
%
 The \emph{left-turn-number} $\tau_l(\mathcal E_{\mu},u,v)$ of $\mathcal E_{\mu}$ is defined as $\sum_{i=1}^{k-1}\lambda(\alpha_i)$, while the \emph{right-turn-number} $\tau_r(\mathcal E_{\mu},u,v)$ of $\mathcal E_{\mu}$ is  $\sum_{i=1}^{h-1}\lambda(\beta_i)$. 
Note that $\alpha_0=\beta_0$ and $\alpha_k=\beta_h$ are the angles at $u$ and $v$ inside $f_0$, respectively. 
The values $\lambda(\alpha_0)$ and $\lambda(\alpha_k)$ are also denoted by $\lambda(\mathcal E_{\mu},u)$ and $\lambda(\mathcal E_{\mu},v)$, respectively. 
Finally, given a vertex $w \in \{u,v\}$, let $\rho_l(\mathcal E_{\mu},w)$ denote the orientation of the edge $e_l$ of $P_l$ incident to $w$, that is, $\rho_l(\mathcal E_{\mu},w) = in$ if $e_l$ is an incoming edge for $w$, $\rho_l(\mathcal E_{\mu},w) = out$ otherwise. 
Analogously, let $\rho_r(\mathcal E_{\mu},w)$ denote the orientation of the edge $e_r$ of $P_r$ incident to $w$. 
The \emph{shape description of $\mathcal E_{\mu}$} is the tuple $\shapeDesc{\tau_l(\mathcal E_{\mu},u,v)}{\tau_r(\mathcal E_{\mu},u,v)}{\lambda(\mathcal E_{\mu},u)}{\lambda(\mathcal E_{\mu},v)}{\rho_l(\mathcal E_{\mu},u)}{\rho_r(\mathcal E_{\mu},u)}{\rho_l(\mathcal E_{\mu},v)}{\rho_r(\mathcal E_{\mu},v)}$; see Fig.~\ref{fig:shape_desc_example}.
	
	\iflong
	\begin{figure}[htb]
		\centering
		\begin{subfigure}{.4\textwidth}
			\centering
			\includegraphics[width=\columnwidth, page=13]{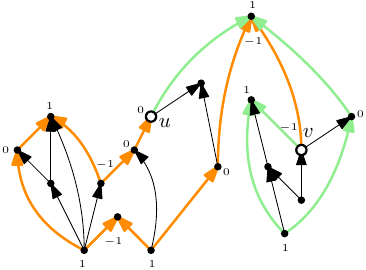}
		\end{subfigure}
		\caption{An upward planar embedding of a split component $G_\mu$ with poles $u$ and $v$ and shape description \shapeDesc{3}{0}{0}{-1}{out}{in}{out}{out}. The left (right) outer path is shown in green (orange).
		\label{fig:shape_desc_example}
		}
	\end{figure}
	\fi 
	\ifshort
	\begin{figure}[htb]
	\centering
	\begin{subfigure}{.4\textwidth}
			\centering
			\includegraphics[width=\columnwidth, page=13]{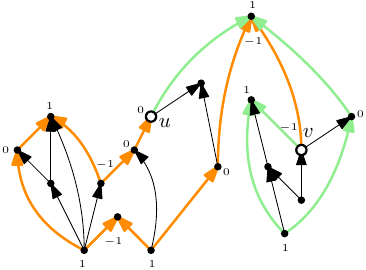}
		\end{subfigure}
		
		\caption{An upward planar embedding of a split component $G_\mu$ with poles $u$ and $v$ and shape description \shapeDesc{3}{0}{0}{-1}{out}{in}{out}{out}. The left (right) outer path is shown in green (orange).
			\label{fig:shape_desc_example}
		}
	\end{figure}	
	\fi

\iflong		
	We note that there are some dependencies between the values of a shape description. First, $\rho_l(\mathcal E_{\mu},u)$ and $\rho_r(\mathcal E_{\mu},u)$ have the same value if $\lambda(\mathcal E_{\mu},u)\in \{-1,1\}$, while they have different values if $\lambda(\mathcal E_{\mu},u)=0$. 
Similarly, $\rho_r(\mathcal E_{\mu},v)$ is implied by the values $\rho_l(\mathcal E_{\mu},v)$ and $\lambda(\mathcal E_{\mu},v)$. 
Furthermore, the values of $\tau_l(\mathcal E_{\mu},u,v)$ and $\rho_l(\mathcal E_{\mu},u)$ imply the value of $\rho_l(\mathcal E_{\mu},v)$; indeed, if $\tau_l(\mathcal E_{\mu},u,v)$ is even, then $\rho_l(\mathcal E_{\mu},u)$ and $\rho_l(\mathcal E_{\mu},v)$ are different, while if $\tau_l(\mathcal E_{\mu},u,v)$ is odd, then $\rho_l(\mathcal E_{\mu},u)$ and $\rho_l(\mathcal E_{\mu},v)$ are the same. Finally, $\lambda(\mathcal E_{\mu},v)$ is implied by $\tau_l(\mathcal E_{\mu},u,v)$,  $\tau_r(\mathcal E_{\mu},u,v)$, and $\lambda(\mathcal E_{\mu},u)$. This is a consequence of the following observation, which adopts the notation of Theorem~\ref{th:upward-conditions} and follows from the definition of the involved parameters. 

	\begin{observation}\label{obs:dependence-parameters}   
	We have $\tau_l(\mathcal E_{\mu},u,v)+\tau_r(\mathcal E_{\mu},u,v)+\lambda(\mathcal E_{\mu},u)+\lambda(\mathcal E_{\mu},v)=n_1(f_0)-n_{-1}(f_0)$.
\end{observation} 

Observation~\ref{obs:dependence-parameters}, together with Property~\textbf{UP3} of Theorem~\ref{th:upward-conditions}, implies the following.

\begin{corollary}\label{cor:dependence-parameters}
	We have $\tau_l(\mathcal E_{\mu},u,v)+\tau_r(\mathcal E_{\mu},u,v)+\lambda(\mathcal E_{\mu},u)+\lambda(\mathcal E_{\mu},v)=2$. 
\end{corollary}	
\fi

\ifshort
	There are some dependencies between the values of a shape description. For example, $\rho_l(\mathcal E_{\mu},u) \neq \rho_r(\mathcal E_{\mu},u)$ if $\lambda(\mathcal E_{\mu},u)=0$. As a further example, we have the following observation, which comes from Property~\textbf{UP3} of Theorem~\ref{th:upward-conditions} and uses the notation of this theorem.
	
	 \begin{observation}\label{obs:dependence-parameters}
	 	We have $\tau_l(\mathcal E_{\mu},u,v)+\tau_r(\mathcal E_{\mu},u,v)+\lambda(\mathcal E_{\mu},u)+\lambda(\mathcal E_{\mu},v)=2$. 
	 \end{observation}	
\fi

	Recall that if $u$ is a top or bottom vertex of $G$, then it has at most one incoming edge or at most one outgoing edge, respectively, which is called the \emph{special edge} of $u$. If $G_{\mu}$ contains this edge, then $G_{\mu}$ is a \emph{special} component for $u$, otherwise we say that $G_\mu$ is a \emph{normal} component  for $u$. Note that, if $u$ is a source or a sink of $G$, then it has no special component. 
\iflong	The following lemma bounds the right-turn number of a $uv$-external upward planar embedding $\mathcal{E}_\mu$ of $G_{\mu}$ with respect to its left-turn number.\fi

	\iflong \begin{lemma} \label{lem:shape_desc_values}
		Let $G_{\mu}$ be a split component of an upward plane digraph $G$ with respect to a split pair $\{u,v\}$. Let $\mathcal{E}_\mu$ be the induced uv-external upward planar embedding of $G_\mu$ with shape description $\shapeDesc{\tau_l(\mathcal{E}_{\mu},u,v)}{\tau_r(\mathcal{E}_{\mu},u,v)}{\lambda(\mathcal{E}_{\mu},u)}{\lambda(\mathcal{E}_{\mu},v)}{\rho_l(\mathcal{E}_{\mu},u)}{\rho_r(\mathcal{E}_{\mu},u)}{\rho_l(\mathcal{E}_{\mu},v)}{\rho_r(\mathcal{E}_{\mu},v)}$. 
If $\tau_l(\mathcal{E}_{\mu},u,v)=c$ is the left turn-number of $\mathcal{E}_{\mu}$, then the right turn-number $\tau_r(\mathcal{E}_{\mu},u,v)$ is $-c+h$ where $h \in \{0,1,2,3,4\}$. 
In particular, if $G_{\mu}$ is a normal component for both $u$ and $v$, then   $\{\lambda(\mathcal{E}_{\mu},u),\lambda(\mathcal{E}_{\mu},v)\}\subseteq \{-1,1\}$ and
		\begin{enumerate}
			\item $\tau_r(\mathcal{E}_{\mu},u,v)=-c$, if $\lambda(\mathcal{E}_{\mu},u)=\lambda(\mathcal{E}_{\mu},v)=1$ (see Fig.~\ref{fig:normal_1});
			\item $\tau_r(\mathcal{E}_{\mu},u,v)=-c+2$, if $\{\lambda(\mathcal{E}_{\mu},u),\lambda(\mathcal{E}_{\mu},v)\}=\{-1,1\}$ (see Figs.~\ref{fig:normal_2} and~\ref{fig:normal_3}); and
			\item $\tau_r(\mathcal{E}_{\mu},u,v)=-c+4$, if $\lambda(\mathcal{E}_{\mu},u)=\lambda(\mathcal{E}_{\mu},v)=-1$ (see Fig.~\ref{fig:normal_4}).
		\end{enumerate}
	\end{lemma}
\fi

\ifshort \begin{lemma}\label{lem:shape_desc_values}
We have $\tau_r(\mathcal E_{\mu},u,v)=-\tau_l(\mathcal E_{\mu},u,v)+h,$ with $h \in \{0,1,2,3,4\}$.
	\end{lemma}
\fi

	\iflong \begin{proof}
	If $G_{\mu}$ is a normal component for $u$, then either all the edges incident to $u$ are incoming $u$ or they are all outgoing $u$, hence the angle incident to $u$ in the outer face $f_0$ of $\mathcal{E}_{\mu}$ is a switch angle and $\lambda(\mathcal{E}_{\mu},u)\in \{-1,1\}$.  On the other hand, if $G_{\mu}$ is a special component for $u$, then either $\lambda(\mathcal{E}_{\mu},u)\in\{-1,0\}$, if $G_{\mu}$ contains at least one incoming and at least one outgoing edge of $u$, or $\lambda(\mathcal{E}_{\mu},u)=1$, if $G_{\mu}$ contains only the special edge of $u$. Similar considerations hold true for $v$. The rest of the statement follows by Corollary~\ref{cor:dependence-parameters}. Indeed, we have that $\tau_l(\mathcal{E}_{\mu},u,v)+\tau_r(\mathcal{E}_{\mu},u,v)+\lambda(\mathcal{E}_{\mu},u)+\lambda(\mathcal{E}_{\mu},v)=2$, hence $\tau_r(\mathcal{E}_{\mu},u,v)=2-\tau_l(\mathcal{E}_{\mu},u,v)-\lambda(\mathcal{E}_{\mu},u)-\lambda(\mathcal{E}_{\mu},v)$, from which the statements follow for $\tau_l(\mathcal{E}_{\mu},u,v)=c$ and by substituting the values of $\lambda(\mathcal{E}_{\mu},u)$ and $\lambda(\mathcal{E}_{\mu},v)$ of each case.
	\end{proof}\fi
	
\iflong
	\begin{figure}[!t]
		\centering
		{\centering
		\begin{subfigure}{.15\textwidth}
			\includegraphics[page=1]{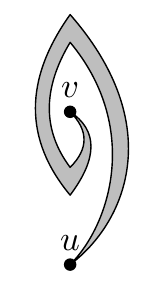}
			\subcaption{}
			\label{fig:normal_1}
		\end{subfigure}
		}
		\hfil
		{\centering
		\begin{subfigure}{.15\textwidth}
			\includegraphics[page=2]{figures/general_comps}
			\subcaption{}
			\label{fig:normal_2}
		\end{subfigure}
		}
		\hfil
		{\centering
		\begin{subfigure}{.15\textwidth}
			\includegraphics[page=3]{figures/general_comps}
			\subcaption{}
			\label{fig:normal_3}
		\end{subfigure}
		}
		\hfil
		{\centering
		\begin{subfigure}{.15\textwidth}
			\includegraphics[page=4]{figures/general_comps}
			\subcaption{}
			\label{fig:normal_4}
		\end{subfigure}
		}
		\caption{Illustrations of the various cases of Lemma~\ref{lem:shape_desc_values} for normal components.}
		\label{fig:normal_comps}
	\end{figure}
\fi

\iflong	
	The following two lemmata establish some useful properties for the shape descriptions of components that are parallel compositions of sub-components. Let $G_\mu$ be the pertinent digraph of a P-node $\mu$. Let $\nu'_1, \nu'_2 \dots, \nu'_{k'}$ be a subset of the children of $\mu$ in $T$ and let $G'_{\mu}$ be the subgraph of $G_{\mu}$ consisting of the union of $G_{\nu'_1}$, $G_{\nu'_2}$, $\dots$, $G_{\nu'_{k'}}$. Assume that the upward planar embedding $\mathcal{E}$ is such that the pertinent graphs $G_{\nu'_1}$, $G_{\nu'_2}$, $\dots$, $G_{\nu'_{k'}}$ appear in this order from left to right around $u$. 
	Denote by $\mathcal{E}_{\nu}$ the upward planar embedding of $G_\nu$ in $\mathcal{E}$. Then the outer boundary of $G'_{\mu}$ in $\mathcal{E}$ is formed by the left outer path of $\mathcal{E}_{\nu'_1}$ and the right outer path of $\mathcal{E}_{\nu'_{k'}}$. 
	Let $f'_i$ be the face bounded by the right outer path of $\mathcal{E}_{\nu'_i}$ and by the left outer path of $\mathcal{E}_{\nu'_{i+1}}$, for $i=1,2,\dots,k'-1$, and let $f'_0$ be the outer face of $G'_{\mu}$ in $\mathcal{E}$.
	
	\iflong \begin{lemma} \fi \ifshort \begin{lemma}\fi\label{lem:shape_face}
		Suppose $f'_i$ is an internal face of $G'_\mu$ with respect to an upward planar embedding $\mathcal{E}$ where $c$ is the right-turn-number of $\mathcal{E}_{\nu'_i}$, that is $c=\tau_r(\mathcal{E}_{\nu'_i}, u, v)$, and $\alpha_u(f'_i)$ and $\alpha_v(f'_i)$ are the angles at vertices $u$ and $v$ inside $f'_i$. Then the left-turn-number $\tau_l(\mathcal{E}_{\nu'_{i+1}}, u, v)$ of $G_{\nu'_{i+1}}$ satisfies:
		\begin{enumerate}
			\item $\tau_l(\mathcal{E}_{\nu'_{i+1}}, u, v)=-c$, if $\lambda(\alpha_u(f'_i))=\lambda(\alpha_v(f'_i))=1$
			\item $\tau_l(\mathcal{E}_{\nu'_{i+1}}, u, v)=-c-2$, if $\{\lambda(\alpha_u(f'_i)),\lambda(\alpha_v(f'_i))\}=\{-1,1\}$
			\item $\tau_l(\mathcal{E}_{\nu'_{i+1}}, u, v)=-c-4$, if $\lambda(\alpha_u(f'_i))=\lambda(\alpha_v(f'_i))=-1$.
		\end{enumerate}
	\end{lemma}

	\iflong \begin{proof}
		By  \textbf{UP3} we have that $n_1(f'_i)+n_{-1}(f'_i)=-2$. Since the boundary of $f'_i$ consists of the right outer path of $\mathcal{E}_{\nu'_{i}}$ and the right outer path of $\mathcal{E}_{\nu'_{i+1}}$, it follows that $n_1(f'_i)+n_{-1}(f'_i)=\tau_r(\mathcal{E}{\nu'_i}, u, v)+\tau_l(\mathcal{E}_{\nu'_{i+1}})+\lambda(\alpha_u(f'_i))+\lambda(\alpha_v(f'_i))=-2$. The lemma's statements follow by replacing $\tau_r(\mathcal{E}_{\nu'_i}, u, v)$ with $c$ and substituting the considered values of $\lambda(\alpha_u(f'_i))$ and $\lambda(\alpha_v(f'_i))$.
	\end{proof}\fi
	
	Let $\nu$ be a child of $\mu$ and let $G_{\nu}$ be the corresponding pertinent graph with poles $u$ and $v$. Suppose that $G_{\nu}$ admits a $uv$-external  upward planar embedding $\mathcal{E}_\nu$ such that the shape description of $G_{\nu}$ is $s$. We say that $s$ is a \emph{thin shape description} if it satisfies the following conditions, which we call \emph{thin conditions}:
	\begin{itemize}
		\item $\tau_r(\mathcal{E}_{\nu},u,v)=-\tau_l(\mathcal{E}_{\nu},u,v)$;
		\item $\lambda(\mathcal{E}_{\nu},u)=\lambda(\mathcal{E}_{\nu},v)=1$;
		\item $\rho_l(\mathcal{E}_{\nu},u)=\rho_r(\mathcal{E}_{\nu},u)$;
		\item $\rho_l(\mathcal{E}_{\nu},v)=\rho_r(\mathcal{E}_{\nu},v)$.
	\end{itemize}
	
	Let $N_\mu$ be the set of children $\nu_1, \nu_2 \dots, \nu_{k}$ of $\mu$. For any subset $N$ of $N_\mu$, we denote by $G_N$ the subgraph of $G_\mu$ that consists of all component $G_\nu$ with $\nu\in N$. Consider two subsets $N_1$ and $N_2$ of $N$ such that $N_2\subseteq N_1$ and assume that $G_{N_1}$ and $G_{N_2}$ are upward planar. We say that the upward planar embeddings $\mathcal{E}_{N_1}$ and $\mathcal{E}_{N_2}$ of $G_{N_1}$ and $G_{N_2}$ respectively, \emph{form an equivalence pair} if the components of $G_{N_2}$ appear in the same order in  $G_{N_1}$ around $u$ starting from the outer face, and for every $\nu\in N_2$ component $G_\nu$ has the same shape description in both $\mathcal{E}_{N_1}$ and $\mathcal{E}_{N_2}$.
Let $s$ be a thin shape description, and let $N_s$ be a subset of $N_\mu$ such that the pertinent graph $G_\nu$, for $\nu\in N_s$ admits an upward planar embedding with shape description $s$. For $N\subseteq N_\mu$ with $N\cap N_s\neq \emptyset$, we say that an upward planar embedding $\mathcal{E}_N$ of $G_N$ \emph{respects shape description }$s$ \emph{for the subset} $N_s$ if the components $G_\nu$ for $\nu\in N\cap N_s$ are consecutive around $u$ when starting from the outer face of $\mathcal{E}_N$, and all have shape description $s$ in $\mathcal{E}_N$.

	\iflong \begin{lemma} \fi \ifshort \begin{lemma}\fi\label{lem:parallel_reduce}
		Let $N_s$ be a subset of $N_\mu$ such that the pertinent graph $G_\nu$, for $\nu\in N_s$ admits an upward planar embedding with thin shape description $s$. Let also $N_1$ and $N_2$ be two subsets of $N$ with $N_2\subseteq N_1$ and $N_1\setminus N_2 \subseteq N_s$. Then, $G_{N_1}$ has an upward planar embedding 	$\mathcal{E}_{N_1}$ that respects $s$ for $N_s$ if and only if $G_{N_2}$ has an upward planar embedding 	$\mathcal{E}_{N_2}$ that respects $s$ for $N_s$, and embeddings $\mathcal{E}_{N_1}$ and $\mathcal{E}_{N_2}$ form an equivalence pair. Furthermore, the shape descriptions of $G_{N_1}$ and $G_{N_2}$ in $\mathcal{E}_{N_1}$ and $\mathcal{E}_{N_2}$, respectively, are the same.
	\end{lemma}
	\iflong \begin{proof}
		Suppose first that $G_{N_1}$ has an upward planar embedding 	$\mathcal{E}_{N_1}$ that respects $s$ for $N_s$. The restriction of $\mathcal{E}_{N_1}$ to $G_{N_2}$ is an upward planar embedding 	$\mathcal{E}_{N_2}$ of $G_{N_2}$ that respects $s$ for $N_s$. Moreover the components shared by  $G_{N_1}$ and $G_{N_2}$ have the same shape description and appear in the same order around $u$, that is, $\mathcal{E}_{N_1}$ and $\mathcal{E}_{N_2}$ form an equivalence pair.
		
		Suppose now that $G_{N_2}$ has an upward planar embedding 	$\mathcal{E}_{N_2}$ that respects $s$ for $N_s$. 
		Let $G_{\nu'_1}$, $G_{\nu'_2}$, $\dots$, $G_{\nu'_{k'}}$ be the sequence of the components of $G_{N_2}$ in the left to right order defined by $\mathcal{E}_{N_2}$.  Denote by $\mathcal{E}_\nu^2$ the upward planar embedding of component $G_\nu$ in $\mathcal{E}_{N_2}$, for $\nu\in N_2$. Let $f'_i$ be the face bounded by the right outer path of $\mathcal{E}^2_{\nu'_{i}}$ and by the left outer path of $\mathcal{E}^2_{\nu'_{i+1}}$, for $i=1,2,\dots,k'-1$, and let $f'_0$ be the outer face of $\mathcal{E}_{N_2}$.
		As $\mathcal{E}_{N_2}$ respects $N_s$, there exist consecutive components $G_{\nu'_\ell}$, $G_{\nu'_{\ell+1}}$, $\dots$, $G_{\nu'_{\ell+j}}$ of $G_{N_2}$ such that $\nu'_{\ell+i}\in N_s$  and $\mathcal{E}^2_{\nu'_{\ell+i}}$ has shape description $s$, for $i=0,\dots,j$. We now change the upward planar embedding $\mathcal{E}_{N_2}$ of $G_{N_2}$ to an upward planar embedding $\mathcal{E}_{N_1}$ of $G_{N_1}$ as follows. Note that, since $N_2 \subseteq N_1\cup N_s$, if component $G_\nu$ of $G_{N_1}$ does not belong to $G_{N_2}$ then $\nu\in N_s$. We replace components  $G_{\nu'_{\ell+i}}$ ($i=0,\dots,j$) with a sequence consisting of all the components $G_\nu$ of $G_{N_1}$ where $\nu\in N_s$. Let $\mathcal{G}_s$ denote the subgraph of $G_{N_1}$ that contains components $G_\nu$ where $\nu\in N_1\cap N_s$. To completely define $\mathcal{E}_{N_1}$ we need to define the upward planar embedding of all the components in $\mathcal{G}_s$, the labels of the angles at $u$ and $v$ inside the faces formed by consecutive elements of $\mathcal{G}_s$, and the labels of the angles at $u$ and $v$ inside the faces $f'_{\ell-1}$ and $f'_{\ell+j}$ (notice that the boundary of both faces is changed). For each component of $\mathcal{G}_s$ we choose an upward planar embedding such that the resulting shape description is $s$. We label $-1$ all the angles at $u$ and $v$ inside the faces formed by consecutive elements of $\mathcal{G}_s$. For the angles at $u$ and $v$ inside $f'_{\ell-1}$ and $f'_{\ell+j}$, we preserve the labels they had in $\mathcal{E}_{N_2}$. We denote by $\mathcal{E}_\nu^1$ the upward planar embedding of $G_\nu$ in $\mathcal{E}_{N_1}$ for $\nu\in N_1$. Note that $\mathcal{E}_\nu^1=\mathcal{E}_\nu^2$ for every $\nu\in N_1\setminus N_s$. In order to prove that the resulting labeling is a valid angle assignment we need to prove that properties \textbf{UP0}--\textbf{UP3} hold. Consider an angle $\alpha$ at a vertex $w$. If $\alpha$ is an angle that exists also in $\mathcal{E}_{N_2}$, then \textbf{UP0} holds for $\alpha$ in $\mathcal{E}_{N_1}$ as well; if $\alpha$ is an angle inside an internal face of $\mathcal{E}^1_{\nu}$ of a component $G_{\nu} \in \mathcal{G}_s$  then \textbf{UP0} holds for $\alpha$ in $\mathcal{E}_{N_1}$ because it holds in the upward planar embedding  chosen for $G_{\nu}$. If $\alpha$ is an angle inside the outer face of $\mathcal{E}^1_{\nu}$ at a vertex different from $u$ and $v$, then again, \textbf{UP0} holds in $\mathcal{E}_{N_1}$ because it holds in the upward planar embedding chosen for $G_{\nu}$. If $\alpha$ is an angle at $u$ or at $v$ inside a face $f_j$ formed by consecutive elements $G_{\nu}$ and $G_{\nu'}$ of $\mathcal{G}_s$ then the two edges defining $\alpha$ have the same orientation because the two upward planar embeddings $\mathcal{E}^1_{\nu}$ and $\mathcal{E}^1_{\nu'}$ have the same shape description; since $\alpha$ is labeled $1$, \textbf{UP0} holds. Finally, if $\alpha$ is an angle at $u$ or at $v$ inside $f'_{\ell-1}$ or $f'_{\ell+j}$, then the two edges defining $\alpha$ have the same orientation with the two edges defining the corresponding angle in $\mathcal{E}_{N_1}$ (because the leftmost component $G_{\nu'_\ell}$ of $G_{N_2}$ that is replaced by the sequence of components in $\mathcal{G}_s$ has the same shape description in $\mathcal{E}_{N_2}$ as the leftmost component of $\mathcal{G}_s$ in $\mathcal{E}_{N_1}$, and the rightmost component $G_{\nu'_{\ell+j}}$ of $G_{N_2}$ has the same shape description in $\mathcal{E}_{N_2}$ with the rightmost component of $\mathcal{G}_s$ in $\mathcal{E}_{N_1}$). Since the label at $\alpha$ is preserved, \textbf{UP0} holds because it holds in $\mathcal{E}_{N_2}$. Concerning \textbf{UP1} and \textbf{UP2}, we observe that for each vertex different from $u$ and $v$ \textbf{UP1} or \textbf{UP2} holds in $\mathcal{E}_{N_1}$ because it either holds in $\mathcal{E}_{N_2}$ or it holds in the upward planar embedding chosen for the components of $\mathcal{G}_s$. So, the only vertices for which we have to prove \textbf{UP1} or \textbf{UP2} are $u$ and $v$. To this aim we observe that either \textbf{UP1} or \textbf{UP2} holds for $u$ and $v$ in $\mathcal{E}_{N_2}$ and that the angles at $u$ and at $v$ that are in $\mathcal{E}_{N_1}$ but not in $\mathcal{E}_{N_2}$ are all labeled $-1$; thus, they cannot create a violation of \textbf{UP1} or of \textbf{UP2}. Concerning \textbf{UP3}, consider a face $f$ of $\mathcal{E}_{N_1}$. If $f$ is also a face of $\mathcal{E}_{N_2}$ then \textbf{UP3} holds for $f$ in $\mathcal{E}_{N_1}$ because it holds in $\mathcal{E}_{N_2}$. If $f$ is an internal face of $\mathcal{E}^1_{\nu}$ of a component $G_{\nu} \in \mathcal{G}_s$, then \textbf{UP3} holds for $\mathcal{E}_{N_1}$ because it holds for the upward planar embedding chosen for $G_{\nu}$. If $f$ is a face $f_j$ formed by consecutive elements $G_{\nu}$ and $G_{\nu'}$ of $\mathcal{G}_s$, then its boundary is formed by the right outer path of $\mathcal{E}^1_{\nu}$ and by the left outer path of $\mathcal{E}^1_{\nu'}$. Since $G_{\nu}$ and $G_{\nu'}$ have the same shape description that satisfies the thin conditions, we have $\tau_r(\mathcal{E}_{\nu}^1,u,v)=-\tau_l(\mathcal{E}_{\nu'}^1,u,v)$; moreover, the two angles at $u$ and $v$ inside $f_j$ are both labeled $-1$. Thus $n_1(f)-n_{-1}(f)=\tau_r(\mathcal{E}_{\nu}^1,u,v)+\tau_l(\mathcal{E}^1_{\nu'},u,v)-2=-2$ and \textbf{UP3} holds. Consider now the case when $f$ is the face $f'_{\ell-1}$ and assume that it is an internal face (a similar proof holds when it is the outer face). The boundary of $f'_{\ell-1}$ in $\mathcal{E}_{N_2}$ consists of the right outer path of  $\mathcal{E}^2_{\nu'_{\ell-1}}$ and of the left outer path of $\mathcal{E}^2_{\nu'_{\ell}}$; since \textbf{UP3} holds for $f'_{\ell-1}$ in $\mathcal{E}_{N_2}$, we have $n_1(f)-n_{-1}(f)=\tau_r(\mathcal{E}^2_{\nu'_{\ell-1}},u,v)+\tau_l(\mathcal{E}^2_{\nu'_{\ell}},u,v)+\alpha_u+\alpha_v=-2$, where $\alpha_u$ and $\alpha_v$ are the angles at $u$ and $v$ inside $f'_{\ell-1}$, respectively. The boundary of $f'_{\ell-1}$ in $\mathcal{E}_{N_1}$ consists of the right outer path of  $\mathcal{E}^1_{\nu'_{\ell-1}}$ and of the left outer path of $\mathcal{E}^1_{\nu}$, where $G_\nu$ is the first component of the set $\mathcal{G}_s$. Thus, we have $n_1(f)-n_{-1}(f)=\tau_r(\mathcal{E}^1_{\nu'_{\ell-1}},u,v)+\tau_l(\mathcal{E}^1_{\nu},u,v)+\alpha_u+\alpha_v$; since $\mathcal{E}^2_{\nu'_{\ell}}$ and $\mathcal{E}^1_{\nu}$ have the same shape description we have $\tau_l(\mathcal{E}^2_{\nu'_{\ell}},u,v)=\tau_l(\mathcal{E}^1_{\nu},u,v)$ and therefore \textbf{UP3} holds for $f'_{\ell-1}$ in $\mathcal{E}_{N_1}$. If $f$ is the face $f'_{\ell+j}$, the proof is analogous to the one for $f'_{\ell-1}$. This concludes the proof that properties \textbf{UP0}--\textbf{UP3} hold for $\mathcal{E}_{N_1}$. Now, since all $\mathcal{G}_s$ are consecutive in $\mathcal{E}_{N_1}$ and have the same shape description $s$, it follows that $\mathcal{E}_{N_1}$ respects $s$ for $N_s$. Additionally $\mathcal{E}_{N_1}$ and $\mathcal{E}_{N_2}$ form an equivalence pair. Hence we proved that if $G_{N_2}$ has an upward planar embedding 	$\mathcal{E}_{N_2}$ that respects $s$ for $N_s$, then  $G_{N_1}$ has an upward planar embedding 	$\mathcal{E}_{N_1}$ that respects $s$ for $N_s$ and the two embeddings form an equivalence pair.

		It remains to prove that the shape descriptions of $G_{N_1}$ in $\mathcal{E}_{N_1}$ and of $G_{N_2}$ in $\mathcal{E}_{N_2}$ are the same. In general, the shape description of a component in an upward planar embedding only depends on the outer boundary of that component. Let $f^1_0$ and $f^2_0$ be the outer faces of the two upward planar embeddings $\mathcal{E}_{N_1}$ and $\mathcal{E}_{N_2}$ described above. The boundary of $f^2_0$ consists of the left outer boundary of $\mathcal{E}^2_{\nu'_{1}}$ and by the right outer boundary of $\mathcal{E}^2_{\nu'_{k'}}$. If neither ${\nu'_{1}}$ nor ${\nu'_{k'}}$ belong to $N_s$, then the boundary of $f^1_0$ coincides with that of $f^2_0$, they have the same labeling and therefore $\mathcal{E}_{N_1}$ and $\mathcal{E}_{N_2}$  have the same shape description. If this is not the case then either ${\nu'_{1}}\in N_s$ or ${\nu'_{k'}}\in N_s$. 
		If ${\nu'_{1}}\in N_s$ the left boundary of $f^2_0$ consists of the left outer path of $\mathcal{E}^2_{\nu'_{1}}$ and the left boundary of $f^1_0$ consists of the left outer path of $\mathcal{E}^1_{\nu}$, where $G_\nu$ is a component  of $\mathcal{G}_s$. Since $\mathcal{E}^2_{\nu'_{1}}$ and $\mathcal{E}^1_\nu$ have the same shape description, they have the same left-turn-number. Similarly, if ${\nu'_{k'}}\in N_s$ the right outer paths of $f^2_0$ and $f^1_0$ have the same right-turn-numbers. Since the angles at $u$ and $v$ inside $f^2_0$ are the same as those inside $f^1_0$, $\mathcal{E}_{N_1}$ and $\mathcal{E}_{N_2}$  have the same shape description.
	\end{proof}\fi
\fi

\section{General Algorithm}\label{sec:general_algo}

\setboolean{longShapeDesc}{true}
\setboolean{longShapeSequence}{true}
\iflong
	Let $G$ be an $n$-vertex biconnected expanded DAG whose underlying undirected graph is planar and let $T$ be an SPQR-tree decomposition of $G$.
	Let also $\tau_{\min}$ and $\tau_{\max}$ be two integers such that $\tau_{\min}\leq \tau_{\max}$.
	 In this section, we present a general algorithm to compute the shape descriptions of all possible upward planar embeddings of $G$ with respect to $T$, and such that the left- and right-turn numbers of the induced upward planar embeddings of all pertinent graphs of $T$ are within the range $[\tau_{\min},\tau_{\max}]$. In case $G$ is not upward planar under these restrictions, the algorithm returns an empty set. 
	We visit the nodes of $T$ bottom-up and we compute for each node $\mu$ the set $\mathcal{F}_{\mu}$ of the shape descriptions of all possible upward planar embeddings of the pertinent graph $G_{\mu}$ of $\mu$.  We call $\mathcal{F}_{\mu}$  the \emph{feasible set} of $\mu$.
	The feasible set $\mathcal{F}_{\mu}$ is computed starting from the feasible sets $\mathcal{F}_{\nu_1}$, $\mathcal{F}_{\nu_2}$, $\ldots$, $\mathcal{F}_{\nu_k}$, where $\nu_1$, $\nu_2$, $\ldots$, $\nu_k$ are the children of $\mu$ in $T$. If at any point the feasible set $\mathcal{F}_{\mu}$ of node $\mu$ is empty, we conclude that $G$ is not upward planar (under the above restrictions) and the process returns an empty set, otherwise we continue the traversal of $T$. 
\fi

\ifshort
	Let $G$ be an $n$-vertex biconnected expanded DAG whose underlying  graph is planar and let $T$ be an SPQR-tree of $G$. Let $\tau_{\min}$ and $\tau_{\max}$ be two integers with $\tau_{\min}\leq \tau_{\max}$ and let $\tau=\tau_{\max}-\tau_{\min}+1$.
	We present a general algorithm to compute all possible shape descriptions of $G$ with respect to $T$, and such that the left- and right-turn numbers of all shape descriptions for all pertinent graphs of $T$ are in the range $[\tau_{\min},\tau_{\max}]$. 
	We visit the nodes of $T$ bottom-up and we compute for each node $\mu$ its \emph{feasible set} $\mathcal{F}_{\mu}$, i.e., the set of all  realizable shape descriptions of its pertinent graph $G_{\mu}$. 
If $\mathcal{F}_{\mu}=\emptyset$, the process stops and $G$ is not upward planar (under the above restrictions), otherwise we continue the traversal of~$T$. 
\fi 
	
	\iflong\paragraph*{Storing feasible sets}\fi
	\ifshort\subparagraph{Storing feasible sets.}\fi
	\iflong
	For each node $\mu$ of $T$ we associate a matrix $M(\mu)$ of size $(\tau_{\max}-\tau_{\min}+1)\times 5$ where the element $M(\mu)[i,j]$ of the matrix contains all shape descriptions for $G_{\mu}$ with left-turn-number $\tau_l=\tau_{\min}+i$ and right-turn-number $\tau_r=-\tau_l+j$. Note that by Lemma~\ref{lem:shape_desc_values}, the right-turn-number of $G_{\mu}$ can only take values in $[-c,-c+4]$, where $c$ is the left-turn-number $\tau_l$. 
	A shape description within $M(\mu)[i,j]$ is stored as a pair $\langle \lambda(\mathcal{E}_{\mu},u), \rho_l(\mathcal{E}_{\mu},u)\rangle$ as its first two values are implied by indices $i$ and $j$ and the other four values of the shape description can be derived from the indices and the stored values (recall that the values of a shape description are not all independent). The following lemma summarizes the space requirements and run-time of basic operations for feasible sets.
	\fi
	\ifshort
	For each node $\mu$ of $T$ we associate a matrix $M(\mu)$ of size $(\tau_{\max}-\tau_{\min}+1)\times 5$ where the element $M(\mu)[i,j]$ of the matrix contains all shape descriptions of $G_{\mu}$ with left turn-number $\tau_l=\tau_{\min}+i$ and right-turn-number $\tau_r=-\tau_l+j$. Note that by Lemma~\ref{lem:shape_desc_values}, $\tau_r$ can only take values in $[-\tau_l,-\tau_l+4]$.
	\fi
	\ifshort
		\begin{lemma}\label{lem:matrix_feasible}
				There are at most $18$ shape descriptions with given left- or right-turn-number.
			\end{lemma}
	\fi
	
	\iflong \begin{lemma} \label{lem:matrix_feasible}
	Let $\tau=\tau_{\max}-\tau_{\min}+1$. Then:
	\begin{enumerate}
	\item creating matrix $M(\mu)$ for all nodes $\mu$ of $T$ takes $\bigoh(|T|\tau)$ time and $\bigoh(|T|\tau)$ space;
	\item retrieving all shape descriptions of $M(\mu)$ takes $\bigoh(\tau)$ time;
	\item adding, deleting or finding a shape description in $M(\mu)$ requires $\bigoh(1)$ time;
	\item there are at most $18$ shape descriptions  with given left- or right-turn-number, and finding them takes $\bigoh(1)$ time.
	\end{enumerate}
	\end{lemma}
	\begin{proof}
	The first part of the lemma is obvious. We now argue that for specific values of $\tau_l(\mathcal{E}_{\mu},u,v)$ and $\tau_r(\mathcal{E}_{\mu},u,v)$ there is a bounded number of possible shape descriptions with the given left- and right-turn-numbers. Indeed, for the pair $\langle \lambda(\mathcal{E}_{\mu},u), \rho_l(\mathcal{E}_{\mu},u)\rangle$, $\lambda(\mathcal{E}_{\mu},u)$ takes a value from $\{-1,0,1\}$, while $\rho_l(\mathcal{E}_{\mu},u)\in \{in,out\}$. Hence the possible values for the pairs are in total six. Hence, each element $M(\mu)[i,j]$ of matrix $M(\mu)$ contains at most six pairs, and parts 2 and 3 of the lemma follows. 
	For the last part of the lemma, we first argue that the number of shape descriptions with given left- or right-turn-number are at most $18$. Indeed, if we specify  the values $\lambda(\mathcal{E}_{\mu},u)$, $\lambda(\mathcal{E}_{\mu},v)$, $\rho_l(\mathcal{E}_{\mu},u)$, and either the left- or the right-turn-number, then the remaining values of a shape description are uniquely determined. As there are $3^2\cdot 2=18$ ways to specify all of $\lambda(\mathcal{E}_{\mu},u)$, $\lambda(\mathcal{E}_{\mu},v)$ and $\rho_l(\mathcal{E}_{\mu},u)$, it follows that there are at most $18$ shape descriptions with specific left- or right-turn-number. Now, the shape descriptions with left-turn-number equal to $\tau_l$ are stored in row $i=\tau_l-\tau_{\min}$ of $M(\mu)$, which has length $5$. On the other hand, if the right-turn-number $\tau_r$ is given, by Lemma~\ref{lem:shape_desc_values} the right-turn-number can only take values in the interval $[\tau_r,-\tau_r+4]$, and therefore it suffices to consider the pairs of at most five elements of the matrix. In both cases the shape descriptions with given left- or right-turn-number can be found in $\bigoh(1)$ time.
	\end{proof}\fi

	We describe how to compute the feasible set $\mathcal{F}_{\mu}$ of a node $\mu$ of $T$ depending on its type.
	
	
	\iflong\paragraph*{Q-node}\fi
	\ifshort\subparagraph{Q-node.}\fi
	\iflong
	In this case, the pertinent graph $G_{\mu}$ is a single directed edge connecting poles $u$ and $v$. If the edge is directed from $u$ to $v$ then the feasible set consists of the tuple 
	$\shapeDesc{0}{0}{1}{1}{{out}}{{out}}{{in}}{{in}}$, 
	while if the edge is directed from $v$ to $u$ then it consists of the tuple 
	$\shapeDesc{0}{0}{1}{1}{{in}}{{in}}{{out}}{{out}}$.
	\fi
	\ifshort
		The pertinent graph $G_{\mu}$ is either the edge $(u,v)$ or $(v,u)$. Hence, the feasible set consists of the tuple 
		$\shapeDesc{0}{0}{1}{1}{{out}}{{out}}{{in}}{{in}}$ or  
		$\shapeDesc{0}{0}{1}{1}{{in}}{{in}}{{out}}{{out}}$, respectively.
	\fi
	
	\begin{lemma}\label{lem:Q_node_general}
		Let $\mu$ be a Q-node of $T$. The feasible set $\mathcal{F}_\mu$ can be computed in $\bigoh(1)$ time.
	\end{lemma}
	
	\iflong\paragraph*{S-node}\fi
	\ifshort\subparagraph{S-node.}\fi 
	\iflong
	Recall that node $\mu$ has exactly two children in $T$, $\nu_1$ and $\nu_2$. Let $u_1$, $v_1$ be the poles of $G_{\nu_1}$ and  $u_2$, $v_2$ be the poles of $G_{\nu_2}$. The poles of $G_\mu$ are $u_1$ and $v_2$, that is $v_1=u_2$. Let 
	$\shapeDesc{\tau_l(\mathcal{E}_{\nu_1},u_1,v_1)}{\tau_r(\mathcal{E}_{\nu_1},u_1,v_1)}{\lambda(\mathcal{E}_{\nu_1},u_1)}{\lambda(\mathcal{E}_{\nu_1},v_1)}{\rho_l(\mathcal{E}_{\nu_1},u_1)}{\rho_r(\mathcal{E}_{\nu_1},u_1)}{\rho_l(\mathcal{E}_{\nu_1},v_1)}{\rho_r(\mathcal{E}_{\nu_1},v_1)}$
	be a tuple in $\mathcal{F}_{\nu_1}$ and  let 
	$\shapeDesc{\tau_l(\mathcal{E}_{\nu_2},u_2,v_2)}{\tau_r(\mathcal{E}_{\nu_2},u_2,v_2)}{\lambda(\mathcal{E}_{\nu_2},u_2)}{\lambda(\mathcal{E}_{\nu_2},v_2)}{\rho_l(\mathcal{E}_{\nu_2},u_2)}{\rho_r(\mathcal{E}_{\nu_2},u_2)}{\rho_l(\mathcal{E}_{\nu_2},v_2)}{\rho_r(\mathcal{E}_{\nu_2},v_2)}$
	be a tuple in $\mathcal{F}_{\nu_2}$.
	
	Let $\alpha_l$ be the angle at the common pole that is created by the two left outer paths of $\mathcal{E}_{\nu_1}$ and $\mathcal{E}_{\nu_2}$ (see Fig.~\ref{fig:series-composition}). Similarly, let $\alpha_r$ be the angle at the common pole that is created by the two right outer paths of $\mathcal{E}_{\nu_1}$ and $\mathcal{E}_{\nu_2}$. 
	We assign the labels $\lambda_l$ and $\lambda_r$ to $\alpha_l$ and $\alpha_r$ respectively as follows:
	$\lambda_l=0$ if $\rho_l(\mathcal{E}_{\mu},v_1)\neq \rho_l(\mathcal{E}_{\mu},u_2)$ otherwise $\lambda_l\in\{-1,1\}$, and 
	$\lambda_r=0$ if $\rho_r(\mathcal{E}_{\mu},v_1)\neq \rho_r(\mathcal{E}_{\mu},u_2)$ otherwise $\lambda_r\in\{-1,1\}$.
	Note that we exclude the case $\lambda_l=\lambda_r=1$ as this would imply that two angles on the same vertex of $G$ are large, which is not possible.
	Given the two values for $\lambda_l$ and $\lambda_r$, we construct a candidate tuple 
	$\shapeDesc{\tau_l(\mathcal{E}_{\mu},u_1,v_2)}{\tau_r(\mathcal{E}_{\mu},u_1,v_2)}{\lambda(\mathcal{E}_{\mu},u_1)}{\lambda(\mathcal{E}_{\mu},v_2)}{\rho_l(\mathcal{E}_{\mu},u_1)}{\rho_r(\mathcal{E}_{\mu},u_1)}{\rho_l(\mathcal{E}_{\mu},v_2)}{\rho_r(\mathcal{E}_{\mu},v_2)}$ with
	\begin{itemize}
	\item $\tau_l(\mathcal{E}_{\mu},u_1,v_2)=\tau_l(\mathcal{E}_{\nu_1},u_1,v_1)+\tau_l(\mathcal{E}_{\nu_2},u_2,v_2)+\lambda_l$,
	\item $\tau_r(\mathcal{E}_{\mu},u_1,v_2)=\tau_r(\mathcal{E}_{\nu_1},u_1,v_1)+\tau_r(\mathcal{E}_{\nu_2},u_2,v_2)+\lambda_r$,
	\item $\lambda(\mathcal{E}_{\mu},u_1)=\lambda(\mathcal{E}_{\nu_1},u_1)$,
	\item $\lambda(\mathcal{E}_{\mu},v_2)=\lambda(\mathcal{E}_{\nu_2},v_2)$,
	\item $\rho_l(\mathcal{E}_{\mu},u_1)=\rho_l(\mathcal{E}_{\nu_1},u_1)$,
	\item $\rho_r(\mathcal{E}_{\mu},u_1)=\rho_r(\mathcal{E}_{\nu_1},u_1)$,
	\item $\rho_l(\mathcal{E}_{\mu},v_2)=\rho_l(\mathcal{E}_{\nu_2},v_2)$,
	\item $\rho_r(\mathcal{E}_{\mu},v_2)=\rho_r(\mathcal{E}_{\nu_2},v_2)$.
	\end{itemize}
	
	We accept the candidate tuple if and only if
	$\tau_l(\mathcal{E}_{\mu},u_1,v_2)+\tau_r(\mathcal{E}_{\mu},u_1,v_2)+\lambda(\mathcal{E}_{\mu},u_1)+\lambda(\mathcal{E}_{\mu},v_2)=2$ holds, that is if \iflong Corollary~\ref{cor:dependence-parameters} \fi \ifshort Observation~\ref{cor:dependence-parameters} \fi holds for the outer face of $\mathcal{E}_{\mu}$.
	
	\begin{figure}[htb]
		\centering
		\begin{subfigure}{.45\textwidth}
			\centering
			\includegraphics[width=\columnwidth, page=2]{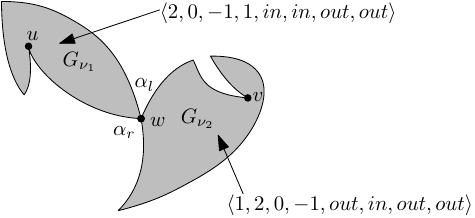}
		\end{subfigure}
		\caption{Series composition. The resulting shape description is $\shapeDesc{2}{2}{-1}{-1}{in}{in}{out}{out}$ \label{fig:series-composition}}
	\end{figure}

	\fi
	
	\ifshort
		Let $\nu_1$ and $\nu_2$ be the children of $\mu$, with poles $u$, $w$ and $w$, $v$, respectively. Let $\shapeDesc{\tau^1_l}{\tau^1_r}{\lambda^1_u}{\lambda^1_w}{\rho^1_{l,u}}{\rho^1_{r,u}}{\rho^1_{l,w}}{\rho^1_{r,w}}$
		be a tuple in $\mathcal{F}_{\nu_1}$ and let 
		$\shapeDesc{\tau^2_l}{\tau^2_r}{\lambda^2_w}{\lambda^2_v}{\rho^2_{l,w}}{\rho^2_{r,w}}{\rho^2_{l,v}}{\rho^2_{r,v}}$
		be a tuple in $\mathcal{F}_{\nu_2}$.
		Let $\alpha_l$ (resp. $\alpha_r$) be the angle at $w$ created by the two left (resp. right) outer paths of $G_{\nu_1}$ and $G_{\nu_2}$ (see Fig.~\ref{fig:series-composition}).
		We assign the labels $\lambda_l$ and $\lambda_r$ to $\alpha_l$ and $\alpha_r$ respectively as follows:
		$\lambda_l=0$ if $\rho^1_{l,w}\neq \rho^2_{l,w}$ otherwise $\lambda_l\in\{-1,1\}$, and 
		$\lambda_r=0$ if $\rho^1_{r,w}\neq \rho^2_{r,w}$ otherwise $\lambda_r\in\{-1,1\}$.
		Note that, by \textbf{UP1} it must be $\lambda_l+\lambda_r < 2$. For all possible values of $\lambda_l$ and $\lambda_r$ satisfying the previous constraints, we construct a candidate tuple 
		$\shapeDesc{\tau_l}{\tau_r}{\lambda_u}{\lambda_v}{\rho_{l,u}}{\rho_{r,u}}{\rho_{l,v}}{\rho_{r,v}}$ with: (i) $\tau_l=\tau^1_l+\tau^2_l+\lambda_l$, (ii) $\tau_r=\tau^1_r+\tau^2_r+\lambda_r$, (iii) $\lambda_u=\lambda^1_u$, (iv) $\lambda_v=\lambda^2_v$, (v) $\rho_{l,u}=\rho^1_{l,u}$, (vi) $\rho_{r,u}=\rho^1_{r,u}$, (vii) $\rho_{l,v}=\rho^2_{l,v}$, (viii) $\rho_{r,v}=\rho^2_{r,v}$. We accept the candidate tuple if and only if it satisfies \iflong Corollary~\ref{cor:dependence-parameters}.\fi \ifshort Observation~\ref{obs:dependence-parameters}.\fi
		
		\begin{figure}[htb]
			\centering
			\begin{subfigure}{.45\textwidth}
				\centering
				\includegraphics[width=\columnwidth, page=2]{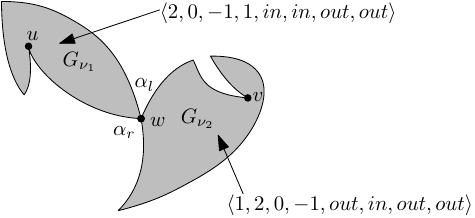}
			\end{subfigure}
			\caption{Series composition. The resulting shape description is $\shapeDesc{2}{2}{-1}{-1}{in}{in}{out}{out}$\label{fig:series-composition}.}
		\end{figure}
	
	\fi

	\iflong \begin{lemma} \fi \ifshort \begin{lemma}\fi\label{lem:S_node_general}
		Let $\mu$ be an S-node of $T$ with children $\nu_1$ and $\nu_2$. The feasible set $\mathcal{F}_\mu$ can be computed in $\bigoh(\tau+|\mathcal{F}_{\nu_1}|\cdot|\mathcal{F}_{\nu_2}|)$ time.
	\end{lemma}
	\iflong \begin{proof}
	By Lemma~\ref{lem:matrix_feasible}, getting the shape descriptions of each feasible set $\mathcal{F}_{\nu_1}$ and $\mathcal{F}_{\nu_2}$ takes $\bigoh(\tau)$ time. There are $|\mathcal{F}_{\nu_1}|\cdot|\mathcal{F}_{\nu_2}|$ pairs of shape descriptions for $\mathcal{E}_{\nu_1}$ and $\mathcal{E}_{\nu_2}$, and for each pair, the candidate tuple is constructed in $\bigoh(1)$ time. Finally, by Lemma~\ref{lem:matrix_feasible} storing each accepted tuple requires $\bigoh(1)$ time. Hence, the feasible set $\mathcal{F}_\mu$ can be computed in $\bigoh(\tau+|\mathcal{F}_{\nu_1}|\cdot|\mathcal{F}_{\nu_2}|)$ time, as claimed.
	\end{proof}\fi

	\renewcommand{\*}{\pmb{^*}}
	\newcommand{\+}{\pmb{^+}}
	
	\iflong\paragraph*{P-node}\fi
	\ifshort\subparagraph{P-node.}\fi	
	\iflong
	Let $\mu$ be a P-node with poles $u$ and $v$ and children $\nu_1, \nu_2 \dots, \nu_{k}$. Let $\nu'_1, \nu'_2 \dots, \nu'_{k'}$ be a subset of the children of $\mu$ and let $G'_{\mu}$ be the subgraph of $G_{\mu}$ consisting of the union of $G_{\nu'_1}$, $G_{\nu'_2}$, $\dots$, $G_{\nu'_{k'}}$. Consider an upward planar embedding $\mathcal{E}'_\mu$ of $G'_{\mu}$ and assume that the pertinent graphs $G_{\nu'_1}$, $G_{\nu'_2}$, $\dots$, $G_{\nu'_{k'}}$ appear in this order from left to right around $u$. Denote by $\mathcal{E}_{\nu'_i}$ the $uv$-external upward planar embedding of $G_{\nu'_i}$, for $i=1,2,\dots,k'$. Then the outer boundary of $\mathcal{E}'_{\mu}$ is formed by the left outer path of $\mathcal{E}_{\nu'_1}$ and the right outer path of $\mathcal{E}_{\nu'_{k'}}$. 
	Denote by $f'_i$ the face bounded by the right outer path of $\mathcal{E}_{\nu'_i}$ and by the left outer path of $\mathcal{E}_{\nu'_{i+1}}$, for $i=1,2,\dots,k'-1$, and let $f'_0$ be the outer face of $G'_{\mu}$. Denote by $S'$ the sequence of shape descriptions of the pertinent digraphs of $\nu'_1, \nu'_2 \dots, \nu'_{k'}$ in the order $G_{\nu'_1}$, $G_{\nu'_2}$, $\dots$, $G_{\nu'_k}$. The sequence $S'$ is the \emph{shape sequence} of $G'_{\mu}$ with respect to the upward planar embedding $\mathcal{E}'_\mu$. To describe the sequence $S'$ we use the following notation: $a\*$ denotes a subsequence of $S'$ consisting of 0 or more elements equal to $a$; $a\+$ denotes a subsequence of $S'$ consisting of 1 or more elements equal to $a$. Not all sequences of shape descriptions of $\mathcal{E}_{\nu'_1}$, $\mathcal{E}_{\nu'_2}$, $\dots$, $\mathcal{E}_{\nu'_{k'}}$ are shape sequences of $G'_{\mu}$. For $G'_{\mu}$ we say that a shape description $s'$ \emph{corresponds} to a shape sequence $S'$ if there exists an upward planar embedding of $G'_{\mu}$ with shape description $s'$ and whose shape sequence is $S'$. Notice that, the same shape sequence might correspond to different upward planar embeddings of $G'_{\mu}$, even if the order of $G_{\nu'_1}$, $G_{\nu'_2}$, $\dots$, $G_{\nu'_{k'}}$ is the same (see Fig.~\ref{fig:shape_seq_desc_diff}).
	Let $S$ be a sequence of shape descriptions; the \emph{reduced sequence} of $S$ is the sequence of shape descriptions obtained from $S$ by replacing each maximal subsequence $a\+$ of $S$ with the single element $a$. We define the \emph{size} of $S$ as the number of elements in its reduced sequence. 
	\fi
	\ifshort
	Let $\mu$ be a P-node with poles $u$ and $v$ and $k$ children $\nu_1, \nu_2 \dots, \nu_{k}$. Let $N'$  be a subset of the children of $\mu$ and let $G'_{\mu}$ be the subgraph of $G_{\mu}$ consisting of components $G_{\nu'}$ for $\nu' \in N'$. Consider a $uv$-external upward planar embedding $\mathcal{E}'_\mu$ of $G'_{\mu}$.
	Denote by $S'$ the sequence of shape descriptions of the components of $G'_{\mu}$ in the clockwise order in which they appear around $u$ starting from the outer face. The sequence $S'$ is the \emph{shape sequence} of $G'_{\mu}$ with respect to $\mathcal{E}'_\mu$. To describe $S'$ we write: $a\*$ (resp. $a\+$) to denote a subsequence of $S'$ consisting of 0 (resp. 1) or more elements equal to $a$. We say that a shape description $s'$ of $G'_{\mu}$ \emph{corresponds} to $S'$ if there exists an upward planar embedding of $G'_{\mu}$ with shape description $s'$ and whose shape sequence is $S'$. 
	Let $S$ be a sequence of shape descriptions; the \emph{reduced sequence} of $S$ is obtained from $S$ by replacing each maximal subsequence $a\+$ of $S$ with the single element $a$. The \emph{size} of $S$ is the number of elements in its reduced sequence. 
	\fi
	
	\iflong
	\begin{figure}[!t]
		\centering
		{\centering
		\begin{subfigure}{.15\textwidth}
			\includegraphics[page=8]{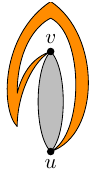}
			\subcaption{}
			\label{fig:shape_seq_desc_1}
		\end{subfigure}
		}
		\hfil
		{\centering
		\begin{subfigure}{.15\textwidth}
			\includegraphics[page=9]{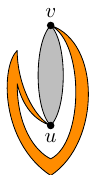}
			\subcaption{}
			\label{fig:shape_seq_desc_2}
		\end{subfigure}
		}
		\caption{Different shape descriptions corresponding to the shape sequence $[s_1,s_2]$, with $s_1=\shapeDesc{0}{0}{1}{1}{out}{out}{in}{in}$ (gray shaded) and $s_2=\shapeDesc{-2}{2}{1}{1}{out}{out}{in}{in}$ (orange shaded). The shape descriptions are (a)~$\shapeDesc{0}{2}{1}{-1}{out}{out}{in}{in}$ and (b)~$\shapeDesc{0}{2}{-1}{1}{out}{out}{in}{in}$.}
		\label{fig:shape_seq_desc_diff}
	\end{figure}
	\fi
	
	\iflong \begin{lemma} \label{lem:p_nodes_seq_valid_test}
		Let $S$ be a sequence of shape descriptions from the feasible sets of $G_{\nu'_1}$, $G_{\nu'_2}$, $\dots$, $G_{\nu'_{k'}}$. We can decide whether $S$ is a shape sequence of $G'_{\mu}$ and compute the corresponding shape descriptions of  $G'_{\mu}$ in $\bigoh(r^3)$ time, where $r$ is the size of $S$. Furthermore there are $\bigoh(r^2)$ computed shape descriptions of  $G'_{\mu}$.
	\end{lemma}
	\begin{proof}
	Assume first that $S$ is reduced. We want to use {\textbf {UP0}-\textbf{UP3}} in order to compute all possible corresponding shape descriptions of $G'_{\mu}$. Let $\mathcal{E}'_{\mu}$ be an upward planar embedding of $G'_\mu$ with shape sequence $S$. Let $G_{\nu'_i}$ be a component of $G'_{\mu}$. Denote by $\mathcal{E}_{\nu'_i}$ its $uv$-external upward embedding induced by $\mathcal{E}'_\mu$ and let $s_i=\shapeDesc{\tau_l(s_i)}{\tau_r(s_i)}{\lambda_u(s_i)}{\lambda_v(s_i)}{\rho_l(u,s_i)}{\rho_r(u,s_i)}{\rho_l(v,s_i)}{\rho_r(v,s_i)}$ be its shape description in $S$. 
	For all faces $f'_i$ of $\mathcal{E}'_{\mu}$, \textbf{UP3} gives $\tau_r(s_i)+\tau_l(s_{i+1})+\lambda(\alpha_u(f'_i))+\lambda(\alpha_v(f'_i))=-2$ holds if $i\neq 0$, and $\tau_r(s_{k'})+\tau_l(s_1)+\lambda(\alpha_u(f'_0))+\lambda(\alpha_v(f'_0))=2$ for the outer face $f'_0$. For the labels at the angles $\alpha_u(f'_i)$ and $\alpha_v(f'_i)$, we have that $\lambda(\alpha_w(f'_i))=0$ (i.e. $\alpha_w(f'_i)$ is a flat angle) if and only if $\rho_r(w,s_i)\neq \rho_l(w,s_{i+1})$, where $w$ is a pole of $\mu$. Hence the label of $\alpha_w(f'_i)$ is not determined only if it is a switch angle and its label can be either $-1$ or $1$. Note that by \textbf{UP2} there is exactly one angle with label $1$. Then, if $u$ is a switch vertex, the possible labelings of all angles $\alpha_u(f'_i)$ are $r+1$; namely there are $r$ faces $f'_i$ that can have the unique label $1$, or none of them has label $1$.  Similarly if $v$ is a switch vertex, the possible labelings of all angles $\alpha_v(f'_i)$ are $r+1$. This implies that there exist in total at most $(r+1)^2$ different labelings for the angles at $u$ and $v$ inside all faces $f'_i$ of $G'_\mu$. For each labeling, we want to test whether \textbf{UP3} holds for all faces $f'_i$ of $G'_\mu$. This can be done in constant time for each face, and thus in $\bigoh(r)$ time for all faces. Still we need to guarantee that \textbf{UP1} and \textbf{UP2} also hold. 
	
	Assume that pole $w$ is a switch vertex of $G'_\mu$. Note that if the shape description for a component $G_{\nu'_i}$ has label at pole $w$ equal to $-1$, that is $\lambda(w,\mathcal{E}_{\nu'_i})=-1$, then the large angle of $w$ is inside an interior face of $\mathcal{E}_{\nu'_i}$. Then $\lambda(\alpha_w(f'_i))=-1$ must hold for all faces $f'_i$ of $\mathcal{E}'_{\mu}$,  and $\lambda(w,\mathcal{E}_{\nu'_j})=1$ for all components $G_{\nu'_j}$ different from $G_{\nu'_i}$. Otherwise, $\lambda(w,\mathcal{E}_{\nu'_i})=1$ holds for all components $G_{\nu'_i}$, and the large angle of $w$ must be inside a face $f'_i$ of $\mathcal{E}'_{\mu}$, that is $\lambda(\alpha_w(f'_i))=1$ and $\lambda(\alpha_w(f'_j))=-1$ for all other faces $f'_j$. Note that all shape descriptions used for the components of $G'_\mu$ belong to their feasible sets. Then  all angles of $w$ in the interior of all components (except possibly for one angle in the interior of one $\mathcal{E}_{\nu'_i}$) have label $-1$. This implies that the above restrictions are sufficient to test whether \textbf{UP1} holds for $w$, and this can be done in $\bigoh(r)$ time for each labeling.
	Now, assume that pole $w$ is a non-switch vertex of $G'_\mu$. Note that if the shape description for a component $G_{\nu'_i}$ has label at pole $w$ equal to $-1$, that is $\lambda(w,\mathcal{E}_{\nu'_i})=-1$, then both flat angles of $u$ are inside 
	$\mathcal{E}_{\nu'_i}$, while if the label equals $0$ that is $\lambda(w,\mathcal{E}_{\nu'_i})=0$, then one flat angle of $u$ is inside $\mathcal{E}_{\nu'_i}$. In the first case, we have that $\lambda(\alpha_w(f'_i))=-1$ must hold for all faces $f'_i$,  and $\lambda(w,\mathcal{E}_{\nu'_j})=1$ for all components $G_{\nu'_j}$ different from $G_{\nu'_i}$. In the second case, $G_{\nu'_i}$ is the unique special component of pole $w$, and $\lambda(w,\mathcal{E}_{\nu'_j})=1$ must hold for all other components $G_{\nu'_j}$, while $\lambda(\alpha_w(f'_i))=1$ must hold for all faces $f'_i$, except for one that must have label $0$. In both cases, since all shape descriptions used for the components of $G'_\mu$ belong to their feasible sets, the switch angles of $w$ in the interior of all components have label $-1$. Hence the above restrictions are sufficient to test whether \textbf{UP2} holds for $w$, and this can be done in $\bigoh(r)$ time for each labeling.

	Overall, there exist $\bigoh(r^2)$ possible labelings for the angles $\alpha_u(f'_i)$ and $\alpha_v(f'_i)$ of $u$ and $v$ respectively inside all faces $f'_i$ of $\mathcal{E}'_\mu$. Testing whether \textbf{UP1}-\textbf{UP3} holds for each of them requires $\bigoh(r)$ time. If this is the case, by Theorem~\ref{th:upward-conditions} $G'_\mu$ is upward planar and the shape description of $\mathcal{E}'_\mu$ can be computed in constant time from the first element of $S$, its last element, and the labels at $u$ and $v$ chosen for the outer face $f'_0$. Thus, the overall complexity of deciding whether a sequence $S$ is a shape sequence for $G'_{\mu}$ is $\bigoh(r^3)$. In the same time all $\bigoh(r^2)$ shape descriptions for $G'_\mu$ that correspond to $S$ are computed.

	Assume now that $S$ is not reduced. We describe how this case can be brought down to the previous one. Let $s\+$ be a maximal subsequence of $S$. Let $G_{\nu'_j}$, $G_{\nu'_{j+1}}$, $\dots$, $G_{\nu'_{j+\ell}}$ be the components of $G'_{\mu}$ that correspond to the subsequence $s\+$ and let $G^+_{\mu}$ be the subgraph of $G'_{\mu}$ consisting of these components. Let $S^-$ be the sequence obtained from $S$ by replacing $s\+$ with $s$ and let $G^-_{\mu}$ the graph obtained from $G'_{\mu}$ by replacing $G^+_{\mu}$ with $G_{\nu'_j}$. 
	
	First we check under which conditions $s\+$ is a shape sequence for $G^+_{\mu}$. Let $G_{\nu'_i}$ and $G_{\nu'_{i+1}}$ be two consecutive components of $G'_{\mu}$ that both have shape description $s$ in $S$. The internal face $f'_i$ of $\mathcal{E}'_{\mu}$ is defined by the right outer path of $\mathcal{E}_{\nu'_i}$ and the left outer path of $\mathcal{E}_{\nu'_{i+1}}$. Let $\alpha_u(f'_i)$ and $\alpha_v(f'_i)$ be the angles at $u$ and $v$ respectively inside $f'_i$. If $S$ is a shape sequence of $G'_{\mu}$, {\textbf{UP3}} must be satisfied for $f'_i$, that is $\tau_r(s)+\tau_l(s)+\lambda(\alpha_u(f'_i))+\lambda(\alpha_v(f'_i))=-2$. Recall that by Lemma~\ref{lem:shape_desc_values} the right-turn-number of $s$ satisfies $\tau_r(s)=-\tau_l(s)+h$ with $h=0,1,2,3,4$. Also the labels $\lambda(\alpha_u(f'_i))$ and $\lambda(\alpha_v(f'_i))$ take values in $\{-1,0,1\}$. The previous equation gives $h+\lambda(\alpha_u(f'_i))+\lambda(\alpha_v(f'_i))=-2$ which is only satisfied if $h=0$ and $\lambda(\alpha_u(f'_i))=\lambda(\alpha_v(f'_i))=-1$. Hence the shape description $s=\shapeDesc{\tau_l(s)}{\tau_r(s)}{\lambda_u(s)}{\lambda_v(s)}{\rho_l(u,s)}{\rho_r(u,s)}{\rho_l(v,s)}{\rho_r(v,s)}$ must satisfy the thin conditions presented in Section~\ref{sec:comp_shapes}.
	Checking whether $s$ satisfies the thin conditions can be done in constant time. Let $N_1$ be the set $\{\nu'_1,\dots,\nu'_{k'}\}$ of children of $\mu$, $N_s=\{\nu'_j,\dots,\nu'_{j+\ell}\}$ and $N_2=N_1\setminus N_s\cup \{\nu'_j\}$. Note that $G'_\mu$ and $G^-_\mu$ are the graphs $G_{N_1}$ and $G_{N_2}$ of Lemma~\ref{lem:parallel_reduce}. Also, an upward embedding $\mathcal{E}_{N_1}$ of $G_{N_1}$ with shape sequence $S$, and an upward embedding $\mathcal{E}_{N_2}$ of $G_{N_2}$ with shape sequence $S^-$, both respect $s$ for the set $N_s$, and they form an equivalence pair. Hence, by Lemma~\ref{lem:parallel_reduce}, if $s$ satisfies the thin conditions, then every shape description for $G'_{\mu}$ that corresponds to $S$ is also a shape description for $G^-_{\mu}$ that corresponds to $S^-$ and vice versa. 
	
	We try to replace every maximal subsequence $s\+$ of $S$ with $s$. If this is not possible, that is $s$ does not satisfy the thin conditions, then $S$ is not a shape sequence of $G'_{\mu}$. Otherwise, we have a subgraph $G^-_{\mu}$ of $G'_{\mu}$ and a sequence $S^-$ which is reduced. By using the procedure described at the beginning of the proof, we can compute the shape descriptions of $G^-_{\mu}$ that corresponds to $S^-$ in $\bigoh(r^3)$ time.
	\end{proof}

	Let $S'$ be the shape sequence of $G'_{\mu}$ with respect to  $\mathcal{E}'_\mu$. Let $\nu$ be a child of $G_{\mu}$ that is not in the set $\{\nu'_1, \nu'_2 \dots, \nu'_{k'}\}$, let $s$ be a shape description from the feasible set of $G_{\nu}$, and let $G''_{\mu}$ be the subgraph of $G_{\mu}$ consisting of the union of $G_{\nu}$ and $G'_{\mu}$. We say that $S'$ \emph{can be extended} with $s$ to a shape sequence $S''$ of $G''_{\mu}$ if $S''$ is a shape sequence of $G''_{\mu}$, $s$ belongs to $S''$, and removing $s$ from $S''$ we obtain $S'$. Notice that $s$ can be either the first, or the last, or an intermediate element of $S''$.

	\begin{lemma} \label{lem:p_nodes_seq_extend_test}
		Let $S'$ be a shape sequence of $G'_{\mu}$. Given a shape description $s$ from the feasible set of  $G_{\nu}$, we can decide whether $S'$ can be extended with $s$ to a shape sequence $S''$ of $G''_{\mu}$ and compute the corresponding shape descriptions for $G''_{\mu}$ in $\bigoh(r^4)$ time, where $r$ is the size of $S'$.
	\end{lemma}
	\begin{proof}
	The shape description $s$ can be placed either at the beginning, at the end or between two consecutive elements of $S'$. In the first two cases, $S''=[s,S']$ or $S''=[S',s]$. The shape descriptions that correspond to $S''$ can be computed using Lemma~\ref{lem:p_nodes_seq_valid_test} in $\bigoh(r^3)$ time each. For the last case, let $a$ and $b$ be the  two consecutive elements of $S'$ where $s$ is placed. If $a\neq b$ there are at most $r-1$ possible placements for $s$, that is, $r-1$ different sequences $S''$. Note that if $s=a$ or $s=b$, then $S''$ might be the same as $S'$. On the other hand, if $a=b$ there are at most $r$ such sequences $S''$. In fact, arguing similarly with the first part of the proof of Lemma~\ref{lem:p_nodes_seq_valid_test}, if $s\neq a$ the sequence $S''$ is not a shape sequence of $G''_{\mu}$ and if $s=a$ then $S''$ is the same sequence with $S'$. By Lemma~\ref{lem:p_nodes_seq_valid_test}, we can test in $\bigoh(r^3)$ time if each produced sequence $S''$ is a shape sequence of  $G''_{\mu}$ or not. The total require time is $\bigoh(r^4)$ as claimed.
	\end{proof}\fi

	\ifshort
	 
	\begin{lemma}\label{lem:p_nodes_seq_valid_test}
		Let $S'$ be a sequence of shape descriptions from the feasible sets of every $G_{\nu'}$, with $\nu' \in N'$. We can decide whether $S'$ is a shape sequence of $G'_{\mu}$ and compute the corresponding shape descriptions of  $G'_{\mu}$ in $\bigoh(r^3)$ time, where $r$ is the size of $S'$. Furthermore there are $\bigoh(r^2)$ computed shape descriptions of  $G'_{\mu}$.
	\end{lemma}

	Let $\nu$ be a child of $G_{\mu}$ with $\nu \not \in N'$, let $s$ be a shape description of $G_{\nu}$, and let $G''_{\mu}$ be the union of $G_{\nu}$ and $G'_{\mu}$. We say that $S'$ \emph{can be extended} with $s$ to a shape sequence $S''$ of $G''_{\mu}$ if $S''$ is a shape sequence of $G''_{\mu}$, $s$ belongs to $S''$, and removing $s$ from $S''$ we obtain $S'$.
	
	\begin{lemma}\label{lem:p_nodes_seq_extend_test}
		Let $S'$ be a shape sequence of $G'_{\mu}$. Given a shape description $s$ of $G_{\nu}$, we can decide whether $S'$ can be extended with $s$ to a shape sequence $S''$ of $G''_{\mu}$ and compute the corresponding shape descriptions of $G''_{\mu}$ in $\bigoh(r^4)$ time, where $r$ is the size of $S'$.
	\end{lemma}
	\fi
	
	\iflong
	Suppose that $G_{\mu}$ is upward planar and consider an upward planar embedding $\mathcal{E}_\mu$ of $G_{\mu}$. We first remove the special components of $u$ and $v$ and the normal components $G_{\nu}$ whose $uv$-external upward planar embedding $\mathcal{E}_{\nu}$ induced by $\mathcal{E}_{\mu}$ has shape description with $\lambda(\mathcal{E}_{\nu},u)=-1$ or $\lambda(\mathcal{E}_{\nu},v)=-1$. 
	Denote by $G'_{\mu}$ the subgraph of $G_{\mu}$ obtained after this removal. We say that $G'_{\mu}$ is the \emph{thin subgraph} of $G_{\mu}$ with respect to $\mathcal{E}_\mu$. Denote by $\mathcal{E}'_{\mu}$ the upward planar embedding of $G'_{\mu}$
	Let $\nu'_1, \nu'_2 \dots, \nu'_{k'}$ be the children of $\mu$ corresponding to the remaining components, that is, normal components of both $u$ and $v$ where $\lambda(\mathcal{E}_{\nu'_i},u)=\lambda(\mathcal{E}_{\nu'_i},v)=1$ for $i=1,\cdots,k'$. Assume that the pertinent graphs $G_{\nu'_1}$, $G_{\nu'_2}$, $\dots$, $G_{\nu'_{k'}}$ appear in this order from left to right around $u$, such that the outer boundary of $\mathcal{E}'_{\mu}$ is formed by the left outer path of $\mathcal{E}_{\nu'_1}$ and the right outer path of $\mathcal{E}_{\nu'_{k'}}$.
	\fi
	\ifshort
	Suppose that $G_{\mu}$ is upward planar and consider a $uv$-external upward planar embedding $\mathcal{E}_\mu$ of $G_{\mu}$. We  remove the special components of $u$ and $v$ and the normal components $G_{\nu}$ whose shape description labels the angle at $u$ or $v$ with $-1$. There are at most two such components, as each one labels an internal angle at a pole with $1$.
	Let $G'_{\mu}$ be the subgraph of $G_{\mu}$ obtained after this removal; $G'_{\mu}$ is the \emph{thin subgraph} of $G_{\mu}$ with respect to $\mathcal{E}_\mu$.
	\fi
	\iflong 
	We first present few technical lemmata that will be used later in this section.
	
	\begin{lemma} \label{lem:at_most_two_comps}
	The thin subgraph of $G_{\mu}$ with respect to an upward planar embedding $\mathcal{E}_\mu$ contains $k'$ components of $G_\mu$, where $k-2\leq k'\leq k$.
	\end{lemma}
	\begin{proof}
	As $G_\mu$ has $k$ components, the thin subgraph $G'_\mu$ contains $k'\leq k$ components of $G_\mu$. In order to prove that $k'\geq k-2$ we argue that there are most two components of $G_\mu$ that don't belong to the thin subgraph $G'_\mu$.
	Note that if component $G_{\nu_i}$ does not belong to $G'_\mu$, then either $G_{\nu_i}$ is a special component of $u$ or $v$, or it is a normal component for both poles $u$ and $v$, with $\lambda(\mathcal{E}_{\nu_i},u)=-1$ or $\lambda(\mathcal{E}_{\nu_i},v)=-1$. We say that $G_{\nu_i}$ is \emph{associated} with pole $w$ if $G_{\nu_i}$ is either a special component of $w$ or it is a normal component for both poles with $\lambda(\mathcal{E}_{\nu_i},w)=-1$. Note that if $G_{\nu_i}$ is a special component of $w$ then $w$ is a non-switch vertex of $G_\mu$. On the other hand, if it is a normal component of $w$ with $\lambda(\mathcal{E}_{\nu_i},w)=-1$, then there is a large angle at $w$ inside $\mathcal{E}_{\nu_i}$ and $w$ is a switch vertex of $G_\mu$. Assume for a contradiction that there exist at least three components that don't belong to the thin subgraph $G'_\mu$. Then at least two of them are associated with the same pole, say $w$. If $w$ is a non-switch vertex, then both components are special components of $w$; contradiction, since there is at most one special component of $w$. Hence $w$ is a switch vertex and both components are normal components with a large angle at $w$ in their interior; contradiction to \textbf{UP2}, since, in any upward planar embedding, each switch vertex has exactly one large angle. 
	\end{proof}

	\begin{figure}[htb]
		\centering
		\begin{subfigure}{.2\textwidth}
			\centering
			\includegraphics[width=\columnwidth, page=1]{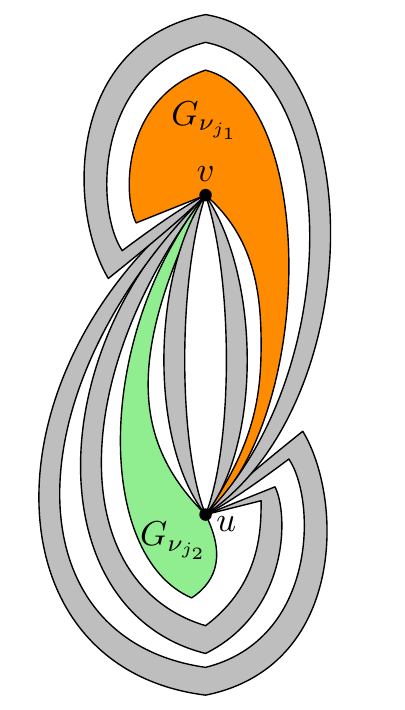}
		\end{subfigure}
		\caption{\label{fig:2-special} A graph $G_{\mu}$ that is a parallel composition of seven components; the thin subgraph is show in gray. The component $G_{\nu_{j_1}}$ (shown in orange) is a normal component for $u$ and $v$ with $\lambda(G_{\nu_{j_1}},v)=-1$; the component $G_{\nu_{j_2}}$ (shown in green) is a special component of $u$.}
	\end{figure}
	\fi
In the next lemma, if $w\in\{u,v\}$ is a top vertex then $\rho_w=out$, otherwise $\rho_w=in$. 
	
	\iflong \begin{lemma} \fi \ifshort \begin{lemma}\fi\label{lem:shapes-seq-necessity}
		Let $\mu$ be a $P$-node such that $G_{\mu}$ is upward planar and let $\mathcal{E}_\mu$ be \iflong an \fi \ifshort a $uv$-external \fi upward planar embedding of $G_{\mu}$ such that the left-turn-number of $G'_{\mu}$ is $c$. Then the shape sequence of $G'_{\mu}$ with respect to $\mathcal{E}_\mu$ is 
		$[s_1\+,s_2\*,s_3\*]$, with 
		$s_1=\shapeSequence{c}{-c}{1}{1}{\rho_u}{\rho_u}{\rho_v}{\rho_v}$,
		$s_2=\shapeSequence{c-2}{-c+2}{1}{1}{\rho_u}{\rho_u}{\rho_v}{\rho_v}$,			
		$s_3=\shapeSequence{c-4}{-c+4}{1}{1}{\rho_u}{\rho_u}{\rho_v}{\rho_v}$.
	\end{lemma}
	\iflong \begin{proof}
	Since all components of $G'_{\mu}$ are normal components of both poles $u$ and $v$, it follows that $u$ and $v$ are switch vertices of $G'_{\mu}$. Hence, for every $i=1,\cdots,k'$,  $\rho_l(\mathcal{E}_{\nu'_i},u)=\rho_r(\mathcal{E}_{\nu'_i},u)=\rho_u$ and $\rho_l(\mathcal{E}_{\nu'_i},v)=\rho_r(\mathcal{E}_{\nu'_i},v)=\rho_v$. Additionally, at each pole $u$ and $v$ there is exactly one large angle. By our assumptions $\lambda(\mathcal{E}_{\nu'_i},u)=\lambda(\mathcal{E}_{\nu'_i},v)=1$ for $i=1,\cdots,k'$, hence the large angles at $u$ and $v$ lie inside two  (not necessarily distinct) faces of $\mathcal{E}'_{\mu}$. Let $f_i$ and $f_j$ be the two faces where the large angle at $u$ and $v$ lie in respectively, where $0\leq i,j<k'$.
		We distinguish some cases depending on whether $i=0$ or $j=0$ holds.
		
		\begin{figure}[htbp]
			\centering
			\begin{subfigure}{.18\textwidth}
				\centering
				\includegraphics[width=\textwidth, page=1]{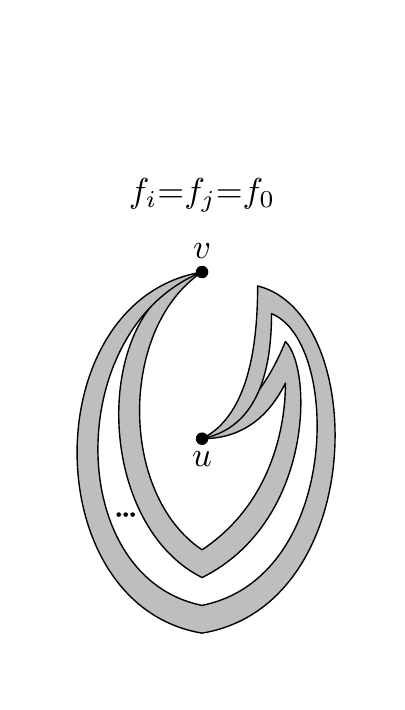}
				\subcaption{}
				\label{fig:p_case_1}
			\end{subfigure}
			\hfil
			\begin{subfigure}{.18\textwidth}
				\centering
				\includegraphics[width=\textwidth, page=2]{figures/p_nodes}
				\subcaption{}
				\label{fig:p_case_2a}
			\end{subfigure}
			\hfil
			\begin{subfigure}{.18\textwidth}
				\centering
				\includegraphics[width=\textwidth, page=3]{figures/p_nodes}
				\subcaption{}
				\label{fig:p_case_2b}
			\end{subfigure}
			\hfil
			\begin{subfigure}{.18\textwidth}
				\centering
				\includegraphics[width=\textwidth, page=4]{figures/p_nodes}
				\subcaption{}
				\label{fig:p_case_4a}
			\end{subfigure}
			\hfil
			\begin{subfigure}{.18\textwidth}
				\centering
				\includegraphics[width=\textwidth, page=5]{figures/p_nodes}
				\subcaption{}
				\label{fig:p_case_4b}
			\end{subfigure}
			\caption{\label{fig:p_nodes}Illustration for Lemma~\ref{lem:shapes-seq-necessity}. Components with shape description $s_1$, $s_2$ and $s_3$ are shaded with gray, orange and green respectively.}
		\end{figure}
		
		\begin{description}
			\item[Case 1:] $i=j=0$. The large angles at $u$ and $v$ are both inside the outer face of $\mathcal{E}'_{\mu}$, that is, $\lambda(G'_{\mu},u)=\lambda(G'_{\mu},v)=1$ (see Fig.~\ref{fig:p_case_1}). In this case, by Lemma~\ref{lem:shape_desc_values}, the left and the right turn number of $\mathcal{E}'_{\mu}$ are the same in absolute value, that is $\tau_l(\mathcal{E}'_{\mu},u,v)=-\tau_r(\mathcal{E}'_{\mu},u,v)$. Since the angles at $u$ and $v$ inside the face $f_h$ is small, it follows from Lemma~\ref{lem:shape_face} that $\tau_r(\mathcal{E}_{\nu'_h},u,v)=-\tau_l(\mathcal{E}_{\nu'_{h+1}},u,v)$, for each $h=1,2,\dots, k-1$. Also, since the angles at $u$ and $v$ inside the internal faces of $\mathcal{E}_{\nu'_h}$ are small, it follows that $\lambda(\mathcal{E}_{\nu'_h},u)=\lambda(\mathcal{E}_{\nu'_h},v)=1$, which implies by Lemma~\ref{lem:shape_desc_values} that $\tau_l(\mathcal{E}_{\nu'_h},u,v)=-\tau_r(\mathcal{E}_{\nu'_{h}},u,v)$. Combining these equalities we have that  all $\mathcal{E}_{\nu_h}$ (gray components in Fig.~\ref{fig:p_case_1}) have the same shape description 
			$\shapeDesc{c}{-c}{1}{1}{\rho_u}{\rho_u}{\rho_v}{\rho_v}$, 
			where $c=\tau_l(\mathcal{E}'_{\mu},u,v)$. 
			Hence the shape sequence is $[\shapeSequence{c}{-c}{1}{1}{\rho_u}{\rho_u}{\rho_v}{\rho_v}\+]$.
			
			\item[Case 2:] $i=0$ and $j>0$. The large angle at $u$ is inside the outer face of $\mathcal{E}'_{\mu}$ while the large angle at $v$ is inside  face $f'_j$, with $1 \leq j < k'$, that is, $\lambda(\mathcal{E}'_{\mu},u)=1$ and $\lambda(\mathcal{E}'_{\mu},v)=-1$ (see Fig.~\ref{fig:p_case_2a}).  
			By Lemma~\ref{lem:shape_face} we have that $\tau_r(\mathcal{E}_{\nu'_{i}},u,v)=-\tau_l(\mathcal{E}_{\nu'_{i+1}},u,v)-2$. 
			For all $h\neq j$, the angles at $u$ and $v$ inside the face $f'_h$ are small, hence from Lemma~\ref{lem:shape_face} it follows that $\tau_r(\mathcal{E}_{\nu'_h},u,v)=-\tau_l(\mathcal{E}_{\nu'_{h+1}},u,v)$. Also, since the angles at $u$ and $v$ inside the internal faces of $\mathcal{E}_{\nu'_h}$ are small, it follows that $\lambda(\mathcal{E}_{\nu'_h},u)=\lambda(\mathcal{E}_{\nu'_h},v)=1$, which implies by Lemma~\ref{lem:shape_desc_values} that $\tau_l(\mathcal{E}_{\nu'_h},u,v)=-\tau_r(\mathcal{E}_{\nu'_{h}},u,v)$.
			Combining these equations, we conclude that all $\mathcal{E}_{\nu_h}$ with $h \leq j$ (gray components in Fig.~\ref{fig:p_case_2a}) have the same shape description 
			$\shapeDesc{c}{-c}{1}{1}{\rho_u}{\rho_u}{\rho_v}{\rho_v}$,
			where $c=\tau_l(\mathcal{E}'_{\mu},u,v)$, 
			while all $\mathcal{E}_{\nu_h}$ with $h > j$ (orange components in Fig.~\ref{fig:p_case_2a}) have the same, yet different, shape description 
			$\shapeDesc{c-2}{-c+2}{1}{1}{\rho_u}{\rho_u}{\rho_v}{\rho_v}$.
			Then the shape sequence for $\mathcal{E}'_{\mu}$ is 
			$[\shapeSequence{c}{-c}{1}{1}{\rho_u}{\rho_u}{\rho_v}{\rho_v}\+,
					\shapeSequence{c-2}{-c+2}{1}{1}{\rho_u}{\rho_u}{\rho_v}{\rho_v}\+]$
			
			\item[Case 3:] $i>0$ and $j=0$. The large angle at $v$ is inside the outer face of $\mathcal{E}'_{\mu}$ while the large angle at $u$ is inside  face $f'_i$, with $1 \leq i < k'$, that is, $\lambda(\mathcal{E}'_{\mu},u)=-1$ and $\lambda(\mathcal{E}'_{\mu},v)=1$ (see Fig.~\ref{fig:p_case_2b}). 
			Symmetrically to the previous case, we can prove that all $\mathcal{E}_{\nu'_h}$ with $h\leq i$ (gray components in Fig.~\ref{fig:p_case_2b}) have the same shape description 
			$\shapeDesc{c}{-c}{1}{1}{\rho_u}{\rho_u}{\rho_v}{\rho_v}$,
			where $c=\tau_l(\mathcal{E}'_{\mu},u,v)$, while all $\mathcal{E}_{\nu'_h}$ with $h > i$ (orange components in Fig.~\ref{fig:p_case_2b}) have the same, yet different, shape description 
			$\shapeDesc{c-2}{-c+2}{1}{1}{\rho_u}{\rho_u}{\rho_v}{\rho_v}$.
			As before, the shape sequence for $\mathcal{E}'_{\mu}$ is 
			$[\shapeSequence{c}{-c}{1}{1}{\rho_u}{\rho_u}{\rho_v}{\rho_v}\+,
					\shapeSequence{c-2}{-c+2}{1}{1}{\rho_u}{\rho_u}{\rho_v}{\rho_v}\+]$.

			\item[Case 4:] $i,j>0$. None of the large angles at $u$ and $v$ is inside the outer face of $\mathcal{E}'_{\mu}$, that is, 	$\lambda(\mathcal{E}'_{\mu},u)=\lambda(\mathcal{E}'_{\mu},v)=-1$. The large angle at $u$ is  inside  face $f'_i$, with $1 \leq i < k'$ and  the large angle at $v$ is  inside  face $f'_j$, with $1 \leq j < k'$.
			Assume first that $i<j$ holds. For each $h\neq i,j$, we have that the angles at $u$ and $v$ inside the face $f'_h$ and inside the internal faces of $\mathcal{E}_{\nu'_h}$ are small. By Lemmata~\ref{lem:shape_face} and \ref{lem:shape_desc_values} it follows that $\tau_r(\mathcal{E}_{\nu'_h},u,v)=-\tau_l(\mathcal{E}_{\nu'_{h+1}},u,v)$ and $\tau_l(\mathcal{E}_{\nu'_h},u,v)=-\tau_r(\mathcal{E}_{\nu'_{h}},u,v)$. This means that all $\mathcal{E}_{\nu'_h}$ with $h \leq i$ (gray components in Fig.~\ref{fig:p_case_4a}) have the same shape description 
			$\shapeDesc{c}{-c}{1}{1}{\rho_u}{\rho_u}{\rho_v}{\rho_v}$,
			where $c=\tau_l(\mathcal{E}'_{\mu},u,v)$, all $\mathcal{E}_{\nu'_h}$ with $i< h\leq j$ (orange components in Fig.~\ref{fig:p_case_4a}) have the same shape description 
			$\shapeDesc{c-2}{-c+2}{1}{1}{\rho_u}{\rho_u}{\rho_v}{\rho_v}$,
			and all $\mathcal{E}_{\nu'_h}$ with $j<h$ (green components in Fig.~\ref{fig:p_case_4a}) have the same shape description 
			$\shapeDesc{c-4}{-c+4}{1}{1}{\rho_u}{\rho_u}{\rho_v}{\rho_v}$; see Fig.~\ref{fig:p_case_4a}.
			The shape sequence defined is $[\shapeSequence{c}{-c}{1}{1}{\rho_u}{\rho_u}{\rho_v}{\rho_v}\+,
					\shapeSequence{c-2}{-c+2}{1}{1}{\rho_u}{\rho_u}{\rho_v}{\rho_v}\+, 			
					\shapeSequence{c-4}{-c+4}{1}{1}{\rho_u}{\rho_u}{\rho_v}{\rho_v}\+]$.
			Note that the shape sequence is the same when $i>j$, while for $i=j$ it becomes 
			$[\shapeSequence{c}{-c}{1}{1}{\rho_u}{\rho_u}{\rho_v}{\rho_v}\+, 			
					\shapeSequence{c-4}{-c+4}{1}{1}{\rho_u}{\rho_u}{\rho_v}{\rho_v}\+]$; see Fig.~\ref{fig:p_case_4b}. \qedhere
		\end{description}
	\end{proof}\fi

	\iflong
	Based on Lemma~\ref{lem:shapes-seq-necessity}, our algorithm computes the shape descriptions of all possible upward planar embeddings of $G_{\mu}$ that match some fixed left- and right-turn-numbers in three steps. In the first step, we decide which sequences $S'$ of shape descriptions given in Lemma~\ref{lem:shapes-seq-necessity} are shape sequences for some maximal subgraph $G'_{\mu}$ of $G_{\mu}$. If $S'$ is such a sequence, then, by Lemma~\ref{lem:at_most_two_comps}, there are at most two children whose shape descriptions are not in $S'$. So, in the second step we compute all possible extensions of $S'$ with the missing shape descriptions to shape sequences of $G_{\mu}$. In the last step, an extended sequence $S$ of $G_{\mu}$ is modified by appropriately changing the shape descriptions for at most two components of $G'_{\mu}$ in order to match the given left- and right-turn-numbers.
	
	We now describe in detail how to compute the feasible set $\mathcal F_{\mu}$ of a P-node $\mu$. For each possible value $c_l$ of $\tau_l(\mathcal{E}_{\mu},u,v)$ and $c_r$ of $\tau_r(\mathcal{E}_{\mu},u,v)$, that is $c_l,c_r\in[\tau_{min}, \tau_{max}]$, we compute the shape descriptions for all possible upward planar embeddings of $G_{\mu}$ with the given left- and right-turn-numbers. All normal components of $G_{\mu}$ have $uv$-external upward planar embeddings with left- and right-turn-number odd if both $u$ and $v$ are either top or bottom vertices, and even if one of them is a top vertex and the other one is a bottom vertex.
	Then any $uv$-external upward planar embedding of every subgraph $G'_{\mu}$ that contains only normal components has a maximal left-turn-number $c'_l$ that is either equal to $c_l$ or to $c_l-1$ depending on the parity of both $c_l$ and $c'_l$. For the first step of the algorithm we consider all sequences of shape descriptions $S'=[s_1\*,s_2\*,s_3\*]$  where $s_1=
	\shapeSequence{c_l'}{-c_l'}{1}{1}{\rho_u}{\rho_u}{\rho_v}{\rho_v}$,
	$s_2=\shapeSequence{c_l'-2}{-c_l'+2}{1}{1}{\rho_u}{\rho_u}{\rho_v}{\rho_v}$,		
	$s_3=\shapeSequence{c_l'-4}{-c_l'+4}{1}{1}{\rho_u}{\rho_u}{\rho_v}{\rho_v}$.
	For each of them we identify a maximal subgraph $G'_{\mu}$ of $G_{\mu}$ such that $S'$ is a shape sequence of $G'_{\mu}$. To this aim, we check whether it is possible to choose a shape description from the feasible set of the pertinent graph $G_{\nu_i}$ of each child $\nu_i$ of $\mu$ (with $i=1,2,\dots,k$) and to sort these pertinent graphs so that the resulting shape sequence is $S'$. We observe that $S'$ contains at most three different shape descriptions, namely $s_1$, $s_2$ and $s_3$ which can occur multiple times in the sequence. For each child $\nu_i$ of $\mu$ (with $i=1,2,\dots,k$) we check whether the feasible set $\mathcal F_{\nu_i}$ contains one of the three shape descriptions of $S'$ in the order that they appear in $S'$; if so, we choose it for $G_{\nu_i}$. Note that the shape description assigned to a child is the first one among $s_1$, $s_2$ and $s_3$ that is contained in its feasible set. However, this does not necessarily produce the desired sequence $S'$ for $G'_{\mu}$, especially if $S'$ has two or three elements. Thus, we check whether it is possible to modify the assigned shape descriptions so as to produce the sequence $S'$. 
	If $S'$ consists of two elements and all components of $G'_{\mu}$ are assigned to the first one, we check whether one of them can be reassigned the second one. If this is not possible, then no subgraph $G'_{\mu}$ has $S'$ as its shape sequence. If $S'$ consists of three elements, that is $S'=[s_1\+,s_2\+,s_3\+]$, then we first make sure that $s_2$ has at least one assigned component by reassigning a component from $s_1$, as in the previous case. If this is not possible, then no subgraph $G'_{\mu}$ has $S'$ as its shape sequence. Otherwise, if $s_3$ has at least one assigned component we are done. If not, we check whether a component from $s_1$ or $s_2$ can be reassigned to $s_3$. If such a component exists, and both $s_1$ and $s_2$ have still at least one element, we are done. Otherwise if $s_1$ is empty, then no subgraph $G'_{\mu}$ has $S'$ as its shape sequence; if $s_2$ is empty we check whether a component from $s_1$ can be reassigned to $s_2$. 
	
For the second step of the algorithm, assume that for some sequence $S'$ we have identified a maximal subgraph $G'_{\mu}$ of $G_{\mu}$ such that $S'$ is a shape sequence of $G'_{\mu}$. Consider now those  children of $\mu$ where no shape description could be assigned for their pertinent graph. Recall that by Lemma~\ref{lem:at_most_two_comps} there can be at most two such children of $\mu$, say $\nu$ and $\nu'$. If this is not true, we conclude that $G_{\mu}$ does not admit an upward drawing whose thin subgraph has $S'$ as its shape sequence.
Otherwise, let $s_\nu$ ($s_{\nu'}$) be a shape description in the feasible set of $\nu$ ($\nu'$, respectively), if it exists. Using Lemma~\ref{lem:p_nodes_seq_extend_test} we compute all possible extensions of $S'$ with $s_\nu$ to shape sequences of $G''_{\mu}$, where $G''_{\mu}$ is the subgraph of $G_{\mu}$ consisting of $G'_{\mu}$ and $G_\nu$. If $\nu'$ does not exist, then $G''_{\mu}=G_{\mu}$ and each computed shape sequence $S''$ is a shape sequence of $G_{\mu}$. Otherwise, using the same lemma, for each computed shape sequence $S''$, we further compute all possible extensions of $S''$ with $s_{\nu'}$ to shape sequences of $G_{\mu}$. For every computed shape sequence $S$ of $G_{\mu}$ we check whether it matches the given left- and right-turn-numbers, that is, the left-turn-number of the first element of $S$ equals $c_l$ and the right-turn-number of its last element equals $c_r$. If this is the case, then we add all shape descriptions for $G_{\mu}$ that correspond to $S$ to the feasible set $\mathcal{F}_{\mu}$ of $\mu$. Otherwise, we proceed to the third step of the algorithm with $S$.

For the last step of our algorithm, let $S$ be a shape sequence of $G_{\mu}$ given by the previous step. That is, $S$ is an extension of $S'$ and contains $s_{\nu}$ and $s_{\nu'}$ for children $\nu$ and $\nu'$ (if they exist) at specific positions, and such that the given left-turn-number and/or the right-turn-numbers are not matched. We consider cases depending on whether only one or both of the two turn-numbers are not matched.
Assume first the the first element of $S$ has left-turn-number different from $c_l$, while the last element of $S$ has right-turn-number equal to $c_r$. Then the first element has left-turn-number smaller than $c_l$.
Our goal is to find a component $G_{\nu'_i}$ of the thin subgraph $G'_{\mu}$, remove its current shape description from $S$ and use another one from its feasible set at the beginning of the sequence in order to match the desired left-turn-number. Note that this is possible only if $c_l$ and $c'_l$ have the same parity.
So, consider component $G_{\nu'_i}$ of the thin subgraph $G'_{\mu}$ and let $s_i\in\{s_1,s_2,s_3\}$ be the shape description assigned to $G_{\nu'_i}$ in $S$. Let $s'_i$ be another shape description in the feasible set of $G_{\nu'_i}$ that is different from $s_1$, $s_2$ and $s_3$. We denote by $S^-$ the sequence $S$ after removing $s_i$. Before proceeding, we want to assure that $S^-$ contains the same elements with $S$. This is not true if $G_{\nu'_i}$ was the only component assigned to $s_i$ during the first step of the algorithm. Note that all components of $G'_{\mu}$ are normal components for both $u$ and $v$ and the shape descriptions in their feasible sets have maximum left-turn-number equal to $c_l'$. Hence $s_1$ does not belong to $S$, that is, the shape sequence of $G'_{\mu}$ is $S'=[s_2\*,s_3\*]$ and $G_{\nu'_i}$ is initially assigned either to $s_2$ or to $s_3$. If $G_{\nu'_i}$ was the only component of $G'_{\mu}$  assigned to $s_2$, then $s_2$ does not belong to the feasible set of the other components of $G'_{\mu}$ and $S^-$ can't be modified to contain all elements of $S'$. If $G_{\nu'_i}$ was the only component of $G'_{\mu}$  assigned to $s_3$, then we check whether $s_2$ belongs to the feasible set of the other components of $G'_{\mu}$. If we can't find such a component, then $S^-$ can't be modified to contain all elements of $S'$, otherwise, we reassign that component to $s_2$ and $S^-$ contains the same elements with $S$. Using Lemma~\ref{lem:p_nodes_seq_valid_test}, we can compute the shape descriptions of all upward planar embeddings of $G_{\mu}$ that correspond to the sequence $[s'_i,S^-]$. The computed shape descriptions have the desired left- and right-turn-number and are added to the feasible set $\mathcal{F}_{\mu}$ of $\mu$.

For the case where $S$ matches the left-turn-number $c_l$ but the last element of $S$ has right-turn-number different from $c_r$, we follow a similar approach. In this case, $s_3$ does not belong to $S$ and for each component $G_{\nu'_i}$ of the thin subgraph $G'_{\mu}$ we first remove its current shape description $s_i$ from $S$. If the new sequence $S^-$ does not contain all elements of $S'$, then $s_i=s_2$ and we check if we can reassign another component of $G'_{\mu}$ from $s_1$ to $s_2$. If this is not possible, then $S^-$ can't be modified to contain all elements of $S'$. Otherwise, by Lemma~\ref{lem:p_nodes_seq_valid_test}, we can compute all possible shape descriptions for $G_{\mu}$ that correspond to the sequence $[S^-,s'_i]$. The computed shape descriptions match the desired left- and right-turn-number and are added to the feasible set $\mathcal{F}_{\mu}$ of $\mu$.

In the special case where both the left- and right-turn-number of $S$ do not match $c_l$ and $c_r$ respectively, then $S$ does not contain neither $s_1$  nor $s_3$. Hence, the initial sequence for the thin subgraph $G'_{\mu}$ is $S'=[s_2\+]$. Now, the shape sequence of a normal component for both $u$ and $v$ that has left-turn-number equal to $c'_l$, right-turn-number at most $-c'_l+2$ and at least one label of the angles at $u$ or $v$ is $-1$, is either $t_l^u=\shapeDesc{c'_l}{-c'_l+2}{-1}{1}{\rho_u}{\rho_u}{\rho_v}{\rho_v}$ or $t_l^v=\shapeDesc{c'_l}{-c_l'+2}{1}{-1}{\rho_u}{\rho_u}{\rho_v}{\rho_v}$ (see Lemma~\ref{lem:shape_desc_values}). Similarly the shape sequence of such a component with left-turn-number at most $c_l'-2$ and right-turn-number $-c_l'+4$ is either $t_r^u=\shapeDesc{c'_l-2}{-c'_l+4}{-1}{1}{\rho_u}{\rho_u}{\rho_v}{\rho_v}$ or $t_r^v=\shapeDesc{c'_l-2}{-c_l'+4}{1}{-1}{\rho_u}{\rho_u}{\rho_v}{\rho_v}$. In order to match both the given left- and  right-turn-numbers, $S$ needs to be extended with $t_l^u$ or $t_l^v$ at the beginning, and with $t_r^u$ or $t_r^v$ at the end. As $t_l^u$ and $t_r^u$ both have $\lambda(u)=-1$, and $t_l^v$ and $t_r^v$ both have $\lambda(v)=-1$, the possible extensions of $S$ are $[t_l^u,S,t_r^v]$ and $[t_l^v,S,t_r^u]$. We consider all components $G_{\nu'_i}$ of the thin subgraph $G'_{\mu}$ and check whether $t_l^u$, $t_l^v$, $t_r^u$ and $t_r^v$ are contained in their feasible set. Let $T_l^u$ be the set of components $G_{\nu'_i}$ that contain $t_l^u$ in their feasible set; sets $T_l^v$, $T_r^u$ and $T_r^v$ are defined similarly. For the sequence $[t_l^u,S,t_r^v]$ we want to select one component from $T_l^u$ and a different one from $T_r^v$ and reassign their shape descriptions from $s_2$ to $t_l^u$ and $t_r^v$ respectively. This is not possible only if either one of  $T_l^u$ or $T_r^v$ is empty, or if $T_l^u$ is the same as $T_r^v$ and they both contain one element. If this is the case, $S$ can't be extended to $[t_l^u,S,t_r^v]$, otherwise, we use Lemma~\ref{lem:p_nodes_seq_valid_test} to decide whether this is a shape sequence of $G_{\mu}$ and to compute the corresponding shape descriptions for $\mathcal{F}_\mu$. Similarly, $S$ can't be extended to the second sequence $[t_l^v,S,t_r^u]$ if either one of $T_l^v$ and $T_r^u$ is empty or if they both have size one and contain the same element, otherwise the shape descriptions for $\mathcal{F}_\mu$ that correspond to $[t_l^v,S,t_r^u]$ are computed with Lemma~\ref{lem:p_nodes_seq_valid_test} and added to the feasible set $\mathcal{F}_{\mu}$ of $\mu$.
	\fi
	
	\ifshort
		Based on Lemma~\ref{lem:shapes-seq-necessity}, our algorithm computes in three steps the shape descriptions of $G_{\mu}$ that match some fixed left-turn-number $c_l$ and right-turn-number $c_r$. 
		Let $c'_l$ be  equal to $c_l$ or $c_l-1$, depending on whether exactly one of $u$ and $v$ is a bottom vertex or not. For the first step, we consider all sequences $S'=[s_1\*,s_2\*,s_3\*]$  where $s_i=
		\shapeSequence{c_l'-2(i-1)}{-c_l'+2(i-1)}{1}{1}{\rho_u}{\rho_u}{\rho_v}{\rho_v}$, for $i=1,2,3$.
		For each of them we identify a maximal subgraph $G'_{\mu}$ of $G_{\mu}$ such that $S'$ is a shape sequence of $G'_{\mu}$. For each child $\nu_i$ of $\mu$ (with $i=1,2,\dots,k$), we check whether the feasible set $\mathcal F_{\nu_i}$ contains shape descriptions of $S'$ in the order that they appear in $S'$; if so, we choose it for $G_{\nu_i}$. This greedy process does not necessarily produce the desired sequence $S'$. By reassigning at most two components of $G'_{\mu}$ either we get $S'$ or no subgraph $G'_{\mu}$ has $S'$ as its shape sequence. 
	
		For the second step, we focus on the children of $\mu$ that, when considering a shape sequence $S'$, have not been assigned a shape description so far. There are at most two such children, say $\nu$ and $\nu'$, otherwise $G_{\mu}$ does not admit an upward planar embedding whose thin subgraph has $S'$ as its shape sequence. Let $s_\nu$ (resp. $s_{\nu'}$) be a shape description in $\mathcal{F}_\nu$ (resp. $\mathcal{F}_{\nu'}$). Using Lemma~\ref{lem:p_nodes_seq_extend_test} we compute all possible extensions of $S'$ with $s_\nu$ and $s_{\nu'}$ to shape sequences of $G_{\mu}$ {(in $\bigoh(1)$ time since the size $r$ of $S'$ is at most $3$)}. For every computed shape sequence $S$ of $G_{\mu}$ we check whether it matches $c_l$ and $c_r$. If so, we add to $\mathcal{F}_{\mu}$ all shape descriptions of $G_{\mu}$ that correspond to $S$ {(in $\bigoh(1)$ time since the size $r$ of $S$ is at most $5$)}. Otherwise, we proceed to the third step with $S$.
	
	To complete the procedure, we perform a case analysis to handle situations where one or both of $c_l$ and $c_r$ are not matched. Intuitively, our goal is to find a component of the thin subgraph $G'_{\mu}$, remove its current shape description from $S$ and use another one from its feasible set at the beginning or at the end of the sequence in order to match $c_l$ or $c_r$. If none of the components of $G'_{\mu}$ can be used for this purpose, we conclude that the pair $c_l$ and $c_r$ cannot be realized. Otherwise, using Lemma~\ref{lem:p_nodes_seq_valid_test} (where the size $r$ is at most 5), we compute all corresponding shape descriptions of $G_{\mu}$ and add them to $\mathcal{F}_{\mu}$. 
	\fi

	\iflong \begin{lemma} \label{lem:P_node_correctness}
		Let $\mu$ be an P-node of $T$. The described algorithm computes correctly the feasible set $\mathcal{F}_\mu$  of $\mu$.
	\end{lemma}
    \begin{proof}
	Our algorithm computes sequences of shape descriptions for $G_{\mu}$ based on the feasible sets of the pertinent graphs of the children of $\mu$. For each computed sequence Lemma~\ref{lem:p_nodes_seq_valid_test} tests whether it is a shape sequence of $G_{\mu}$ and the corresponding shape descriptions are added in the feasible set $\mathcal{F}_\mu$  of $\mu$. Hence, whenever a shape description $s$ is computed for a sequence $S$, there exists an upward planar embedding $\mathcal{E}_{\mu}$ of $G_{\mu}$ with shape sequence $S$ and whose shape description is $s$. 
	
	It remains to prove that if $s_\mu$ is a shape description of $G_{\mu}$, then our algorithm finds a shape sequence $S$ such that $s_\mu$ corresponds to $S$. Let $\mathcal{E}_{\mu}$ be an upward planar embedding of $G_\mu$ with shape sequence $S_\mu$ and such that the shape description of $G_\mu$ is $s_\mu$. Let $G'_{\mathcal{E}_{\mu}}$ be the thin subgraph of $G_\mu$ with respect to $\mathcal{E}_{\mu}$ and let $S'_\mu$  be its shape sequence. 
	
For the left- and right-turn-numbers of $s_\mu$, the first step of the algorithm computes a maximal thin subgraph $G'_\mu$ for $S'_\mu$ assigning each child of $G_\mu$ to the first element of $S'_\mu$ that is contained in its feasible set. If an element of $S'_{\mu}$ is not assigned, the algorithm tries to reassign some components of $G'_\mu$. This is not possible only if an element of $S'_{\mu}$ is not contained in the feasible sets of all components of $G_{\mu}$. Note that the children of $\mu$ whose pertinent graphs belong to the thin subgraph $G'_{\mathcal{E}_{\mu}}$ contain at least one element of $S'_\mu$ in their feasible set. Hence, the algorithm will successfully compute a maximal thin subgraph $G'_\mu$ for $S'_\mu$. Additionally, the computed maximal thin subgraph $G'_\mu$ contains all components of $G'_{\mathcal{E}_{\mu}}$, although they may be assigned to different elements of $S'_\mu$. 

If $G'_{\mathcal{E}_{\mu}}=G_\mu$, that is $S_\mu=S'_\mu$, then all components of $G_\mu$ are assigned to elements of $S'_\mu$, the second step can't extend $S'_\mu$, and the left- and right-turn-numbers of $S'_\mu$ match the given values. Hence, during the third step of the algorithm all shape descriptions for $G_\mu$ that correspond to $S'_\mu$ will be computed, including $s_\mu$. So, assume that there is at least one component of $G_\mu$, say $G_\nu$, whose shape description $s_\nu$ in $S_\mu$ is not an element of $S'_\mu$. If $G_\nu$ is a special component of one of the poles, or a normal component for both poles whose feasible set contains no element of $S'_\mu$, then the second step of the algorithm tests whether $S'_\mu$ can be extended with all shape descriptions in the feasible set of $G_\nu$, including $s_\nu$. Since $S_\mu$ is a shape sequence of $G_\mu$, the algorithm will extend $S'_\mu$ with $s_\nu$ to all possible shape sequences of the subgraph consisting of $G'_{\mu}$ and $G_\nu$. In particular the sequence where $s_\nu$ is placed between the same elements of $S'_\mu$ as in $S_\mu$, will be one of the computed shape sequences. If there is another component $G_{\nu'}$ that is either a special component of one of the poles, or a normal component for both poles whose feasible set contains no element of $S'_\mu$, then $S'_\mu$ will be further extended to contain the shape description $s_{\nu'}$ of $G_{\nu'}$ at the correct position. At the end of the second step, either the computed shape sequence $S$ is the same with $S_\mu$, or all elements of $S$ appear in $S_\mu$ in the same order from left to right.

If $S$ has the same left- and right-turn-numbers with $S_\mu$, then the set of shape descriptions that correspond to $S$ will include $s_\mu$, since all elements of $S$ are contained in $S'_\mu$. So, we examine the case where the left- and/or the right-turn-numbers of $S$ do not match those of $S_\mu$. Assume first that the left-turn-number is different from $S_\mu$, while the right-turn-number is the same. Let $s_1$ be the first element of $S_\mu$ that corresponds to component $G_{\nu_1}$ of $G_\mu$ with respect to $\mathcal{E}_\mu$. Then,  $G_{\nu_1}$ is a component $G_{\nu'_i}$ of the maximal thin subgraph $G'_\mu$ computed in the first step, and its shape description in $S$ belongs to $S'_\mu$. During the third step of the algorithm, it will be tested whether the sequence $[s_1,S]$ is a shape sequence for $G_\mu$. Note that since all elements of $[s_1,S]$ belong to $S_\mu$ in the same order, and the left- and right-turn-numbers of the two sequences are the same, the shape description $s_\mu$ of $G_\mu$ corresponds to both $S_\mu$ and $[s_1,S]$ and will be therefore computed by the algorithm. Similarly, if $S_\mu$ and $S$ have the same left-turn-number and different right-turn-number, the third step of the algorithm will identify $s_\mu$ as a shape description corresponding to the sequence $[S,s_k]$, where $s_k$ is the last element of $S_\mu$ and it is the shape description used for the last component $G_{\nu_k}$ of $G_\mu$ with respect to $\mathcal{E}_\mu$. Hence it remains to consider the case where $S$ and $S_\mu$ have different left- and right-turn-numbers. In this case $S_\mu$ is one of $[t_l^u,S'_\mu,t_r^v]$ or $[t_l^v,S'_\mu,t_r^u]$, and $S'_\mu=[s_2\+]$. The last part of the algorithm tests whether it is possible to reassign two components of $G_{\mu}$ from $s_2$, one to $t_l^u$ and the other to $t_r^v$ for the first sequence, or one to $t_l^v$ and the other to $t_r^u$ for the second one. Hence $S_\mu$ will be successfully identified as a shape sequence of $G_\mu$, concluding the lemma.
	\end{proof}

 	\begin{lemma} \label{lem:P_node_general}
		Let $\mu$ be an P-node of $T$. The feasible set $\mathcal{F}_\mu$ can be computed in $\bigoh(\tau\cdot k)$ time, where $k$ is the number of children of $\mu$.
	\end{lemma}
	\begin{proof}
	We prove that the time complexity of the algorithm described above is  $\bigoh(\tau\cdot k)$.
	The number of choices for the left-turn-number $c_l$ of a $uv$-external upward planar embedding of $G_{\mu}$ is $\tau$, and for each value of $c_l$ by Lemma~\ref{lem:shape_desc_values} the right-turn-number $c_r$ can only take values $-c_l+h$ for $h=0,1,2,3,4$. Hence there are $5\tau$ possible pairs for $c_l$ and $c_r$. For each pair, 
		we compute the shape descriptions of all possible  $uv$-external upward planar embeddings of $G_{\mu}$ with left-turn-number equal to $c_l$ and right-turn-number equal to $c_r$. For given $c_l$ and $c_r$, first we retrieve the feasible sets of the children of $\mu$. For each feasible set we are only interested in those shape descriptions that have left-turn-number in $[c_l-4,c_l]$ and right-turn-number in $[c_r-4,c_r]$. Let $|\mathcal{F}|$ denote the size of the largest feasible set under these constraints. By Lemma~\ref{lem:matrix_feasible} the shape descriptions with specific left- and right-turn-numbers are at most six and can be fetched in $\bigoh(1)$ time. Hence $|\mathcal{F}|$ is a constant and all feasible sets can be computed in $\bigoh(k|\mathcal{F}|)$ time. 
		
		For the first step of the algorithm, we consider all possible sequences $S'=[s_1\*,s_2\*,s_3\*]$. There are eight such sequences depending on whether each of $s_1$, $s_2$ and $s_3$ belongs to the sequence or not. Let $S'$ be one of them.
		For $i=1,\ldots,k$ we check if $\mathcal F_{\nu_i}$ contains one of the at most three elements of $S'$; by Lemma~\ref{lem:matrix_feasible}  each check takes $\bigoh(1)$ time, and we store which nodes among $\nu_i$ are assigned to $s_1$, $s_2$ and $s_3$ (if they belong to $S'$). This step takes $\bigoh(k)$ time and the maximal thin subgraph $G'_{\mu}$ is identified. If an element of $S'$ has no nodes assigned to it, at most two components of $G'_{\mu}$ are reassigned. Each reassignment process takes $\bigoh(k)$ time, since a specific shape description is looked for in the feasible sets of the components of $G'_{\mu}$. Hence the maximal thin subgraph $G'_{\mu}$ with shape sequence $S'$ is identified in $\bigoh(k)$ time or the algorithm concludes that there is no subgraph of $G_{\mu}$ with shape sequence $S'$.
 
The second step considers children $\nu$ and $\nu'$ of $\mu$ (if they exist). For each $s_\nu\in\mathcal{F}_{\nu}$ and $s_{\nu'}\in\mathcal{F}_{\nu'}$, shape sequence $S'$ is extended in all possible ways with $s_\nu$ and $s_{\nu'}$ to a shape sequence $S$ for $G_{\mu}$. Each extension uses Lemma~\ref{lem:p_nodes_seq_extend_test} and requires $\bigoh(1)$ time, since $S'$ has size at most three. Hence, computing all possible sequences $S$ for $G_{\mu}$ that extend a specific shape sequence $S'$ of the first step can be done in $\bigoh(|\mathcal{F}|^2)$ time. 

Checking whether the left- or right-turn-number of a sequence is equal to a given number takes $\bigoh(1)$ time. If both the left- and right-turn-numbers are matched by $S$, the shape descriptions that correspond to $S$ are computed with Lemma~\ref{lem:p_nodes_seq_valid_test} in $\bigoh(1)$ time (as $S$ has size at most five). Otherwise the algorithm proceeds with the last step. We consider two cases depending on whether one of $c_l$ and $c_r$ are not matched or both. Consider first the case where one of them is not matched. Assume that the first element of $S$ has left-turn-number different from $c_l$ and the right-turn-number equals $c_r$; the case where $S$ matches $c_l$ and not $c_r$ is similar. There are at most $k$ choices for selecting component $G_{\nu'_i}$. Computing the sequence $S^-$ takes constant time since its size is at most five. For at most one component, say $G_{\nu'_j}$, $S^-$ might not contain all elements of $S'$; for this component only, we reassign in $(k'-1)\bigoh(|\mathcal{F}|)$ time another component from $s_2$ to $s_3$. We check whether all elements of $S'$ have at least one node assigned to them in $\bigoh(1)$ time. If this is the case, then the new sequence $S^-$ contains all elements of $S'$. For each shape description $s'_i$ of the feasible set of $G_{\nu'_i}$ with left-turn-number equal to $c_l$  we compute all possible shape descriptions for $G_{\mu}$ that correspond to the sequence $[s'_i,S^-]$ in $\bigoh(1)$ time by Lemma~\ref{lem:p_nodes_seq_valid_test}. As the computed shape descriptions match the desired left- and right-turn-number, they are added to the feasible set $\mathcal{F}_{\mu}$ of $\mu$ in $\bigoh(1)$ time each. The process for a component $G_{\nu'_i}$ requires $\bigoh(k|\mathcal{F}|^2)$ time if the sequence $S^-$ does not contain all elements of $S'$, and $\bigoh(|\mathcal{F}|)$ time otherwise. Since this can occur for at most one component of $G'_{\mu}$, the third step runs in $\bigoh(k|\mathcal{F}|^2)$ time, for the case where one of $c_l$ or $c_r$ are not matched. 

In the special case where both the left- and right-turn-number of $S$ do not match $c_l$ and $c_r$ respectively, then the four sets $T_l^u$, $T_l^v$, $T_r^u$ and $T_r^v$ are computed, by checking whether $t_l^u$, $t_l^v$, $t_r^u$ or $t_r^v$ respectively belong to the feasible set of all components of $G'_{\mu}$. By Lemma~\ref{lem:matrix_feasible} each check takes $\bigoh(1)$ time, hence the sets are computed in $\bigoh(k)$ time. There are two possible sequences for $G_{\mu}$, namely $[t_l^u,S,t_r^v]$ and $[t_l^v,S,t_r^u]$. For each of them, we check if any of the corresponding sets is empty, or if they have both size one and contain the same element. If this is the case  $S$ can't be extended to match the given left- and right-turn-numbers, otherwise, by Lemma~\ref{lem:p_nodes_seq_valid_test} the corresponding shape descriptions for $\mathcal{F}_\mu$ can be computed in $\bigoh(1)$ time; each shape description  is added to $\mathcal{F}_\mu$ in $\bigoh(1)$ time. Hence, if $S$ matches none of $c_l$ and $c_r$ the last step of the algorithm requires $\bigoh(k)$ time. 

The overall time complexity is $5\tau(\bigoh(k|\mathcal{F}|)+8\cdot( \bigoh(k) + \bigoh(|\mathcal{F}|^2)\cdot (\bigoh(k|\mathcal{F}|^2)+\bigoh(k))))$, where $|\mathcal{F}|=\bigoh(1)$, that is $\bigoh(\tau\cdot k)$, as claimed.
	\end{proof}\fi

	\ifshort \begin{lemma} \label{lem:P_node_combined}
			Let $\mu$ be an P-node of $T$ with $k$ children. The feasible set $\mathcal{F}_\mu$ can be computed in $\bigoh(\tau\cdot k)$ time.
		\end{lemma}
	\fi	

	\iflong\paragraph*{R-node}\fi
	\ifshort\subparagraph{R-node.}\fi 
	\iflong
	While for S-, P- and Q-nodes we were able to provide polynomial-time procedures to compute the feasible sets, the \NP-hardness of \UP\ prevents the existence of such a polynomial-time procedure for R-nodes. Indeed, obtaining a suitable algorithm that can compute the feasible sets for an R-node given the feasible sets for all of its children will be the task of our parameterized algorithms presented in Sections~\ref{sec:fpt_sources} and~\ref{sec:tw}. In order to complete the description of our framework and also present it in a modular fashion which can be reused by future work, here we formalize the notion of an \emph{R-node subprocedure} which, intuitively, is an algorithm that can compute a feasible set of an R-node.
	
	Formally, an \emph{R-node subprocedure} is an algorithm which takes as input an $R$-node $\mu$ of $T$ with skeleton $H$, a mapping $\mathcal{S}_\mu$ which assigns each child of $\mu$ to its feasible set, runs in time at most $\alpha(G_\mu,\mathcal{S}_\mu)$, and computes the feasible set $\mathcal{F}_\mu$. For a digraph $G$ with SPQR-tree $T$, let $\alpha(G)$ be $\sum\limits_{\text{R-node }\mu}\alpha(G_\mu,\mathcal{S}_\mu)$. We say that $\alpha(G)$ is the \emph{total time complexity} of the R-node subprocedure for graph $G$.
	\fi
	\ifshort
	The R-nodes will be handled differently in Sections~\ref{sec:fpt_sources} and~\ref{sec:tw} depending on the parameter we use. To complete the description of our framework we introduce the notion of an R-node subprocedure.
	Formally, an \emph{R-node subprocedure} is an algorithm which takes as input an $R$-node $\mu$ of $T$ and a mapping $\mathcal{S}_\mu$ which assigns each child of $\mu$ to its feasible set, and computes the feasible set $\mathcal{F}_\mu$ in at most $\alpha(G_\mu,\mathcal{S}_\mu)$ time. For a DAG $G$ with SPQR-tree $T$, $\alpha(G) =\sum\limits_{\text{R-node }\mu}\alpha(G_\mu,\mathcal{S}_\mu)$ is the \emph{total time complexity} of the R-node subprocedure for $G$.
	\fi

\iflong\paragraph*{Root node}\fi
\ifshort\subparagraph{Root node.}\fi
\iflong
We now focus on the root node of the SPQR-tree $T$.  The root node $r$ corresponds to an edge $e=(u,v)$ of $G$ that lies on its outer face. Node $r$ has only one child node $\mu$ with poles $u$ and $v$. Let $s$ be a shape description in the feasible set $\mathcal{F}_{\mu}$ of $\mu$. Then, there exists a $uv$-external upward planar embedding $\mathcal{E}_{\mu}$ of the pertinent graph $G_{\mu}$ of node $\mu$, such that $G_{\mu}$ has shape description $s$, and poles $u$ and $v$ are on the outer face of $G_{\mu}$. Note that $G_{\mu}=G\setminus e$. Hence, $r$ can be treated as a P-node with poles $u$ and $v$ and two children; one of them is $\mu$ and the other one is a Q-node, say $\mu'$ for the edge $(u,v)$. Then the feasible set of $\mu'$ contains only one shape description, namely $s'=\shapeDesc{0}{0}{1}{1}{{out}}{{out}}{{in}}{{in}}$. By Lemma~\ref{lem:P_node_general} we can compute in $\bigoh(\tau)$ time the feasible set of $r$. Since $r$ has only two children, any shape sequence $S$ of $r$ contains $s'$ as its first or last element. This implies that for any upward planar embedding $\mathcal{E}$ of $G_r$ with shape sequence $S$, either the left or the right outer path of $G_{\mu'}$ belongs to the outer face of $G_r$ with respect to $\mathcal{E}$. As $G_{\mu'}$ consists only of edge $(u,v)$, it follows that $(u,v)$ belongs to the outer face of $G$ in all possible upward planar embeddings.
\fi

\ifshort
The root node $r$ corresponds to an edge $e=(u,v)$ of $G$ that lies on its outer face and has only one child $\mu$ with poles $u$ and $v$. Since $G_{\mu}=G\setminus e$, node $r$ can be treated as a P-node with poles $u$ and $v$ and two children; one of them is $\mu$ and the other one is a Q-node for the edge $(u,v)$. 
By Lemma~\ref{lem:P_node_combined}, we can compute the feasible set of $r$ in $\bigoh(\tau)$ time. 
\fi

\begin{lemma} \label{lem:root_node_biconnected}
	The feasible set $\mathcal{F}_r$ of the root node $r$ of $T$  can be computed in $\bigoh(\tau)$ time.
\end{lemma}

\iflong\paragraph*{Time complexity}\fi
\iflong From the discussion above and in particular from Lemmata~\ref{lem:Q_node_general},~\ref{lem:S_node_general},~\ref{lem:P_node_general} and~\ref{lem:root_node_biconnected}, we obtain the following lemma. \fi

\ifshort
Combining Lemmata~\ref{lem:Q_node_general},~\ref{lem:S_node_general},~\ref{lem:P_node_combined} and~\ref{lem:root_node_biconnected}, we obtain the following lemma. 
\fi

\begin{lemma}\label{lem:R_node_biconnected}
	Let $G$ be a biconnected DAG with $n$ vertices and let $T$ be an SPQR-tree of $G$ rooted at a $Q$-node corresponding to an edge $e = (u, v)$. Let $\tau_{\min}$ and $\tau_{\max}$ be two given integer values, and let $\tau=\tau_{\max}-\tau_{\min}+1$. Given an R-node subprocedure with total time complexity $\alpha(G)$, it is possible to compute in time $\bigoh(\alpha(G)+\tau^2\cdot n)$ the shape descriptions of every upward planar embedding with $e$ on the outer face, such that the left- and right-turn-numbers of the pertinent graph of every node of $T$ are in the range $[\tau_{\min},\tau_{\max}]$.
\end{lemma}

\section{Extension to the Single-Connected Case}\label{se:singly-connected}

\iflong
In this section, we show that Lemma~\ref{lem:R_node_biconnected} can be used as a blackbox to extend our results from the biconnected case to arbitrary graphs.
Our goal is to prove Interface Lemma~\ref{lem:R_node_general}, which reduces the task of solving \textsc{Upward Planarity} to the one of obtaining an R-node subprocedure, for \emph{all} graphs including single-connected ones. 
First, in the next definition we encapsulate the property that the shapes occuring in SPQR-tree decompositions of a graph have a bounded range of left- and right-turn-numbers, as in Sections~\ref{sec:fpt_sources} and~\ref{sec:tw} we obtain efficient algorithms precisely in this regime.

\begin{definition}
    Let $G$ be a digraph and $\tau_{\min}$, $\tau_{\max}$ be two integers. We call $G$ \emph{$[\tau_{\min}, \tau_{\max}]$-turn-bounded} if the following holds for every upward planar embedding $\mathcal{E}$ of $G$. For any biconnected component $B$ of $G$ and every SPQR-tree decomposition $T$ of $B$ that is rooted at a Q-node corresponding to an edge that lies on the outer face of $B$ in the embedding induced by $\mathcal{E}$, left- and right-turn-numbers of the pertinent graph of every node of $T$ are in range $[\tau_{\min}, \tau_{\max}]$.
\end{definition}

We now give the precise statement of our Interface Lemma.
\fi

\iflong \begin{lemma}[Interface Lemma] \fi \ifshort \begin{lemma}[Interface Lemma]\fi\label{lem:R_node_general}
        Let $G$ be an $n$-vertex digraph, and $\tau_{\min}$, $\tau_{\max}$ be integers such that $G$ is $[\tau_{\min}, \tau_{\max}]$-turn-bounded.
        Given an R-node subprocedure with total time complexity $\alpha(G)$, it is possible to determine whether $G$ admits an upward planar embedding in time $\bigoh(n (\alpha(G)+\tau^2\cdot n))$ where $\tau = \tau_{\max} - \tau_{\min} + 1$. 
\end{lemma}

Note that, for a single-connected graph $G$, we define the total time complexity $\alpha(G)$ of an R-node subprocedure to be the sum of $\alpha(B)$ over all biconnected components $B$ of $G$.

\ifshort
To give an intuition of the proof, consider a fixed rooting of the block-cut tree of $G$. 
The core of our algorithm is a procedure that, given suitable embeddings of leaf components containing the same cut-vertex,
attaches these embeddings to an arbitrary upward planar embedding of the rest of the graph.  
This allows us to process the block-cut tree upwards: we iteratively verify that there exist desired embeddings for a group of leaf blocks via the biconnected algorithm (Lemma~\ref{lem:R_node_biconnected}), and reduce to a smaller tree by removing these blocks.
\fi

\iflong

To give an intuition of the proof, we start by guessing the root of the block-cut tree of $G$, which is assumed to see the outer face in the desired upward planar embedding of $G$. 
The core of the proof is the following lemma, which states that leaf components can always be added on top of an upward planar embedding, up to simple conditions on the cut-vertex.

\begin{lemma} 
    Consider a rooted block-cut tree of G, its cut vertex $v$ that is adjacent to leaf blocks $B_1$,...,$B_\ell$, and the parent block $P$. Denote by $G^P$ the corresponding component of $G - \{v\}$.
Any upward planar embedding of $G^P$ where the root block sees the outer face can be extended to an embedding of $G$ with the same property if the following holds:
\begin{enumerate}
    \item Each $B_i$ has an upward-planar embedding with $v$ on the outer face.
    \item If $v$ is a non-switch vertex in $P$, each $B_i$ has an upward-planar embedding with $v$ on the outer face where the large angle of $v$ is on the outer face.
\end{enumerate}
    \label{lem:children_block}
\end{lemma}

As the conditions of Lemma~\ref{lem:children_block} turn out to be necessary as well, our algorithm simply proceeds upwards along the block-cut tree, checking every time the conditions of the lemma and removing the respective leaf components. If at some point the conditions are not satisfied, we conclude that $G$ is a no-instance; otherwise, any upward planar embedding of the root block can be iteratively extended to an embedding of the whole graph. Finally, finding upward planar embeddings of the blocks is performed via Lemma~\ref{lem:R_node_biconnected}; it can also be used to verify conditions of Lemma~\ref{lem:children_block}.
\fi

\iflong \begin{proof}[Proof of Lemma~\ref{lem:children_block}]
    Fix an upward planar embedding $\mathcal{E}^P$ of $G^P$, and upward plannar embeddings $\mathcal{E}^{B_1}$, \dots, $\mathcal{E}^{B_\ell}$ of $B_1$, \dots, $B_\ell$ subject to the conditions of the lemma.
    W.l.o.g. assume that the leaf blocks are indexed in such a way that for the first $p$ blocks, the angle of $v$ at the outer face is small or flat, for some $0 \le p \le \ell$, and for the remaining blocks the outer angle at $v$ is large. If $p > 0$, it follows from the conditions of the lemma that $v$ is a switch vertex in $P$ and in all the blocks $B_1$, \dots, $B_\ell$ except maybe one; if there is a block where $v$ is a non-switch vertex w.l.o.g. let that block be $B_p$.

    We first consequently deal with the blocks $B_1$, \dots, $B_p$, that is, we construct an upward planar embedding $\mathcal{E}_i$ of $G_i = G[V(G_p) \bigcup_{j = 1}^i V(B_j)]$ for each $i \in [p]$ from an upward planar embedding $\mathcal{E}_{i - 1}$ of $G_{i - 1}$, starting with the embedding $\mathcal{E}_0 = \mathcal{E}^P$ of $G_0 = G^P$. For each vertex of $G_i$ that is not $v$, we infer both the circular order of the incident edges and the angle assignment from $\mathcal{E}_{i - 1}$ or $\mathcal{E}^{B_i}$, depending on whether the vertex belongs to $B_i$ or not. To define a circular order on edges incident to $v$, split the ordering on neighbors of $v$ in $G_{i-1}$ at a large angle of $v$ (according to $\mathcal{E}_{i - 1}$), which necessarily exists as $v$ is a switch vertex of $G_{i-1}$, and join it with the ordering of neighbors of $v$ in $B_i$ split at the outer angle of $v$ in $\mathcal{E}^{B_i}$. This results in two new angles, for each of them, we assign a value of $-1$ if it is a switch angle, and $0$ if it is a flat angle. All the other angles around $v$ keep their value from either $\mathcal{E}_{i - 1}$ or $\mathcal{E}^{B_i}$. We now verify that the conditions \textbf{UP0}--\textbf{UP3} from Theorem~\ref{th:upward-conditions} hold for the vertex $v$ and the new face of the constructed embedding. \textbf{UP0}--\textbf{UP2} hold by construction,
next we verify \textbf{UP3} by considering two cases. In what follows, $f$ denotes the newly formed face. We assume that \textbf{UP3} is already verified for $\mathcal{E}_{i - 1}$.

    \textbf{Case 1.} The large angle at $v$ is inside the outer face of $\mathcal{E}_{i - 1}$. This necessarily implies $i = 1$, as each $B_i$ for $i \in [p]$ is embedded with a small angle at $v$ on the outer face. Since we embed $B_1$ in the outer angle of $v$, the newly formed face is the outer face. Observe that the set of angles at $f$ is precisely the union of the angles of $G_0$ and the angles of $B_1$ on the outer faces in the respective embeddings, except for the angles at $v$. The angle at $v$ on the outer face is necessarily $1$ in $\mathcal{E}_0$, and there are two cases for the value in $\mathcal{E}^{B_1}$. If the corresponding angle is a switch angle, then the value is $-1$, and the two newly formed angles at $v$ in $f$ by construction receive the value $-1$ (see Fig.~\ref{fig:single_connected_example_1}). If the angle is flat in $B_1$, its value is $0$ and one of the newly formed angles is also flat, receiving a value of $0$, while the other is a switch angle, receiving a value of $-1$ (see Fig.~\ref{fig:single_connected_example_2}). In either case, the sum of angle values on $f$ at $v$ is two less than the sum of two outer angles at $v$ in $G_0$ and $B_1$.
    Thus, summing up all angle values along $f$ we obtain two (sum of angles on the outer face of $G_0$) plus two (sum of angles on the outer face of $B_1$ minus two (difference in the angles at $v$), which satisfies \textbf{UP3} since $f$ is the outer face.

    \textbf{Case 2.} The large angle at $v$ is inside an inner face of $\mathcal{E}_{i - 1}$. In this case, $f$ is also an inner face in $\mathcal{E}_i$. Analogously to Case 1, summing up the angles on $f$ in $G_i$ is equivalent to summing angles on the corresponding faces in $G_{i - 1}$ and $B_i$ accounting for the change of angles at $v$. By the same arguments, the change at $v$ results in minus two, while the sum in $G_{i - 1}$ is minus two (inner face) and the sum in $B_i$ is two (outer face) by \textbf{UP3} in $\mathcal{E}_{i - 1}$ and $\mathcal{E}^{B_i}$, respectively. This sums up to minus two, fulfilling \textbf{UP3} for $f$ in $\mathcal{E}_i$.

\begin{figure}[htb]
		\centering
		\begin{subfigure}{.18\textwidth}
			\centering
			\includegraphics[width=\columnwidth, page=6]{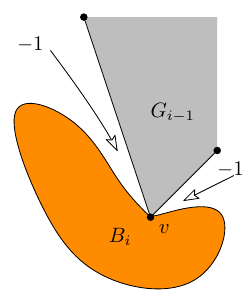}
			\subcaption{}
			\label{fig:single_connected_example_1}
		\end{subfigure}
		\hfil
		\begin{subfigure}{.18\textwidth}
			\centering
			\includegraphics[width=\columnwidth, page=5]{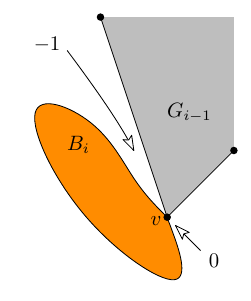}
			\subcaption{}
			\label{fig:single_connected_example_2}
		\end{subfigure}
		\hfil
		\begin{subfigure}{.18\textwidth}
			\centering
			\includegraphics[width=\columnwidth, page=2]{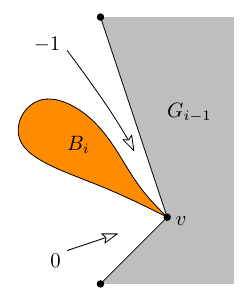}
			\subcaption{}
			\label{fig:single_connected_example_3}
		\end{subfigure}
		\hfil
		\begin{subfigure}{.18\textwidth}
			\centering
			\includegraphics[width=\columnwidth, page=4]{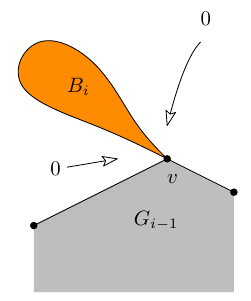}
			\subcaption{}
			\label{fig:single_connected_example_4}
		\end{subfigure}
		\hfil
		\begin{subfigure}{.18\textwidth}
			\centering
			\includegraphics[width=\columnwidth, page=3]{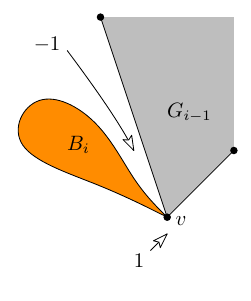}
			\subcaption{}
			\label{fig:single_connected_example_5}
		\end{subfigure}
		\caption{\label{fig:single_connected_example}Illustrations for the proof of Lemma~\ref{lem:children_block}.}
	\end{figure}
	
    It remains to extend the resulting upward planar embedding of $G_p$ by accounting for the block $B_{p + 1}$, \dots, $B_\ell$ for which there is an upward planar embedding with a large angle at $v$ on the outer face. Again, we consequently extend the embedding one block at a time, constructing an upward planar embedding $\mathcal{E}_i$ of $G_i = G[V(G_p) \bigcup_{j = 1}^i V(B_j)]$ for each $i \in [p+1..\ell]$. Fix $B_i$, $v$ is a switch vertex in $B_i$ since it has an adjacent large angle. We first define the embedding on every vertex except $v$, by repeating the respective circular order and angle assignment from $\mathcal{E}_{i - 1}$ or $\mathcal{E}^{B_i}$. Now, consider a flat or large angle $\alpha$ at $v$ in $G_{i - 1}$ as given by $\mathcal{E}_{i - 1}$, at least one necessarily exists by Theorem~\ref{th:upward-conditions}. Split the neighbor list of $v$ in $G_{i - 1}$ at $\alpha$, split the neighbor list of $v$ in $B_i$ at the large angle, then join these lists to obtain a circular order of neighbors of $v$ in $G_i$. Consider the two newly formed angles. If $\alpha$ is flat, one of them is flat and the other is a switch angle, assign $0$ to the former and $-1$ to the latter (see Fig.~\ref{fig:single_connected_example_3}). If $\alpha$ is a large angle, either the edges at $v$ are oriented differently in $G_{i - 1}$ and $B_i$, in which case both new angles are flat (see Fig.~\ref{fig:single_connected_example_4}), or $v$ is a switch vertex in $G_i$, then assign $1$ to one of the newly formed edges and $-1$ to the other (arbitrarily) (see Fig.~\ref{fig:single_connected_example_5}). By construction, the conditions \textbf{UP0}--\textbf{UP2} of Theorem~\ref{th:upward-conditions} are satisfied at $v$, it remains to verify \textbf{UP3} for the newly formed face $f$. Analogously to the case $i \le p$, the sum of angles on $f$ can be computed as the sum of three parts: the sum of angles in $G_{i - 1}$ and $B_i$ on the corresponding faces, and the change in the angle values at $v$. The first part is two or minus two, depending on whether the large angle at $v$ is on the outer face of $G_{i - 1}$ or not, but it matches exactly the target value for $f$. The second part is always two since $v$ is on the outer face of $B_i$. It can be easily seen that the third part is exactly minus two, in each of the three cases above. Thus, \textbf{UP3} is satisfied for $f$, and $\mathcal{E}_i$ is a required upward planar embedding. Setting $i = \ell$ finishes the proof.
\end{proof}\fi

\iflong
By proceeding upwards along the block-cut tree of $G$, Lemma~\ref{lem:R_node_general} follows shortly from Lemma~\ref{lem:children_block}.

\begin{proof}[Proof of Lemma~\ref{lem:R_node_general}]
    
    The algorithm proceeds as follows. First, compute the block-cut tree $T$ of $G$. Then, branch on the choice of the root block $B^R$ of $T$. The algorithm now verifies whether there is an upward planar embedding of $G$ where at least one edge of $B^R$ is incident to the outer face. To this end, the algorithm iteratively checks the conditions of Lemma~\ref{lem:children_block} in the bottom-up fashion. Specifically, the algorithm picks a cut vertex $v$ in $T$, such that all its adjacent blocks are leaves in $T$, except for the parent block $P$. For each of the leaf blocks $B_1$, \dots, $B_\ell$, the algorithm uses Lemma~\ref{lem:R_node_biconnected} to compute possible shape descriptions for every upward planar drawing where $v$ lies on the outer face. That is, for every edge $e$ incident to $v$ in $B_i$, the algorithm constructs an SPQR-tree decomposition of $B_i$ rooted at a Q-node corresponding to $e$, and runs Lemma~\ref{lem:R_node_biconnected} on this decomposition.

    If some of the blocks does not admit any feasible shape at this stage, the algorithm immediately concludes that $G$ is a no-instance. Additionally, if $v$ is a non-switch vertex in $P$, the algorithm checks whether for each $i \in [\ell]$ there exists a feasible shape description which assigns a large angle to the outer face at $v$. Since for each $e = (v, u) \in E(B_i)$, we have computed the shape descriptions of every upward planar embedding of $B_i$ with edge $e$ on the outer face via Lemma~\ref{lem:R_node_biconnected}, this information is readily available to the algorithm. 
If for some component the condition does not hold, the algorithm finishes with reporting that $G$ is a no-instance. Otherwise, the algorithm removes the blocks $B_1$, \dots, $B_\ell$ from the graph and proceeds. The algorithm continues in this fashion until only the root block $B^R$ is left, where it finally uses Lemma~\ref{lem:R_node_biconnected} to check whether $B^R$ admits an upward planar drawing.
That is, for every edge $e \in E(B^R)$, the algorithm constructs an SPQR-tree decomposition of $B^R$ rooted at $e$ and runs Lemma~\ref{lem:R_node_biconnected} on this decomposition. If for each $e$ this results in an empty set of feasible shapes, the algorithm concludes that $G$ is a no-instance. Otherwise, the algorithm returns that $G$ is a yes-instance.

\textbf{Correctness.} 
    First, assume that the algorithm concluded that $G$ admits an upward planar embedding.
    Then, there is at least one feasible shape description of $B^R$ returned by the algorithm of Lemma~\ref{lem:R_node_biconnected}, which means that there is an upward planar embedding of $B^R$. By going through the sequence of considered cut vertices $v$ in the reverse order, the embedding can be extended to all children blocks containing $v$ by Lemma~\ref{lem:children_block}. Specifically, fix a cut vertex $v$ that splits the graph into the leaf blocks $B_1$, \dots, $B_\ell$ and the subgraph $G^P$ that contains the parent block, and assume that the existence of an upward planar embedding for $G^P$ is already shown. Now, by Lemma~\ref{lem:children_block} this embedding extends to an upward planar embedding of $G[V(G^P) \cup V(B_1) \cup \dots \cup V(B_\ell)]$. When the whole block-cut tree $T$ is traversed back in this way, we obtain an upward planar embedding of $G$.
    
    In the other direction, we argue that if $G$ has an upward planar embedding $\mathcal{E}^*$, the algorithm necessarily returns one as well (not necessarily the same). In $\mathcal{E}^*$, there is an edge $e$ that is incident to the outer face, consider the block $B^R$ that contains $e$ and the branch of the algorithm where $B^R$ is set to be the root block of the block-cut tree $T$ of $G$. Since $G$ admits an upward planar embedding, each block $B$ in $T$ admits one as well. Here we assume that a particular upward planar drawing of $G$ is fixed as well, and an upward planar drawing for its subgraph is induced naturally by removing the respective part of the drawing of $G$. We now show that no step of the algorithm can result with reporting a no-instance, this immediately implies that the algorithm returns an upward planar embedding of $G$.

    Fix a graph $G'$ that is obtained from $G$ by removing a sequence of leaf blocks (as during the work of our algorithm), $\mathcal{E}^*$ induces an upward planar embedding of $G'$ as well. If $G'$ consists of a single biconnected block, our algorithm successfully decides whether there exists an upward planar embedding of this block by Lemma~\ref{lem:R_node_biconnected}. Otherwise, consider a cut-vertex $v$ with adjacent children leaf blocks $B_1$, \dots, $B_\ell$ and the parent block $P$, denote the component of $G' - v$ that contains $P$ by $G_P$. For each $B_1$, \dots, $B_\ell$, $G_P$ fix an upward planar embedding that is induced by $\mathcal{E}^*$. Since $G_P$ contains $B_R$, there is an edge $e$ in $G_P$ that is incident to the outer face. This immideately implies that for each $B_i$, $v$ is on the outer face in its embedding --- otherwise the whole $G_P$ together with $e$ lies in an inner face of an embedding of $B_i$, thus contradicting the choice of $\mathcal{E}^*$ and $e$. This verifies the first condition of Lemma~\ref{lem:children_block}, we now move to the second condition. Let $v$ be a non-switch vertex in $P$, then by the expansion procedure $v$ is a switch vertex in all $B_1$, \dots, $B_\ell$. Assume there is a block $B_i$ where $v$ has a small angle on the outer face, in its fixed embedding. Then by the condition \textbf{UP1} of Theorem~\ref{th:upward-conditions}, $v$ has a large angle in one of the inner faces of $B_i$. However, since by the discussion above $P$ is embedded in the outer face of $B_i$, the flat angles at $v$ will end up being inside the small angle on the outer face of $B_i$. This contradicts the fact that $\mathcal{E}^*$ is an upward planar embedding. Thus, at every intermediate step of the algorithm Lemma~\ref{lem:children_block} is invoked successfully.

    \textbf{Running time.}
    For a fixed rooting of the block-cut tree $T$, on each non-root block $B$ our algorithm runs the algorithm of Lemma~\ref{lem:R_node_biconnected} on $B$ for each edge incident to $v$, where $v$ is the parent cut-vertex of this block. Additionally, on the root block $B^R$ Lemma~\ref{lem:R_node_biconnected} is invoked once for every edge of $B^R$. Thus, over all possible rootings of $T$, Lemma~\ref{lem:R_node_biconnected} is run once on each block per edge of the block, and once for each edge incident to a cut vertex.
    The latter holds since if our algorithm queries Lemma~\ref{lem:R_node_biconnected} several times with the same block and the same root edge $e$, there is no need to recompute the set of admissible shape descriptions as it only depends on the block itself.
    Therefore, the running time is dominated by invocations of Lemma~\ref{lem:R_node_biconnected}, where the number of invocations on a block is proportional to the number of vertices of the block, and that takes time
    $$\sum_{B \in T} \bigoh(|B| (\alpha(B) + \tau^2 \cdot |B|)),$$
    Since the blocks of $T$ intersect only in cut-vertices, and the amount of intersections is linear in $n$, the running time can be upper-bounded by $n (\alpha(G) + \tau^2 \cdot n)$, as required (here we also use that by definition $\alpha(B) \le \alpha(G)$ for each block $B$ of $G$).
   Every other part of the algorithm (e.g. constructing the block-cut tree) takes linear time.
\end{proof}\fi

\section{An Algorithm Parameterized by the Number of Sources}\label{sec:fpt_sources}
\setboolean{longShapeDesc}{true}

\ifshort
Let $G$ be an acyclic digraph with $n$ vertices and $\sigma$ sources, whose underlying graph is planar. In order to obtain an algorithm for {\sc Upward Planarity} parameterized by $\sigma$, in view of Lemma~\ref{lem:R_node_general}, we devise an R-node subprocedure whose runtime depends on $\sigma$ and, polynomially, on $n$. 
We hence assume that $G$ is biconnected and that has been expanded. Let $e^*$ be any edge of $G$; we compute an SPQR-tree $T$ of $G$ rooted at the Q-node representing $e^*$ in $\bigoh(n)$ time \cite{dt-opl-96,gm-lti-00}. A key ingredient of our algorithm is the following.


\begin{lemma} \label{le:turn-set-range}
Let $\mu$ be a node of $T$, let $u$ and $v$ be the poles of $\mu$, let $\sigma_{\mu}$ be the number of sources of $G_{\mu}$ different from its poles, and let $\mathcal E_{\mu}$ be any $uv$-external upward planar embedding of $G_{\mu}$.  The left- and right-turn-numbers of $\mathcal E_{\mu}$ are in the interval $[-2\sigma_{\mu}-1,2\sigma_{\mu}+1]$. Furthermore, the size of the feasible set $\mathcal{F}_{\mu}$ of $\mu$ is at most $72\sigma_{\mu}+54$.
\end{lemma}	

Let $\mu$ be an R-node of $T$ with children  $\nu_1,\dots,\nu_k$. Let $u$ and $v$ be the poles of $\mu$,  $\sigma_{\mu}$ be the number of sources of $G_{\mu}$ different from its poles; for $i=1,\dots,k$, let $u_i$ and $v_i$ be the poles of $\nu_i$ and $e_i$ be the virtual edge representing $\nu_i$ in the skeleton $\textrm{sk}(\mu)$ of $\mu$. We give an algorithm that computes $\mathcal{F}_{\mu}$ from the feasible sets $\mathcal{F}_{\nu_1},\dots,\mathcal{F}_{\nu_k}$ in $\bigoh(\sigma 1.45^\sigma \cdot k\log^3 k)$ time.

We introduce two classifications of the components of $G_\mu$. A component $G_{\nu_i}$ is \emph{interesting} if it contains sources other than its poles, and \emph{boring} otherwise. Because $G$ has $\sigma$ sources, at most $\sigma$ components among $G_{\nu_1},\dots,G_{\nu_k}$ are interesting, while any number of components can be boring. Second, a component $G_{\nu_i}$ is \emph{extreme} if $e_i$ is incident to a pole of $\mu$ and is incident to the
face containing $u$ and $v$ of any planar embedding of $\textrm{sk}(\mu)$, and \emph{non-extreme} otherwise. Note that there are four
 extreme components among $G_{\nu_1},\dots,G_{\nu_k}$, because there are exactly two virtual edges incident to each of $u$ and $v$ in the considered face. We can order $G_{\nu_1},\dots,G_{\nu_k}$ in $O(k\log k)$ time so that all the extreme or interesting components come first. 

\begin{figure}[!t]
	\centering
	{\centering
		\begin{subfigure}{.19\textwidth}
			\includegraphics[page=1]{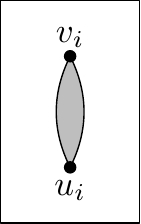}
			\subcaption{}
			\label{fig:boring_sausage-short}
		\end{subfigure}
	}		
	\hfil
	{\centering
		\begin{subfigure}{.19\textwidth}
			\includegraphics[page=2]{figures/boring_comps}
			\subcaption{}
			\label{fig:boring_right-short}
		\end{subfigure}
	}\hfil
	{\centering
		\begin{subfigure}{.19\textwidth}
			\includegraphics[page=3]{figures/boring_comps}
			\subcaption{}
			\label{fig:boring_left-short}
		\end{subfigure}
	}\hfil
	{\centering
		\begin{subfigure}{.19\textwidth}
			\includegraphics[page=5]{figures/boring_comps}
			\subcaption{}
			\label{fig:boring_hat-short}
		\end{subfigure}
	}\hfil
	{\centering
		\begin{subfigure}{.19\textwidth}
			\includegraphics[page=4]{figures/boring_comps}
			\subcaption{}
			\label{fig:boring_heart-short}
		\end{subfigure}
	}
	\caption{Shape descriptions of boring components.}\label{fig:boring_comps-short}
\end{figure}

Despite their name, boring components play an important role in our algorithm. A key feature is that a $u_iv_i$-external upward planar embedding of a boring component $G_{\nu_i}$ can only have one of $O(1)$ shape descriptions: the {\bf sausage} $\shapeDesc{0}{0}{1}{1}{\textrm{out}}{\textrm{out}}{\textrm{in}}{\textrm{in}}$, see Fig.~\ref{fig:boring_sausage-short}; the {\bf inverted-sausage} $\shapeDesc{0}{0}{1}{1}{\textrm{in}}{\textrm{in}}{\textrm{out}}{\textrm{out}}$, see Fig.~\ref{fig:boring_sausage-short} with $u_i$ and $v_i$ inverted; the {\bf right-wing} $\shapeDesc{0}{1}{1}{0}{\textrm{out}}{\textrm{out}}{\textrm{in}}{\textrm{out}}$, see Fig.~\ref{fig:boring_right-short}; the {\bf inverted-right-wing} $\shapeDesc{1}{0}{0}{1}{\textrm{out}}{\textrm{in}}{\textrm{out}}{\textrm{out}}$, see Fig.~\ref{fig:boring_right-short} with $u_i$ and $v_i$ inverted; the {\bf left-wing} $\shapeDesc{1}{0}{1}{0}{\textrm{out}}{\textrm{out}}{\textrm{out}}{\textrm{in}}$, see Fig.~\ref{fig:boring_left-short}; the {\bf inverted-left-wing} $\shapeDesc{0}{1}{0}{1}{\textrm{out}}{\textrm{in}}{\textrm{out}}{\textrm{out}}$, see Fig.~\ref{fig:boring_left-short} with $u_i$ and $v_i$ inverted; the {\bf hat} $\shapeDesc{-1}{1}{1}{1}{\textrm{out}}{\textrm{out}}{\textrm{out}}{\textrm{out}}$, see Fig.~\ref{fig:boring_hat-short}; the {\bf inverted-hat} $\shapeDesc{1}{-1}{1}{1}{\textrm{out}}{\textrm{out}}{\textrm{out}}{\textrm{out}}$, see Fig.~\ref{fig:boring_hat-short}  with $u_i$ and $v_i$ inverted; the {\bf heart} $\shapeDesc{1}{1}{1}{-1}{\textrm{out}}{\textrm{out}}{\textrm{out}}{\textrm{out}}$, see Fig.~\ref{fig:boring_heart-short}; and the {\bf inverted-heart} $\shapeDesc{1}{1}{-1}{1}{\textrm{out}}{\textrm{out}}{\textrm{out}}{\textrm{out}}$, see Fig.~\ref{fig:boring_heart-short}  with $u_i$ and $v_i$ inverted. Furthermore, we can prove that not all such shape descriptions can occur simultaneously in the feasible set  of a node $\nu_i$ and that some shape descriptions are ``better'' than others. This allows us to assume that the feasible set of a node $\nu_i$ contains: only the sausage, or only the inverted-sausage, or only the left-wing and the right-wing, or only the inverted-left-wing and the inverted-right-wing, or only the hat and the inverted-hat, or only the heart, or only the inverted-heart, or only the heart and the inverted-heart. 

We test independently whether each shape description $s=\shapeDesc{\tau_l}{\tau_r}{\lambda_u}{\lambda_v}{\rho_{l,u}}{\rho_{r,u}}{\rho_{l,v}}{\rho_{r,v}}$, where $\tau_l \in [-2\sigma_{\mu}-1,2\sigma_{\mu}+1]$, $\tau_r \in [-\tau_l,-\tau_l+4]$, $\lambda_u\in \{-1,0,1\}$, $\lambda_v\in \{-1,0,1\}$, $\rho_{l,u}\in \{\textrm{in},\textrm{out}\}$, $\rho_{r,u}\in \{\textrm{in},\textrm{out}\}$, $\rho_{l,v}\in \{\textrm{in},\textrm{out}\}$, and $\rho_{r,v}\in \{\textrm{in},\textrm{out}\}$ belongs to $\mathcal{F}_{\mu}$ or not. Note that $\tau_l \in [-2\sigma_{\mu}-1,2\sigma_{\mu}+1]$ and $\tau_r \in [-\tau_l,-\tau_l+4]$ can be assumed without loss of generality by Lemmata~\ref{le:turn-set-range} and~\ref{lem:shape_desc_values}, respectively, thus the number of shape descriptions to be tested is in $\bigoh(\sigma_{\mu})$. We select shape descriptions $s_1\in \mathcal{F}_{\nu_1},\dots, s_h\in \mathcal{F}_{\nu_h}$ for the extreme or interesting components $G_{\nu_1},\dots, G_{\nu_h}$ of $G_{\mu}$. Clearly, the number $\ell$ of ways this selection can be done is $\ell=\prod_{i=1}^h |\mathcal{F}_{\nu_i}|$; by exploiting the bound on $|\mathcal{F}_{\nu_i}|$ given by Lemma~\ref{le:turn-set-range}, we can prove that $\ell\in \bigoh(1.45^\sigma)$.
 We also fix $\mathcal S_{\mu}$ to be a planar embedding of the skeleton $\textrm{sk}(\mu)$ of $\mu$ in which $u$ and $v$ are incident to the outer face. Since $\mu$ is an R-node, there are two such planar embeddings, which are flips of each other. The goal now becomes the one of testing whether $G_{\mu}$ admits a $uv$-external upward planar embedding $\mathcal E_{\mu}$ such that: (P1) the shape description of $\mathcal E_{\mu}$ is $s$;  (P2) for $i=1,\dots,h$, the $u_iv_i$-external upward planar embedding $\mathcal E_{\nu_i}$ of $G_{\nu_i}$ in $\mathcal E_{\mu}$ has shape description $s_i$; and (P3) the planar embedding of $\textrm{sk}(\mu)$ induced by $\mathcal E_{\mu}$ is $\mathcal S_{\mu}$. Then we have that $s$ belongs to $\mathcal{F}_{\mu}$ if and only if this test is successful for at least one selection of the shape descriptions $s_1,\dots, s_h$ and of the planar embedding $\mathcal S_{\mu}$.

We now borrow ideas from an algorithm by Bertolazzi et al.~\cite{bdl-udtd-94} for testing the upward planarity of a digraph $D$ with a prescribed planar embedding $\mathcal E$. 
The algorithm in~\cite{bdl-udtd-94} constructs a bipartite flow network $\mathcal N(S,T,A)$, where each source $s_w\in S$ corresponds to a switch vertex $w$ of $D$, each sink $t_f\in T$ corresponds to a face $f$ of $\mathcal E$, and $A$ has an arc from $s_w$ to $t_f$ if $w$ is incident to $f$. A unit of flow passing from $s_w$ to $t_f$ corresponds to a large angle at $w$ in $f$. Each source  supplies $1$ unit of flow, each arc has capacity  $1$, and each sink $t_f$ demands  $n_f/2-1$ units of flow if $f$ is an internal face of $\mathcal E$  and $n_f/2+1$ if $f$ is the outer face of $\mathcal E$, where $n_f$ is the number of switch angles incident to $f$. Then $D$ has an upward planar embedding which respects $\mathcal E$ if and only if $\mathcal N$ has a flow in which each sink is supplied with a number of units of flow equal to its demand.

After some preliminary checks, which ensure that the values $s,s_1,\dots,s_h,\mathcal S_{\mu}$ are ``coherent'' with each other, we also construct a flow network $\mathcal N(S,T,A)$. Note that the skeleton $\textrm{sk}(\mu)$ of our R-node $\mu$ has a prescribed planar embedding $\mathcal S_{\mu}$. However, the edges of $\textrm{sk}(\mu)$ are not actual edges, but rather virtual edges that correspond to components of $G_{\mu}$. These components introduce new sources, sinks, and arcs in $\mathcal N$, and contribute to the demands of their incident faces. As we have already fixed the shape description $s_i$ of each extreme or interesting component $G_{\nu_i}$, we know the excess of large angles with respect to small angles ``on the sides'' of $G_{\nu_i}$, as these are the first two values of $s_i$. These values introduce sources (if they are positive) and contribute to the demands of the faces of $\mathcal S_{\mu}$ incident to $e_i$. Handling non-extreme boring components is more challenging. Each boring component has at most two shape descriptions in its feasible set, however the number of such components is not, in general, bounded by a function of $\sigma$ only, hence we cannot try all possible combinations for their shape descriptions. Rather, we plug the freedom of choosing a shape description for each non-extreme boring component directly into the flow network. For example, a component $G_{\nu_i}$ such that $\mathcal F_{\nu_i}$ contains the hat and the inverted-hat is modeled by a source with two incident arcs to the faces of $\mathcal S_{\mu}$ incident to $e_i$, reflecting the fact that each of the two shape descriptions provides a large angle in a different face incident to $e_i$. As another example, a component $G_{\nu_i}$ such that $\mathcal F_{\nu_i}$ contains the left-wing and the right-wing also provides a large angle in a different face incident to $e_i$ depending on the choice of the shape description, however in this case the choice might also affect whether a pole of the component creates a switch angle in a face of $\mathcal S_{\mu}$ or not, which affects the demand of the face. This is solved either by ``ignoring'' the component, or by transfering its effect to an adjacent non-switch vertex. 

\begin{figure}[htb]
	\centering
	\begin{subfigure}{.5\textwidth}
		\includegraphics[page=14]{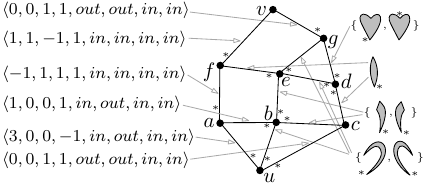}
		\subcaption{}
		\label{fig:network_example_1-short}
	\end{subfigure}
	\hfil
	\begin{subfigure}{.3\textwidth}
		\includegraphics[page=15]{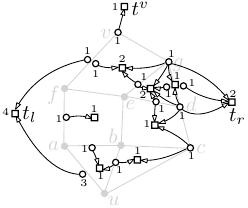}
		\subcaption{}
		\label{fig:network_example_2-short}
	\end{subfigure}
	\\
	\begin{subfigure}{.3\textwidth}
		\includegraphics[page=16]{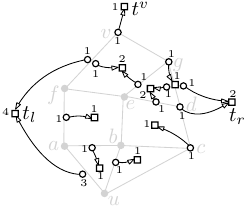}
		\subcaption{}
		\label{fig:network_example_3-short}
	\end{subfigure}
	\hfil
	\begin{subfigure}{.3\textwidth}
		\includegraphics[page=13]{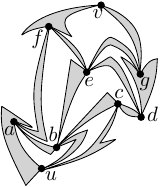}
		\subcaption{}
		\label{fig:network_example_4-short}
	\end{subfigure}
	\caption{The construction of a flow network $\mathcal N$ that allows us to determine whether a shape description $s=\langle 1,0,-1,0,\textrm{in},\textrm{out},\textrm{in},\textrm{in}\rangle$ belongs to $\mathcal F_{\mu}$. (a) shows the input: a shape description $s_i$ for each extreme or interesting component $G_{\nu_i}$ of $G_{\mu}$ and the feasible set $\mathcal F_{\nu_i}$ for each non-extreme boring component $G_{\nu_i}$ of $G_{\mu}$. (b) shows $\mathcal N$; arc capacities are not shown (each of them is equal to the supply of the source of the arc). (c) shows a flow for $\mathcal N$ in which every sink receives an amount of flow equal to its demand; each shown arc is traversed by a flow equal to its capacity. (d) shows a $uv$-external upward planar embedding of $G_{\mu}$ with shape description $s$ corresponding to the flow.}\label{fig:network_example-short}
\end{figure}

Figure~\ref{fig:network_example-short} 
shows an example of the construction of $\mathcal N$. We have that $\mathcal N$ has $\bigoh(k)$ nodes and arcs. We test whether every sink has a non-negative demand and whether $\mathcal N$ admits a flow in which every sink receives an amount of flow equal to its demand. The latter can be done in $\bigoh(k\log^3 k)$ time by means of an algorithm by Borradaile et al.~\cite{bkm-msms-17}. We conclude that $G_{\mu}$ admits a $uv$-external upward planar embedding satisfying Properties~P1--P3 if and only if the tests are successful. This leads to the following.

\begin{lemma}  \label{lem:R_node_sources-short}
The feasible set $\mathcal F_{\mu}$ of an R-node $\mu$ of $T$ can be computed in $\bigoh(\sigma 1.45^\sigma \cdot k\log^3 k)$ time, where $k$ is the number of children of $\mu$ in $T$ and $\sigma$ is the number of sources of $G$. 
\end{lemma}

Lemmata~\ref{lem:R_node_general} and~\ref{lem:R_node_sources-short} imply the following main result.

  \begin{theorem}  \label{th:sources-short}
 	{\sc Upward Planarity} can be solved in $\bigoh(\sigma 1.45^\sigma \cdot n^2\log^3 n)$ time for a digraph with $n$ vertices and $\sigma$ sources. 
 \end{theorem}

 

\fi


\iflong
Let $G$ be an $n$-vertex digraph whose underlying graph is planar and let $\sigma$ be the number of sources of $G$. First, we expand $G$. As discussed in Section~\ref{sec:prelim}, this operation does not alter whether $G$ is upward planar, it preserves the number of sources and at most doubles the number of vertices of $G$.
In view of Lemma~\ref{lem:R_node_general}, in order to obtain an algorithm for {\sc Upward Planarity} parameterized by $\sigma$, it suffices to devise an R-node subprocedure whose runtime depends on $\sigma$ and, polynomially, on $n$. We hence consider a single biconnected component $B$ of $G$ and observe the following.

\begin{observation} \label{obs:maintains-sources}
The digraph $B$ is expanded and it has at most $\sigma$ sources.	
\end{observation} 

\begin{proof}
First, $B$ is expanded since $G$ is expanded and since each vertex of $B$ has a set of incident edges which is a subset of the set of edges incident to the same vertex in $G$. 

We prove that $B$ has at most $\sigma$ sources by induction on the number $b_G$ of biconnected components of $G$. Indeed, if $b_G=1$, then $B=G$ and the statement is trivial. If $b_G>1$, remove from $G$ a biconnected component $B'\neq B$ that has a single cut-vertex $c$; let $G'$ be the resulting graph. 
We prove that $G'$ has at most $\sigma$ sources. This is obvious if $c$ is a source of $G$ or if $c$ is not a source of $G'$. Otherwise, suppose that $c$ is not a source of $G$ and it is a source of $G'$. The two things imply that $c$ is not a source of $B'$; since $G$ is acyclic,  it follows that $B'$ contains a source $s\neq c$ of $B'$; note that $s$ is also a source of $G$. Thus, although $G'$ contains the ``new'' source $c$, it does not contain the source $s$, hence $G'$ has at most $\sigma$ sources. Since $G'$ has one less biconnected component than $G$, the inductive hypothesis applies and we can conclude that $B$ has $\sigma$ sources.
\end{proof}

For the sake of familiarity of notation, in the following we denote by $G$ the considered biconnected digraph, by $n$ the number of its vertices, and by $\sigma$ the number of its sources. Let $e^*$ be any edge of $G$; we compute an SPQR-tree $T$ of $G$ rooted at the Q-node representing $e^*$ in $\bigoh(n)$ time \cite{dt-opl-96,gm-lti-00}. We will show how to compute, for each R-node $\mu$ of $T$, the feasible set $\mathcal F_{\mu}$, assuming that the feasible sets for the children of $\mu$ have been already computed. However, before doing that, we need to establish some combinatorial lemmata. 


\subsection{Turn-Numbers and Size of Feasible Sets}

We start by proving that a bounded number of sources implies bounded left- and right-turn-numbers. Let $\mu$ be any node of $T$, let $u$ and $v$ be the poles of $\mu$, let $\sigma_{\mu}$ be the number of sources of $G_{\mu}$ different from its poles, and let $\mathcal E_{\mu}$ be any $uv$-external upward planar embedding of $G_{\mu}$.  We have the following.

\begin{lemma}\label{lem:spirality_range}
The left- and right-turn-numbers of $\mathcal E_{\mu}$ are in the interval $[-2\sigma_{\mu}-1,2\sigma_{\mu}+1]$.
\end{lemma}

\begin{proof}
Let $\langle v_0=u, v_1, v_2, \dots, v_k=v \rangle$ and $\langle w_0=u,w_1,w_2,\dots,w_h=v\rangle$ be the left and right outer paths of $\mathcal E_{\mu}$, respectively. For $i=0,1,\dots,k$, denote by $\alpha_i$ the angle at $v_i$ inside the outer face of $\mathcal E_{\mu}$ and, for $i=0,1,\dots,h$, by $\beta_i$ the angle at $w_i$ inside the outer face of $\mathcal E_{\mu}$.

We first prove that the left-turn-number of $\mathcal E_{\mu}$ is at most $2\sigma_{\mu}+1$. This is done by proving, for every $m=1,\dots,k-1$, the following claim. Let $\sigma_m$ be the number of sources of $G_{\mu}$ among $v_1,\dots,v_m$ and let $p_m\leq m$ be the largest index such that $\alpha_{p_m}$ is a large angle, where $p_m$ is undefined if each angle among $\alpha_1,\dots,\alpha_m$ is small or flat. Note that, if $p_m$ is defined, then $v_{p_m}$ is either a source or a sink of $G_{\mu}$, by properties \textbf{UP1}-\textbf{UP2} of Theorem~\ref{th:upward-conditions}. Then the claim is that:

\begin{itemize}
	\item $\sum_{i=1}^{m}\lambda(\alpha_i)$ is smaller than or equal to $2\sigma_m+1$, if the index $p_m$ is defined and $v_{p_m}$ is a sink; and
	\item $\sum_{i=1}^{m}\lambda(\alpha_i)$ is smaller than or equal to $2\sigma_m$, if the index ${p_m}$ is undefined or $v_{p_m}$ is a source.
\end{itemize} 
Observe that the claim, applied with $m=k-1$, implies that the left-turn-number of $\mathcal E_{\mu}$ is at most $2\sigma_{\mu}+1$, since $\sigma_{k-1}\leq \sigma_{\mu}$.  

The proof is by induction on $m$. In the base case $m=1$. We distinguish some cases, according to the type of $\alpha_1$.

\begin{itemize}
	\item If $\alpha_1$ is a non-large angle, then $\lambda(\alpha_1)\leq 0=2\sigma_1$.
	\item If $\alpha_1$ is a large angle and $v_1$ is a source of $G_{\mu}$, then $\lambda(\alpha_1)=1<2=2\sigma_1$.
	\item If $\alpha_1$ is a large angle and $v_1$ is a sink of $G_{\mu}$, then $\lambda(\alpha_1)=1=2\sigma_1+1$.
\end{itemize}

For the inductive case, suppose that the claim is true for some $m \in \{1,\dots,k-2\}$, we prove that the claim is true for $m+1$ as well. We distinguish some cases, depending on the type of $\alpha_{m+1}$ and on whether $p_m$ is defined or not. 

\begin{itemize}
	\item First, suppose that $\alpha_{m+1}$ is a non-large angle and that $p_m$ is undefined; then $p_{m+1}$ is undefined, $\sigma_{m+1}=\sigma_m$, and $\lambda(\alpha_{m+1})\leq 0$. It follows that $\sum_{i=1}^{m+1}\lambda(\alpha_i)\leq \sum_{i=1}^m\lambda(\alpha_i)\leq 2\sigma_m=2\sigma_{m+1}$, where the second inequality uses the inductive hypothesis.
	\item Second, suppose that $\alpha_{m+1}$ is a non-large angle, that $p_m$ is defined, and that $v_{p_m}$ is a source; then $v_{p_{m+1}}=v_{p_m}$ is a source, $\sigma_{m+1}=\sigma_m$, and $\lambda(\alpha_{m+1})\leq 0$. As in the previous case, it follows that $\sum_{i=1}^{m+1}\lambda(\alpha_i)\leq \sum_{i=1}^m\lambda(\alpha_i)\leq 2\sigma_m=2\sigma_{m+1}$.
	\item Third, suppose that $\alpha_{m+1}$ is a non-large angle, that $p_m$ is defined, and that $v_{p_m}$ is a sink; then $v_{p_{m+1}}=v_{p_m}$ is a sink, $\sigma_{m+1}=\sigma_m$, and $\lambda(\alpha_{m+1})\leq 0$. It follows that $\sum_{i=1}^{m+1}\lambda(\alpha_i)=\sum_{i=1}^m\lambda(\alpha_i)\leq 2\sigma_m+1=2\sigma_{m+1}+1$.
	\item Fourth, suppose that $\alpha_{m+1}$ is a large angle and that $v_{m+1}$ is a source; then $\sigma_{m+1}=\sigma_m+1$ and $\lambda(\alpha_{m+1})=1$. It follows that $\sum_{i=1}^{m+1}\lambda(\alpha_i)=1+\sum_{i=1}^m\lambda(\alpha_i)\leq 1+(2\sigma_m+1)=2\sigma_{m+1}$. 
	\item Fifth, suppose that $\alpha_{m+1}$ is a large angle, that $v_{m+1}$ is a sink, and that $p_m$ is undefined; then $\lambda(\alpha_{m+1})=1$. Further, each angle among $\alpha_1,\dots,\alpha_m$ is small or flat, hence $\sum_{i=1}^m\lambda(\alpha_i)\leq 0$. It follows that $\sum_{i=1}^{m+1}\lambda(\alpha_i)\leq 1\leq 2\sigma_{m+1}+1$.
	\item Sixth, suppose that $\alpha_{m+1}$ is a large angle, that $v_{m+1}$ is a sink, that $p_m$ is defined, and that $v_{p_{m}}$ is a source; then $\sigma_{m+1}=\sigma_m$ and $\lambda(\alpha_{m+1})=1$. It follows that $\sum_{i=1}^{m+1}\lambda(\alpha_i)=1+\sum_{i=1}^m\lambda(\alpha_i)\leq 1+2\sigma_m=2\sigma_{m+1}+1$. 
	\item Finally, suppose that $\alpha_{m+1}$ is a large angle, that $v_{m+1}$ is a sink, that $p_m$ is defined, and that $v_{p_{m}}$ is a sink; then $\sigma_{m+1}=\sigma_m$ and $\lambda(\alpha_{m+1})=1$. 
	
	With respect to the previous cases, here we need to prove the following observation: there exists (at least) one angle among $\alpha_{p_m+1},\alpha_{p_m+2},\dots,\alpha_{m}$ that is small. Indeed, by definition of $p_m$, each angle among $\alpha_{p_m+1},\alpha_{p_m+2},\dots,\alpha_{m}$ is either small or flat. Further, since $v_{p_{m}}$ is a sink, the edge between $v_{p_{m}}$ and $v_{p_{m}+1}$ is incoming $v_{p_{m}}$. Now consider the smallest index $i\geq 1$ such that the edge between $v_{p_{m}+i}$ and $v_{p_{m}+i+1}$ is outgoing $v_{p_{m}+i}$; such an index exists, as otherwise $v_{m+1}$ would not be a sink. Then the edges connecting $v_{p_{m}+i}$ to $v_{p_{m}+i-1}$ and $v_{p_{m}+i+1}$ are both outgoing $v_{p_{m}+i}$; since $\alpha_{p_m+i}$ is not large, it is small, which proves the claimed observation. 
	
	It follows that $\sum_{i=1}^{m+1}\lambda(\alpha_i)=\sum_{i=1}^{p_m}\lambda(\alpha_i)+\sum_{i=p_{m+1}}^{m}\lambda(\alpha_i)+\lambda(\alpha_{m+1})\leq (2\sigma_{m}+1) -1+1=2\sigma_{m+1}+1$. 
\end{itemize}

This completes the induction, hence proves the claim and that the left-turn-number of $\mathcal E_{\mu}$ is at most $2\sigma_{\mu}+1$. An analogous proof shows that the right-turn-number of $\mathcal E_{\mu}$ is also at most $2\sigma_{\mu}+1$. From this, it follows that $\sum_{i=0}^{k}\lambda(\alpha_i)\leq 2\sigma_{\mu}+3$ and $\sum_{i=0}^{h}\lambda(\beta_i)\leq 2\sigma_{\mu}+3$ (the two sums are the left- and right-turn-number of $\mathcal E_{\mu}$, respectively, increased by the angles at $u$ and $v$ inside the outer face of $\mathcal E_{\mu}$).  By property~\textbf{UP3} of Theorem~\ref{th:upward-conditions}, we have that $\sum_{i=1}^{k-1}\lambda(\alpha_i)+\sum_{i=0}^{h}\lambda(\beta_i)=2$, hence $\sum_{i=1}^{k-1}\lambda(\alpha_i)=2-\sum_{i=0}^{h}\lambda(\beta_i)\geq 2-(2\sigma_{\mu}+3)=-2\sigma_{\mu}-1$, hence the left-turn-number of $\mathcal E_{\mu}$ is at least $-2\sigma_{\mu}-1$. Analogously, the right-turn-number of $\mathcal E_{\mu}$ is at least $-2\sigma_{\mu}-1$, which concludes the proof of the lemma.
\end{proof} 

A direct consequence of Lemma~\ref{lem:spirality_range} is that we can bound the size of a feasible set $\mathcal{F}_{\mu}$ of $\mu$.

\begin{lemma}\label{lem:feasible_size}
	The size of the feasible set $\mathcal{F}_{\mu}$ of $\mu$ is at most $72\sigma_{\mu}+54$.
\end{lemma}

\begin{proof}
	Note that $G_{\mu}$ contains at most $\sigma_{\mu}$ sources other than $u$ and $v$. By Lemma~\ref{lem:spirality_range}, the shape description of $\mathcal E_{\mu}$ has the value $\tau_l(\mathcal E_{\mu},u,v)$ in $[-2\sigma_{\mu}-1,2\sigma_{\mu}+1]$, hence $\tau_l(\mathcal E_{\mu},u,v)$ can have one of $4\sigma_{\mu}+3$ possible values. 
	By Lemma~\ref{lem:matrix_feasible} there are at most 18 shape descriptions with given $\tau_l(\mathcal E_{\mu},u,v)$. Hence, $\mathcal F_{\mu}$ has at most $(4\sigma_{\mu}+3)\cdot 18=72\sigma_{\mu}+54$ possible shape descriptions.
\end{proof}

\subsection{An R-node Subprocedure} \label{subse:sources_R}

%

In the following, we show an algorithm running in $\bigoh(\sigma 1.45^\sigma \cdot k\log^3 k))$ time that computes the feasible set $\mathcal{F}_{\mu}$ of an R-node $\mu$ of $T$ with $k$ children. Let $u$ and $v$ be the poles of $\mu$. Let $\sigma_{\mu}$ be the number of sources of $G_{\mu}$ different from its poles. Furthermore, let $\nu_1,\dots,\nu_k$ be the children of $\mu$ and, for $i=1,\dots,k$, let $u_i$ and $v_i$ be the poles of $\nu_i$ and $e_i$ be the virtual edge representing $\nu_i$ in the skeleton $\textrm{sk}(\mu)$ of $\mu$. We assume that we have the feasible set $\mathcal{F}_{\nu_i}$ for each child $\nu_i$ of $\mu$.

We first introduce two classifications of the components of $G_\mu$ (that is, $G_{\nu_1},\dots,G_{\nu_k}$). First, a component $G_{\nu_i}$ is \emph{boring} if it does not contain any sources except, possibly, for its poles, while it is \emph{interesting} if it contains at least one source different from its poles. Because $G_{\mu}$ has $\sigma_{\mu}$ sources, there are at most $\sigma_{\mu}$ interesting components among $G_{\nu_1},\dots,G_{\nu_k}$, while there can be any number of boring components. Second, a component $G_{\nu_i}$ is \emph{extreme} if $e_i$ is incident to a pole of $\mu$ and is incident to the face of the (unique) planar embedding of $\textrm{sk}(\mu)$ containing $u$ and $v$, it is \emph{non-extreme} otherwise. Note that there are four extreme components among $G_{\nu_1},\dots,G_{\nu_k}$. Clearly, an $\bigoh(n)$-time pre-processing of $G$ allows us to equip every virtual edge $e$ in the skeleton of every R-node $\tau$ of $T$ with a label, indicating whether the component of $G_{\tau}$ corresponding to $e$ is boring, interesting, extreme, or non-extreme. With such a labeling, we can produce in $\bigoh(k)$ time (and thus in $\bigoh(n)$ time over all the R-nodes of $T$) an ordering of the components of $G_{\mu}$ such that all the extreme or interesting components precede all the non-extreme boring components. We assume, w.l.o.g., that $G_{\nu_1},\dots,G_{\nu_k}$ is such an order, where $G_{\nu_1},\dots,G_{\nu_h}$ are the extreme or interesting components. 

The outline of our algorithm to compute the feasible set $\mathcal{F}_{\mu}$ of $\mu$ is as follows. For each shape description $s=\shapeDesc{\tau_l}{\tau_r}{\lambda^u}{\lambda^v}{\rho^u_l}{\rho^u_r}{\rho^v_l}{\rho^v_r}$, where $\tau_l \in [-2\sigma_{\mu}-1,2\sigma_{\mu}+1]$, $\tau_r \in [-\tau_l,-\tau_l+4]$, $\lambda^u\in \{-1,0,1\}$, $\lambda^v\in \{-1,0,1\}$, $\rho^u_l\in \{\textrm{in},\textrm{out}\}$, $\rho^u_r\in \{\textrm{in},\textrm{out}\}$, $\rho^v_l\in \{\textrm{in},\textrm{out}\}$, and $\rho^v_r\in \{\textrm{in},\textrm{out}\}$, we independently test whether $G_{\mu}$ admits a $uv$-external upward planar embedding with shape description $s$; then $\mathcal{F}_{\mu}$ contains the shape descriptions for which the test was successful. In order to test whether $G_{\mu}$ admits a $uv$-external upward planar embedding with shape description $s$, we consider all possible combinations of shape descriptions for the extreme and interesting components of $G_{\mu}$. For each of these combinations, we create a flow network that admits a flow of a certain size if and only if $G_{\mu}$ admits a $uv$-external upward planar embedding whose shape description is $s$ and in which the shape descriptions of the extreme and interesting components are those in the considered combination.

\subsubsection{Boring components} \label{subsub:boring_components}

Despite their name, boring components play an important role in our algorithm. In the following we discuss their features. 

\begin{lemma} \label{lem:shape_desc_boring}
		If $G_{\nu_i}$ is a boring component, then the shape description of any $u_iv_i$-external upward planar embedding of $G_{\nu_i}$ is one of the following:
		\begin{enumerate}
			\item $\shapeDesc{0}{0}{1}{1}{\textrm{out}}{\textrm{out}}{\textrm{in}}{\textrm{in}}$ (see Fig.~\ref{fig:boring_sausage});
			\item $\shapeDesc{0}{0}{1}{1}{\textrm{in}}{\textrm{in}}{\textrm{out}}{\textrm{out}}$ (see Fig.~\ref{fig:boring_sausage} with $u_i$ and $v_i$ inverted);
			\item $\shapeDesc{0}{1}{1}{0}{\textrm{out}}{\textrm{out}}{\textrm{in}}{\textrm{out}}$ (see Fig.~\ref{fig:boring_right});
			\item $\shapeDesc{1}{0}{0}{1}{\textrm{out}}{\textrm{in}}{\textrm{out}}{\textrm{out}}$ (see Fig.~\ref{fig:boring_right} with $u_i$ and $v_i$ inverted);
			\item $\shapeDesc{1}{0}{1}{0}{\textrm{out}}{\textrm{out}}{\textrm{out}}{\textrm{in}}$ (see Fig.~\ref{fig:boring_left});
			\item $\shapeDesc{0}{1}{0}{1}{\textrm{out}}{\textrm{in}}{\textrm{out}}{\textrm{out}}$ (see Fig.~\ref{fig:boring_left} with $u_i$ and $v_i$ inverted);
			\item $\shapeDesc{-1}{1}{1}{1}{\textrm{out}}{\textrm{out}}{\textrm{out}}{\textrm{out}}$ (see Fig.~\ref{fig:boring_hat}); \item $\shapeDesc{1}{-1}{1}{1}{\textrm{out}}{\textrm{out}}{\textrm{out}}{\textrm{out}}$ (see Fig.~\ref{fig:boring_hat}  with $u_i$ and $v_i$ inverted);
			\item $\shapeDesc{1}{1}{1}{-1}{\textrm{out}}{\textrm{out}}{\textrm{out}}{\textrm{out}}$ (see Fig.~\ref{fig:boring_heart}); and
			\item $\shapeDesc{1}{1}{-1}{1}{\textrm{out}}{\textrm{out}}{\textrm{out}}{\textrm{out}}$ (see Fig.~\ref{fig:boring_heart}  with $u_i$ and $v_i$ inverted);
		\end{enumerate}
	\end{lemma}

	\begin{figure}[!t]
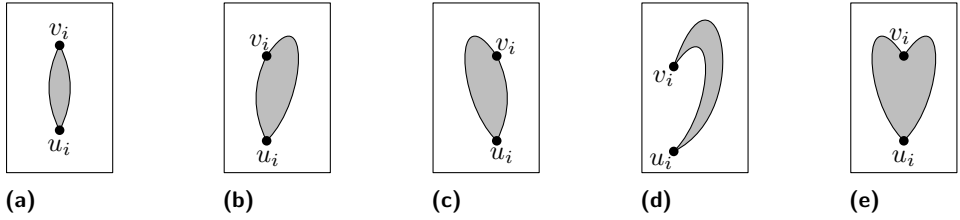

		\centering
		{\centering
			\begin{subfigure}{.19\textwidth}
				\includegraphics[page=1]{figures/boring_comps}
				\subcaption{}
				\label{fig:boring_sausage}
			\end{subfigure}
		}		
		\hfil
		{\centering
			\begin{subfigure}{.19\textwidth}
				\includegraphics[page=2]{figures/boring_comps}
				\subcaption{}
				\label{fig:boring_right}
			\end{subfigure}
		}\hfil
		{\centering
			\begin{subfigure}{.19\textwidth}
				\includegraphics[page=3]{figures/boring_comps}
				\subcaption{}
				\label{fig:boring_left}
			\end{subfigure}
		}\hfil
		{\centering
			\begin{subfigure}{.19\textwidth}
				\includegraphics[page=5]{figures/boring_comps}
				\subcaption{}
				\label{fig:boring_hat}
			\end{subfigure}
		}\hfil
		{\centering
			\begin{subfigure}{.19\textwidth}
				\includegraphics[page=4]{figures/boring_comps}
				\subcaption{}
				\label{fig:boring_heart}
			\end{subfigure}
		}
		\caption{Shape descriptions of boring components.}\label{fig:boring_comps}
	\end{figure}

  \begin{proof}
Consider any $uv$-external upward planar embedding $\mathcal E_{\nu_i}$ of $G_{\nu_i}$. Let $s_{\nu_i}$ be the shape description of $\mathcal E_{\nu_i}$. Since $G_{\nu_i}$ is a boring component, only $u_i$ and $v_i$ can be sources of $G_{\nu_i}$. Then, by Lemma~\ref{lem:spirality_range}, the left- and right-turn-numbers $\tau_l(\mathcal E_{\nu_i},u_i,v_i)$ and $\tau_r(\mathcal E_{\nu_i},u_i,v_i)$ of $\mathcal E_{\nu_i}$ can only take the values $-1$, $0$, and $1$. 

\begin{itemize}
	\item Assume first that both $u_i$ and $v_i$ are sources of $G_{\nu_i}$. It follows that $\{\lambda(\mathcal E_{\nu_i},u_i),\lambda(\mathcal E_{\nu_i},v_i)\}\subseteq\{-1,1\}$, that the last four values of $s_{\nu_i}$ are all $\textrm{out}$, and that $\tau_l(\mathcal E_{\nu_i},u_i,v_i)$ and $\tau_r(\mathcal E_{\nu_i},u_i,v_i)$ are both odd. By Corollary~\ref{cor:dependence-parameters}, we have that $\tau_l(\mathcal E_{\nu_i},u_i,v_i)+\tau_r(\mathcal E_{\nu_i},u_i,v_i)+\lambda(\mathcal E_{\nu_i},u_i)+\lambda(\mathcal E_{\nu_i},v_i)=2$. Hence:
	
	\begin{itemize}
		\item If $\lambda(\mathcal E_{\nu_i},u_i)=\lambda(\mathcal E_{\nu_i},v_i)=1$, then one between $\tau_l(\mathcal E_{\nu_i},u_i,v_i)$ and $\tau_r(\mathcal E_{\nu_i},u_i,v_i)$ is $-1$ and the other one is $1$ (see Fig.~\ref{fig:boring_hat}, possibly with $u_i$ and $v_i$ inverted).
		\item If $\{\lambda(\mathcal E_{\nu_i},u_i),\lambda(\mathcal E_{\nu_i},v_i)\}=\{-1,+1\}$, then $\tau_l(\mathcal E_{\nu_i},u_i,v_i)$ and $\tau_r(\mathcal E_{\nu_i},u_i,v_i)$ are both $1$ (see Fig.~\ref{fig:boring_heart}, possibly with $u_i$ and $v_i$ inverted).
		\item Note that $\lambda(\mathcal E_{\nu_i},u_i)$ and $\lambda(\mathcal E_{\nu_i},v_i)$ cannot be both $-1$, as otherwise $\tau_l(\mathcal E_{\nu_i},u_i,v_i)+\tau_r(\mathcal E_{\nu_i},u_i,v_i)+\lambda(\mathcal E_{\nu_i},u_i)+\lambda(\mathcal E_{\nu_i},v_i)=2$ could not be satisfied, as $\tau_l(\mathcal E_{\nu_i},u_i,v_i)+\tau_r(\mathcal E_{\nu_i},u_i,v_i)\leq 2$.
	\end{itemize}  
	\item Assume next that $u_i$ is a source of $G_{\nu_i}$ and $v_i$ is not. It follows that the fifth and sixth values of $s_{\nu_i}$ are both $\textrm{out}$.
	\begin{itemize}
		\item If $v_i$ is a sink of $G_{\nu_i}$, then the seventh and eighth values of $s_{\nu_i}$ are both $\textrm{in}$. Hence, $\tau_l(\mathcal E_{\nu_i},u_i,v_i)$ and $\tau_r(\mathcal E_{\nu_i},u_i,v_i)$ are even and thus they are both $0$. By Corollary~\ref{cor:dependence-parameters}, we have that $\tau_l(\mathcal E_{\nu_i},u_i,v_i)+\tau_r(\mathcal E_{\nu_i},u_i,v_i)+\lambda(\mathcal E_{\nu_i},u_i)+\lambda(\mathcal E_{\nu_i},v_i)=2$, hence $\lambda(\mathcal E_{\nu_i},u_i)=\lambda(\mathcal E_{\nu_i},v_i)=1$ (see Fig.~\ref{fig:boring_sausage}). 
		\item If $v_i$ is a non-switch vertex of $G_{\nu_i}$, then $\lambda(\mathcal E_{\nu_i},v_i)=0$; further, either the seventh value of $s_{\nu_i}$ is $\textrm{in}$ and the eight value is $\textrm{out}$, or vice versa. In the former case (see Fig.~\ref{fig:boring_right}), $\tau_l(\mathcal E_{\nu_i},u_i,v_i)$ is even (and thus it is $0$) and $\tau_r(\mathcal E_{\nu_i},u_i,v_i)$ is odd (and thus it is $1$, as if it were $-1$, then $\tau_l(\mathcal E_{\nu_i},u_i,v_i)+\tau_r\mathcal E_{\nu_i},u_i,v_i)+\lambda(\mathcal E_{\nu_i},u_i)+\lambda(\mathcal E_{\nu_i},v_i)=2$ could not be satisfied), thus $\lambda(\mathcal E_{\nu_i},u_i)=1$. In the latter case (see Fig.~\ref{fig:boring_left}), $\tau_r(\mathcal E_{\nu_i},u_i,v_i)$ is even (and thus it is $0$) and $\tau_l(\mathcal E_{\nu_i},u_i,v_i)$ is odd (and thus it is $1$, as if it were $-1$, then $\tau_l(\mathcal E_{\nu_i},u_i,v_i)+\tau_r(\mathcal E_{\nu_i},u_i,v_i)+\lambda(\mathcal E_{\nu_i},u_i)+\lambda(\mathcal E_{\nu_i},v_i)=2$ could not be satisfied), thus $\lambda(\mathcal E_{\nu_i},u_i)=1$.
	\end{itemize}
\end{itemize}
The case in which $v_i$ is a source of $G_{\nu_i}$ and $u_i$ is not is symmetric to the previous one (see Figs.~\ref{fig:boring_sausage}--\ref{fig:boring_left} with $u_i$ and $v_i$ inverted). Finally, one of $u_i$ and $v_i$ is a source, given that $G_{\nu_i}$ is acyclic.
\end{proof} 

In the following, we call (inverted-) \emph{sausage} the shape description in item 1 (resp.\ item 2) of Lemma~\ref{lem:shape_desc_boring}, (inverted-) \emph{right-wing} the shape descriptions in item 3 (resp.\ item 4) of Lemma~\ref{lem:shape_desc_boring}, (inverted-) \emph{left-wing} the shape description in item 5 (resp.\ item 6) of Lemma~\ref{lem:shape_desc_boring}, (inverted-) \emph{hat} the shape descriptions in item 7 (resp.\ item 8) of Lemma~\ref{lem:shape_desc_boring}, and (inverted-) \emph{heart} the shape descriptions in item 9 (resp.\ item 10) of Lemma~\ref{lem:shape_desc_boring}. 
The next lemma describes the combinations of shape descriptions that can appear in the feasible set of a non-extreme boring component and which shape descriptions are ``better'' than others.


  \begin{lemma}  \label{lem:boring-relationship}
	For every $i=1,\dots,h$, the following statements hold true:
	\begin{description}
		\item[S1] If $\mathcal F_{\nu_i}$ contains the sausage or the inverted-sausage, then $|\mathcal F_{\nu_i}|=1$, i.e., $\mathcal F_{\nu_i}$ contains no other shape description. 
		\item[S2] If $\mathcal F_{\nu_i}$ contains the right-wing (the left-wing), then it contains the left-wing (resp.\ the right-wing), it might contain the heart, and it does not contain any other shape description.
		\item[S3] If $\mathcal F_{\nu_i}$ contains the inverted-right-wing (the inverted-left-wing), then it contains the inverted-left-wing (the inverted-right-wing), it might contain the inverted-heart, and it does not contain any other shape description.
		\item[S4] If $\mathcal F_{\nu_i}$ contains the hat (the inverted-hat), then it contains the inverted-hat (resp.\ the hat), it might contain the heart, it might contain the inverted-heart, and it does not contain any other shape description.
		\item[S5] Suppose that $G_{\mu}$ admits a $uv$-external upward planar embedding $\mathcal E_{\mu}$ with shape description $s$, in which the shape description of the $u_iv_i$-external upward planar embedding $\mathcal E_{\nu_i}$ of $G_{\nu_i}$ is the heart. Suppose also that $G_{\nu_i}$ admits a $u_iv_i$-external upward planar embedding $\mathcal E'_{\nu_i}$ whose shape description is the right-wing, the left-wing, the hat, or the inverted-hat. Then replacing $\mathcal E_{\nu_i}$ with $\mathcal E'_{\nu_i}$ in $\mathcal E_{\mu}$ results in a $uv$-external upward planar embedding $\mathcal E'_{\mu}$ of $G_{\mu}$ with shape description $s$.
		\item[S6] Suppose that $G_{\mu}$ admits a $uv$-external upward planar embedding $\mathcal E_{\mu}$ with shape description $s$, in which the shape description of the $u_iv_i$-external upward planar embedding $\mathcal E_{\nu_i}$ of $G_{\nu_i}$ is the inverted-heart. Suppose also that $G_{\nu_i}$ admits a $u_iv_i$-external upward planar embedding $\mathcal E'_{\nu_i}$ whose shape description is the inverted-right-wing, the inverted-left-wing, the hat, or the inverted-hat. Then replacing $\mathcal E_{\nu_i}$ with $\mathcal E'_{\nu_i}$ in $\mathcal E_{\mu}$ results in a $uv$-external upward planar embedding $\mathcal E'_{\mu}$ of $G_{\mu}$ with shape description $s$.
	\end{description}
\end{lemma}
		
  \begin{proof}
	We prove statement S1. Suppose that $\mathcal F_{\nu_i}$ contains the sausage; the case in which $\mathcal F_{\nu_i}$ contains the inverted-sausage is analogous. Let $\mathcal E_{\nu_i}$ be a $u_iv_i$-external upward planar embedding of $G_{\nu_i}$ whose shape description is the sausage. Since both the edges incident to $v_i$ and to the outer face of $\mathcal E_{\nu_i}$ are incoming $v_i$ (as the seventh and eighth values in the sausage are both in) and since the angle at $v_i$ incident to the outer face of $\mathcal E_{\nu_i}$ is large (as the fourth value in the sausage is $1$), it follows that $v_i$ contains no outgoing edge. In all other shape descriptions of a boring component, the seventh and eighth values contain at least one out (see items 2--10 in Lemma~\ref{lem:shape_desc_boring}), hence they require at least one outgoing edge for $v_i$. Hence, the sausage is the only shape description in the feasible set of $\mu$. 
	
	We prove statement S2. Suppose that $\mathcal F_{\nu_i}$ contains the right-wing; the case in which $\mathcal F_{\nu_i}$ contains the left-wing is analogous. Let $\mathcal E_{\nu_i}$ be a $u_iv_i$-external upward planar embedding of $G_{\nu_i}$ whose shape description is the right-wing. Then a $u_iv_i$-external upward planar embedding of $G_{\nu_i}$ whose shape description is the left-wing can be simply obtained by flipping $\mathcal E_{\nu_i}$. By statement S1, we have that $\mathcal F_{\nu_i}$ does not contain the (possibly inverted-) sausage. Further, $v_i$ contains both incoming and outgoing edges (as the seventh and eighth values in the right-wing are one in and one out). On the other hand, each shape description among the inverted-right-wing, the inverted-left-wing, the hat, the inverted-hat, and the inverted-heart requires $v_i$ to only have outgoing edges (as, in each of such shape descriptions, the seventh and eighth values are both out and the fourth value is $1$, see items 4, 6, 7, 8, and 10 in Lemma~\ref{lem:shape_desc_boring}). Hence, $\mathcal F_{\nu_i}$ does not contain any of these shape descriptions. 
	
	The proof of statement S3 is analogous to the one of statement S2.
	
	We prove statement S4. Suppose that $\mathcal F_{\nu_i}$ contains the hat; the case in which $\mathcal F_{\nu_i}$ contains the inverted-hat is analogous. Let $\mathcal E_{\nu_i}$ be a $u_iv_i$-external upward planar embedding of $G_{\nu_i}$ whose shape description is the hat. Then a $u_iv_i$-external upward planar embedding of $G_{\nu_i}$ whose shape description is the inverted-hat can be simply obtained by flipping $\mathcal E_{\nu_i}$. By statements S1--S3, we have that $\mathcal F_{\nu_i}$ contains neither the (possibly inverted-) sausage, nor the (possibly inverted-) right-wing, nor the (possibly inverted-) left-wing. 
	
	We finally prove statement S5; the proof of statement S6 is analogous. Suppose that the shape description of $\mathcal E_{\nu_i}$ is the heart and that the one of $\mathcal E'_{\nu_i}$ is the right-wing; the other cases are analogous. Let $f$ ($g$) be the face of $\mathcal E_{\mu}$ that is to the left (resp.\ to the right) of $\mathcal E_{\nu_i}$ when traversing the left (resp.\ right) outer path of $\mathcal E_{\nu_i}$ from $u_i$ to $v_i$. Similarly, let $f'$ ($g'$) be the face of $\mathcal E'_{\mu}$ that is to the left (resp.\ to the right) of $\mathcal E'_{\nu_i}$ when traversing the left (resp.\ right) outer path of $\mathcal E'_{\nu_i}$ from $u_i$ to $v_i$. 
	
	\begin{figure}[!t]
		\centering
		{\centering
			\begin{subfigure}{.19\textwidth}
				\includegraphics[page=1,scale=0.85]{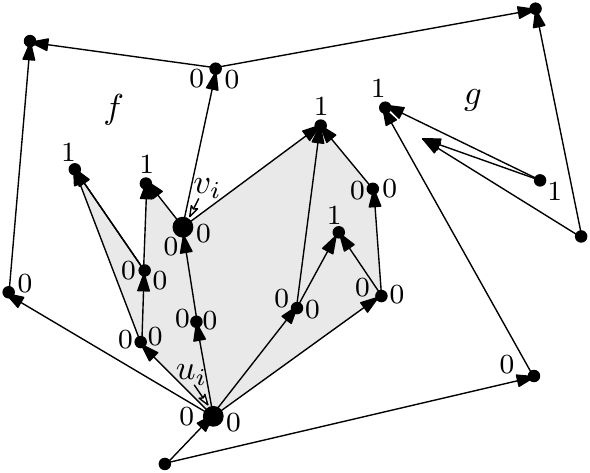}
				\subcaption{}
				\label{fig:boring_better_1}
			\end{subfigure}
		}		
		\hfil
		{\centering
			\begin{subfigure}{.19\textwidth}
				\includegraphics[page=2,scale=0.85]{figures/boring_better}
				\subcaption{}
				\label{fig:boring_better_2}
			\end{subfigure}
		}
		\caption{(a) The upward planar embedding $\mathcal E_{\mu}$, where only $\mathcal E_{\nu_i}$ and the faces $f$ and $g$ are shown. The internal faces of $\mathcal E_{\nu_i}$ are gray. The angles whose values are different from $-1$ are shown. (b) The upward planar embedding $\mathcal E'_{\mu}$, where only $\mathcal E'_{\nu_i}$ and the faces $f'$ and $g'$ are shown. The internal faces of $\mathcal E'_{\nu_i}$ are gray. The angles whose values are different from $-1$ are shown.}\label{fig:boring_better}
	\end{figure}

	Let us specify how the replacement of $\mathcal E_{\nu_i}$ with $\mathcal E'_{\nu_i}$ in $\mathcal E_{\mu}$ defines the angle assignment of $\mathcal E'_{\mu}$; refer to Fig.~\ref{fig:boring_better}. First, every face $h'$ of $\mathcal E'_{\mu}$ different from $f'$ and $g'$ is also a face of $\mathcal E_{\mu}$ or of $\mathcal E'_{\nu_i}$, hence the value in $\{-1,0,1\}$ assigned to a vertex incident to $h'$ in $\mathcal E'_{\mu}$ coincides with the one in $\mathcal E_{\mu}$. Second, for each vertex $w\notin \{u_i,v_i\}$ incident to $f'$ (resp.\ to $g'$), the value in $\{-1,0,1\}$ assigned to $w$ in $f'$ (resp.\ in $g'$) coincides with the value assigned to $w$ in the outer face of $\mathcal E'_{\nu_i}$. Third, the value in $\{-1,0,1\}$ assigned to $u_i$ in $f'$ (resp.\ in $g'$) coincides with the one assigned to $u_i$ in $f$ (resp.\ in $g$). Finally, the values assigned to $v_i$ in $f'$ and in $g'$ are $0$ and $-1$, respectively.   
	
	In order to prove that $\mathcal E'_{\mu}$ is a $uv$-external upward planar embedding of $G_{\mu}$, we need to prove that Properties \textbf{UP0}--\textbf{UP3} of Theorem~\ref{th:upward-conditions} are satisfied by $\mathcal E'_{\mu}$. 
	
	\begin{description}
		\item[UP0] Consider the angle $\alpha=\langle e'_1,w,e'_2\rangle$ at a vertex $w$ in a face $h'$ of $\mathcal E'_{\mu}$. 
		
		If $h'\notin \{f',g'\}$ or if $w\notin \{u_i,v_i\}$, then $\alpha$ is also an angle of $\mathcal E_{\mu}$ or of $\mathcal E'_{\nu_i}$ and, by construction, the value assigned to $\alpha$ in $\mathcal E'_{\mu}$ is the same as the one assigned to $\alpha$ in $\mathcal E_{\mu}$ or of $\mathcal E'_{\nu_i}$. Hence, from the fact that $\mathcal E_{\mu}$ and $\mathcal E'_{\nu_i}$ are upward planar embeddings, it follows that the value assigned to $\alpha$ is in $\{-1,1\}$ if $\alpha$ is a switch angle, and $0$ if $\alpha$ is a flat angle.
		
		If $h'\in \{f',g'\}$, say $h'=f'$, and $w=u_i$, then $e'_1$ is the edge incident to $u_i$ in the left outer path of $\mathcal E'_{\nu_i}$ and $e'_2$ is the edge incident to $u_i$ and different from $e'_1$ in the boundary of $f'$. Let $\beta=\langle e_1,w,e_2\rangle$ be the angle at $w$ in $f$, where $e'_1$ is the edge incident to $u_i$ in the left outer path of $\mathcal E_{\nu_i}$ and $e_2$ is the edge incident to $u_i$ and different from $e_1$ in the boundary of $f$. Note that $e'_2=e_2$. Further, since the shape descriptions of $\mathcal E_{\nu_i}$ and $\mathcal E'_{\nu_i}$ both have the fifth value equal to out (see items 9 and 3 in Lemma~\ref{lem:shape_desc_boring}, respectively), $e_1$ and $e'_1$ are both outgoing $u_i$. Hence, $\alpha$ is a switch angle or a flat angle if and only if $\beta$ is a switch angle or a flat angle, respectively. By construction, the value assigned to $\alpha$ in $\mathcal E'_{\mu}$ coincides with the value assigned to $\beta$ in $\mathcal E_{\mu}$. Since $\mathcal E_{\mu}$ is an upward planar embedding, the value assigned to $\alpha$ is in $\{-1,1\}$ if $\alpha$ is a switch angle, and $0$ if $\alpha$ is a flat angle.
		
		If $h'=f'$ and $w=v_i$, then $e'_1$ is the edge incident to $v_i$ in the left outer path of $\mathcal E'_{\nu_i}$ and $e'_2$ is the edge incident to $v_i$ and different from $e'_1$ in the boundary of $f'$. Since the shape description of $\mathcal E'_{\nu_i}$ has the seventh value equal to in (see item 3 in Lemma~\ref{lem:shape_desc_boring}), the edge $e'_1$ is incoming $v_i$. Further, we have that $e'_2$ is outgoing $v_i$. Indeed, since the shape description of $\mathcal E_{\nu_i}$ has the fourth value equal to $-1$ (see item 9 in Lemma~\ref{lem:shape_desc_boring}), all the edges that are incident to $v_i$ in $G_{\mu}$ and that do not belong to $G_{\nu_i}$ are outgoing $v_i$. It follows that $\alpha$ is a flat angle; by construction, $\alpha$ is assigned the value $0$. 
		
		If $h'=f'$ and $w=v_i$, then $e'_1$ is the edge incident to $v_i$ in the right outer path of $\mathcal E'_{\nu_i}$ and $e'_2$ is the edge incident to $v_i$ and different from $e'_1$ in the boundary of $g'$. Since the shape description of $\mathcal E'_{\nu_i}$ has the eighth value equal to out (see item 3 in Lemma~\ref{lem:shape_desc_boring}), the edge $e'_1$ is outgoing $v_i$. Further, we have that $e'_2$ is outgoing $v_i$, as all the edges that are incident to $v_i$ in $G_{\mu}$ and that do not belong to $G_{\nu_i}$ are outgoing $v_i$. It follows that $\alpha$ is a small angle; by construction, $\alpha$ is assigned the value $-1$.
		
		\item[UP1] Consider a switch vertex $w$ of $G_{\mu}$ and let $n_1(w)$, $n_0(w)$, and $n_{-1}(w)$ be the number of angles labeled $1$, $0$, and $-1$ in $\mathcal E'_{\mu}$, respectively. 
		
		First, suppose that $w\notin \{u_i,v_i\}$ of $G_{\mu}$ and suppose that $w$ belongs to $G_{\nu_i}$ (to $G_{\mu}\setminus G_{\nu_i}$). Then we have $n_1(w)=1$, $n_0(w)=0$, and $n_{-1}(w)=\textrm{deg}(w)-1$, because the same equalities hold true in $\mathcal E'_{\nu_i}$ (resp.\ in $\mathcal E_{\mu}$) and because the number of angles incident to $w$ assigned with the value $1$, $0$, and $-1$ in $\mathcal E'_{\mu}$ coincides with the number of angles incident to $w$ assigned with the value $1$, $0$, and $-1$, respectively, in $\mathcal E'_{\nu_i}$ (resp.\ in $\mathcal E_{\mu}$). 
		
		Second, suppose that $w=u_i$. Since the shape descriptions of $\mathcal E_{\nu_i}$ and $\mathcal E'_{\nu_i}$ have the third value equal to $1$ (see items 9 and 3 in Lemma~\ref{lem:shape_desc_boring}, respectively), all the angles incident to $u_i$ in the internal faces of $\mathcal E_{\nu_i}$ and $\mathcal E'_{\nu_i}$ are assigned the value $-1$. Furthermore, the angle incident to $u_i$ in $f'$ (in $g'$) concides with the angle incident to $u_i$ in $f$ (resp.\ in $g$); finally, every face $h'$ of $\mathcal E'_{\mu}$ that is incident to $u_i$ and that is not incident to edges of $\mathcal E'_{\nu_i}$ is also a face of $\mathcal E_{\mu}$, hence the value assigned to the angle at $u_i$ in $h'$ coincides with the one in $\mathcal E_{\mu}$. Thus, we have $n_1(w)=1$, $n_0(w)=0$, and $n_{-1}(w)=\textrm{deg}(w)-1$ in $\mathcal E'_{\mu}$ as the same equalities hold true in $\mathcal E_{\mu}$.
		
		Finally, we have that $w\neq v_i$. Indeed, the shape description of $\mathcal E'_{\nu_i}$ has the seventh and eighth values equal to in and out, respectively (see item 3 in Lemma~\ref{lem:shape_desc_boring}), hence $v_i$ has at least one outgoing and one incoming edge, thus it is not a switch vertex.
		
		\item[UP2] Consider a non-switch vertex $w$ of $G_{\mu}$ and let $n_1(w)$, $n_0(w)$, and $n_{-1}(w)$ be the number of angles labeled $1$, $0$, and $-1$ in $\mathcal E'_{\mu}$, respectively. If $w\neq v_i$, then the proof that $n_1(w)=0$, $n_0(w)=2$, and $n_{-1}(w)=\textrm{deg}(w)-2$ is analogous to the proof that $n_1(w)=1$, $n_0(w)=0$, and $n_{-1}(w)=\textrm{deg}(w)-1$ in the case in which $w$ is a switch vertex. We prove that $n_1(v_i)=0$, $n_0(v_i)=2$, and $n_{-1}(v_i)=\textrm{deg}(v_i)-2$. First, by construction, the angle at $v_i$ in $f'$ has value $0$, while the one at $v_i$ in $g'$ has value $-1$. Second, since $v_i$ is a non-switch vertex of $\mathcal E'_{\nu_i}$, there are two angles labeled $0$ and incident to $v_i$ in $\mathcal E'_{\nu_i}$; one of them is the angle in the outer face of $\mathcal E'_{\nu_i}$ (which is not a face of $\mathcal E'_{\mu}$) and the other one is an angle in an internal face of $\mathcal E'_{\nu_i}$ (this is also an angle labeled $0$ in $\mathcal E'_{\mu}$). Every other angle incident to $v_i$ in an internal face of $\mathcal E'_{\nu_i}$ has label $-1$. Finally, since the shape description of $\mathcal E_{\nu_i}$ has the fourth value equal to $-1$ (see item 9 in Lemma~\ref{lem:shape_desc_boring}), all the angles incident to $v_i$ inside a face of $\mathcal E_{\mu}$ that lies outside $\mathcal E_{\nu_i}$ and that is different from $f'$ and $g'$ have label $-1$.

		\item[UP3] Consider any face  $h'$ of $\mathcal E'_{\mu}$. Denote by $n_1(h')$ and $n_{-1}(h')$ the number of angles of $h'$ labeled $1$ and $-1$, respectively. If $h'\notin \{f',g'\}$, then $h'$ is also a face of $\mathcal E_{\mu}$ or of $\mathcal E'_{\nu_i}$; further, by construction, the value assigned to each angle of $h'$ in $\mathcal E'_{\mu}$ is the same as the one assigned to it in $\mathcal E_{\mu}$ or in $\mathcal E'_{\nu_i}$. Hence, we have $n_1(h')=n_{-1}(h')-2$ if $h'$ is an internal face of $\mathcal E'_{\mu}$ or $n_1(h')=n_{-1}(h')+2$ if $h'$ is the outer face of $\mathcal E'_{\mu}$, because the same is true in $\mathcal E_{\mu}$ or in $\mathcal E'_{\nu_i}$.

		We now prove that $n_1(f')=n_{-1}(f')-2$ if $f'$ is an internal face of $\mathcal E'_{\mu}$ or $n_1(f')=n_{-1}(f')+2$ if $f'$ is the outer face of $\mathcal E'_{\mu}$; the proof that $n_1(g')=n_{-1}(g')-2$ if $g'$ is an internal face of $\mathcal E'_{\mu}$ or $n_1(g')=n_{-1}(g')+2$ if $g'$ is the outer face of $\mathcal E'_{\mu}$ is analogous. Let $(u_i=w_1,w_2,\dots,w_x=v_i)$ and $(u_i=w'_1,w'_2,\dots,w'_{x'}=v_i)$ be the left outer paths of $\mathcal E_{\nu_i}$ and $\mathcal E'_{\nu_i}$, respectively. Further, let $z_1,\dots,z_{y}$ be the vertices of $G_{\mu}$ incident to $f$ (and to $f'$) and not in $G_{\nu_i}$. For $i=1,\dots,x$, let $\alpha(w_i)$ be angle incident to $w_i$ inside $f$. Furthermore, for $i=1,\dots,x'$, let $\beta(w'_i)$ be angle incident to $w'_i$ inside $f'$. Similarly, for $i=1,\dots,y$, let $\alpha(z_i)$ the angle incident to $z_i$ inside $f$ and $f'$. 
		
		Denote by $\lambda$ the angle assignment of $\mathcal E_{\mu}$. Since $\mathcal E_{\mu}$ is an upward planar embedding, we have $\sum_{i=2}^{x-1} \lambda(\alpha(w_i)) + \sum_{i=1}^{y} \lambda(\alpha(z_i)) + \lambda(\alpha(u_i))+ \lambda(\alpha(v_i))=-2$ if $f$ is an internal face of $\mathcal E_{\mu}$, or $\sum_{i=2}^{x-1} \lambda(\alpha(w_i)) + \sum_{i=1}^{y} \lambda(\alpha(z_i)) + \lambda(\alpha(u_i))+ \lambda(\alpha(v_i))=2$ if $f$ is the outer face of $\mathcal E_{\mu}$. The sum $\sum_{i=2}^{x-1} \lambda(\alpha(w_i))$ is the left-turn-number of $\mathcal E_{\nu_i}$, hence it is equal to $1$ (see the first value in the shape description at item 9 of Lemma~\ref{lem:shape_desc_boring}). Further, the value $\lambda(\alpha(v_i))$ is equal to $-1$, given that the shape description of $\mathcal E_{\nu_i}$ has the third value equal to $-1$  (see again item 9 of Lemma~\ref{lem:shape_desc_boring}). It follows that $\sum_{i=1}^{y} \lambda(\alpha(z_i)) + \lambda(\alpha(u_i))=-2$ if $f$ is an internal face of $\mathcal E_{\mu}$ or $\sum_{i=1}^{y} \lambda(\alpha(z_i)) + \lambda(\alpha(u_i))=2$ if $f$ is the outer face of $\mathcal E_{\mu}$.
		
		Denote by $\lambda'$ the angle assignment of $\mathcal E'_{\mu}$. We need to prove that $\sum_{i=2}^{x'-1} \lambda'(\beta(w'_i)) + \sum_{i=1}^{y} \lambda'(\alpha(z_i)) + \lambda'(\beta(u_i))+ \lambda'(\beta(v_i))=-2$ if $f'$ is an internal face of $\mathcal E'_{\mu}$, or $\sum_{i=2}^{x'-1} \lambda(\beta(w'_i)) + \sum_{i=1}^{y} \lambda'(\alpha(z_i)) + \lambda'(\beta(u_i))+ \lambda'(\beta(v_i))=2$ if $f'$ is the outer face of $\mathcal E'_{\mu}$. First, note that $f$ is the outer face of $\mathcal E_{\mu}$	if and only if $f'$ is the outer face of $\mathcal E'_{\mu}$. Second, by construction we have $\lambda'(\beta(u_i))=\lambda(\alpha(u_i))$, we have $\lambda'(\beta(v_i))=0$, and, for $i=1,\dots,y$, we have $\lambda'(\alpha(z_i))=\lambda(\alpha(z_i))$. Third, the sum $\sum_{i=2}^{x'-1} \lambda'(\beta(w'_i))$ is the left-turn-number of $\mathcal E'_{\nu_i}$, hence it is equal to $0$ (see the first value in the shape description at item 3 of Lemma~\ref{lem:shape_desc_boring}). It follows that the sum $\sum_{i=2}^{x'-1} \lambda'(\beta(w'_i)) + \sum_{i=1}^{y} \lambda'(\alpha(z_i)) + \lambda'(\beta(u_i))+ \lambda'(\beta(v_i))$ is equal to $\sum_{i=1}^{y} \lambda(\alpha(z_i)) + \lambda(\alpha(u_i))$, hence it is equal to $-2$ if $f'$ is an internal face of $\mathcal E'_{\mu}$ or to $2$ if $f'$ is the outer face of $\mathcal E'_{\mu}$, as required.
	\end{description}	
	This concludes the proof that $\mathcal E'_{\mu}$ is a $uv$-external upward planar embedding of $G_{\mu}$. 
	
	It remains to prove that the shape description of $\mathcal E'_{\mu}$ is $s$. Since $G_{\nu_i}$ is a non-extreme component, the replacement of $\mathcal E_{\nu_i}$ with $\mathcal E'_{\nu_i}$ does not alter the last six values of the shape description of the $uv$-external upward planar embedding of $G_{\mu}$. We prove that the left-turn-number of $\mathcal E'_{\mu}$ is equal to the one of $\mathcal E_{\mu}$; the proof that the right-turn-number of $\mathcal E'_{\mu}$ is equal to the one of $\mathcal E_{\mu}$. First, if the virtual edge $e_i$ of the skeleton of $\mu$ corresponding to $G_{\nu_i}$ does not belong to the left outer path of the planar embedding of the skeleton of $\mu$ in $\mathcal E_{\mu}$, the replacement of $\mathcal E_{\nu_i}$ with $\mathcal E'_{\nu_i}$ does not alter the left-turn-number of the $uv$-external upward planar embedding of $G_{\mu}$. Otherwise, the contribution of $\mathcal E_{\nu_i}$ to the left-turn-number of $\mathcal E_{\mu}$ is given by the label of $u_i$ in the outer face of $\mathcal E_{\mu}$, plus the left-turn-number of $\mathcal E_{\nu_i}$, plus the label of $v_i$ in the outer face of $\mathcal E_{\mu}$. Similarly, the contribution of $\mathcal E'_{\nu_i}$ to the left-turn-number of $\mathcal E'_{\mu}$ is given by the label of $u_i$ in the outer face of $\mathcal E'_{\mu}$, plus the left-turn-number of $\mathcal E'_{\nu_i}$, plus the label of $v_i$ in the outer face of $\mathcal E'_{\mu}$. The label of $u_i$ in the outer face of $\mathcal E_{\mu}$ coincides with the label of $u_i$ in the outer face of $\mathcal E'_{\mu}$. Further, the left-turn-number of $\mathcal E_{\nu_i}$ is $1$, while the left-turn-number of $\mathcal E'_{\nu_i}$ is $0$. Finally, the label of $v_i$ in the outer face of $\mathcal E_{\mu}$ is $-1$, while the label of $v_i$ in the outer face of $\mathcal E'_{\mu}$ is $0$. Hence, the contribution of $\mathcal E_{\nu_i}$ to the left-turn-number of $\mathcal E_{\mu}$ is equal to the contribution of $\mathcal E'_{\nu_i}$ to the left-turn-number of $\mathcal E'_{\mu}$. It follows that the left-turn-number of $\mathcal E_{\mu}$ is equal to the one of $\mathcal E'_{\mu}$, hence the shape description of $\mathcal E'_{\mu}$ is $s$.	
\end{proof}

Based on Lemma~\ref{lem:boring-relationship}, we can associate to each boring component $\nu_i$ a \emph{preferred set} $\mathcal P_{\nu_i}$, which is a subset of its feasible set $\mathcal F_{\nu_i}$ defined as follows. 

\begin{itemize}
	\item If $\mathcal F_{\nu_i}$ contains the sausage or the inverted-sausage, then $\mathcal P_{\nu_i}$ coincides with $\mathcal F_{\nu_i}$; see statement S1 of Lemma~\ref{lem:boring-relationship}.
	\item If $\mathcal F_{\nu_i}$ contains the left-wing, then $\mathcal P_{\nu_i}$ contains the left-wing and the right-wing; see statement S2 of Lemma~\ref{lem:boring-relationship}.
	\item If $\mathcal F_{\nu_i}$ contains the inverted-left-wing, then $\mathcal P_{\nu_i}$ contains the inverted-left-wing and the inverted-right-wing; see statement S3 of Lemma~\ref{lem:boring-relationship}.
	\item If $\mathcal F_{\nu_i}$ contains the hat, then $\mathcal P_{\nu_i}$ contains the hat and the inverted-hat; see statement S4 of Lemma~\ref{lem:boring-relationship}.
	\item If $\mathcal F_{\nu_i}$ does not satisfy any of the previous conditions, then, by statements S1--S4 of Lemma~\ref{lem:boring-relationship}, it contains only the heart, or the inverted-heart, or both. Then we let $\mathcal P_{\nu_i}$ coincide with $\mathcal F_{\nu_i}$.
\end{itemize}

By Lemma~\ref{lem:boring-relationship}, if there exists a $uv$-external upward planar embedding of $G_{\mu}$ with some shape desription $s$, then there exists a $uv$-external upward planar embedding of $G_{\mu}$ with some shape description $s$ in which each boring component $G_{\nu_i}$ has a shape description in its preferred set.

\subsubsection{The Algorithm} \label{subsub:algorithm}

We now present an algorithm to construct the feasible set $\mathcal{F}_{\mu}$ of $\mu$. 

Consider a shape description $s=\shapeDesc{\tau_l}{\tau_r}{\lambda^u}{\lambda^v}{\rho^u_l}{\rho^u_r}{\rho^v_l}{\rho^v_r}$, where $\tau_l \in [-2\sigma_{\mu}-1,2\sigma_{\mu}+1]$, $\tau_r \in [-\tau_l,-\tau_l+4]$, $\lambda^u\in \{-1,0,1\}$, $\lambda^v\in \{-1,0,1\}$, $\rho^u_l\in \{\textrm{in},\textrm{out}\}$, $\rho^u_r\in \{\textrm{in},\textrm{out}\}$, $\rho^v_l\in \{\textrm{in},\textrm{out}\}$, and $\rho^v_r\in \{\textrm{in},\textrm{out}\}$; we show how to test whether $s\in \mathcal{F}_{\mu}$. 

\paragraph{Checks} 
In order to test whether $s\in \mathcal{F}_{\mu}$, we start by performing some preliminary checks, which might let us conclude that $s\notin \mathcal{F}_{\mu}$.

We start with a \emph{coherence check}. Recall that only four out of the eight values in $s$ are independent. Thus, we check:
\begin{itemize}
	\item whether $\rho^u_l$ and $\rho^u_r$ have the same value if $\lambda^u\in \{-1,1\}$;
	\item whether $\rho^u_l$ and $\rho^u_r$ have different values if $\lambda^u=0$;
	\item whether $\rho^v_l$ and $\rho^v_r$ have the same value if $\lambda^v\in \{-1,1\}$;
	\item whether $\rho^v_l$ and $\rho^v_r$ have different values if $\lambda^v=0$;
	\item whether $\rho^u_l$ and $\rho^v_l$ have the same value if $\tau_l$ is odd;
	\item whether $\rho^u_l$ and $\rho^v_l$ have different values if $\tau_l$ is even; and
	\item whether $\tau_l+\tau_r+\lambda^u+\lambda^v=2$ (see Corollary~\ref{cor:dependence-parameters}). 
\end{itemize}   

If any of these checks fails, we conclude that $\mathcal{F}_{\mu}$ does not contain $s$, otherwise we proceed.  

Second, for each extreme or interesting component $G_{\nu_i}$ of $G_{\mu}$, we select a shape description $s_i\in \mathcal{F}_{\nu_i}$. Clearly, the number $\ell$ of ways this selection can be done is the product of the cardinalities of the sets $\mathcal{F}_{\nu_i}$ of the extreme and interesting components of $G_{\mu}$; that is, $\ell=\prod_{i=1}^h |\mathcal{F}_{\nu_i}|$ (recall that $G_{\nu_1},G_{\nu_2},\dots, G_{\nu_h}$ are the components of $G_{\mu}$ that are extreme or interesting). We will argue later that $\ell$ can be bounded by a simply-exponential function of $\sigma$ (and, in particular, is independent of the number of vertices of $G_{\mu}$). For now, assume that shape descriptions $s_1\in \mathcal{F}_{\nu_1},\dots, s_h\in \mathcal{F}_{\nu_h}$ have been fixed. We also fix $\mathcal S_{\mu}$ to be a planar embedding of the skeleton $\textrm{sk}(\mu)$ of $\mu$ in which $u$ and $v$ are incident to the outer face. Since $\mu$ is an R-node, there are two such planar embeddings, which are one the flip of the other one. The goal now becomes the one of testing whether $G_{\mu}$ admits a $uv$-external upward planar embedding $\mathcal E_{\mu}$ such that: 
\begin{description}
	\item[P1] the shape description of $\mathcal E_{\mu}$ is $s$;  
	\item[P2] for $i=1,\dots,h$, the $u_iv_i$-external upward planar embedding $\mathcal E_{\nu_i}$ of $G_{\nu_i}$ in $\mathcal E_{\mu}$ has shape description $s_i$; and
	\item[P3] the planar embedding of $\textrm{sk}(\mu)$ induced by $\mathcal E_{\mu}$ is $\mathcal S_{\mu}$.
\end{description}
Then we have that $s$ belongs to $\mathcal{F}_{\mu}$ if and only if this test is successful for at least one selection of the shape descriptions $s_1,\dots, s_h$ and of the planar embedding $\mathcal S_{\mu}$.

We start by performing an \emph{extreme-edge check}. The rationale here is that, since the shape descriptions of the extreme components of $G_{\mu}$ have been already selected, the direction of the edges incident to the outer face of any $uv$-external upward planar embedding $\mathcal E_{\mu}$ of $G_{\mu}$ satisfying Property~P2 is already determined. These directions have to be coherent with the last four values in $s$ in order to satisfy Property~P1. Let $e_i$ be the virtual edge of $\textrm{sk}(\mu)$ that is incident to $u$ and that belongs to the left outer path of $\mathcal S_{\mu}$; assume, w.l.o.g., that $u_i=u$. Then we check whether the fifth value of $s_i$ coincides with the fifth value of $s$. Three more checks are performed for the other edges incident to $u$ and $v$ and incident to the outer face of $\mathcal E_{\mu}$. If any of these checks fails, then there is no $uv$-external upward planar embedding of $G_{\mu}$ satisfying Properties~P1--P3. Otherwise, we proceed (while being sure that, if a $uv$-external upward planar embedding $\mathcal E_{\mu}$ of $G_{\mu}$ is found satisfying Property~P2, then the last four values of the shape description of $\mathcal E_{\mu}$ are those specified in $s$).

We next perform an \emph{angle check}. This verifies whether the shape descriptions $s_1,\dots,s_h$ of the extreme or interesting components and the feasible sets $\mathcal{F}_{\nu_{h+1}},\dots,\mathcal{F}_{\nu_k}$ of the non-extreme boring components allow for an angle assignment that satisfies Properties \textbf{UP0}-\textbf{UP2} of Theorem~\ref{th:upward-conditions}. Consider each vertex $w$ of $\textrm{sk}(\mu)$. 

We first deal with the case in which $w$ is a switch vertex of $G_{\mu}$. Let $l_w$ be a counter, initially set to $0$. The counter $l_w$ is going to represent, in any $uv$-external upward planar embedding $\mathcal E_{\mu}$ of $G_{\mu}$ satisfying Property P2, the number of angles incident to $w$ that are required to be large inside some component $G_{\nu_i}$ of $G_{\mu}$. 

\begin{itemize}
	\item For each extreme or interesting component $G_{\nu_i}$ such that $w=u_i$ (or $w=v_i$) and such that the third (resp.\ the fourth) value of $s_i$ is $-1$, increase $l_w$ by $1$. Namely, if the third (resp.\ the fourth) value of $s_i$ is equal to $-1$, there is a large angle at $w$ in an internal face of every $u_iv_i$-external upward planar embedding of $G_{\nu_i}$ whose shape description is $s_i$.
	\item For each non-extreme boring component $G_{\nu_i}$ such that $w=v_i$ (or $w=u_i$), if $\mathcal{F}_{\nu_i}$ contains only the heart (resp.\ the inverted-heart), then increase $l_w$ by $1$. Namely, if $\mathcal{F}_{\nu_i}$ contains only the heart (resp.\ the inverted-heart), there is a large angle at $w$ in an internal face of every $u_iv_i$-external upward planar embedding of $G_{\nu_i}$ whose shape description is the heart (resp.\ the inverted-heart).
\end{itemize}


We now check whether $l_w\leq 1$. If the check fails, we conclude that there is no $uv$-external upward planar embedding of $G_{\mu}$ satisfying Properties~P1--P3. Indeed, in order to satisfy Property~P2, there should be at least two large angles incident to $w$, which is not possible by Property~\textbf{UP1}. If $l_w=1$, then we label $w$ as \emph{unavailable} (meaning that $w$ is already incident to a large angle inside some component of $G_{\mu}$, hence it cannot be incident to another large angle),  whereas if $l_w=0$, then we label $w$ as \emph{available} (meaning that $w$ still has to be assigned a large angle). 

We next deal with the case in which $w$ is a non-switch vertex of $G_{\mu}$. Recall that $G$ is expanded, hence $G_{\mu}$ contains a unique edge incoming $w$ or a unique edge outgoing $w$. Let $G_{\nu_i}$ be the component of $G_{\mu}$ containing such an edge. Then we check whether the angles at $w$ in the internal faces of every other component are all small. Thus, for each component $G_{\nu_j}$ of $G_{\mu}$ with $j\neq i$ such that $w=u_j$ (or $w=v_j$), we check whether the third value of $s_j$ (resp.\ the fourth value of $s_j$) is equal to $1$. If any of these checks fails, we conclude that there is no $uv$-external upward planar embedding of $G_{\mu}$ satisfying Properties~P1--P3. 

There is a final angle check that needs to be performed on the poles of $\mu$. Namely, if $\lambda^u=1$ ($\lambda^v=1$) and either $u$ (resp.\ $v$) is not a switch vertex or it is an unavailable switch vertex, then we conclude that there is no $uv$-external upward planar embedding of $G_{\mu}$ satisfying Properties~P1--P3. 

\paragraph{Flow network} 

We are now going to borrow ideas from the algorithm presented by Bertolazzi et al.\ in~\cite{bdl-udtd-94} for testing the upward planarity of a plane digraph. As described in Section~\ref{subse:upward-definitions}, Bertolazzi et al.\ reduce the problem of testing whether a digraph with a fixed planar embedding $\mathcal E$ admits an upward planar embedding to the one of testing whether the vertex-face incidence network $\mathcal N$ of $\mathcal E$ (in which the capacities of the arcs of $\mathcal N$ and the supplies and demands of the nodes of $\mathcal N$ are suitably set) admits a flow with a certain value. The skeleton $\textrm{sk}(\mu)$ of our R-node $\mu$ indeed has a fixed planar embedding $\mathcal S_{\mu}$. However, the edges of $\textrm{sk}(\mu)$ are not actual edges, but rather virtual edges that correspond to components of $G_{\mu}$. This is not really problematic for the extreme or interesting components, as we already fixed their shape description. In fact, as the shape description of a component $G_{\nu_i}$ contains all the relevant information that describes how $G_{\nu_i}$ interacts with the rest of $G_{\mu}$, we do not even need to represent $G_{\nu_i}$ in the flow network, but rather we just need to take into account the constant amount of information in its shape description. Handling non-extreme boring components is more challenging. According to Lemma~\ref{lem:shape_desc_boring}, each boring component has $\bigoh(1)$ shape descriptions in its feasible set, however the number of such components is not, in general, bounded by a function of $\sigma$ only, hence we cannot try all possible combinations for their shape descriptions, as we are doing for the interesting components. Lemma~\ref{lem:boring-relationship} comes here to the rescue in order to restrict the possible shape descriptions of each boring component to at most two, namely those in its preferred set; the choice between these two shape descriptions can be plugged into the flow network. 

We now define a bipartite network $\mathcal N(S,T,A)$, where $S$ is a set of sources, $T$ is a set of sinks, and $A$ is a set of arcs from the sources to the sinks such that $G_{\mu}$ admits a $uv$-external upward planar embedding satisfying Properties~P1--P3 if and only if $\mathcal N$ admits a flow of a certain size. Intuitively, and following the approach in~\cite{bdl-udtd-94}, each unit of flow going from a source $s$ to a sink $t$ represents the choice that a vertex of $G_{\mu}$ (represented by $s$) has a large angle in a face (represented by $t$) of the planar embedding of $G_{\mu}$ we construct. The number of such large angles has to be equal to $\frac{n_t}{2}-1$, where $n_t$ is the number of switch angles in $t$.

The network $\mathcal N$ is formally defined as follows. We start by defining the set $S$ of sources.

\begin{figure}[!t]
	\centering
	{\centering
		\begin{subfigure}{.19\textwidth}
			\includegraphics[page=1]{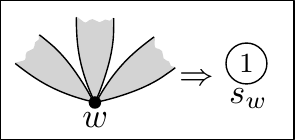}
			\subcaption{}
			\label{fig:network_sources_1}
		\end{subfigure}
	}		
	\hfil
	{\centering
		\begin{subfigure}{.19\textwidth}
			\includegraphics[page=2]{figures/Flow-construction}
			\subcaption{}
			\label{fig:network_sources_2}
		\end{subfigure}
	}
	\hfil
	{\centering
		\begin{subfigure}{.19\textwidth}
			\includegraphics[page=3]{figures/Flow-construction}
			\subcaption{}
			\label{fig:network_sources_3}
		\end{subfigure}
	}\\
	{\centering
		\begin{subfigure}{.19\textwidth}
			\includegraphics[page=4]{figures/Flow-construction}
			\subcaption{}
			\label{fig:network_sources_4}
		\end{subfigure}
	}
	\hfil
	{\centering
		\begin{subfigure}{.19\textwidth}
			\includegraphics[page=5]{figures/Flow-construction}
			\subcaption{}
			\label{fig:network_sources_5}
		\end{subfigure}
	}
	\hfil
	{\centering
		\begin{subfigure}{.19\textwidth}
			\includegraphics[page=6]{figures/Flow-construction}
			\subcaption{}
			\label{fig:network_sources_6}
		\end{subfigure}
	}
	\caption{(a) Sources for switch vertices. (b) Sources for non-switch vertices. (c)--(d) Sources for extreme or interesting components. (e) Sources for boring hats. (f) Sources for boring hearts.}\label{fig:network_sources}
\end{figure}

\begin{description}
	\item[Switch vertices.] For each available switch vertex $w$ of $G_{\mu}$ that is a vertex of $\textrm{sk}(\mu)$, there is a source $s_w$ in $S$ which can supply a single unit of flow; see Figure~\ref{fig:network_sources_1}. 
	\item[Non-switch vertices.] For each non-switch vertex $w$ of $G_{\mu}$ that is a vertex of $\textrm{sk}(\mu)$, there is a source $s_w$ in $S$ which can supply a single unit of flow if and only if the following three conditions are satisfied (see Figure~\ref{fig:network_sources_2}):
	\begin{enumerate}
		\item The special component for $w$ is a non-extreme boring component $G_{\nu_i}$ of $G_{\mu}$. 
		\item Either $w=v_i$ and the preferred set $\mathcal P_{\nu_i}$ of $\nu_i$ consists of the left-wing and the right-wing, or $w=u_i$ and $\mathcal P_{\nu_i}$ consists of the inverted-left-wing and the inverted-right-wing.
		\item The special edge for $w$ is outgoing $w$. 
	\end{enumerate} 
	\item[Extreme or interesting components.] Consider each extreme or interesting component $G_{\nu_i}$. If the shape description $s_i$ has the first value (i.e., the left-turn-number) $\tau_l^i$ greater than $0$, then there is a source $z_l^i$ in $S$ which can supply $\tau_l^i$ units of flow. Further, if $s_i$ has the second value (i.e., the right-turn-number) $\tau_r^i$ greater than $0$, then there is a source $z_r^i$ in $S$ which can supply $\tau_r^i$ units of flow; see Figures~\ref{fig:network_sources_3} and~\ref{fig:network_sources_4}. 
	\item[Non-extreme boring hats.] For each non-extreme boring component $G_{\nu_i}$ whose preferred set $\mathcal P_{\nu_i}$ consists of the hat and of the inverted-hat, there is a source $b^i$ which can supply a single unit of flow; see Figure~\ref{fig:network_sources_5}.
	\item [Non-extreme boring hearts.] For each non-extreme boring component $G_{\nu_i}$ whose preferred set $\mathcal P_{\nu_i}$ consists of the heart, or of the inverted-heart, or of both, there are two sources $b^i_1$ and $b^i_2$, each of which can supply a single unit of flow; see Figure~\ref{fig:network_sources_6}. 	
\end{description}

The set $T$ contains the following sinks. For a virtual edge $e_i$ of $\textrm{sk}(\mu)$, denote by $f_l(e_i)$ and $f_r(e_i)$ the faces of $\mathcal S_{\mu}$ which are respectively to the left and to the right of $e_i$ when traversing such an edge from $u_i$ to $v_i$. 

\begin{figure}[htb]
	\centering
	{\centering
		\begin{subfigure}{.18\textwidth}
			\includegraphics[page=8]{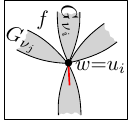}
			\subcaption{}
			\label{fig:network_sinks_1}
		\end{subfigure}
	}		
	\hfil
	{\centering
		\begin{subfigure}{.18\textwidth}
			\includegraphics[page=9]{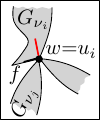}
			\subcaption{}
			\label{fig:network_sinks_2}
		\end{subfigure}
	}
	\hfil
	{\centering
		\begin{subfigure}{.18\textwidth}
			\includegraphics[page=10]{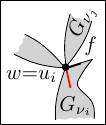}
			\subcaption{}
			\label{fig:network_sinks_3}
		\end{subfigure}
	}
	\hfil
	{\centering
		\begin{subfigure}{.18\textwidth}
			\includegraphics[page=11]{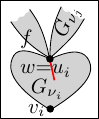}
			\subcaption{}
			\label{fig:network_sinks_4}
		\end{subfigure}
	}
	\hfil
	{\centering
		\begin{subfigure}{.18\textwidth}
			\includegraphics[page=12]{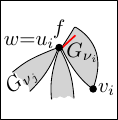}
			\subcaption{}
			\label{fig:network_sinks_5}
		\end{subfigure}
	}
	\caption{Non-switch vertices $w$ that count for the value of $n_f$. The special edge for $w$ is red. (a) $G_{\nu_i}$ is not the special component for $w$. In (b)--(e)  $G_{\nu_i}$ is the special component for $w$, in (b)--(c) it is an extreme or interesting component, while in (d)--(e) it is non-extreme and boring. (b) The fifth value of $s_i$ is in and the special edge for $w$ is outgoing $w$. (c) The fifth value of $s_i$ is out and the special edge for $w$ is incoming $w$. (d) $\mathcal P_{\nu_i}$ consists of the inverted-heart. (e) $\mathcal P_{\nu_i}$ consists of the inverted-right-wing and the special edge is outgoing $w$.}\label{fig:network_sinks}
\end{figure}

\begin{description}
	\item[Internal faces.] For each internal face $f$ of $\mathcal S_{\mu}$, there is a sink $t_f$ in $T$ that demands a number $d_f$ of units of flow equal to $\frac{n_f}{2}-1$, where $n_f$ is equal to the sum of the following numbers (roughly speaking, $n_f$ represents the number of switch angles incident to $f$ which are not yet ``paired'' with large angles incident to $f$; what follows allows us to count the number of such switch angles although a planar embedding of $G_{\mu}$ is not specified):
	\begin{itemize}
		\item For each virtual edge $e_i$ incident to $f$ such that $G_{\nu_i}$ is an extreme or interesting component and such that $f=f_l(e_i)$ ($f=f_r(e_i)$), the absolute value of $\tau_l^i$ (resp.\ of $\tau_r^i$).
		\item The number of virtual edges $e_i$ incident to $f$ such that $G_{\nu_i}$ is a non-extreme boring component whose preferred set consists of the hat and of the inverted-hat, or contains the heart and/or the inverted-heart.  
		\item The number of switch vertices of $\textrm{sk}(\mu)$ incident to $f$.
		\item The number of non-switch vertices $w$ of $\textrm{sk}(\mu)$ incident to $f$ that satisfy one of the following three conditions. Let $e_i$ and $e_j$ be the virtual edges incident to $w$ and $f$. Assume, w.l.o.g.\ up to relabeling of $e_i$ with $e_j$, that $G_{\nu_j}$ is not the special component for $w$. Further, assume that $w=u_i$; the case $w=v_i$ is analogous. Finally, assume that $f=f_l(e_i)$; the case $f=f_r(e_i)$ is analogous. The conditions are as follows. 
		\begin{enumerate}
			\item $G_{\nu_i}$ is not the special component for $w$ (see Figure~\ref{fig:network_sinks_1}); or
			\item $G_{\nu_i}$ is the special component for $w$, it is an extreme or interesting component, and the fifth value of $s_i$ is in if and only if the special edge for $w$ is outgoing $w$ (see Figures~\ref{fig:network_sinks_2} and~\ref{fig:network_sinks_3}); note that fifth value of $s_i$ determines the direction of the edge of $G_{\nu_i}$ that is incident to $u_i$ and to $f$ in any $u_iv_i$-external upward planar embedding of $G_{\nu_i}$ whose shape description is $s_i$; or
			\item $G_{\nu_i}$ is the special component for $w$, it is a non-extreme boring component, and either (a) the preferred set $\mathcal P_{\nu_i}$ of $G_{\nu_i}$ consists of the inverted-heart (see Figure~\ref{fig:network_sinks_4}), or (b) $\mathcal P_{\nu_i}$ consists of the inverted-right-wing and the special edge is outgoing $w$ (see Figure~\ref{fig:network_sinks_5}). 			
		\end{enumerate}
	\end{itemize}   
	\item[Left- and right-turn-number.] There are two sinks $t_l$ and $t_r$ which demand $\frac{\tau_l+n^l_f}{2}$ and $\frac{\tau_r+n^r_f}{2}$ units of flow, respectively, where $\tau_l$ and $\tau_r$ are the left- and right-turn-numbers of the shape description $s$ whose membership in $\mathcal F_{\mu}$ we are establishing, and where $n^l_f$ ($n^r_f$) is defined as $n_f$ in the previous item, however the face $f$ is the outer face, and the virtual edges and the vertices to be considered in the count for $n^l_f$ (resp.\ $n^r_f$) are only those in the left (resp.\ right) outer path of $\mathcal S_{\mu}$ and different from $u$ and $v$. 
	\item[Heart or inverted-heart.] For each non-extreme boring component $G_{\nu_i}$ whose preferred set consists of the heart and of the inverted-heart, there is a sink $t_i$ which demands a single unit of flow; see Figure~\ref{fig:network_sinks_3}. Note that either an angle at $u_i$ or an angle at $v_i$ inside the upward planar embedding of $G_{\nu_i}$ has to be large.
	\item[Angle at $u$.] If $\lambda^u=1$, there is a sink $t^u$ which demands a single unit of flow.
	\item[Angle at $v$.] If $\lambda^v=1$, there is a sink $t^v$ which demands a single unit of flow.
\end{description}

Finally, the set $A$ contains the following arcs. 


\begin{description}
	\item[Switch vertices to faces.] For each available switch vertex $w\notin \{u,v\}$ in $\textrm{sk}(\mu)$, there is an arc with capacity $1$ from the corresponding source $s_w\in S$ to each sink $t_f$ corresponding to a face $f$ of $\mathcal S_{\mu}$ incident to $w$ (among these sinks there is $t_l$ or $t_r$, if $w$ belongs to the left or to the right outer path of $\mathcal S_{\mu}$, respectively).
	\item[Non-switch vertices to faces.] For each non-switch vertex $w$ of $G_{\mu}$ that corresponds to a source $s_w$ in $S$, let $e_i$ be the virtual edge of $\textrm{sk}(\mu)$ corresponding to the special component for $w$. Then there is an arc with capacity $1$ from $s_w$ to each of the two sinks corresponding to the two faces $f_l(e_i)$ and $f_r(e_i)$ of $\mathcal S_{\mu}$ incident to $e_i$ (one of these sinks is $t_l$ or $t_r$ if $e_i$ belongs to the left or to the right outer path of $\mathcal S_{\mu}$, respectively).  
	\item[Extreme or interesting components to faces on the left.] For each extreme or interesting component $G_{\nu_i}$ such that the first value $\tau_l^i$ of $s_i$ is greater than $0$, there is an arc with capacity $\tau_l^i$ from $z_l^i$ to the sink $t_f$ corresponding to the face $f=f_l(e_i)$ of $\mathcal S_{\mu}$; note that $t_f=t_l$ or $t_f=t_r$ if $e_i$ belongs to the left or to the right outer path of $\mathcal S_{\mu}$, respectively, and $f$ is the outer face of $\mathcal S_{\mu}$. 
	\item[Extreme or interesting components to faces on the right.] For each extreme or interesting component $G_{\nu_i}$ such that the  second value $\tau_r^i$ of $s_i$ is greater than $0$, there is an arc with capacity $\tau_r^i$ from $z_r^i$ to the sink $t_f$ corresponding to the face $f=f_r(e_i)$ of $\mathcal S_{\mu}$; note that $t_f=t_l$ or $t_f=t_r$ if $e_i$ belongs to the left or to the right outer path of $\mathcal S_{\mu}$, respectively, and $f$ is the outer face of $\mathcal S_{\mu}$.
	\item[Non-extreme boring hats to faces.] For each non-extreme boring component $\mathcal G_{\nu_i}$ whose preferred set $\mathcal P_{\nu_i}$ consists of the hat and of the inverted-hat, there is an arc with capacity $1$ from $b^i$ to each of the two sinks corresponding to the faces $f_l(e_i)$ and $f_r(e_i)$ of $\mathcal S_{\mu}$ incident to $e_i$ (one of these sinks is $t_l$ or $t_r$ if $e_i$ belongs to the left or to the right outer path of $\mathcal S_{\mu}$, respectively). 
	\item[Non-extreme boring hearts to faces.] For each non-extreme boring component $\mathcal G_{\nu_i}$ whose preferred set $\mathcal P_{\nu_i}$ consists of the heart, or of the inverted-heart, or of both, there is an arc with capacity $1$ from $b^i_1$ to the sink corresponding to the face $f_l(e_i)$ of $\mathcal S_{\mu}$ and there is an arc with capacity $1$ from $b^i_2$ to the sink corresponding to the face $f_r(e_i)$ of $\mathcal S_{\mu}$ (one of these sinks is $t_l$ or $t_r$ if $e_i$ belongs to the left or right outer path of $\mathcal S_{\mu}$, respectively).
	\item[Switch vertices to non-extreme boring hearts and inverted-hearts.] For each non-extreme boring component $G_{\nu_i}$ whose preferred set contains the heart and the inverted-heart, there is an arc with capacity $1$ from each pole of $G_{\nu_i}$ that is an available switch vertex to the sink $t_i$. 
	\item[First pole.] If $\lambda^u=1$, then there is an arc with capacity $1$ from the source $s_u$ to the sink $t^u$. If $u$ is a switch vertex and $\lambda^u\in \{-1,0\}$, there is an arc with capacity $1$ from $s_u$ to each sink $t_f$ corresponding to an internal face $f$ of $\mathcal S_{\mu}$ incident to $u$. 
	\item[Second pole.] If $\lambda^v=1$, then there is an arc with capacity $1$ from the source $s_v$ to the sink $t^v$. If $v$ is a switch vertex and $\lambda^v\in \{-1,0\}$, there is an arc with capacity $1$ from $s_v$ to each sink $t_f$ corresponding to an internal face $f$ of $\mathcal S_{\mu}$ incident to $v$.
\end{description}

\begin{figure}[htb]
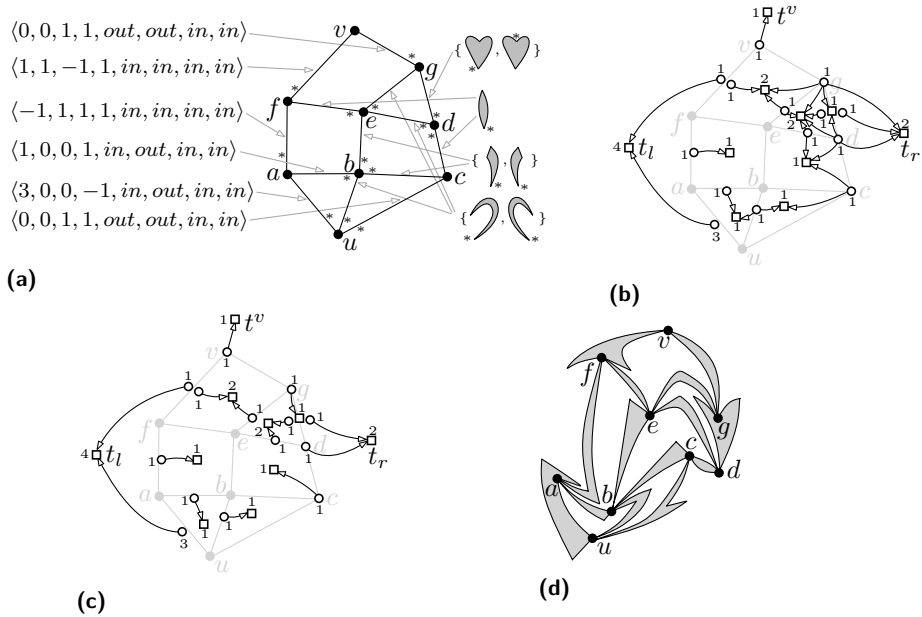

	\centering
	\begin{subfigure}{.5\textwidth}
		\includegraphics[page=14]{figures/Flow-construction}
		\subcaption{}
		\label{fig:network_example_1}
	\end{subfigure}
	\hfil
	\begin{subfigure}{.3\textwidth}
		\includegraphics[page=15]{figures/Flow-construction}
		\subcaption{}
		\label{fig:network_example_2}
	\end{subfigure}
	\\
	\begin{subfigure}{.3\textwidth}
		\includegraphics[page=16]{figures/Flow-construction}
		\subcaption{}
		\label{fig:network_example_3}
	\end{subfigure}
	\hfil
	\begin{subfigure}{.3\textwidth}
		\includegraphics[page=13]{figures/Flow-construction}
		\subcaption{}
		\label{fig:network_example_4}
	\end{subfigure}
	\caption{Illustration for the construction of a flow network $\mathcal N$ that allows us to determine whether a shape description $s=\langle 1,0,-1,0,\textrm{in},\textrm{out},\textrm{in},\textrm{in}\rangle$ belongs to $\mathcal F_{\mu}$. (a) shows the input, which contains of a shape description $s_i$ for each extreme or interesting component $G_{\nu_i}$ of $G_{\mu}$ and the preferred set $\mathcal P_{\nu_i}$ for each non-extreme boring component $G_{\nu_i}$ of $G_{\mu}$. (b) shows the corresponding flow network; arc capacities are not shown (each of them is equal to the supply of the source of the arc). (c) shows a flow for $\mathcal N$ whose value coincides with the total number of units of flow that are demanded by the sinks of $\mathcal N$; each shown arc is traversed by a flow equal to its capacity. (d) shows a $uv$-external upward planar embedding of $G_{\mu}$ with shape description $s$ corresponding to the flow.}\label{fig:network_example}
\end{figure}

This concludes the construction of the network $\mathcal N$. Note that $\mathcal N$ has $\bigoh(k)$ nodes and arcs, as $\textrm{sk}(\mu)$ has $\bigoh(k)$ vertices and virtual edges. We define $F$ as the total number of units of flow that are demanded by the sinks in $T$. We test whether every sink has a non-negative demand and whether $\mathcal N$ admits a flow with value $F$. The latter can be done in $\bigoh(k\log^3 k)$ time by means of an algorithm by Borradaile et al.~\cite{bkm-msms-17}. We conclude that $G_{\mu}$ admits a $uv$-external upward planar embedding satisfying Properties~P1--P3 if and only if the tests are successful. 

Figure~\ref{fig:network_example} for an example of the described algorithm to construct $\mathcal N$. In this example, the demand of each sink $x$ is the following.

\begin{itemize}
	\item For the sink corresponding to the internal face $x$ incident to $u$, $a$, and $b$, the demand is $1$, as $n_x=4$: One unit comes from the switch vertex $a$, one from the non-switch vertex $u$, one from the boring component corresponding to the virtual edge $(u,b)$, and one from the interesting component corresponding to the virtual edge $(b,a)$.
	\item For the sink corresponding to the internal face $x$ incident to $u$, $b$, and $c$, the demand is $1$, as $n_x=4$: Three units come from the non-switch vertices $u$, $b$, and $c$, and one from the boring component corresponding to the virtual edge $(u,b)$.
	\item For the sink corresponding to the internal face $x$ incident to $a$, $b$, $e$, and $f$, the demand is $1$, as $n_x=4$: Two units come from the switch vertices $a$ and $f$, one from the non-switch vertex $b$, and one from the interesting component corresponding to the virtual edge $(a,f)$.
	\item For the sink corresponding to the internal face $x$ incident to $b$, $c$, $d$, and $e$, the demand is $1$, as $n_x=4$: One unit comes from the switch vertex $d$, two from the non-switch vertices $b$ and $c$, and one from the boring component corresponding to the virtual edge $(d,e)$.
	\item For the sink corresponding to the internal face $x$ incident to $d$, $e$, and $g$, the demand is $2$, as $n_x=6$: Two units come from the switch vertices $d$ and $g$, one from the non-switch vertex $e$, and three from the boring components corresponding to the virtual edges $(d,e)$, $(d,g)$, and $(e,g)$.
	\item For the sink corresponding to the internal face $x$ incident to $e$, $f$, $g$, and $v$, the demand is $2$, as $n_x=6$: Three units come from the switch vertices $f$, $g$, and $v$, one from the non-switch vertex $e$, one from the boring component corresponding to the virtual edge $(e,g)$, and one from the interesting component corresponding to the virtual edge $(f,v)$.
	\item For the sink $t_l$, the demand is $4$, as $\tau_l$ is $1$ and $n^l_x=7$, where $x$ is the outer face: Two units come from the switch vertices $a$ and $f$, three units from the interesting component corresponding to the virtual edge $(u,a)$, two units from the interesting components corresponding to the virtual edges $(a,f)$ and $(f,v)$. 
	\item For the sink $t_r$, the demand is $2$, as $\tau_r$ is $0$ and $n^r_x=4$, where $x$ is the outer face: Two units come from the switch vertices $d$ and $g$, one from the non-switch vertex $c$, and one from the boring component corresponding to the virtual edge $(d,g)$. 
	\item For the sink corresponding to the boring component whose preferred set consists of the heart and the inverted-heart, the demand is $1$. 
	\item For the sink $t^v$, the demand is $1$, as $\lambda^v=1$. 
\end{itemize}

\paragraph{Correctness and Running Time} 

In Lemma~\ref{lem:R_node_sources} below, we argue about the correctness and the running time of our algorithm. For the latter, we need the following auxiliary lemma.

  \begin{lemma}  \label{le:number_combinations}
		We have $\ell\in \bigoh(1.45^\sigma)$. 
	\end{lemma}
	
	  \begin{proof}
		For $i=1,\dots,h$, let $\sigma_i$ be the number of sources of $G_{\nu_i}$ different from $u_i$ and $v_i$ and note that $\sum_{i=1}^h \sigma_i=:\overline{\sigma}\leq \sigma$. By Lemma~\ref{lem:feasible_size}, the size of $\mathcal F_{\nu_i}$ is at most $72\sigma_i+54$. Hence, we have that $\ell \leq \prod_{i=1}^h (72\sigma_i+54)$. 
		
		Recall that each component among $G_{\nu_1},\dots,G_{\nu_{h}}$ is extreme or interesting, and that at most four components are extreme; each extreme component is either boring or interesting. Let $0\leq m\leq 4$ be the number of boring components among $G_{\nu_1},\dots,G_{\nu_{h}}$ and assume, w.l.o.g., that such components come first in the order $G_{\nu_1},\dots,G_{\nu_{h}}$. Hence, we have $\sigma_{i}=0$ for $i=1,\dots,m$, and $\sigma_i>0$ for $i=m+1,\dots,h$. Then we have $\ell \leq \left(\prod_{i=1}^{m} 54 \right)\cdot \left(\prod_{i=m+1}^{h} (72\sigma_i+54)\right)\leq 54^4 \cdot 73 \prod_{i=m+1}^{h} \sigma_i$. We prove that $\prod_{i=m+1}^{h} \sigma_i \leq e^{\sigma/e}$. 
		
		First, we have $\prod_{i=m+1}^{h} \sigma_i\leq (\overline{\sigma}/x)^x$, where $x:=h-m$; that is, $\prod_{i=m+1}^{h} \sigma_i$ is maximum when all the values $\sigma_i$ are equal. Indeed, if there is a value $\sigma_i>\overline{\sigma}/x$ and a value $\sigma_j<\overline{\sigma}/x$, then $\sigma_i\cdot \sigma_j<(\overline{\sigma}/x)\cdot (\sigma_i+\sigma_j-\overline{\sigma}/x)$, as can be proved by substituting $\sigma_i= \overline{\sigma}/x + y$ and $\sigma_j= \overline{\sigma}/x - y'$. Hence, the values $\sigma_i$ can be modified into $\overline{\sigma}/x$ one by one without altering the sum $\sum_{i=m+1}^h \sigma_i=\overline{\sigma}$ and while increasing the product $\prod_{i=m+1}^{h} \sigma_i$.
		
		Second, we have that $f(x):=(\overline{\sigma}/x)^x$ has its maximum when $x=\overline{\sigma}/e$, where $e$ is the base of the natural logarithm. Indeed, we have $\frac{\partial f(x)}{\partial x}=(\overline{\sigma}/x)^x(\ln(\overline{\sigma}/x)-1)$, which is positive for $x\in (0,\overline{\sigma}/e)$, null for $x=\overline{\sigma}/e$, and negative for $x\in (\overline{\sigma}/e,\overline{\sigma}]$. The value of $f(x)$ at its maximum is hence $e^{\overline{\sigma}/e}<1.45^{\overline{\sigma}}$. It follows that $\ell \leq 54^4 \cdot 73 \prod_{i=m+1}^{h} \sigma_i \in \bigoh(1.45^\sigma)$. 
	\end{proof} 
	
	We are now ready to prove the following.
			
  \begin{lemma}  \label{lem:R_node_sources}
		Let $\mu$ be an R-node of $T$. The feasible set $\mathcal F_{\mu}$ of $\mu$ can be computed in $\bigoh(\sigma 1.45^\sigma \cdot k\log^3 k))$ time, where $k$ is the number of children of $\mu$ in $T$ and $\sigma$ is the number of sources of $G$. 
	\end{lemma}

	\begin{proof}
		First, we argue about the running time of the algorithm we presented. The number of shape descriptions that are checked for membership in $\mathcal F_{\mu}$ is in $\bigoh(\sigma)$, by Lemma~\ref{lem:feasible_size}; this determines the multiplicative factor of $\sigma$ in the running time. For each shape description $s$, we consider $\ell$ possible choices for the shape descriptions of the extreme and interesting components of $G_{\mu}$ and two possible choices for the planar embedding of $\textrm{sk}(\mu)$. By Lemma~\ref{le:number_combinations}, we have $\ell\in \bigoh(1.45^\sigma)$; this determines the multiplicative factor of $1.45^\sigma$ in the running time. Finally, for each choice for the shape descriptions of the extreme and interesting components of $G_{\mu}$ and for each choice for the planar embedding of $\textrm{sk}(\mu)$, we construct in $\bigoh(k)$ time the planar bipartite network $\mathcal N$ and determine in $\bigoh(k \log^3 k)$ time its maximum flow, by means of the algorithm by Borradaile et al.~\cite{bkm-msms-17}. 
		
		We now prove the correctness of the algorithm we presented for computing $\mathcal F_{\mu}$. Namely, we prove that, for any shape description $s$, the graph $G_{\mu}$ admits a $uv$-external upward planar embedding $\mathcal E_{\mu}$ if and only if our algorithm determines that $s\in \mathcal F_{\mu}$.
		
		($\Longrightarrow$) Suppose first that $G_{\mu}$ admits a $uv$-external upward planar embedding $\mathcal E_{\mu}$ with shape description $s$. We prove that our algorithm determines that $s\in \mathcal F_{\mu}$. First, for each non-extreme boring component $G_{\nu_i}$ of $G_{\mu}$, we replace the $u_iv_i$-external upward planar embedding of $G_{\nu_i}$ contained in $\mathcal E_{\mu}$ with a $u_iv_i$-external upward planar embedding of $G_{\nu_i}$ whose shape description is one in the preferred set of $G_{\nu_i}$. By Lemma~\ref{lem:boring-relationship}, this modification preserves the property that $\mathcal E_{\mu}$ is a $uv$-external upward planar embedding with shape description $s$.
		
		By Lemma~\ref{lem:spirality_range}, the first value $\tau_l$ of $s$ (i.e., the left-turn-number of $\mathcal E_{\mu}$) is an integer in the interval $[-2\sigma-1,2\sigma+1]$. By Lemma~\ref{lem:shape_desc_values}, the second value $\tau_r$ of $s$ (i.e., the right-turn-number of $\mathcal E_{\mu}$) is a value in  $[-\tau_l,-\tau_l+4]$. Further, by definition, the third and fourth values $\lambda^u$ and $\lambda^v$ of $s$ are in the set $\{-1,0,1\}$ and the fifth, sixth, seventh, and eighth values $\rho_l^u$, $\rho_r^u$, $\rho_l^v$, and $\rho_r^v$ of $s$ are in the set 	 $\{\textrm{in},\textrm{out}\}$. It follows that our algorithm considers $s$ as a possible shape description which is to be tested for membership in $\mathcal F_{\mu}$.
		
		As discussed after the definition of shape description, we have that $\rho^u_l$ and $\rho^u_r$ have the same value if and only if $\lambda^u\in \{-1,1\}$, that $\rho^v_l$ and $\rho^v_r$ have the same value if and only if $\lambda^v\in \{-1,1\}$, that $\rho^u_l$ and $\rho^v_l$ have the same value if and only if $\tau_l$ is odd, and that $\tau_l+\tau_r+\lambda^u+\lambda^v=2$. Hence, $s$ does not fail the coherence check. 
		
		Let $\mathcal S_{\mu}$ be the planar embedding of $\textrm{sk}(\mu)$ in $\mathcal E_{\mu}$. Further, for each component $G_{\nu_i}$ of $G_{\mu}$, let $\mathcal E_{\nu_i}$ be the  $u_iv_i$-external upward planar embedding of $G_{\nu_i}$ contained in $\mathcal E_{\mu}$, where $u_i$ and $v_i$ are the poles of $\nu_i$, and let $s_i$ be the shape description of $\mathcal E_{\nu_i}$. 
		
		Recall that $G_{\nu_1},\dots,G_{\nu_h}$ are the extreme or interesting components of $G_{\mu}$. For $i=1,\dots,h$, since a $u_iv_i$-external upward planar embedding of $G_{\nu_i}$ with shape description $s_i$ exists, we have that $s_i\in \mathcal F_{\nu_i}$. Hence, our algorithm tests whether $G_{\mu}$ admits a $uv$-external upward planar embedding whose shape description is $s$, whose restriction to $G_{\nu_i}$ is a $u_iv_i$-external upward planar embedding with shape description $s_i$, for $i=1,\dots,h$, and such that the corresponding planar embedding of $\textrm{sk}(\mu)$ is $\mathcal S_{\mu}$. We now prove that, when considering the shape descriptions $s_1,\dots,s_h$ for $G_{\nu_1},\dots,G_{\nu_h}$, respectively, and the planar embedding $\mathcal S_{\mu}$ for $\textrm{sk}(\mu)$, our algorithm indeed concludes that $s\in \mathcal F_{\mu}$.
		
		The extreme-edge check does not fail, as the shape description $s_i$ of each extreme component $G_{\nu_i}$ of $G_{\mu}$ is the one of a $u_iv_i$-external upward planar embedding contained in a $uv$-external upward planar embedding of $G_{\mu}$ with shape description $s$. 
		
		The angle check does not fail either. Indeed, consider any vertex $w$ of $\textrm{sk}(\mu)$. Suppose first that $w$ is a switch vertex of $G_{\mu}$. Since $\mathcal E_{\mu}$ is an upward planar embedding, by Property~\textbf{UP1} of Theorem~\ref{th:upward-conditions}, there is one angle incident to $w$ in $\mathcal E_{\mu}$ which is labeled $1$. Hence, at most one component $G_{\nu_i}$ of $G_{\mu}$ is such that the angle at $w$ in one of the internal faces of $\mathcal E_{\nu_i}$ is large; hence, the counter $l_w$ is at most $1$. Suppose next that $w$ is a non-switch vertex of $G_{\mu}$. Since $\mathcal E_{\mu}$ is an upward planar embedding, by Property~\textbf{UP2} of Theorem~\ref{th:upward-conditions}, no angle incident to $w$ in $\mathcal E_{\mu}$ is labeled $1$. Hence, for every non-special component $G_{\nu_i}$ of $G_{\mu}$, the angle at $w$ in every internal face of $\mathcal E_{\nu_i}$ is small, hence the angle check does not fail.
		
		The final check at $u$ and $v$ also does not fail. Namely, if $\lambda^u=1$ ($\lambda^v=1$), then by Properties~\textbf{UP1} and~\textbf{UP2} of Theorem~\ref{th:upward-conditions}, we have that $u$ (resp.\ $v$) is a switch vertex; further, since the angle incident to $u$ (resp.\ to $v$) in the outer face of $\mathcal E_{\mu}$ is large, it follows that every angle at $u$ (resp.\ at $v$) inside every internal face of $\mathcal E_{\mu}$, and consequently inside every internal face of $\mathcal E_{\nu_i}$, is small, hence $u$ (resp.\ $v$) is available.
		
		It remains to prove that the network $\mathcal N(S,T,A)$ admits a flow whose value is equal to the sum of the demands of the sinks in $T$. We compute such a flow as follows. 
		
		\begin{description}
			\item[Switch vertices to faces.] For each available switch vertex $w$ in $\textrm{sk}(\mu)$, if the large angle incident to $w$ in $\mathcal E_{\mu}$ lies in the interior of a face of $\mathcal E_{\mu}$ corresponding to a face $f$ of $\mathcal S_{\mu}$, we let one unit of flow pass from $s_w$ to the sink $t_f$ corresponding to $f$ (this sink is $t_l$ or $t_r$, if $w$ belongs to the left or to the right outer path of $\mathcal S_{\mu}$ and the large angle incident to $w$ lies in the outer face of $\mathcal E_{\mu}$).
			\item[Non-switch vertices to faces.] For each non-switch vertex $w$ of $G_{\mu}$ that corresponds to a source $s_w$ in $S$, let $e_i$ be the virtual edge of $\textrm{sk}(\mu)$ corresponding to the special component for $w$. Then we let one unit of flow pass from $s_w$ to the sink corresponding to $f_l(e_i)$ if $w=v_i$ and the shape description of $\mathcal E_{\nu_i}$ is the left-wing or if $w=u_i$ and the shape description of $\mathcal E_{\nu_i}$ is the inverted-right-wing, or we let one unit of flow pass from $s_w$ to the sink corresponding to $f_r(e_i)$ if $w=v_i$ and the shape description of $\mathcal E_{\nu_i}$ is the right-wing or if $w=u_i$ and the shape description of $\mathcal E_{\nu_i}$ is the inverted-left-wing.
			\item[Extreme or interesting components to faces on the left.] For each extreme or interesting component $G_{\nu_i}$ such that the first value $\tau_l^i$ of $s_i$ is greater than $0$, we let $\tau_l^i$ units of flow pass from $z_l^i$ to the sink $t_f$ corresponding to the face $f=f_l(e_i)$ of $\mathcal S_{\mu}$. 
			\item[Extreme or interesting components to faces on the right.] For each extreme or interesting component $G_{\nu_i}$ such that the  second value $\tau_r^i$ of $s_i$ is greater than $0$, we let $\tau_r^i$ units of flow pass from $z_r^i$ to the sink $t_f$ corresponding to the face $f=f_r(e_i)$ of $\mathcal S_{\mu}$. 
			\item[Non-extreme boring hats to faces.] For each non-extreme boring component $\mathcal G_{\nu_i}$ whose preferred set $\mathcal P_{\nu_i}$ consists of the hat and of the inverted-hat, we let one unit of flow pass from $b^i$ to the sink corresponding to the face $f_r(e_i)$ or to the face $f_l(e_i)$ of $\mathcal S_{\mu}$, depending on whether the shape description of $\mathcal E_{\nu_i}$ is the hat or the inverted-hat, respectively. 
			\item[Non-extreme boring hearts to faces.] For each non-extreme boring component $\mathcal G_{\nu_i}$ whose preferred set $\mathcal P_{\nu_i}$ consists of the heart, or of the inverted-heart, or of both, we let one unit of flow pass from $b^i_1$ to the sink corresponding to the face $f_l(e_i)$ of $\mathcal S_{\mu}$ and unit of flow pass from $b^i_2$ to the sink corresponding to the face $f_r(e_i)$ of $\mathcal S_{\mu}$.
			\item[Switch vertices to non-extreme boring hearts and inverted-hearts.] For each non-extreme boring component $G_{\nu_i}$ whose preferred set contains the heart and the inverted-heart, we let one unit of flow pass from $v_i$ to the sink $t_i$ or from $u_i$ to the sink $t_i$, depending on whether the shape description of $\mathcal E_{\nu_i}$ is the heart or the inverted-heart, respectively (note that this satisfies the demand of the sink $t_i$). 
			\item[First pole.] If $\lambda^u=1$, then we let one unit of flow pass from the source $s_u$ to the sink $t^u$ (note that this satisfies the demand of the sink $t^u$). 
			\item[Second pole.] If $\lambda^v=1$, then we let one unit of flow pass from the source $t_u$ to the sink $t^v$ (note that this satisfies the demand of the sink $t^v$). 
		\end{description}
		
		It is easy to verify that each source in $S$ sends a number of units of flow equal to the maximum number of units of flow it can supply. The only interesting sources to consider are: \begin{itemize}
			\item The source $s_u$ (similar arguments deal with $s_v$); note that $s_u$ exists only if $\lambda^u=1$. If $\lambda^u=1$, then $s_u$ sends one unit of flow to $t^u$ and it does not send flow to any sink corresponding to an internal face of $\mathcal S_{\mu}$, as the angle at $u$ in the outer face of $\mathcal E_{\mu}$ is large, hence the angle at $u$ in any internal face of $\mathcal E_{\mu}$ is small.
			\item Each source $s_w$ that sends one unit of flow to a sink $t_i$ corresponding to a component $G_{\nu_i}$ of $G_{\mu}$ whose preferred set contains the heart and the inverted-heart. By the definition of the flow, the angle at the vertex $w$ corresponding to $s_w$ in the outer face of $\mathcal E_{\nu_i}$ is small (as either $w=v_i$ and the shape of $\mathcal E_{\nu_i}$ is the heart or $w=u_i$ and the shape of $\mathcal E_{\nu_i}$ is the inverted-heart). It follows that the angle at $w$ in any internal face of $\mathcal E_{\mu}$ that is not an internal face of $\mathcal E_{\nu_i}$ is small, hence $s_w$ does not send further units of flow.
		\end{itemize}
		
		We now prove that each sink receives a number of units of flow equal to its demand. This was already observed during the construction of the flow for the sinks $t^u$ and $t^v$, as well as for each sink corresponding to a non-extreme boring component of $G_{\mu}$ whose preferred set contains the heart and the inverted-heart. We now consider the sinks corresponding to internal faces of $\mathcal S_{\mu}$ and for the sinks $t_l$ and $t_r$.
		
		{\bf Consider any sink $t_f$ corresponding to an internal face $f$ of $\mathcal S_{\mu}$.} First, recall that the demand of $t_f$ is equal to $\frac{n_f}{2}-1$, by construction.
		
		For each component $G_{\nu_i}$ of $G_{\mu}$, let $s_i$ be the shape description of $\mathcal E_{\nu_i}$, and let $\tau_l^i$ and $\tau_r^i$ be the left- and right-turn-number of $\mathcal E_{\nu_i}$, respectively. Let $E_1$ ($E_2$) be the set that contains every virtual edge $e_i$ of $\textrm{sk}(\mu)$ such that $f=f_l(e_i)$ (resp.\ $f=f_r(e_i)$). Let $E'_1$ ($E'_2$) be the subset of $E_1$ (resp.\ $E_2$) that only contains the virtual edges $e_i$ such that $\tau_l^i>0$ (resp.\ $\tau_r^i>0$). Let $f_{\mathcal E}$ be the face of $\mathcal E_{\mu}$ corresponding to $f$. Furthermore, let $n'_1(f_{\mathcal E})$ and $n'_{-1}(f_{\mathcal E})$ be the number of large and small angles of $f_{\mathcal E}$, respectively, incident to vertices of $\textrm{sk}(\mu)$. Finally, let $g$ be the number of \emph{ignored} virtual edges; a virtual edge $e_i$ is ignored if: (a) $e_i\in E'_1$, $s_i$ is the left-wing, and the special edge for $v_i$ is incoming $v_i$; or (b) $e_i\in E'_2$, $s_i$ is the right-wing, and the special edge for $v_i$ is incoming $v_i$;  or (c) $e_i\in E'_1$, $s_i$ is the inverted-right-wing, and the special edge for $u_i$ is incoming $u_i$; or or (d) $e_i\in E'_2$, $s_i$ is the inverted-left-wing, and the special edge for $u_i$ is incoming $u_i$.
		
		We now prove a sequence of statements.
		
		\begin{claim} \label{cl:sent-flow}
			The amount of flow sent to $t_f$ is equal to $\sum_{e_i\in E'_1} \tau_l^i + \sum_{e_i\in E'_2} \tau_r^i + n'_1(f_{\mathcal E})-g$.
		\end{claim}
		
		\begin{proof}	
			The number of units of flow that are sent to $t_f$ from the sources in $S$ corresponding to available switch vertices of $\textrm{sk}(\mu)$ is equal to $n'_1(f_{\mathcal E})$. Indeed, if a vertex of $\textrm{sk}(\mu)$ is incident to a large angle in $f_{\mathcal E}$, it is an available switch vertex, by Property~\textbf{UP1} of Theorem~\ref{th:upward-conditions}.
			
			
			
			We now deal with the flow sent to $t_f$ from the sources in $S$ corresponding to components of $G_{\mu}$. We need to prove that the total amount of such flow is $\sum_{e_i\in E'_1} \tau_l^i + \sum_{e_i\in E'_2} \tau_r^i-g$.
			
			\begin{itemize}
				\item Each extreme or interesting component $G_{\nu_i}$ such that $e_i\in E'_1$(such that $e_i\in E'_2$) sends to $t_f$ a number of units of flow equal to $\tau_l^i$ (resp.\ to $\tau_r^i$).
				\item Each non-extreme boring component $G_{\nu_i}$ sends one unit of flow to $t_f$ if $e_i\in E'_1$ and $s_i$ is the left-hat, the heart, or the inverted-heart (in each case we have $\tau_l^i=1$), or if $e_i\in E'_2$ and $s_i$ is the right-hat, the heart, or the inverted-heart (in each case we have $\tau_r^i=1$).
				\item A non-extreme boring component $G_{\nu_i}$ such that $e_i\in E'_1$ ($e_i\in E'_2$), such that $e_i$ is not ignored, and such that $s_i$ is the left-wing or the inverted-right-wing (resp.\ the right-wing or the inverted-left-wing) does not send any unit of flow to $t_f$, while we have $\tau_l^i=1$ (resp.\ $\tau_r^i=1$). However, this ``missing'' unit of flow is ``compensated'' by the non-switch vertex $v_i$ or $u_i$, for which $G_{\nu_i}$ is the special component, which sends one unit of flow to $t_f$. 
				\item  Finally, a non-extreme boring component $G_{\nu_i}$ such that $e_i\in E'_1$ ($e_i\in E'_2$), such that $e_i$ is ignored, and such that $s_i$ is the left-wing or the inverted-right-wing (resp.\ the right-wing or the inverted-left-wing) does not send any unit of flow to $t_f$, while we have $\tau_l^i=1$ (resp.\ $\tau_r^i=1$). This ``missing'' unit of flow is not ``compensated'' by any non-switch vertex, which results in the $g$ term subtracted from the amount of flow sent to $t_f$.
			\end{itemize}  
			
			This concludes the proof of the claim.
		\end{proof}	
		
		\begin{claim} \label{cl:nf}
			We have $n_f=\sum_{e_i\in E_1} |\tau_l^i| + \sum_{e_i\in E_2} |\tau_r^i| + n'_1(f_{\mathcal E})+ n'_{-1}(f_{\mathcal E})-2g$.
		\end{claim}
		
		\begin{proof}
			By Properties~\textbf{UP1} and~\textbf{UP2} of Theorem~\ref{th:upward-conditions}, every large angle of $f_{\mathcal E}$ is incident to a switch vertex. Denote by $p'_{-1}(f_{\mathcal E})$ and $q'_{-1}(f_{\mathcal E})$ the number of small angles of $f_{\mathcal E}$ incident to switch and non-switch vertices of $\textrm{sk}(\mu)$, respectively. Then we have $n'_{-1}(f_{\mathcal E})=p'_{-1}(f_{\mathcal E})+q'_{-1}(f_{\mathcal E})$. By construction, each switch vertex of $\textrm{sk}(\mu)$ incident to $f$ contributes $1$ to $n_f$. Since a switch vertex does not have incident flat angles, the total contribution of the switch vertices to $n_f$ is $n'_1(f_{\mathcal E})+p'_{-1}(f_{\mathcal E})$. 
			
			Each non-switch vertex incident to a small angle in $f_{\mathcal E}$ contributes $1$ to $n_f$, except for the non-switch vertices $w$ such that the special component $G_{\nu_i}$ for $w$ corresponds to an ignored virtual edge $e_i$. By definition, the number of such non-switch vertices is $g$, hence the total contribution of the non-switch vertices to $n_f$ is $q'_{-1}(f_{\mathcal E})-g$. 
			
			We now deal with the contributions to $n_f$ stemming from the components of $G_{\mu}$.
			
			\begin{itemize}
				\item Each extreme or interesting component $G_{\nu_i}$ such that $e_i\in E_1$ ($e_i\in E_2$) contributes to $n_f$ by an amount of $|\tau_l^i|$ (resp.\ $|\tau_r^i|$). 
				\item Similarly, each non-extreme boring component $G_{\nu_i}$ such that $e_i\in E_1$ ($e_i\in E_2$) and such that $s_i$ is the hat, the inverted-hat, the heart, or the inverted-heart contributes to $n_f$ by an amount of $|\tau_l^i|=1$ (resp.\ $|\tau_r^i|=1$). 
				\item Further, each non-extreme boring component $G_{\nu_i}$ such that $e_i\in E_1$ and such that $s_i$ is the sausage, the inverted-sausage, the right-wing, or the inverted-left-wing contributes to $n_f$ by an amount of $|\tau_l^i|=0$. 
				\item Similarly, each non-extreme boring component $G_{\nu_i}$ such that $e_i\in E_2$ and such that $s_i$ is the sausage, the inverted-sausage, the left-wing, or the inverted-right-wing contributes to $n_f$ by an amount of $|\tau_r^i|=0$.
				\item Each non-extreme boring component $G_{\nu_i}$ such that $e_i\in E_1$ ($e_i\in E_2$), such that $e_i$ is not ignored, and such that $s_i$ is the left-wing or the inverted-right-wing (resp.\ the right-wing or the inverted-left-wing) contributes by an amount of $0$ to $n_f$, while we have $\tau_l^i=1$ (resp.\ $\tau_r^i=1$). However, this ``missing'' contribution to $n_f$ is ``compensated'' by the non-switch vertex $v_i$ or $u_i$, for which $G_{\nu_i}$ is the special component, which contributes $1$ to $n_f$ while being incident to a flat angle in $f_{\mathcal E}$. 
				\item Finally, each non-extreme boring component $G_{\nu_i}$ such that $e_i\in E'_1$ ($e_i\in E'_2$), such that $e_i$ is ignored, and such that $s_i$ is the left-wing or the inverted-right-wing (resp.\ the right-wing or the inverted-left-wing) contributes by an amount of $0$ to $n_f$, while we have $\tau_l^i=1$ (resp.\ $\tau_r^i=1$). This ``missing'' contribution to $n_f$ is not ``compensated'' by any non-switch vertex, which results in a second $g$ term subtracted from the value of $n_f$. 
			\end{itemize}
			
			This concludes the proof of the claim.
		\end{proof}
		
		\begin{claim} \label{cl:mixing}
			We have $\sum_{e_i\in E'_1} \tau_l^i - \sum_{e_i\in E_1\setminus E'_1} |\tau_l^i| + \sum_{e_i\in E'_2} \tau_r^i - \sum_{e_i\in E_2\setminus E'_2} |\tau_r^i| + n'_1(f_{\mathcal E}) - n'_{-1}(f_{\mathcal E})=-2$.
		\end{claim}
		
		\begin{proof}
			By definition, for each component $G_{\nu_i}$ of $G_{\mu}$ such that $e_i\in E_1$ ($e_i\in E_2$), the value $\tau_l^i$ (resp.\ $\tau_r^i$) is equal to the number of non-pole vertices of $G_{\nu_i}$ incident to large angles in $f_{\mathcal E}$ minus  the number of non-pole vertices of $G_{\nu_i}$ incident to small angles in $f_{\mathcal E}$. Hence, the sum $\sum_{e_i\in E_1} \tau_l^i + \sum_{e_i\in E_2} \tau_r^i + n'_1(f_{\mathcal E}) - n'_{-1}(f_{\mathcal E}) = \sum_{e_i\in E'_1} \tau_l^i - \sum_{e_i\in E_1\setminus E'_1} |\tau_l^i| + \sum_{e_i\in E'_2} \tau_r^i - \sum_{e_i\in E_2\setminus E'_2} |\tau_r^i| + n'_1(f_{\mathcal E}) - n'_{-1}(f_{\mathcal E})$ is equal to the number of vertices of $G_{\mu}$ incident to large angles in $f_{\mathcal E}$ minus the number of vertices of $G_{\mu}$ incident to small angles in $f_{\mathcal E}$. By Property~\textbf{UP3} of Theorem~\ref{th:upward-conditions}, this difference is equal to $-2$.
		\end{proof}
		
		By Claim~\ref{cl:mixing}, we have $\sum_{e_i\in E_1\setminus E'_1} |\tau_l^i| + \sum_{e_i\in E_2\setminus E'_2} |\tau_r^i|= \sum_{e_i\in E'_1} \tau_l^i + \sum_{e_i\in E'_2} \tau_r^i + n'_1(f_{\mathcal E})-n'_{-1}(f_{\mathcal E}) +2$. This can be replaced into the equality $n_f=\sum_{e_i\in E'_1} \tau_l^i+\sum_{e_i\in E_1\setminus E'_1} |\tau_l^i| + \sum_{e_i\in E'_2} \tau_r^i+\sum_{e_i\in E_2\setminus E'_2} |\tau_r^i| + n'_1(f_{\mathcal E})+ n'_{-1}(f_{\mathcal E})-2g$ given by Claim~\ref{cl:nf}, in order to get that $n_f=2\sum_{e_i\in E'_1}\tau_l^i+2\sum_{e_i\in E'_2}\tau_r^i+2n'_1(f_{\mathcal E})-2g+2$. Hence, the demand of $t_f$ is equal to $\frac{n_f}{2}-1= \sum_{e_i\in E'_1}\tau_l^i+\sum_{e_i\in E'_2}\tau_r^i+n'_1(f_{\mathcal E})-g$, which by Claim~\ref{cl:sent-flow} is equal to the flow sent to $t_f$. 
		

		{\bf Consider now the sink $t_l$} (the arguments for the sink $t_r$ are analogous). First, recall that the demand of $t_l$ is equal to $\frac{\tau_l+n_f^l}{2}$. We need to prove that the amount of flow which is sent to $t_l$ is equal to the same amount. The proof is very similar to the one for the sink $t_f$ corresponding to an internal face $f$ of $\mathcal S_{\mu}$, and is hence only sketched here. 
		
		For each component $G_{\nu_i}$ of $G_{\mu}$, let $\tau_l^i$ and $\tau_r^i$ be defined as before. Let $E_1$ ($E_2$) be the set that contains every virtual edge $e_i$ of $\textrm{sk}(\mu)$ that belongs to the left outer path of $\mathcal S_{\mu}$ and such that the outer face of $\mathcal S_{\mu}$ is $f_l(e_i)$ (resp.\ $f_r(e_i)$). Let $E'_1$ ($E'_2$) be the subset of $E_1$ (resp.\ $E_2$) that only contains the virtual edges $e_i$ such that $\tau_l^i>0$ (resp.\ $\tau_r^i>0$). Let $f_{\mathcal E}$ be the outer face of $\mathcal E_{\mu}$. Further, let $n'_1(f_{\mathcal E})$ and $n'_{-1}(f_{\mathcal E})$ be the number of large and small angles, respectively, incident to $f_{\mathcal E}$ and to non-pole vertices of the left outer path of $\mathcal S_{\mu}$. Finally, let $g$ be defined as before.
		
		As in the proof for the sink $t_f$ corresponding to an internal face $f$ of $\mathcal S_{\mu}$, we need three statements. The first one is the analogue of Claim~\ref{cl:sent-flow}: The amount of flow which is sent to $t_l$ is equal to $\sum_{e_i\in E'_1} \tau_l^i + \sum_{e_i\in E'_2} \tau_r^i + n'_1(f_{\mathcal E})-g$; the proof is identical to the one of Claim~\ref{cl:sent-flow}. The second statement is the analogue of Claim~\ref{cl:nf}: $n_f^l=\sum_{e_i\in E_1} |\tau_l^i| + \sum_{e_i\in E_2} |\tau_r^i| + n'_1(f_{\mathcal E})+ n'_{-1}(f_{\mathcal E})-2g$; again, the proof is identical to the one of Claim~\ref{cl:nf}. Finally, the third statement is similar to Claim~\ref{cl:mixing}: $\sum_{e_i\in E'_1} \tau_l^i - \sum_{e_i\in E_1\setminus E'_1} |\tau_l^i| + \sum_{e_i\in E'_2} \tau_r^i - \sum_{e_i\in E_2\setminus E'_2} |\tau_r^i| + n'_1(f_{\mathcal E}) - n'_{-1}(f_{\mathcal E})=\tau_l$; again, its proof follows the same lines of the proof of Claim~\ref{cl:mixing}, while using that the left-turn-number of $\mathcal E_{\mu}$ is $\tau_l$ in place of Property~\textbf{UP3} of Theorem~\ref{th:upward-conditions}. 
		
		The three statements imply that the amount of flow which is sent to $t_l$ is equal to $\frac{\tau_l+n_f^l}{2}$, as requested. Indeed, by the third statement we have $\sum_{e_i\in E_1\setminus E'_1} |\tau_l^i| + \sum_{e_i\in E_2\setminus E'_2} |\tau_r^i|= \sum_{e_i\in E'_1} \tau_l^i + \sum_{e_i\in E'_2} \tau_r^i + n'_1(f_{\mathcal E})-n'_{-1}(f_{\mathcal E}) -\tau_l$. This can be replaced into the equality $n_f^l=\sum_{e_i\in E'_1} \tau_l^i+\sum_{e_i\in E_1\setminus E'_1} |\tau_l^i| + \sum_{e_i\in E'_2} \tau_r^i+\sum_{e_i\in E_2\setminus E'_2} |\tau_r^i| + n'_1(f_{\mathcal E})+ n'_{-1}(f_{\mathcal E})-2g$ given by the second statement, in order to get that $n_f^l=2\sum_{e_i\in E'_1}\tau_l^i+2\sum_{e_i\in E'_2}\tau_r^i+2n'_1(f_{\mathcal E})-2g-\tau_l$. Hence, the demand of $t_l$ is equal to $\frac{n_f+t_l}{2}= \sum_{e_i\in E'_1}\tau_l^i+\sum_{e_i\in E'_2}\tau_r^i+n'_1(f_{\mathcal E})-g$, which by the first statement is equal to the flow sent to $t_l$.  
		
		
		($\Longleftarrow$) Suppose next that our algorithm determines that $s\in \mathcal F_{\mu}$. We prove that $G_{\mu}$ admits a $uv$-external upward planar embedding $\mathcal E_{\mu}$ with shape description $s$.
		
		Our algorithm determines that $s\in \mathcal F_{\mu}$ when examining some shape descriptions $s_1,\dots,s_h$ for the extreme and interesting components $G_{\nu_1},\dots,G_{\nu_h}$ of $G_{\mu}$, together with a planar embedding $\mathcal S_{\mu}$ for the skeleton $\textrm{sk}(\mu)$ of $\mu$. If the corresponding flow network $\mathcal N$ admits a flow in which each sink is supplied with a number of units of flow equal to its demand, then our algorithm indeed concludes that $s\in \mathcal F_{\mu}$. 
		
		The $uv$-external upward planar embedding $\mathcal E_{\mu}$ is constructed starting from $\mathcal S_{\mu}$, by selecting a $u_iv_i$-external upward planar embedding $\mathcal E_{\nu_i}$ for each component $G_{\nu_i}$ of $G_{\mu}$, and by assigning a label to each vertex of $\textrm{sk}(\mu)$ inside the faces of $\mathcal S_{\mu}$. Selecting a $u_iv_i$-external upward planar embedding $\mathcal E_{\nu_i}$ for each component $G_{\nu_i}$ of $G_{\mu}$ implies that the label of each angle at a vertex $w\notin \{u_i,v_i\}$ of $G_{\nu_i}$ in $\mathcal E_{\mu}$ is the same as in $\mathcal E_{\nu_i}$; further, the label of each angle in an internal face of $\mathcal E_{\nu_i}$ at $u_i$ or $v_i$ in $\mathcal E_{\mu}$ is the same as in $\mathcal E_{\nu_i}$. Hence, after this selection, the only missing labels for an angle assignment are those for the vertices of $\textrm{sk}(\mu)$ inside the faces of $\mathcal S_{\mu}$. 
		
		For $i=1,\dots,h$, we select $\mathcal E_{\nu_i}$ to be any $u_iv_i$-external upward planar embedding of $G_{\nu_i}$ with shape description $s_i$. This embedding exists, as our algorithm only consider shape descriptions in $\mathcal F_{\nu_i}$ for the component $G_{\nu_i}$. It remains to select a $u_iv_i$-external upward planar embedding $\mathcal E_{\nu_i}$ for each non-extreme boring component $G_{\nu_i}$ of $G_{\mu}$ and a label for each vertex of $\textrm{sk}(\mu)$ inside each face of $\mathcal S_{\mu}$. This is done based on the solution for the flow network $\mathcal N$. 
		
		Let $\gamma$ be a function that assigns to each arc $a\in A$ a value $\gamma(a)$ equal to the amount of flow traversing $a$ in the solution for $\mathcal N$. By the integral flow theorem, we can assume that $\gamma(a)$ is a non-negative integer, for each arc $a\in A$. Similarly to~\cite{bdl-udtd-94}, it can be proved that the total demand of the sinks in $T$ is equal to the total amount of units of flow that can be supplied from the sources in $S$. This implies that each source sends out the maximum number of units of flow it can supply. Since the flow is integral and the only sources that can supply more than one unit of flow have a single incident arc, it follows that each source sends its flow on a single incident arc. This allows us to complete the construction of $\mathcal E_{\mu}$ as follows. Consider each non-extreme boring component $G_{\nu_i}$ of $G_{\mu}$ and its preferred set $\mathcal P_{\nu_i}$.
		
		\begin{itemize}
			\item If $\mathcal P_{\nu_i}$ contains a single shape description, then we select $\mathcal E_{\nu_i}$ to be any $u_iv_i$-external upward planar embedding of $G_{\nu_i}$ with that shape description.
			\item If $\mathcal P_{\nu_i}$ contains the hat and the inverted-hat, we select $\mathcal E_{\nu_i}$ to be any $u_iv_i$-external upward planar embedding of $G_{\nu_i}$ with the former or the latter shape description, depending on whether the source $b^i$ sends its unit of flow to the sink corresponding to the face $f_r(e_i)$ or to the sink corresponding to the face $f_l(e_i)$, respectively. 
			\item If $\mathcal P_{\nu_i}$ contains both the heart and the inverted-heart, we select $\mathcal E_{\nu_i}$ to be any $u_iv_i$-external upward planar embedding of $G_{\nu_i}$ with the former or the latter shape description, depending on whether the unit of flow demanded by the sink $t_i$ is supplied by a source corresponding to $v_i$ or by a source corresponding to $u_i$, respectively. 
			\item If $\mathcal P_{\nu_i}$ contains both the left-wing and the right-wing (the case in which it contains both the inverted-left-wing and the inverted-right-wing is analogous), then we further distinguish two cases. 
			\begin{description}
				\item[$v_i$ is a bottom vertex:] then we select $\mathcal E_{\nu_i}$ to be any $u_iv_i$-external upward planar embedding of $G_{\nu_i}$ whose shape description is the left-wing or the right-wing (the choice on whether to use one or the other shape description is arbitrary);
				\item[$v_i$ is a top vertex:] then we select $\mathcal E_{\nu_i}$ to be any $u_iv_i$-external upward planar embedding of $G_{\nu_i}$ whose shape description is the left-wing or the right-wing, depending on whether the source corresponding to the non-switch vertex $v_i$ sends its unit of flow to the sink corresponding to $f_l(e_i)$ or to the sink corresponding to $f_r(e_i)$, respectively. 
			\end{description}
		\end{itemize} 
	
		We now have a planar embedding (which is not yet an upward planar embedding) of $G_{\mu}$. For each vertex $w$ of $\textrm{sk}(\mu)$ that is a non-switch vertex of $G_{\mu}$, we label each angle at $w$ in a face of $\mathcal S_{\mu}$ incident to precisely one incoming edge of $w$ with the value $0$, and we label every other angle at $w$ in a face of $\mathcal S_{\mu}$ with the value $-1$. For each vertex $w$ of $\textrm{sk}(\mu)$ that is an unavailable switch vertex of $G_{\mu}$, we label each angle at $w$ in a face of $\mathcal S_{\mu}$ with the value $-1$. For each vertex $w$ of $\textrm{sk}(\mu)$ that is an available switch vertex of $G_{\mu}$, let $f$ be the unique face of $\mathcal S_{\mu}$ such that the source corresponding to $w$ sends its unit of flow to the sink corresponding to $f$; then we label the angle at $w$ in $f$ with the value $1$ and we label every other angle at $w$ in a face of $\mathcal S_{\mu}$ with the value $-1$.
				
		This completes the construction of the upward planar embedding $\mathcal E_{\mu}$ of $G_{\mu}$. Since $u$ and $v$ are incident to the outer face of $\mathcal S_{\mu}$, by construction, it follows that $\mathcal E_{\mu}$ is $uv$-external. We now prove that $\mathcal E_{\mu}$ is actually an upward planar embedding of $G_{\mu}$, i.e., it satisfies Properties~\textbf{UP0}--\textbf{UP3} of Theorem~\ref{th:upward-conditions}. Let $\lambda$ be the angle assignment in $\mathcal E_{\mu}$; for each vertex $w$ of $G_{\mu}$ and each face $f$ of $\mathcal E_{\mu}$, we define $n_{-1}(w)$, $n_{0}(w)$, $n_{1}(w)$, $n_{-1}(f)$, and $n_{1}(f)$ as before Theorem~\ref{th:upward-conditions}. 
	
		Property~\textbf{UP0} (if $\alpha$ is a switch angle of $\mathcal E_{\mu}$, then $\lambda(\alpha)\in\{-1,1\}$, and if $\alpha$ is a flat angle, then $\lambda(\alpha)=0$) is true for each angle $\alpha$ internal to $\mathcal E_{\nu_i}$ and for each angle $\alpha$ at a non-pole vertex of $G_{\nu_i}$, for any $i\in \{1,\dots,k\}$, since $\mathcal E_{\nu_i}$ is an upward planar embedding. Further, it is true for each angle $\alpha$ at a vertex of $\textrm{sk}(\mu)$ inside a face of $\mathcal S_{\mu}$, by construction. 
		
		 Property~\textbf{UP1} (if $w$ is a switch vertex of $G_{\mu}$, then $n_1(w)=1$, $n_{-1}(w)=\deg(w)-1$, $n_0(w)=0$) is true for each non-pole switch vertex of $G_{\nu_i}$, for any $i\in \{1,\dots,k\}$, since $\mathcal E_{\nu_i}$ is an upward planar embedding. Further, Property~\textbf{UP1} is true for each available switch vertex $w$ of $\textrm{sk}(\mu)$, as the only angle at $w$ which is labeled $1$ is the one in the face $f$ of $\mathcal S_{\mu}$ such that the source in $S$ corresponding to $w$ sends its unit of flow to the sink in $T$ corresponding to $f$; indeed, every other angle at $w$ in a face of $\mathcal S_{\mu}$ is labeled $-1$, by construction, and every other angle at $w$ in an internal face of an embedding $\mathcal E_{\nu_i}$ is labeled $-1$ since $w$ is available. Moreover, Property~\textbf{UP1} is true for each unavailable switch vertex $w$ of $\textrm{sk}(\mu)$, as every angle at $w$ in a face of $\mathcal S_{\mu}$ is labeled $-1$, by construction, and there is exactly one angle at $w$ in an internal face of an embedding $\mathcal E_{\nu_i}$ that is labeled $1$, while all the others are labeled $-1$, since $w$ is unavailable. Finally, observe that a switch vertex is either available or unavailable, since the angle check does not fail.
		 
		 Property~\textbf{UP2} (if $w$ is a non-switch vertex of $G_{\mu}$, then $n_1(w)=0$, $n_{-1}(w)=deg(w)-2$, $n_0(w)=2$) is true for each non-switch vertex $w$ of $G_{\mu}$. Indeed, by Property~\textbf{UP1} and since $G$ is expanded, there are two angles at $w$ labeled $0$. Moreover, every other angle $\alpha$ at $w$ is labeled $-1$. This is by construction if $\alpha$ is in a face of $\mathcal S_{\mu}$; further, it is since $\mathcal E_{\nu_i}$ satisfies Property~\textbf{UP2} and the angle check did not fail if $\alpha$ is in an internal face of $\mathcal E_{\nu_i}$.
		 
		 We prove Property~\textbf{UP3} (if $f_{\mathcal E}$ is an internal face of $\mathcal E_{\mu}$, then $n_1(f_{\mathcal E})=n_{-1}(f_{\mathcal E})-2$, and if $f_{\mathcal E}$ is the outer face of $\mathcal E_{\mu}$, then $n_{1}(f_{\mathcal E})=n_{-1}(f_{\mathcal E})+2$). Consider an internal face $f_{\mathcal E}$ of $\mathcal E_{\mu}$; the outer face of $\mathcal E_{\mu}$ can be dealt with similarly. If $f_{\mathcal E}$ is also an internal face of $\mathcal E_{\nu_i}$, for some $i\in \{1,\dots,k\}$, then we have $n_1(f_{\mathcal E})=n_{-1}(f_{\mathcal E})-2$, given that $\mathcal E_{\nu_i}$ is an upward planar embedding. Otherwise, $f_{\mathcal E}$ corresponds to an internal face $f$ of $\mathcal S_{\mu}$. Let $E_1$, $E_2$, $E'_1$, $E'_2$, $n'_1(f_{\mathcal E})$, $n'_{-1}(f_{\mathcal E})$, $\tau_l^i$, $\tau_r^i$, and $g$ be defined as in the proof of necessity. 
		 
		 A key observation is the following.
		 
		 \begin{claim} \label{cl:turn-angles}
		 If the number of small angles in $f_{\mathcal E}$ at non-pole vertices of the components $G_{\nu_i}$ of $G_{\mu}$ is $\chi$, for some $\chi\geq 0$, then the number of large angles in $f_{\mathcal E}$ at non-pole vertices of the components $G_{\nu_i}$ of $G_{\mu}$ is equal to $\chi+\sum_{e_i\in E_1} \tau_l^i +\sum_{e_i\in E_2} \tau_r^i$.
		 \end{claim}
	 
	 	\begin{proof}		 
		For every edge $e_i\in E_1$, we have that, if the number of small angles in $f_{\mathcal E}$ at non-pole vertices of $G_{\nu_i}$ is equal to $x_i$, for some integer $x_i\geq 0$, then the number of large angles in $f_{\mathcal E}$ at non-pole vertices of $G_{\nu_i}$ is equal to $\tau_l^i+x_i$; this just comes from the definition of the left-turn-number of $\mathcal E_{\nu_i}$. A similar relationship (with $\tau_r^i$ in place of $\tau_l^i$) holds true for every edge $e_i\in E_2$. Then the claim follows with $\chi= \sum_{e_i\in E_1\cup E_2}x_i$.
		\end{proof}
		
		By Claim~\ref{cl:turn-angles}, the number of large angles in $f_{\mathcal E}$ is equal to $\sum_{e_i\in E_1} \tau_l^i +\sum_{e_i\in E_2} \tau_r^i+\chi+n'_1(f_{\mathcal E})$, while the number of small angles in $f_{\mathcal E}$ is equal to $\chi+n'_{-1}(f_{\mathcal E})$. Hence, we want to prove that $\sum_{e_i\in E_1} \tau_l^i +\sum_{e_i\in E_2} \tau_r^i+n'_1(f_{\mathcal E})-n'_{-1}(f_{\mathcal E})=-2$.
		
		The rest of the proof proceeds similarly to the proof of sufficiency. First, we have that 	the amount of flow sent to the sink $t_f$ corresponding to $f$ is equal to $\sum_{e_i\in E'_1} \tau_l^i + \sum_{e_i\in E'_2} \tau_r^i + n'_1(f_{\mathcal E})-g$ (as in Claim~\ref{cl:sent-flow}). This is because: 
		
		\begin{description}
			\item[(a)] we make the angle at a switch vertex of $\textrm{sk}(\mu)$ large if and only if the source corresponding to the switch vertex sends one unit of flow to $t_f$;
			\item[(b)] $t_f$ receives $\tau_l^i$ ($\tau_r^i$) units of flow from an extreme or interesting component $G_{\nu_i}$ such that $e_i \in E_1$ (resp.\ $e_i \in E_2$) if and only if $\tau_l^i>0$ (resp.\ $\tau_r^i>0$);
			\item[(c)] for each edge $e_i \in E_1$ such that $G_{\nu_i}$ is non-extreme, boring, and non-ignored, we choose $\mathcal E_{\nu_i}$ in such a way that $\tau_l^i=1$ if and only if $t_f$ receives one unit of flow from a source corresponding to $G_{\nu_i}$ or corresponding to a non-switch vertex for which $G_{\nu_i}$ is the special component;
			\item[(d)] for each edge $e_i \in E_2$ such that $G_{\nu_i}$ is non-extreme, boring, and non-ignored, we choose $\mathcal E_{\nu_i}$ in such a way that $\tau_r^i=1$ if and only if $t_f$ receives one unit of flow from a source corresponding to $G_{\nu_i}$ or corresponding to a non-switch vertex for which $G_{\nu_i}$ is the special component; and
			\item[(e)] for each edge $e_i \in E_1$ ($e_i \in E_2$) such that $G_{\nu_i}$ is non-extreme, boring, and ignored, we have $\tau_l^i=1$ (resp.\ $\tau_r^i=1$), although this does not correspond to any unit of flow received by $t_f$; this determines the $-g$ term in the amount of flow received by $t_f$. 
		\end{description}

		Similar considerations on the value of $n_f$ establish that $n_f=\sum_{e_i\in E_1} |\tau_l^i| + \sum_{e_i\in E_2} |\tau_r^i| + n'_1(f_{\mathcal E})+ n'_{-1}(f_{\mathcal E})-2g$ (as in Claim~\ref{cl:nf}). Since the demand of $t_f$ is equal to $\frac{n_f}{2}-1$ and since the amount of flow received by $t_f$ is equal to its demand, we get $\sum_{e_i\in E'_1} \tau_l^i + \sum_{e_i\in E'_2} \tau_r^i + n'_1(f_{\mathcal E})-g=\frac{\sum_{e_i\in E_1} |\tau_l^i| + \sum_{e_i\in E_2} |\tau_r^i| + n'_1(f_{\mathcal E})+ n'_{-1}(f_{\mathcal E})-2g}{2}-1$, which is $2\sum_{e_i\in E'_1} \tau_l^i + 2\sum_{e_i\in E'_2} \tau_r^i + 2n'_1(f_{\mathcal E})=\sum_{e_i\in E'_1} \tau_l^i - \sum_{e_i\in E_1\setminus E'_1} \tau_l^i + \sum_{e_i\in E'_2} \tau_r^i - \sum_{e_i\in E_2\setminus E'_2} \tau_r^i + n'_1(f_{\mathcal E})+ n'_{-1}(f_{\mathcal E})-2$. This gives us $\sum_{e_i\in E_1} \tau_l^i +\sum_{e_i\in E_2} \tau_r^i+n'_1(f_{\mathcal E})-n'_{-1}(f_{\mathcal E})=-2$, as desired.


	Finally, we need to prove that the shape description of $\mathcal E_{\mu}$ is $s$. The proof that the left-turn-number of $\mathcal E_{\mu}$ coincides with the first value of $s$, say $\tau_l$, is very similar to the one we just presented that $n_1(f_{\mathcal E})=n_{-1}(f_{\mathcal E})-2$ (the main difference is that the demand of the sink $t_l$ is equal to $\frac{n_f+\tau_l}{2}$ and this value is used in place of $\frac{n_f}{2}-1$); this proof is hence omitted here. Likewise for the proof that the right-turn-number of $\mathcal E_{\mu}$ coincides with the second value of $s$. That the label of $u$ in the outer face of $\mathcal E_{\mu}$ coincides with the third value of $s$, say $\lambda^u$, comes from the requirement that the source $s_u$ corresponding to $u$ must send one unit of flow to the sink $t^u$, if $\lambda^u=1$, comes from the requirement that $s_u$ must send one unit of flow to an internal face of $\mathcal S_{\mu}$ incident to $u$, if $\lambda^u=-1$, and comes from the extreme-edge check otherwise. A similar argument shows that the label of $v$ in the outer face of $\mathcal E_{\mu}$ coincides with the fourth value of $s$. Finally, that the last four values in the shape description of $\mathcal E_{\mu}$ coincide with   the last four values of $s$ directly comes from the extreme-edge check. 	
	\end{proof} 
	
	We are finally ready to state the main theorem of this section.
	
	  \begin{theorem}  \label{th:sources}
			It is possible to solve {\sc Upward Planarity} in $\bigoh(\sigma 1.45^\sigma \cdot n^2\log^3 n))$ time for a digraph with $n$ vertices and $\sigma$ sources.
		\end{theorem}	
		
		\begin{proof}
			Let $G$ be an $n$-vertex digraph whose underlying graph is planar and let $\sigma$ be the number of sources of $G$.  First, we expand $G$, which is done in $\bigoh(n)$ time. As discussed in Section~\ref{sec:prelim}, an expansion does not alter whether $G$ is upward planar, it preserves the number of sources, and at most doubles the number of vertices of $G$. By Observation~\ref{obs:maintains-sources}, each biconnected component $B$ of $G$ is expanded and has at most $\sigma$ sources. For any edge $e^*$ of $B$, the SPQR-tree $T$ of $B$ can be computed in $\bigoh(n)$ time. By Lemma~\ref{lem:spirality_range}, for any R-node $\mu$ of $T$, the left- and right-turn-numbers of any $uv$-external upward planar embedding of the pertinent graph of $\mu$ are in $\bigoh(\sigma)$. Further, by Lemma~\ref{lem:R_node_sources},  there exists an R-node subprocedure whose total time complexity is $\bigoh(\sigma 1.45^\sigma \cdot n\log^3 n)$. Hence, by Lemma~\ref{lem:R_node_general}, {\sc Upward Planarity} can be solved in  $\bigoh(n((\sigma 1.45^\sigma \cdot n\log^3 n)+(\sigma^2\cdot n)))\subseteq \bigoh(\sigma 1.45^\sigma \cdot n^2\log^3 n)$ time.
		\end{proof}

\fi

\section{An Algorithm Parameterized by Treewidth}
\label{sec:tw}
The aim of this section is to provide an R-node subprocedure which yields parameterized algorithms for \UP\ when parameterized by treewidth and treedepth. The idea behind this is to obtain a combinatorization of the task that is asked in the subprocedure. This will be done by extending the skeleton of the R-node with additional information, notably via a so-called \emph{embedding graph}\ifshort\footnote{Note that this notion differs from the embedding graphs used in recent drawing extension problems~\cite{EibenGHKN20,GanianHKPV21}; unlike in those problems, here it seems impossible to use Courcelle's Theorem~\cite{Courcelle90}.}\fi. The R-node subprocedure is then obtained by performing dynamic programming over the embedding graph. However, to obtain the desired runtime, we will first have to show that the embedding graph has bounded treewidth.

\iflong
It is worth noting that the notion of embedding graph used here shares similarities with the embedding graphs that have been recently used to solve several drawing extension problems~\cite{EibenGHKN20,GanianHKPV21}. However, the notions are not the same due to problem-specific differences, and the use of the combinatorizations differ as well: while in the two aforementioned papers the combinatorizations were used in conjunction with Courcelle's Theorem~\cite{Courcelle90}, this approach cannot be used here due to the fact that we will need to make the turn numbers add up to $- 2$ for each non-outer face. 
\fi

\iflong
	\subsection{A Combinatorial Representation of the Skeleton}
	\label{sub:embeddinggraph}
	\fi
\ifshort
	\subparagraph{A Combinatorial Representation of the Skeleton.}
\fi	
	
		\newcommand{\Vf}{V^{\texttt{f}}}
\newcommand{\Ve}{V^{\texttt{e}}}
\newcommand{\Vv}{V^{\texttt{v}}}

	Let $G$ be a connected graph with a planar embedding $\cG$, and let $F$ be the set of faces of $\cG$. Let $G^-$ be the graph obtained from $G$ by subdividing each edge $e$ once, creating the vertex $v_e$. We define the \emph{embedding graph} \signG\ of $G$ as the graph obtained from $G^-$ by adding a vertex $f$ for each face in $F$, and connecting $f$ to each vertex in $G^-$ incident to $f$. \iflong Formally, $V(\signG)=V(G)\cup F \cup \{v_e~|~e\in E(G)\}=\Vv\cup \Vf\cup \Ve$, and $E(\signG)=E(G^-)\cup  \{fv~|~v\in V(G) \wedge v\text{ is incident to f }\} \cup \{fv_e~|~e\in E(G) \wedge \text{ both endpoints of $e$ are incident to f }\}$.  \fi Observe that \signG\ is tripartite, and we call the three sets of vertices that occur in the definition of $V(\signG)$ the \emph{true vertices}, \emph{face-vertices} and \emph{edge-vertices} of \signG, respectively. An illustration is shown in Fig.~\ref{fig:embeddingGraph}.
	
	\begin{figure}[htb]
		\centering
		\begin{subfigure}{.28\textwidth}
			\centering
			\includegraphics[width=\columnwidth, page=1]{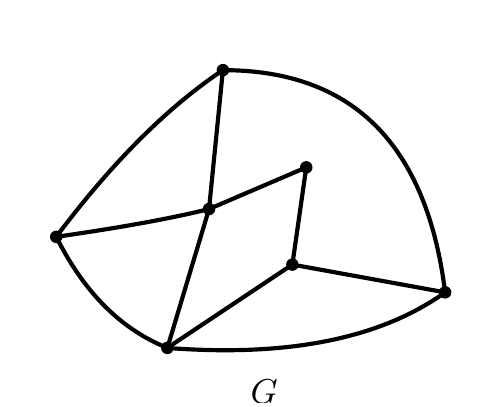}
			\subcaption{}
			\label{fig:embeddingGraph_1}
		\end{subfigure}
		\hfil
		\begin{subfigure}{.28\textwidth}
			\centering
			\includegraphics[width=\columnwidth, page=2]{figures/embeddingGraph}
			\subcaption{}
			\label{fig:embeddingGraph_2}
		\end{subfigure}
		\caption{
			(a) A planar graph $G$. (b) The embedding graph $\tilde{G}$ of $G$. True-, face-, and edge-vertices are shown in black, green, and orange, respectively.
			\label{fig:embeddingGraph}
		}
	\end{figure}
	
	\iflong Furthermore, note that each edge-vertex is adjacent to at most two face-vertices.
	
	 \fi	
\ifshort
Our aim in this section is to show that $\tw(\signG)$ is linearly bounded by $\tw(G)$.
To do so, we identify the faces that are, in some sense, ``relevant'' for a bag in a tree decomposition of $G^-$, and prove that (1) the number of such relevant faces is linearly bounded by the width of that decomposition and (2) adding these faces to the decomposition of $G^-$ results in a tree-decomposition of \signG. We can then prove:
\fi
\iflong
It is well known (and easy to show) that if the treewidth of a graph $G$ is at most $k$, then the treewidth of its subdivision is at most $k$~\cite{CyganFKLMPPS15,DowneyF13}. Our aim in this section, however, will be to show that if the treewidth of an embedded graph is bounded, then the treewidth of its embedding graph is also bounded. 
We begin by identifying a small number of faces that are, in some sense, ``relevant'' for a bag in a tree decomposition of $G^-$. We will then show that it is possible to construct a bounded-width decomposition of the embedding graph by adding the face-vertices of these relevant faces to the appropriate bags.

Given a tree-decomposition $(T,\chi)$ of $G^-$, a face $f\in F$ is called \emph{crossing} for a node $t\in T$ if $f$ is incident to at least one vertex in $\fut(t)$ and to at least one vertex in $\past(t)$. Our first task is to bound the number of crossing faces in a bag.

\newcommand{\cfbound}[1]{2#1}

\iflong \begin{lemma} \fi \ifshort \begin{lemma}\fi
\label{lem:crossfacebound}
If $G$ has treewidth $k$, then there exists a tree-decomposition $(T,\chi)$ of width $k$ such that there are at most $\cfbound{k}$ crossing faces for each node $t$ of $T$.
\end{lemma}
%
%

\iflong \begin{proof}
The lemma follows from previous work on so-called \emph{geometric tree-decompositions}~\cite{Dorn10}. In particular, Dorn~\cite[Theorem 1]{Dorn10} showed that if a plane graph $G$ has treewidth $k$, there exists a tree-decomposition $(T,\chi)$ of $G$ of width at most $k$ with the additional property that for each node $t\in V(T)$, the bag $\chi(t)$ contains at most two vertex-subsets $S_1,S_2\subseteq \chi(t)$ and there is a partitioning of $\past(t)$ into $\past_1(t),\past_2(t)$ with the following properties:
\begin{itemize}
\item $S_1, S_2$ are minimal separators in $G$ that separate $\fut(t)$ from $\past_1(t)$ and $\past_2(t)$, respectively, and
\item there are Jordan curves $C_1, C_2$ which intersect the drawing of $G$ precisely in the vertices of $S_1, S_2$, respectively, and separate $\fut(t)$ from $\past_1(t)$ and $\past_2(t)$, respectively.
\end{itemize}

Given the above, consider such a geometric tree-decomposition $(T,\chi)$ of an arbitrary drawing that realizes the embedding of $G^-$. Each crossing face must, by definition, be intersected by at least one of the Jordan curves $C_1$, $C_2$. Since each of these curves intersects at most $k$ vertices and no edges, the total number of faces intersecting at least one of these curves is upper-bounded by $2k$.
%
%
\end{proof}\fi

%
%

Let us now proceed with a tree-decomposition $(T,\chi)$ satisfying Lemma~\ref{lem:crossfacebound}, and let $F(t)$ be the set of crossing faces for $t$. We say that a face is \emph{inside} a bag $t$ if every vertex on its boundary lies in $t$, and denote the set of faces inside $t$ as $I(t)$. 

\newcommand{\inbound}[1]{2#1-4}
\iflong \begin{observation} \fi \ifshort \begin{observation}\fi
\label{obs:inbound}
$|I(t)|\leq \inbound{k}$.
\end{observation}

\iflong \begin{proof}
Let $H$ be the subgraph of $G$ induced on $\chi(t)$. Since $H$ is planar, by Euler's formula we obtain that it contains at most $2k-4$ faces. Moreover, $I(t)$ is a subset of the faces of~$H$. 
\end{proof}\fi

Finally, in order to obtain a valid tree-decomposition for \signG, we will need to add the set $F(t)$ of crossing faces also to all nodes adjacent to $t$ (this will be important in cases where $t$ is a forget node or when $t$'s parent is an introduce node). Formally, let the \emph{propagation faces} $P(t)$ for $t\in T$ simply be defined as $\bigcup_{\{p~|~pt\in E(T)\}}F(p)$. Since a nice tree-decomposition has degree at most $3$, it follows that $|P(t)|\leq 3\cdot \cfbound{k}$.


Let $\psi(t)=\chi(t)\cup F(t)\cup I(t)\cup P(t)$. While Lemma~\ref{lem:crossfacebound} together with Observation~\ref{obs:inbound} means that adding the face-vertices for the crossing, inside and propagation faces to the relevant bags (i.e., replacing $\chi(t)$ with $\psi(t)$) would preserve a bound on the bag size, we need to show that this construction results in a valid tree-decomposition for \signG. 

\iflong \begin{lemma} \fi \ifshort \begin{lemma}\fi
\label{lem:correcttw}
$(T,\psi(t))$ is a tree-decomposition of \signG.
\end{lemma}

\iflong \begin{proof}
To prove the lemma, it suffices to verify that (1) every edge incident to $F$ occurs in some bag, and (2) for each face-vertex $f\in F$, the set of nodes of $T$ satisfying $f\in \chi(t)$ forms a nonempty subtree of~$T$. But first, we'll establish the following subclaim: for each face-vertex $f$, there exists a node $t$ such that $f\in \psi(t)$. 

Indeed, if all vertices adjacent to $f$ in $\signG$ occur together in some bag, say $\psi(t)$, then $f\in \psi(t)$ and the subclaim holds. Otherwise, let us consider a forget node $t$ with a child $p$ such that $\chi(p)\setminus \chi(t)=\{z\}$ where $zf\in E(\signG)$ with the additional property that $t$ achieves maximum distance from the root $r$ (i.e., $t$ is ``as far away'' from the root as possible). In particular, this implies that no other forget node which forgets a vertex adjacent to $f$ occurs in the subtree rooted at $t$. Since $\chi(p)$ does not contain all the neighbors of $f$ and no neighbor of $f$ was forgotten in the subtree rooted at $t$, there must exist some neighbor of $f$, say $y$, in $\fut(t)$. Since $z\in \past(t)$, we have that $f\in F(t)$ and hence $f\in \psi(t)$, and the subclaim holds.

Let us now turn to establishing property (1). Again, if all vertices adjacent to $f$ in $\signG$ occur in some bag, then the property trivially holds. Otherwise, let $v$ be an arbitrary vertex such that $vf\in E(\signG)$; our aim is to show that $\{v,f\}$ occur together in some bag of $(T,\psi)$. Recall that there must exist some node in $T$ where $f$ occurs as a crossing vertex; let us pick one such node, say $u$. Consider the node $t$ where $v$ was forgotten. We distinguish three cases.

\begin{itemize}
\item $u$ is a descendant of $t$: There must exist some vertex $w\in \past(u)$ adjacent to $f$. Consider the unique $u$-$t$ path in $T$, and let $s'$ be the unique node on this path such that its parent $s'$ introduces $v$. Since $v\in \fut(s)$ and $w\in \past(s)$, it follows that $f\in F(s)$. But then $\{f,v\}\subseteq P(s)$, as desired.

\item $u$ is an ancestor of $t$: There must exist some vertex $w\in \fut(u)$ adjacent to $f$. But then $w\in \fut(t)$ as well, and since $v\in \past(t)$ we obtain that $f\in F(t)$. Then $\{f,v\}\subseteq P(t')$ where $t'$ is the unique child of $t$, as desired.

\item $u$ is neither an ancestor nor a descendant of $t$: There must exist some vertex $w\in \past(u)$ adjacent to $f$. But then $w\in \fut(t)$, and since $v\in \past(t)$ we obtain that $f\in F(t)$. Then $\{f,v\}\subseteq P(t')$ where $t'$ is the unique child of $t$, as desired.
\end{itemize}

Finally, let us establish property (2). In particular, let $t,u$ be nodes of $t$ such that $f\in \psi(t)\cap \psi(u)$, and it suffices to show that the bag of every node on the $t$-$u$ path in $T$ contains $f$. Let us first consider the case where $f\in I(t)$. Since the intersection of connected subtrees on $T$ is itself a connected subtree and each vertex $v$ adjacent to $f$ occurs on a connected subtree of $T$, it follows that the set of bags where $f\in I(t)$ also forms a connected subtree. Moreover, this implies that in this case there can be no node $u$ such that $f\in F(u)$, and so the claim holds.

On the other hand, assume that $f\in F(t)\cup P(t)$. If $f\in P(t)$, then there exists a neighbor $t'$ of $t$ such that $f\in F(t')$ and the $t'$-$u$ path contains the $t$-$u$ path; hence it suffices to establish the case where $f\in F(t)$. Moreover, the previous paragraph established that in this case $f\in F(u)\cup P(u)$, and hence by the same argument as in the previous sentence it also suffices to consider the case where $f\in F(u)$. We now distinguish three cases.

\begin{itemize}
\item $u$ is a descendant of $t$: By the definition of crossing faces, there must exist a vertex $a\in \past(u)$ and a vertex $b\in \fut(t)$ where $a\neq b$. But then each bag $p$ on the $u$-$t$ path in $T$ satisfies that $a\in \past(p)$ and $b\in \fut(p)$, and in particular $f\in F(p)$.

\item $t$ is a descendant of $u$: This case is completely symmetrical to the previous one.

\item $u$ is neither an ancestor nor a descendant of $t$: By the definition of crossing faces, there must exist a vertex $a\in \past(u)$ and a vertex $b\in \past(t)$ where $a\neq b$. Let $s$ by the least common ancestor of $t$ and $u$. Each node $p$ on the unique $t$-$s$ path except for $s$ satisfies that $b\in \past(p)$ and $a\in \fut(p)$, meaning that $f\in F(p)$. The same claim then holds symmetrically for the unique $u$-$s$ path. Finally, since we showed that $s$ has a neighbor $p'$ satisfying $f\in F(p')$, we obtain $f\in P(s)$ as well. In particular, the bag of every node on the $t$-$u$ path contains $f$.\qedhere
\end{itemize}
%
%
%
%
\end{proof}\fi

Lemma~\ref{lem:crossfacebound}, Observation~\ref{obs:inbound} and Lemma~\ref{lem:correcttw}	together with the known fact that subdividing edges does not increase the treewidth~\cite{CyganFKLMPPS15,DowneyF13} immediately imply the following result, which we believe may be of general interest.
\fi

\iflong \begin{theorem} \fi \ifshort \begin{theorem}\fi
\label{thm:embeddingtw}
Let $G$ be a graph with a planar embedding of treewidth $k$ where $k\geq 1$. Then the embedding graph $\signG$ has treewidth at most $11k-4\in \bigoh(k)$.
\end{theorem}

\newcommand{\pos}{\emph{pos}}	
\iflong
\subsection{Problem Reformulation}
\label{sub:problemreformulation}
\fi
\ifshort
\subparagraph{Problem Reformulation.}
\fi

\ifshort
Our second task is to formulate the problem we have to solve on a given embedding graph.
First of all, the R-node subprocedure required by Lemma~\ref{lem:R_node_general} can be straightforwardly reduced to the task of checking whether a specific shape description $\psi$ can be achieved at the R-node. This reduction takes at most $\bigoh(\tau)$ time by Lemma~\ref{lem:matrix_feasible}. At this point, the input consists of (1) an R-node $\mu$ of $T$ with skeleton $H$, (2) a mapping $\mathcal{S}_\mu$ which assigns each virtual edge in $H$ to its feasible set, (3) a bound $\kappa$ on the treewidth of the embedding graph $\signH$ obtained from $H$, and (4) a target shape description $\psi$.

The combinatorial reformulation we obtain can be stated as follows:
\fi
\iflong
While the embedding graph defined in Subsection~\ref{sub:embeddinggraph} is well-suited for our purposes, we still need to formulate the problem we will be solving on these embedding graphs. First of all, the R-node subprocedure required by Lemma~\ref{lem:R_node_general} can be straightforwardly reduced to the task of checking whether a specific shape description $\psi$ can be achieved at the R-node. This reduction takes time at most $\bigoh(\tau)$ by Lemma~\ref{lem:matrix_feasible}. At this point, the input consists of (1) an R-node $\mu$ of $T$ with skeleton $H$, (2) a mapping $\mathcal{S}_\mu$ which assigns each virtual edge in $H$ to its feasible set, (3) a bound $\kappa$ on the treewidth of the embedding graph $\signH$ obtained from $H$, and (4) a target shape description $\psi$.

The main difficulty here---and the task of this subsection---will be to express the question of whether $\psi$ is in the feasible set of $\mu$ as a combinatorial problem over the embedding graph $\signH$ of $H$:
\fi
\ifshort
Determine if there exists an \emph{angle mapping} $\alpha$ and \emph{shape selector} $\beta$ which is \emph{valid}, where

\begin{itemize}
\item an angle mapping $\alpha$ maps each switch vertex $v\in V(\signH)$ to a vertex in $N_{\signH}(v)$; intuitively, this describes where  the large angle at $v$ is  in the upward planar embedding of the pertinent graph (this may be in a face between two virtual edges--and $\alpha$ maps $v$ to the corresponding face vertex---or in a virtual edge--and $\alpha$ maps $v$ to that virtual edge), 
\item a shape selector $\beta$ maps each edge-vertex $v_e$ obtained from the virtual edge $e$ of $H$ to a shape description that occurs in a feasible set in the range of $\mathcal{S}_\mu(e)$, and
\item intuitively, a pair $(\alpha,\beta)$ is valid if it satisfies three Validity Conditions: (1) all face-vertices receive the correct number of small and large angles from $\alpha$ and $\beta$, (2) for each true-vertex $v$ and adjacent edge-vertex $w$, $\alpha(v)$ is consistent with the requirements of the shape selected by $\beta(w)$, and (3) the shape of the outer face is consistent with $\psi$. 
\end{itemize}
\fi
\iflong
Determine if there exists a \emph{valid} angle mapping $\alpha$ and a shape selector $\beta$.
We will need some additional terminology to to elaborate on the meaning of ``valid'', ``angle mapping'' and ``shape selector''. Denote the poles of $\mu$ by $\tilde{u}$, $\tilde{v}$. Fix a $\tilde{u}\tilde{v}$-external planar embedding $\mathcal{E}^H_\mu$ of $H$. Let $\signH$ be the corresponding embedding graph, where $V(\signH) = \Vv \cup \Ve \cup \Vf$. Let $F(H)$ denote the set of faces defined by $\mathcal{E}^H_\mu$, we may assume that the outer face $f_H \in F(H)$ is fixed. For each edge $e \in E(H)$, denote by $\nu^e$ the respective child of the node $\mu$ in the used SPQR decomposition, and recall that $G_{\nu^e}$ is then the corresponding pertinent graph.
Let the set $A_\mu$ of active switch vertices be defined as the set of all switch vertices in $G_\mu$ which are contained in $V(H)$. Then:


\begin{itemize}
\item an angle mapping $\alpha$ maps each active switch vertex $v \in A_\mu$ to a vertex in $N_{\signH}(v)$; intuitively, this describes the placement of the large angle at $v$ in the upward planar embedding of the pertinent graph (which could either be in a face between two virtual edges---in which case $\alpha$ maps $v$ to the corresponding face vertex---or in a virtual edges---in which case $\alpha$ maps $v$ to that virtual edge). 
\item a shape selector $\beta$ maps each vertex $v_e$ in $\Ve$ obtained from virtual edge $e$ of $H$ to a shape description that occurs in a feasible set in the range of $\mathcal{S}_\mu(e)$.
\end{itemize}

\newcommand{\turnt}{\texttt{turn}}

It remains to formalize the meaning of ``valid''. For a true vertex $v \in \Vv$ and an edge-vertex $e \in \Ve$ adjacent to $v$, denote by $\lambda_\beta(v, v_e)$ the value assigned to the outer angle of $v$ in $G_{\nu^e}$ by the shape description $\beta(v_e)$. For a true vertex $v \in \Vv$ and an adjacent face vertex $v_f \in \Vf$, denote by $\lambda_\alpha(v, v_f)$ the value that $\alpha$ induces for the angle at $v$ inside $v_f$. That is, if $v \in A_\mu$, $\lambda_\alpha(v, v_f)$ is equal to $1$ if $\alpha(v) = v_f$, and $-1$ otherwise. If $v \notin A_\mu$, $\lambda_\alpha(v, v_f)$ is equal to $0$ if the corresponding angle is flat angle, and $-1$ otherwise, independently of $\alpha$. For a face vertex $v_f \in \Vf$ and an adjacent edge-vertex $v_e \in \Ve$, denote by $\turnt_\beta(v_f, v_e)$ the turn number that $\beta(v_e)$ induces on the respective part of $f$.
Specifically, if $v_e$ corresponds to the component $(G_{\nu^e}, u, v)$, then $\turnt_\beta(v_f, v_e)$ is the left-turn-number $\tau_l(G_{\nu^e}, u, v)$ if $u$ appears right before $v$ in the cyclic order induced by $\mathcal{E}^H_\mu$ on the vertices of $H$ adjacent to $v_f$, and the right-turn-number $\tau_r(G_{\nu^e}, u, v)$ otherwise, where both $\tau_l(G_{\nu^e}, u, v)$ and $\tau_r(G_{\nu^e}, u, v)$ are defined  by the shape $\beta(v_e)$. Analogously, for $v$ adjacent to $v_e$ define $\rho_\beta(v_f, v_e, v)$ as the value $\rho_l(G_{\nu^e}, v)$ or $\rho_r(G_{\nu^e}, v)$.

The pair $(\alpha,\beta)$ is valid if the following three Validity Conditions hold:
\begin{enumerate}
    \item For each face vertex $f\in \Vf$, we have that
        $$\sum_{v_e\in \Ve\cap N_{\signH}(v_f)}\turnt_\beta(v_f, v_e) + \sum_{v\in \Vv\cap N_{\signH}(v_f)}\lambda_\alpha(v, v_f) = \begin{cases}
            2, & \text{ if } f = f_H\\
            -2, & \text{ otherwise.}
    \end{cases}$$
\item For each true vertex $v \in \Vv$, the following conditions hold on the pole angles in the shapes selected by $\beta$ for the incident edges:
    \begin{enumerate}
        \item If $v$ is a switch vertex in $G_\mu$, for each $v_e \in \Ve \cap N_{\signH}(v) \setminus \{\alpha(v)\}$, $\lambda_\beta(v, v_e) = 1$, and $\lambda_\beta(v, v_e) = -1$ if $v_e = \alpha(v)$. That is, in accordance with $\alpha$, at most one of the virtual edges ``takes'' a large angle inside, meaning that the outside angle of this component at $v$ is small.
        \item If $v$ is a non-switch vertex in $G_\mu$, for each $v_e \in \Ve \cap N_{\signH}(v)$, we distinguish two cases. If the interior of $G_{\nu^e}$ does not contain a flat angle at $v$, then $\lambda_\beta(v, v_e) = 1$. Note that otherwise (i.e., if there is both an ingoing edge and an outgoing edge at $v$ in $G_{\nu^e}$), $\lambda_\beta(v, v_e) \in \{-1, 0\}$.
    \end{enumerate}
\item $(\alpha, \beta)$ agrees with the target shape $\psi$, which denotes the tuple $(\tau_l(G_\mu, \tilde{u}, \tilde{v}), \tau_r(G_\mu, \tilde{u}, \tilde{v}), \lambda(G_\mu, \tilde{u}), \lambda(G_\mu, \tilde{v}), \rho_l(G_\mu, \tilde{u}), \rho_r(G_\mu, \tilde{u}),  \rho_l(G_\mu, \tilde{v}), \rho_r(G_\mu, \tilde{v}))$. The circular order on the vertices and edges incident to $f_H$ induces a right path $v_0' = \tilde{u}$, $e_1'$, $v_1'$, \ldots, $e_h'$, $v_h' = \tilde{v}$ and a left path $v_0 = \tilde{u}$, $e_1$, $v_1$, \ldots, $e_\ell$, $v_\ell = \tilde{v}$ from $\tilde{u}$ to $\tilde{v}$, where the left path is defined as the reverse of the corresponding path from $\tilde{v}$ to $\tilde{u}$ in the circular order. Let $L_e$ and $L_v$ contain all intermediate edge-vertices and true vertices on the left path, respectively. That is, $L_e = \{e_1, e_2, \dots, e_\ell\}$, $L_v = \{v_1, v_2, \dots, v_{\ell - 1}\}$. For the right path, define analogously $R_e$ and $R_v$.
The following holds: 
    \begin{enumerate}
        \item Left-turn and right-turn numbers match, that is,
            \begin{gather*}
            \sum_{v_e\in L_e}\turnt_\beta(f_H, v_e) + \sum_{v_v\in L_v}\lambda_\alpha(v_v, f_H) = \tau_l(G_\mu, \tilde{u}, \tilde{v}),\\
            \sum_{v_e\in R_e}\turnt_\beta(f_H, v_e) + \sum_{v_v\in R_v}\lambda_\alpha(v_v, f_H) = \tau_r(G_\mu, \tilde{u}, \tilde{v}).
            \end{gather*}
        \item Angles at the poles match, 
            \begin{gather*}
                \lambda_\alpha(\tilde{u}, f_H) = \lambda(G_\mu, \tilde{u}),\\
                \lambda_\alpha(\tilde{v}, f_H) = \lambda(G_\mu, \tilde{v}).
            \end{gather*}
        \item Boundary edge directions at the poles match,
            \begin{gather*}
                \rho_\beta(f_H, e_1, \tilde{u}) = \rho_l(G_\mu, \tilde{u}),\\
                \rho_\beta(f_H, e_1', \tilde{u}) = \rho_r(G_\mu, \tilde{u}),\\
                \rho_\beta(f_H, e_\ell, \tilde{v}) = \rho_l(G_\mu, \tilde{v}),\\
                \rho_\beta(f_H, e_h', \tilde{v}) = \rho_r(G_\mu, \tilde{v}).
            \end{gather*}
    \end{enumerate}
\end{enumerate}

We are now ready to prove that the problem of determining whether there exists a valid pair $(\alpha,\beta)$ is equivalent to checking whether $\psi$ is in the feasible set of $\mu$.
\fi

\iflong \begin{lemma} \fi \ifshort \begin{lemma}\fi
    There is an upward planar embedding of $G_\mu$ with the shape description $\psi$ if and only if there is a valid pair $(\alpha, \beta)$.
    \label{lemma:valid_pair}
\end{lemma}
\iflong \begin{proof}
    In the forward direction, take the upward planar embedding $\mathcal{E}_\mu$ of $G_\mu$ that has shape description $\psi$ and produce the corresponding $(\alpha, \beta)$.
    For each children node $\nu$ of $\mu$, consider the upward planar embedding $\mathcal{E}_\nu$ induced by $\mathcal{E}_\mu$ on $G_\nu$ which is a subgraph of $G_\mu$. Let $e$ be the edge of $H$ that corresponds to $\nu$, assign $\beta(e)$ to be the shape in $\mathcal{S}_\mu(e)$ that is defined by the embedding $\mathcal{E}_\nu$. For each active switch vertex $v \in A_\mu$, define $\alpha(v)$ as follows. By Theorem~\ref{th:upward-conditions}, there is precisely one large angle at $v$ in $\mathcal{E}_\mu$. If this angle falls inside an inner face of $G_{\nu^e}$ for some virtual edge $e$, $\alpha(v)$ is set to be $v_e$. Otherwise, this angle belongs to one of the faces induced by a face $f$ of $H$; let $\alpha(v)$ be the corresponding vertex $v_f$. It remains to verify that $\alpha$ and $\beta$ defined in this way form a valid pair.
    \begin{enumerate}
        \item By the condition \textbf{UP3} of Theorem~\ref{th:upward-conditions}, in $\mathcal{E}_\mu$ the angles belonging to a face sum up to two if the face is the outer face, and to $-2$ otherwise. In particular, this holds for any face that is induced by a face $f$ of $H$. The sum of angle values on $f$ is precisely the left part of the Validity Condition 1, since for each $v \in V(H)$, $\lambda_\alpha(v, v_f)$ is by construction equal to the value of this angle in $\mathcal{E}_\mu$, and for each $e \in E(H)$ incident to $v$, $\turnt_\beta(v_f, v_e)$ is precisely the sum of angle values along the boundary walk of $G_{\nu^e}$ lying on $f$, by definition of a shape description.
        \item Validity Condition 2 follows immediately from conditions \textbf{UP0}--\textbf{UP2} of Theorem~\ref{th:upward-conditions} applied to $\mathcal{E}_\mu$.
        \item Validity Condition 3 boils down to verifying that the constructed pair ($\alpha$, $\beta$) is set in accordance with the embedding $\mathcal{E}_\mu$ of $G_\mu$ that has shape description $\psi$: 3a is shown analogously to Validity Condition 1, 3b is immediate by the construction of $\alpha$, in 3c the values of $\rho$ on both sides are taken from the same embedding and thus necessarily match.
    \end{enumerate}

    In the other direction, consider a valid pair $(\alpha, \beta)$. First, fix a planar embedding $\mathcal{E}_\mu$ of $G_\mu$ based on $(\alpha, \beta)$ together with $\mathcal{E}^H_\mu$: For each children component $(G_{\nu^e}, u, v)$ fix a $uv$-external upward planar embedding that has the shape description given by $\beta(v_e)$ on the corresponding virtual edge $v_e$. In particular that fixes a planar embedding $\mathcal{E}_{\nu^e}$ of $G_{\nu^e}$, and an angle assignment $\lambda_{\nu^e}$. It remains to fix the circular order of the edges incident to the vertices of the skeleton $H$ and the walk defining the boundary of the outer face. For the former, the embedding $\mathcal{E}^H_\mu$ defines a circular order on the set of virtual edges adjacent to each vertex $v$ of $H$. On the other hand, the order of the edges incident to $v$ inside each corresponding children component $G_{\nu^e}$ is given by the fixed planar embedding of $G_{\nu^e}$. Together, this defines a circular order on all the edges of $G_\mu$ incident to $v$.
    Finally, the boundary walk of the outer face is also defined unambiguouly from the outer face in $\mathcal{E}^H_\mu$ and the embeddings of the children components. Specifically, consider a sequence of virtual edges $v_{e_1}$,\ldots, $v_{e_\ell}$ forming the boundary of the outer face in $\mathcal{E}^H_\mu$. Obtain a walk in $G_\mu$ by replacing each $v_{e_i}$ in the above sequence by the left- or right-walk from $u_i$ to $v_i$ around the outer face in the fixed embeddding $\mathcal{E}_{\nu^{e_i}}$ of the component $(G_{\nu^{e_i}}, u_i, v_i)$, where the left-walk is taken if $u_i$ appears right before $v_i$ in the clockwise boundary walk of the outer face in $\mathcal{E}_{\nu^{e_i}}$, and the right-walk is taken otherwise. Clearly, $\mathcal{E}_\mu$ defined in this way is a valid planar embedding of $G_\mu$.
    
    Next, from $(\alpha, \beta)$ we construct an angle assignment $\lambda$ for $G_\mu$ with the embedding $\mathcal{E}_\mu$ that satisfies Theorem~\ref{th:upward-conditions}, while simultaneously achieving the shape description $\psi$ for $(G_\mu, \tilde{u}, \tilde{v})$.
    In what follows we specify how $\lambda$ is defined for each pair of a vertex $v$ and an incident face $f$. We distinguish the cases where $v \in V(H)$ and $v \notin V(H)$, and also the cases where $f$ lies inside $G_{\nu^e}$ for some $e \in E(H)$, or $f$ is induced by a face of $H$.
    \begin{description}
        \item[Case 1:] If $f$ lies inside $G_{\nu^e}$ for some $e \in E(H)$, $\lambda(v, f) = \lambda_{\nu^e}(v, f)$, that is, the value of the angle assignment is copied from the corresponding angle assignment on $G_{\nu^e}$ that yields shape description $\beta(v_e)$.
        \item[Case 2:] If $f$ is induced by a face $f'$ of $H$, and $v$ is a vertex of $G_{\nu^e}$ that does not belong to $V(H)$, we put $\lambda(v, f) = \lambda_{\nu^e}(v, f_e)$, where $f_e$ is the external face in the embedding $\mathcal{E}_{\nu^e}$ of $G_{\nu^e}$.
        \item[Case 3:] If $f$ is induced by a face $f'$ of $H$ and $v$ is a vertex of $H$, we consider two cases. If $v$ is an active vertex, we put $\lambda(v, f) = 1$ if $\alpha(v) = v_{f'}$ and $\lambda(v, f) = -1$ otherwise. If $v$ is not an active vertex, $\lambda(v, f) = 0$ if the corresponding angle is flat, and $\lambda(v, f) = -1$ otherwise.
    \end{description}

    Next, we verify that $\lambda$ defined in this way satisfies the conditions of Theorem~\ref{th:upward-conditions}.
	    If $v \in V(G_{\nu^e})$, but not in $V(H)$, $\lambda(v, f)$ is either $\lambda_{\nu^e}(v, f)$ or $\lambda_{\nu^e}(v, f_e)$, and $\lambda_{\nu^e}(v, \cdot)$ satisfies \textbf{UP0} since $(\mathcal{E}_{\nu^e}, \lambda_{\nu^e})$ is an upward planar embedding of $G_{\nu^e}$.
	    Moreover, \textbf{UP1} and \textbf{UP2} are also automatically satisfied  for $v$ since the angle assignment around $v$ is identical to that of $\lambda_{\nu^e}$.
        Analogously, if $f$ is a face inside $G_{\nu^e}$, the condition \textbf{UP3} holds since the angle assignment for the angles of $f$ is identical to that of $\lambda_{\nu^e}$, which is a part of an upward planar embedding and thus satisfies \textbf{UP3}.

        It remains to verify \textbf{UP0--2} for $v \in V(H)$ and \textbf{UP3} for each face $f$ induced by a face of $H$.
	    \textbf{UP0} follows immediately from the definition of $\lambda$, see \textbf{Case 3}.
	    If $v$ is a switch vertex of $G_\mu$, $v$ is an active vertex, and the mapping $\alpha$ picks a certain value for $v$. If $\alpha(v) = v_{f'}$ for some face $f'$ of $H$, then $\lambda(v, f) = 1$ for the corresponding face $f$ of $G_\mu$, and $\lambda(v, f) = -1$ for all other faces of $G_\mu$ that are induced by the faces of $H$.  Moreover, by Validity Condition 2a,  for each virtual edge $e$ incident to $v$ in $H$, the outer angle of $v$ in $G_{\nu^e}$ is large, thus every angle of $v$ inside $G_{\nu^e}$ is small, meaning that every angle at $v$ in $G_{\mu}$ is small, except for the angle at the face $f$. If $\alpha(v) = v_{e'}$ for some edge $e'$ of $H$, by \textbf{Case 3} all the angles of form $(v, f)$, where $f$ is a face induced by $H$, are small.
        Analogously, for each virtual edge $e$ incident to $v$ in $H$, the outer angle of $v$ in $G_{\nu^e}$ is large, meaning that all angles at $v$ inside $G_{\nu^e}$ are small. Finally, for the angles at $v$ inside $G_{\nu^{e'}}$, exactly one of them is assigned to large, as the angle assignment follows that of $\lambda_{\nu^e}$, which satisfies \textbf{UP1}. Thus, \textbf{UP1} holds for $v$.

        If $v$ is in $H$ and not a switch vertex of $G_\mu$, by \textbf{Case 3} among the angles induced by faces of $H$, every flat angle is assigned $0$, and every switch angle is assigned $-1$. Consider a virtual edge $e$ incident to $v$ in $H$. If the corresponding component contains only edges of the same direction from $v$, by Validity condition 2b the outer angle of $G_{\nu^e}$ at $v$ is large, and thus each angle at $v$ inside $G_{\nu^e}$ is small. If, on the other hand, $G_{\nu^e}$ contains a flat angle at $v$, then either both flat angles are inside $G_{\nu^e}$, the outer angle of $G_{\nu^e}$ at $v$ is large, and all the other angles at $v$ in $G_{\nu^e}$ (and thus in $G_\mu$) are small. Or, the second flat angle at $v$ occurs in a face adjacent to $G_{\nu^e}$, the outer angle of $G_{\nu^e}$ is flat, and and all the other angles are small. In either case, the Validity Condition 2b together with the fact that the angles assigned in $G_{\nu^e}$ by $\lambda_{\nu^e}$ satisfy \textbf{UP0--2}, ensure the fact that the condition \textbf{UP2} at $v$ is satisfied.

        For each face $f$ of $G_\mu$ induced by a face $f'$ of $H$, Validity Condition 1 of a valid pair immediately implies that \textbf{UP3} is satisfied. Specifically, for a vertex $v$ in $H$, the angle of $v$ is set exactly to $\lambda_\alpha(v, v_f)$. For a vertex $v$ that is not in $H$, but inside a children component $G_{\nu^e}$, the value of its angle is a part of the respective left- or right-turn number of $G_{\nu^e}$, $\turnt_\beta(v_{f'}, v_e)$. Thus, the sum of angle values along $f$ is precisely encoded in the left part of Validity Condition 1.

        Finally, it remains to see that the $\tilde{u}\tilde{v}$-external upward planar embedding $(\mathcal{E}_\mu, \lambda)$ of $G_\mu$ 
        has shape description $\psi$. Analogously to verifying \textbf{UP3}, the left parts in Validity Condition 3a can be seen to encode the sum of angle values along the left and right boundary walk of $G_\mu$, thus the turn numbers of $G_\mu$ are as required by $\psi$. Validity Condition 2b explicitly ensures that the outer pole angles of $G_\mu$ are in accordance wit $\psi$. Finally, Validity Condition 3c verifies that the edges incident to the poles on the left and tight outer walk have the direction as specified by $\psi$.
   \end{proof}\fi

\iflong
\subsection{Finding Valid Pairs Using Treewidth}
Subsections~\ref{sub:embeddinggraph} and~\ref{sub:problemreformulation} together result in a combinatorial problem that we can solve in order to provide an R-node subprocedure. To provide a more careful analysis of the running time, we introduce a new measure $\zeta$ of the graph which, intuitively, bounds the maximum amount of spiraling that can occur for any face of the pertinent graph. Formally, $\zeta$ is the maximum over $n_1(f)$ and $n_{-1}(f)$ (see Theorem~\ref{th:upward-conditions}), over all faces $f$ of all possible planar embeddings of the pertinent graph $G_\mu$ of $\mu$. Recalling Proposition~\ref{pro:tdfacts}, we obtain:
\fi
\ifshort
\subparagraph{Finding Valid Pairs Using Treewidth.}
At this point, what is left to do is solve this combinatorial problem. For the runtime analysis of the algorithm we will develop, we let $\zeta$ be the maximum over $n_1(f)$ and $n_{-1}(f)$ (see Theorem~\ref{th:upward-conditions}), over all faces $f$ of all possible planar embeddings of the pertinent graph $G_\mu$ of $\mu$. Recalling that no path in $G$ can have length greater than $2^{\td(G_\mu)}$~\cite{sparsity}, we obtain:
\fi
%
%

\begin{observation}
\label{obs:zeta}
$\zeta\leq V(G_\mu)$, and moreover $\zeta\leq 2^{\td(G_\mu)}$.
\end{observation}

\newcommand{\anglet}{\texttt{angle}}
\newcommand{\shape}{\texttt{shape}}
\newcommand{\score}{\texttt{score}}
\newcommand{\assigned}{\texttt{assigned}}
\newcommand{\todot}{\texttt{todo}}
\newcommand{\leftt}{\texttt{left}}
\newcommand{\rightt}{\texttt{right}}
\newcommand{\old}{\texttt{old}}

We can now design a dynamic program that solves the task at hand. The program computes sets of records for each node of a tree-decomposition in a leaf-to-root fashion, where each record is a tuple of the form $(\anglet,\shape,\score,\leftt,\rightt)$ where $\anglet$ and $\shape$ contain snapshots of $\alpha$ and $\beta$ in the given bag, respectively; $\score$ keeps track of the sum of large and small angles for each face in the given bag; and $\leftt, \rightt$ store information about the left-and right-turn-numbers of the outer face.

\iflong \begin{lemma} \fi \ifshort \begin{lemma}\fi
\label{lem:twdynprog}
    There is an algorithm that runs in time $\zeta^{\bigoh(\tw(H))} \cdot (|V(H)|+|\mathcal{S}_\mu|)$ and either computes a valid pair, or correctly determines that no such pair exists.
\end{lemma}

\iflong \begin{proof}
As our first step, we compute the embedding graph $\signH$ of $H$ and then use the well-known $5$-approximation algorithm~\cite{BodlaenderDDFLP16} to obtain a nice tree-decomposition $(T,\chi)$ of the embedding graph of width at most $5\tw(\signH)$, which is in $\bigoh(\tw(H))$ by Lemma~\ref{thm:embeddingtw}. 
We assume that $(T,\chi)$ is rooted at the node $r$ and that $\chi(r)=\emptyset$.
Without loss of generality and for ease of presentation, we will alter this decomposition by adding the face-vertex $v_H$ of the outer face to every bag (other than $\chi(r)$) and apply the standard conversion to nice tree-decompositions in order to obtain a nice tree-decomposition with the property that $v_H$ is the first vertex introduced on each leaf and the last vertex forgotten just below the root. Let the width of the resulting tree-decomposition $(T,\chi)$ be $k$. To complete the proof, we provide a dynamic programming algorithm $\mathcal{A}$ that processes $T$ in a leaf-to-root fashion.

We begin by formalizing the records used by $\mathcal{A}$. For each node $t\in V(T)$, let $\Vv_t, \Ve_t, \Vf_t$ denote the intersection of $\chi(t)$ with $\Vv, \Ve, \Vf$, respectively. At each node $t\in V(T)$, we will store a record set $R_t$ which will be a set of tuples of the form $(\anglet,\shape,\score,\leftt,\rightt)$ where:
\begin{itemize}
\item $\anglet$ maps each switch vertex in $\Vv_t\cap A_\mu$ to an element of $(\Ve_t\cup \Vf_t\cup \{\assigned,\todot\})$,
\item $\shape$ maps each edge-vertex $w\in \Ve_t$ to a shape description in a feasible set in $\mathcal{S}_\mu(w)$, 
\item $\score$ maps each face-vertex in $\Vf_t$ to an integer $z$ such that $-\zeta\leq z\leq \zeta$, and
\item $\leftt,\rightt$ are integers whose absolute value is at most $\zeta$.
\end{itemize}
    
Intuitively, for each node $t\in V(T)$ we will use $R_t$ to store all the relevant combinatorial information about the behavior of a valid pair $(\alpha,\beta)$ of $G$ w.r.t.\ the target shape description $\psi$ in the subgraph induced on $\chi(T_t)$; let us call this subgraph $H_t$. To formalize the semantics of our records, we will first introduce a suitable projection of pairs to $H_t$: a $t$-\emph{pair} is the restriction of pairs to $\chi(T_t)$, i.e., a pair of the form $(\alpha_t,\beta_t)$ where the domain (but not the range) of both mappings $\alpha_t$ and $\beta_t$ is restricted to $\chi(T_t)=\past(t)\cup \chi(t)$. A $t$-pair is \emph{valid} if Validity Conditions 1, 2, 3b and 3c are satisfied for all vertices in $\past(t)$---in particular, (1) each face-vertex in $\past(t)$ achieves the correct sum, (2) each true vertex in $\past(t)$ satisfies the condition on its pole angles, and (3b,c) the behavior of the poles and the edges on the outer face incident to the poles in $\past(t)$ match $\psi$. Note that Validity Condition 3a is not included in this definition, since it is a global condition that we will keep track of separately and check at the end.

Our records for each node $t$ will capture information about valid $t$-pairs as follows: for each tuple $(\anglet,\shape,\score,\leftt,\rightt)$ of the form described above, $(\anglet,\shape,\score,$ $\leftt,\rightt)\in R_t$ if and only if there exists a valid $t$-pair $(\alpha_t, \beta_t)$ with the following properties:

\begin{itemize}
\item For each switch vertex $v\in \Vv_t\cap A_\mu$ such that $\alpha_t(v)=w$, 
\begin{itemize}
\item $\anglet(v)=w$ if and only if $w\in \chi(t)$,
\item $\anglet(v)=\assigned$ if and only if $w\in \past(t)$, and
\item $\anglet(v)=\todot$ otherwise, i.e., if and only if $w\in \fut(t)$.
\end{itemize}
\item For each $v_e\in \Ve_t$, $\shape(v_e)=\beta_t(v_e)$.
\item For each $v_f\in \Vf_t$ representing a face $f$ of $H$, $\score(v_f)$ is the intermediate sum for the term in Validity Condition 1; formally, $\score(v)=\sum_{v_e\in \Ve\cap N_{\signH}(v_f)\cap \chi(T_t)}\turnt_\beta(v_f, v_e) + \sum_{v\in \Vv\cap N_{\signH}(v_f)\cap \chi(T_t)}\lambda_\alpha(v, v_f)$.
\item $\rightt$ and $\leftt$ contain the intermediate sums of Validity Condition 3a, formally:
            \begin{gather*}
            \sum_{v_e\in L_e\cap \chi(T_t)}\turnt_\beta(f_H, v_e) + \sum_{v_v\in L_v\cap \chi(T_t)}\lambda_\alpha(v_v, f_H) = \leftt,\\
            \sum_{v_e\in R_e\cap \chi(T_t)}\turnt_\beta(f_H, v_e) + \sum_{v_v\in R_v\cap \chi(T_t)}\lambda_\alpha(v_v, f_H) = \rightt.
            \end{gather*}
\end{itemize}

Observe that since $H_r=H$ and $\chi(r)=\emptyset$, the set of valid $r$-pairs is precisely the set of pairs for the instance satisfying all Validity Conditions except for 3a. Hence, in view of Lemma~\ref{lemma:valid_pair}, in order to determine whether there exists a valid pair it suffices to verify whether there is a valid $r$-pair satisfying Validity Condition 3a---i.e., whether there $R_r$ contains a record such that $\leftt=\tau_l(G_\mu, \tilde{u}, \tilde{v})$ and $\rightt=\tau_r(G_\mu, \tilde{u}, \tilde{v})$.

Having defined the syntax and semantics of the records, we proceed to a description of the dynamic programming steps themselves---in particular, we describe how one can compute the record set at each of the four kinds of nodes in a nice tree-decomposition.

\smallskip
\noindent \textbf{Leaf node.}\quad
If $t\neq r$ is a leaf in $T$ with $\chi(t)=\{v_H\}$, we recall that each leaf contains a single vertex $v_H$ and we simply set $R_t$ to $\{(\emptyset, \emptyset, v_H\mapsto 0,0,0\}$.

\smallskip
\noindent \textbf{Introduce node.}\quad
If $t$ is an introduce node with child $t_\old$ and $\chi(t)\setminus \chi(t_\old)=\{v\}$, we compute $R_t$ as follows. 

If $v\in \Vv$, then for each tuple $(\anglet,\shape,\score,\leftt,\rightt)\in R_{t_\old}$ and each $w\in (N(v)\cap \chi(t))\cup \{\todot\}$, we construct the tuple $(\anglet\cup \{v\mapsto w\},\shape,\score',\leftt',\rightt')$ and add it to $R_t$ where 
\begin{itemize}
\item $\score'$ is obtained from $\score$ by updating the values for each face $v_f\in \chi(t)$, i.e., by adding the value of $\lambda_\alpha(v,v_f)$ to each $\score(v_f)$. Notably, if $v\in A_\mu$ then every face $f\in \chi(t)$ such that $\anglet(v)\neq f$ has its score decreased by $1$ while the face $\anglet(v)$ has its score increased by $1$ (the latter only applies if $\anglet(v)\neq \todot$). On the other hand, if $v\not \in A_\mu$ then every face in $\chi(t)$ has its score decreased by $1$ for every small angle it has at $v$ (which can be verified directly from $\signH$, regardless of the records). 
\item \leftt' and \rightt' are obtained from \leftt and \rightt, respectively, by updating the values of the intermediate sums on the left side in Validicty Condition 3a based on the impact on the newly introduced vertex $v$. This update is only carried out if $v\in L_v$ or $R_v$, respectively, and is done by simply adding $\lambda_\alpha(v,f_H)$ to \score.
\end{itemize}

If $v\in \Ve$, then for each tuple $(\anglet,\shape,\score,\leftt,\rightt)\in R_{t_\old}$ and each $\iota\in \mathcal{S}_\mu$, we construct the tuple $(\anglet',\shape\cup\{v\mapsto\iota\},\score',\leftt',\rightt')$ where \score', \leftt' and \rightt' are updated based on the value of $\turnt_\beta$ of the newly introduced shape description $\iota$ in the same way as in the previous case. Moreover, for each $v_v\in \Vv_t\cap N(v)$ such that $\anglet(v_v)=\todot$, we perform an additional branching to determine whether $\anglet'(v_v)=\todot$ or $\anglet'(v_v)=v$; we set $\anglet'=\anglet$ for all vertices not satisfying the condition of this sentence. In each case, we verify whether Validity Condition 2 holds for the pair $v_v,v_e$; if not, we discard the non-compliant choice of \anglet', and otherwise we add the tuple $(\anglet',\shape\cup\{v\mapsto\iota\},\score',\leftt',\rightt')$ to $R_{t_\old}$.

If $v\in \Vf$, then for each tuple $(\anglet,\shape,\score,\leftt,\rightt)\in R_{t_\old}$ and each $v_v\in \Vv_t\cap N(v)$ such that $\anglet(v_v)=\todot$, we branch to determine whether $\anglet'(v_v)=\todot$ or $\anglet'(v_v)=v$. In each branch, we construct the tuple $(\anglet',\shape,\score',\leftt,\rightt)$ where $\score'$ matches score on all vertices other than $v$, and maps $v$ to the intermediate sum of $v$ obtained in an analogous way as in the previous two cases (in particular, we compute the sum $\turnt_\beta(v,v_e)$ over all $v_e\in \Ve_t$ plus the sum $\lambda_\alpha(v_v,v)$ over all $v_V\in \Vv_t$ based on the information in \anglet' and \shape).

\smallskip
\noindent \textbf{Join node.}\quad
If $t$ is a join node with children $t_1$ and $t_2$, we compute $R_t$ as follows. For each $\rho_1=(\anglet_1,\shape_1,\score_1,\leftt_1,\rightt_1)\in R_{t_1}$ and each $\rho_2=(\anglet_2,\shape_2,$ $\score_2,\leftt_2,\rightt_2)\in R_{t_2}$, we perform two consistency checks:
\begin{itemize}
\item Check the consistency of $\anglet_1$ and $\anglet_2$ for each $v_v\in \Vv_t\cap A_\mu$: if $\anglet_1(v_v)\in \chi(t)$ then $\anglet_1(v_v)= \anglet_2(v_v)$, if $\anglet_1(v_v)=\assigned$ then $\anglet_2(v_v)=\todot$, and if $\anglet_1(v_v)=\todot$ then $\anglet_2(v_v)\in \{\todot, \assigned\}$. We perform the analogous checks also for $\anglet_2(v_v)$. If any of these checks fail, we discard the current pair of $(\rho_1, \rho_2)$.
\item Check the consistency of $\shape_1$ and $\shape_2$ for each $v_e\in \Ve_t$: it must hold that $\shape_1(v_e)=\shape_2(v_e)$. If any of these checks fail, we discard the current pair of $(\rho_1, \rho_2)$.
\end{itemize}

If the above checks succeeded, we construct a new entry $(\anglet,\shape,\score,\leftt,\rightt)$ and add it to $R_t$. This is carried out as follows:
\begin{enumerate}
\item For each $v_v\in \Vv_t\cap A_\mu$ such that $\anglet_1(v_v)=\anglet_2(v_v)$, set $\anglet(v_v)=\anglet_1(v_v)$. For each $v_v\in \Vv_t\cap A_\mu$ such that $\anglet_1(v_v)\neq \anglet_2(v_v)$, based on the previous checks we know that one was set to \assigned\ while the other to \todot; we hence set $\anglet(v_v)=\assigned$.
\item For each $v_e\in \Ve_t$, set $\shape(v_e)=\shape_1(v_e)$.
\item For each $v_f\in \Vf_t$, set $\score(v_f)=\score_1(v_f)+\score_2(v_f)-\emph{corr}(v_f)$ where $\emph{corr}(v_f)$ nullifies the double-counting of vertices in $\Vv_t\cup \Ve_t $ and is computed as $\sum_{v_e\in \Ve_t\cap N_{\signH}(v_f)}\turnt_\beta(v_f, v_e) + \sum_{v\in \Vv_t\cap N_{\signH}(v_f)}\lambda_\alpha(v, v_f)$ (where $\turnt_\beta$ and $\lambda_\alpha$ are obtained from $\shape$ and $\anglet$ in the same way as for the introduce nodes).
\item Similarly as for $\score$, set $\leftt=\leftt_1+\leftt_2-\emph{leftcorr}$ and $\rightt=\rightt_1+\rightt_2-\emph{rightcorr}$. Here, $\emph{leftcorr}$ nullifies the double-counting of vertices in $\chi(t)\cap (L_v\cup L_e)$ and is equal to $\sum_{v_e\in L_e\cap \chi(t)}\turnt_\beta(f_H, v_e) + \sum_{v_v\in L_v\cap \chi(t)}\lambda_\alpha(v_v, f_H)$. The value of $\emph{rightcorr}$ is defined and computed analogously. 
\end{enumerate}

\smallskip
\noindent \textbf{Forget node.}\quad
If $t$ is a forget node with child $t_\old$ and $\chi(t_\old)\setminus \chi(t)=\{v\}$, we compute $R_t$ as follows. 

If $v\in \Vv$, for each tuple $(\anglet,\shape,\score,\leftt,\rightt)\in R_{t_\old}$ we construct a new tuple $(\anglet',\shape,\score,\leftt,\rightt)$ where $\anglet'$ is obtained from $\anglet$ by omitting $v$ from the domain. If $v$ is a pole of $H$, we also check that Validity Conditions 3b and 3c are satisfied for each neighbor of $v$, and if not we discard the computed tuple. Otherwise, we add the tuple to $R_t$.

If $v\in \Ve$, for each tuple $(\anglet',\shape,\score,\leftt,\rightt)\in R_{t_\old}$ we construct a new tuple $(\anglet,\shape',\score,\leftt,\rightt)$ where $\shape'$ is obtained from $\shape$ by omitting $v$ from the domain and $\anglet'$ is obtained from $\anglet$ by mapping every vertex that \anglet mapped to $v$ to $\assigned$ instead. If $v$ is incident to a pole of $H$, we also check that Validity Condition 3c is satisfied. If it is, we add the tuple to $R_t$.

If $v\in \Vf$, for each tuple $(\anglet,\shape,\score,\leftt,\rightt)\in R_{t_\old}$ we construct a new tuple $(\anglet',\shape,\score',\leftt,\rightt)$ where $\score'$ is obtained from $\score$ by omitting $v$ from the domain and $\anglet'$ is obtained from $\anglet$ by mapping every vertex that \anglet mapped to $v$ to $\assigned$ instead. If Validity Condition 1 is satisfied (i.e., if $\score(v)=-2$ or, in the case of $v=f_h$, $+2$), we add the resulting tuple to $R_t$, and discard it otherwise.

\smallskip
Using the above procedures, we compute the record set of all nodes up to $R_r$, for which we check if there is a record such that $\leftt=\tau_l(G_\mu, \tilde{u}, \tilde{v})$ and $\rightt=\tau_r(G_\mu, \tilde{u}, \tilde{v})$. If yes, the algorithm outputs ``\texttt{Yes}'', and otherwise it outputs ``\texttt{No}''. This concludes the description of the algorithm. 

Regarding the running time, the size of each record set is upper-bounded by $\zeta^{\bigoh(tw)}$, and the procedures described in the computation of each type of node can be completed in time at most $\zeta^{\bigoh(tw)}$ per node. The resulting runtime can hence be upper-bounded by $\zeta^{\bigoh(tw)}\cdot V(H)$. Correctness follows from the correctness of the computation of the record set for each node $t\in V(T)$, which can be verified directly from the definitions and the semantics of $R_t$.
\end{proof}\fi

\iflong
As an consequence of Lemma~\ref{lem:twdynprog} together with Observation~\ref{obs:zeta}, Lemma~\ref{lemma:valid_pair} and Theorem~\ref{thm:embeddingtw}, we obtain an R-node subprocedure that runs in \XP-time parameterized by treewidth and fixed-parameter time parameterized by treedepth. Combining this with the Interface Lemma~\ref{lem:R_node_general}, we conclude:
\fi

\ifshort
We now have an R-node subprocedure that runs in \XP-time parameterized by treewidth and fixed-parameter time parameterized by treedepth. By invoking Lemma~\ref{lem:R_node_general}, we conclude:
\fi
\begin{theorem}
\label{thm:twtd}
It is possible to solve \UP\ in time $n^{\bigoh(\tw(G))}$ and time $2^{\bigoh(\td(G)^2)} \cdot n^2$, where $n$ is the number of vertices of the input digraph $G$.
\end{theorem}

\section{Concluding Remarks}\label{sec:conclusions}

The presented results show that the combination of SPQR-trees with parameterized techniques is a promising algorithmic tool for geometric graph problems. Indeed, for the case of upward planarity, our framework allows us to reduce the general problem to a similar one on 3-connected graphs, at which point it is possible to use parameter-specific approaches such as dynamic programming or flow networks to obtain a solution. We believe not only that the framework developed here can help obtain other algorithms for \UP, but that the idea behind the framework can be adapted to solve other problems of interest as well---a candidate problem in this regard would be constrained level planarity testing~\cite{BrucknerR17}.

All algorithms and arguments given within this paper are constructive and can be extended to output an upward planar drawing for each yes-instance of \UP. An open problem is whether \UP\ is \W[1]-hard when parameterized by treewidth, or fixed-parameter tractable. 
\iflong
It seems that improving the current dynamic programming algorithm to a fixed-parameter one would require non-trivial and non-obvious insights, but the problem has so far also proven resilient to our attempts at obtaining a \W[1]-hardness reduction.
\fi
Another question is whether the fixed-parameter tractability of \UP\ parameterized by the number of sources can be lifted to parameterizing by the maximum turn number of a face in the final drawing.

	\bibliography{biblio}

\end{document}